%% file: spiderstate.tex
\documentclass[aps,pra,nofootinbib,superscriptaddress,longbibliography,eprint]{revtex4-2}

\usepackage{bookmark}
\usepackage{xcolor}
\definecolor{citegreen}{RGB}{34,110,60}
\hypersetup{colorlinks, citecolor=citegreen}
\usepackage{csquotes}
\usepackage{multirow}
\usepackage{booktabs}
\usepackage{makecell}
\usepackage{ragged2e}
\usepackage{subcaption}

\usepackage[normalem]{ulem}

\usepackage{amsmath}
\usepackage{stmaryrd}
\SetSymbolFont{stmry}{bold}{U}{stmry}{m}{n}

\usepackage{amsthm,thmtools}
\theoremstyle{plain}
\newtheorem{theorem}{Theorem}[section]

\newtheorem{lemma}[theorem]{Lemma}

\newtheorem{proposition}[theorem]{Proposition}

\theoremstyle{definition}
\newtheorem{definition}[theorem]{Definition}
\newtheorem{example}[theorem]{Example}

\theoremstyle{remark}
\newtheorem{remark}[theorem]{Remark}
\usepackage{xspace}

\usepackage{algorithm}
\usepackage{algpseudocodex}

\usepackage[noabbrev,nameinlink]{cleveref}

\makeatletter
\renewcommand\paragraph{\@startsection{paragraph}{4}{\parindent}{\z@}{-1em}{\normalfont\normalsize\bfseries}}
\makeatother

\usepackage{geometry}
\usepackage{todonotes}
\usepackage{tcolorbox}
\tcbuselibrary{breakable,skins}
\makeatletter
\patchcmd{\tcb@handle@nobreak}
  {\tcb@h@page+3.5\baselineskip+\parskip+\lineskip}
  {\tcb@h@page+3.5\baselineskip+\parskip+\lineskip+.5\textheight}
  {}{\PackageWarning{spiderstate}{Could not patch tcolorbox nobreak overfill}}
\patchcmd{\tcb@handle@nobreak}
  {\vspace{\tcb@h@page}{\parskip\z@\nointerlineskip\noindent}\vspace{-\tcb@h@page}\par}
  {\vskip\tcb@h@page\nointerlineskip\hbox{}\vskip-\tcb@h@page}
  {}{\PackageWarning{spiderstate}{Could not patch tcolorbox nobreak paragraph}}
\makeatother

\AddToHook{cmd/appendix/before}{\crefalias{section}{appendix}}

\usepackage{macros/tikzfig}
\input{macros/floquet.tikzstyles}
\input{macros/floquet.tikzdefs}

\newcommand{\code}[1]{\left\llbracket#1\right\rrbracket}

\newcommand{\pMem}{p_{\mathrm{mem}}}
\newcommand{\pTwo}{p_{\mathrm{2}}}
\newcommand{\pOne}{p_{\mathrm{1}}}
\newcommand{\pSPAM}{p_{\mathrm{SPAM}}}
\newcommand{\CAT}{\textsf{CAT}\xspace}

\usepackage{ulem}

\newcounter{boldicomment}

\newcounter{andreycomment}

\newcounter{johncomment}

\tcbset{combox/.style={
    enhanced, breakable, shrink break goal=4pt, parbox=false, coltext=black,
    fonttitle=\bfseries, boxrule=0.6pt, arc=2pt,
    left=4pt, right=4pt, top=4pt, bottom=4pt,
}}
\newtcolorbox{comrichie}{combox, colback=blue!4,   colframe=blue!60!black,   title={Richie:}}
\newtcolorbox{comsarah}{combox,  colback=orange!6, colframe=orange!80!black, title={Sarah:}}

\begin{document}

\author{Boldizsár Poór}
\affiliation{University of Oxford, Oxford, United Kingdom}

\author{Andrey Boris Khesin}
\affiliation{University of Oxford, Oxford, United Kingdom}

\author{Sarah Meng Li}
\affiliation{University of Amsterdam, Amsterdam, Netherlands}
\affiliation{QuSoft, Amsterdam, Netherlands}

\author{Benjamin Rodatz}
\affiliation{University of Oxford, Oxford, United Kingdom}

\author{John van de Wetering}
\affiliation{University of Amsterdam, Amsterdam, Netherlands}
\affiliation{QuSoft, Amsterdam, Netherlands}

\author{Richie Yeung}
\affiliation{University of Oxford, Oxford, United Kingdom}

\title{SpiderCSS: Scalable Fault-Tolerant CSS State Preparation}

\begin{abstract}
  Fault-tolerant preparation of logical states is a critical primitive for the realization of large-scale fault-tolerant quantum computing.
  This paper introduces SpiderCSS, a scalable compilation pipeline that generates highly optimized, fault-tolerant state preparation circuits for arbitrary Calderbank-Shor-Steane (CSS) codes.
  Starting from an idealized specification of a CSS state, we use the ZX-calculus and fault-equivalent rewrites to construct a diagram that prepares the given state while preserving the fault tolerance of the initial ideal ZX-diagram.
  From this diagram, \protect{SpiderCSS} constructs a fault-tolerant CSS-state preparation circuit using a modular construction that is CNOT-minimal both within each component and also in the way these components are connected together.
  By only using fault-equivalent rewrites, the final circuit is provably fault-tolerant without needing expensive verification.
  Monte Carlo simulations across a variety of CSS codes up to distance $15$ show that SpiderCSS scales efficiently and constructs circuits with substantially reduced depth and qubit footprint compared to existing scalable and heuristic methods, while simultaneously achieving an average decrease of $33.7\%$ in logical error rate and an increase $9.5\%$ in acceptance rate.
\end{abstract}

\keywords{QEC, FTQC, CSS, fault tolerance, state preparation, ZX-calculus, fault equivalence}

\maketitle

\section{Introduction}

Fault-tolerant quantum computation requires reliable preparation of encoded states~\cite{gottesman2024surviving}. To ensure fault tolerance we require controlling fault propagation, since a single circuit fault can affect several output qubits. The challenge is to achieve fault tolerance with low quantum resource overhead while keeping circuit construction tractable as code size and distance increase. In this paper, we will focus on preparing states in CSS codes.

\emph{Calderbank--Shor--Steane (CSS) codes}~\cite{CalderbankShor1996,Steane1996} are stabilizer codes with purely $X$- and $Z$-type generators, whose logical states $\ket{\overline{0}}$ and $\ket{\overline{+}}$ are examples of \emph{CSS states}. These states are used as logical inputs to a computation and as ancillas for syndrome extraction in Steane-style error correction~\cite{Steane1997ActiveStabilization,Steane1998SpaceTime,Steane2003Overhead,Steane2002Filtering,SteaneIbinson2005}.

The current best methods for CSS state preparation either use hand-optimized circuits, SAT solvers, or heuristic modular constructions. Hand-optimized circuits can achieve low overhead but rely on code-specific human insight that is difficult to scale systematically~\cite{Goto2016,ChaoReichardt2018,Butt2024CodeSwitching}. Automated synthesis using satisfiability (SAT) solvers or Integer Linear Programming (ILP) can produce highly-optimized circuits for small codes, but the search becomes expensive as the code size and required level of fault tolerance increase~\cite{Peham2025StatePreparation,Schmid2025Deterministic,criger2026automatedflagbasedfaulttolerantstate}.

The current best modular construction, Flag at Origin (FaO)~\cite{forlivesi2025flag}, improves the scalability of circuit construction through reusable flag gadgets and its resource overhead depends on how these components are assembled and scheduled. FaO works by first constructing a library of fault-tolerant flagged gadgets which makes it possible to fault-tolerantly apply a sequence of control (or target) operations of a CNOT gate on a single qubit.
Utilizing the fact that any CSS state can be constructed in a way that only control or target operations are applied to each qubit, FaO composes these flagged gadgets to obtain a fully fault-tolerant CSS-state preparation circuit.
While this method scales much better than methods relying on SAT solvers or ILP, it is still computationally limited by the (offline) generation of the flagged gadgets and the resulting circuits can have a large depth due to the depth of the flagged gadgets.

In this paper, we present \emph{SpiderCSS}, a new modular construction that significantly improves upon the state-of-the-art FaO in terms of scalability, resource usage, depth, acceptance rate, and logical error rate. First, each component of our pipeline runs in polynomial time so that our method is scalable beyond what was possible before. Second, our method requires fewer entangling gates and achieves lower depth. This is achieved by using a better approach to construct fault-tolerant \CAT states based on the SpiderCat~\cite{spidercat2026} and Fault Tolerance by Construction~\cite{rodatz2025fault} frameworks, then finding novel ways to route the qubit lines to minimize the cost of connecting the states together to form the target CSS state. Third, our method has an improved acceptance rate (meaning the average number of shots needed to get one that passes all the checks is lower) and a lower logical error probability (meaning the chance we accept a run that actually contains a non-trivial logical error is lower). We achieve this by minimizing the number of flags in the \CAT states in our construction and by letting our qubit routing algorithm try to minimize qubit lifetime and maximize qubit reuse.

In total, we benchmarked our pipeline on codes ranging from $7$ to $144$ qubits and from distance $3$ to distance $15$. The results are shown in \Cref{fig:comparison}. In all directly comparable cases, we match or improve on the CNOT count, with reductions of up to $53.8\%$. When prioritizing depth, our circuits are consistently shallower than those of FaO, with nearly $10\times$ reduction in circuit depth in the best case.
Estimated logical error rates are lower for most codes with reported results, with the largest reduction being approximately $90\%$, while acceptance rates are higher across all benchmarks.
\begin{figure}[htb]
    \centering
    \begin{subfigure}[t]{0.49\linewidth}
    \includegraphics[width=\linewidth]{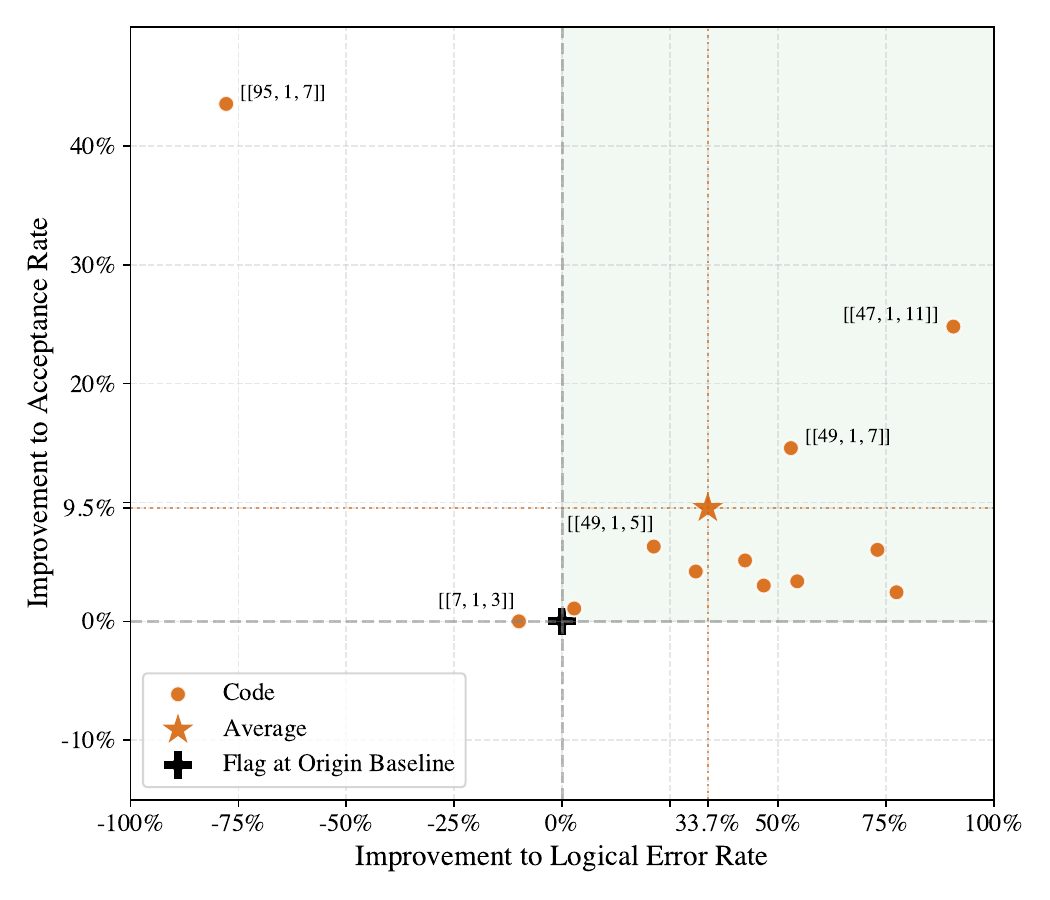}
    \caption{Compare the logical error rate and acceptance rate of SpiderCSS and FaO.}
    \label{fig:noise-improvements}
    \end{subfigure}
    \hfill
    \begin{subfigure}[t]{0.49\linewidth}
    \includegraphics[width=\linewidth]{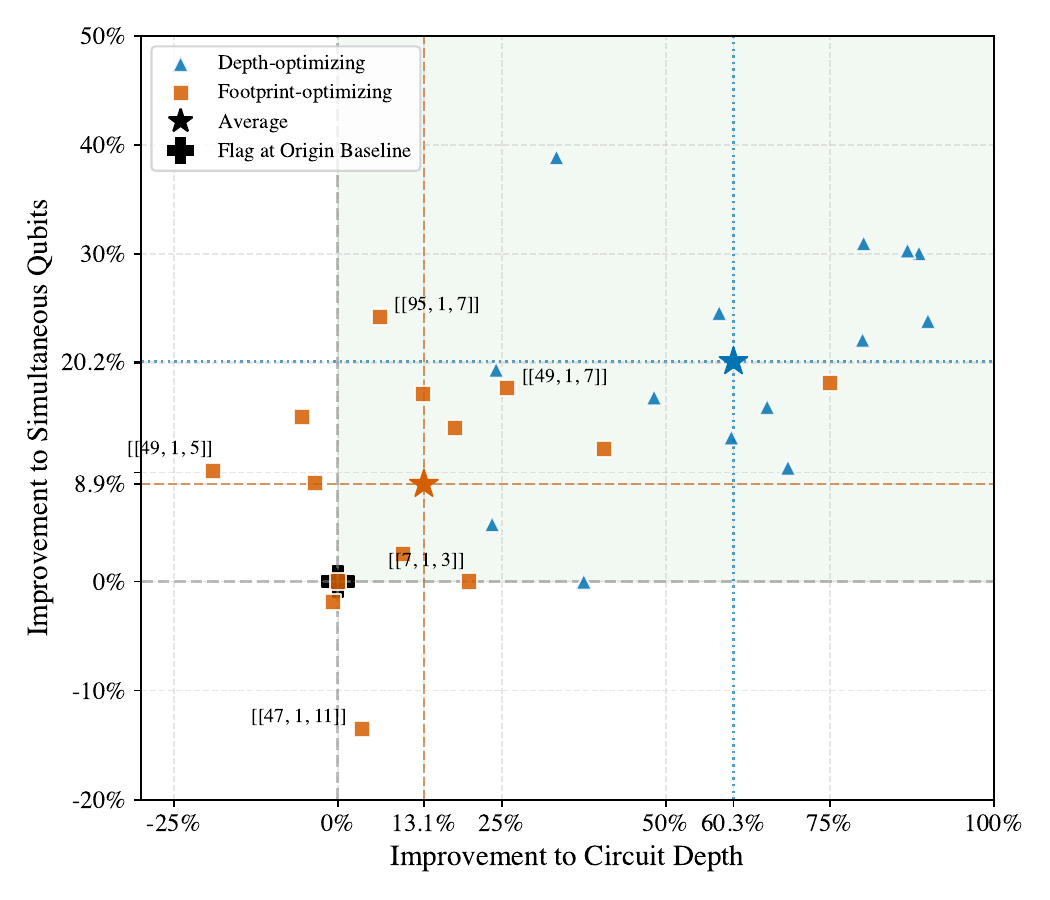}
    \caption{Circuit depth and peak simultaneous qubit count. Triangles compare SpiderCSS and FaO when both use depth-optimizing qubit reuse, while squares compare them when both use footprint-optimizing reuse.}
    \label{fig:size-improvements}
    \end{subfigure}
    \caption{Percentage improvements of SpiderCSS over Flag at Origin (FaO) across the benchmarked CSS states. Each point represents a CSS state, and the cross at the origin is the FaO baseline. Positive values mean SpiderCSS outperforms FaO (lower logical error rate, depth or qubit count; higher acceptance rate), so points in the shaded top-right quadrant improve on both metrics. Stars mark the mean improvements, with their values indicated on the axes.}
    \label{fig:comparison}
\end{figure}

\subsection{The SpiderCSS Pipeline}
\label{subsec:pipeline}
Our pipeline is inspired by Flag at Origin~\cite{forlivesi2025flag}: to construct an arbitrary CSS state, it suffices to fault-tolerantly prepare \CAT states $\left(\ket{0}^{\otimes n}+\ket{1}^{\otimes n}\right)/\sqrt{2}$ and connect them together in a particular configuration as efficiently as possible. We augment their results by using the ZX-calculus~\cite{coeckeInteractingQuantumObservables2008} to represent the intermediate steps. This allows us to produce a circuit in a less constrained way, and use the framework of Fault Tolerance by Construction~\cite{rodatz2025fault} to rewrite our diagrams in a way that provably preserves fault tolerance.

The SpiderCSS pipeline takes as input a code distance $d$ and a description of a logically encoded state in terms of the stabilizers of the code and the stabilizers of the logically encoded state. Then, it proceeds as follows.

\begin{enumerate}
    \item Begin by representing the target state as a ZX-diagram in the normal form of~\cite{kissinger2022phasefreezxdiagramscss}. Then, rewrite the diagram until we have a bipartite graph of $Z$- and $X$-spiders. This diagram provides an idealized description of the state we aim to prepare fault-tolerantly.
    
    \begin{equation}
    \scalebox{.8}{\input{figures/representation.tikz}}
    \label{eq:representation}
    \end{equation}
    
    \item To obtain a circuit that can be executed on a quantum device, we transform \eqref{eq:representation} into a circuit-like ZX-diagram as follows:

    \begin{enumerate}
        \item Unidealize the internal edges using the fault-equivalent rewrite rules.
        \begin{equation}
            \scalebox{.8}{\input{figures/rewrites.tikz}}
            \label{eq:rewrites}
        \end{equation}
        \item Decompose each $Z$- and $X$-spider in \eqref{eq:rewrites}, which correspond to \CAT states in the $Z$- and $X$-basis, into a network of 3-ary $Z$- and $X$-spiders using a modified procedure from the SpiderCat framework~\cite{spidercat2026}.
        \begin{equation}
            \scalebox{.8}{\input{figures/decompose.tikz}}
            \label{eq:decompose}
        \end{equation}
    \end{enumerate}
    \item Generate a layout for \eqref{eq:decompose} that assigns each spider to a qubit and a timestep.
    \begin{equation}
    \scalebox{.5}{\input{figures/layout.tikz}}
    \label{eq:layout}
    \end{equation}
    \item The diagram together with the layout in \eqref{eq:layout} is interpreted as a circuit containing ancillas, CNOT gates, and measurements.
    \[  
        \large
        \hspace*{-1.75cm}\scalebox{.45}{\input{figures/sarah-figure.tikz}}
    \]
\end{enumerate}

Each of the steps 1--3 can be performed in many different ways, which result in vastly different performance characteristics. Step 1 corresponds to performing Gaussian elimination on a non-square Boolean matrix where we try to maximize the sparsity of the matrix. Step 2 corresponds to constructing 3-regular graphs satisfying certain expander-like properties. Step 3 is a scheduling problem where we try to find a minimum diameter spanning forest covering the ZX-diagram, while minimizing several other relevant properties.

The ZX-calculus brings these distinct optimization problems together within a single framework of fault-equivalent rewrites. This unified perspective also allows us to make additional observations that would otherwise be highly non-trivial. Building on these observations, we develop a new compilation pipeline for constructing highly-optimized, resource-efficient circuits that outperform other polynomial-time methods.

Our paper is organized as follows. In \Cref{sec:preliminaries}, we review CSS states, the ZX-calculus, fault-equivalent rewrites, and fault-tolerant \CAT-state constructions. In \Cref{sec:the-theory}, we develop the ZX transformations and circuit extraction procedures underlying our construction. \Cref{sec:the-pipeline} describes the SpiderCSS compilation pipeline and its optimization strategies. In \Cref{sec:simulation}, we benchmark the resulting circuits against FaO in terms of resource overhead, logical error rates, and acceptance rates. Further details on safe hook errors are provided in \Cref{apdx:safe-hook-errors}.

\section{Preliminaries}
\label{sec:preliminaries}
\subsection{Fault-Tolerant CSS State Preparation}
\label{sec:CSS-states}

This work focuses on the fault-tolerant preparation of CSS states.
We assume familiarity with the stabilizer formalism and the Pauli group. We refer the reader to~\cite{gottesman2024surviving} for a comprehensive introduction.
\begin{definition}[CSS States]
  \label{def:css-code-state}
  A \emph{Calderbank--Shor--Steane (CSS) state} is a stabilizer state whose stabilizers can be generated by a set of purely $X$-type and purely $Z$-type Pauli strings: tensor products containing only $I$ and $X$, or only $I$ and $Z$, respectively, apart from an overall sign.
\end{definition}
CSS states are precisely those states that can be built from $\ket{0}$ and $\ket{+}$ state preparations, single-qubit Pauli $X$ and $Z$ rotations, CNOT gates, and $Z$- and $X$-basis measurements.
Importantly, CSS states encoded in CSS quantum error correction codes remain CSS states (which is not the case for CSS states encoded in non-CSS stabilizer codes).
As such, this class of states includes, for example, logical $\ket{0}^k$ for any CSS code, a key component for Steane-style syndrome extraction.

We will be interested in the behaviour of CSS states under noise, so we define noisy circuits.
We consider Clifford circuits under Pauli noise, where faults can occur at the time steps in between the gates:
\begin{definition}[Faults]
  Let $C$ be a Clifford circuit and let $\mathcal{L}$ be the spacetime locations of all the time steps in between the gates.
  Then a \emph{fault} $F$ on $C$ assigns each qubit at a timestep an element in the Pauli group $\mathcal{P} = \{I, X, Y, Z\}$.
  Formally, we can express faults as elements in $\mathcal{P}^{|\mathcal{L}|}$.
  We write $C^F$ for the circuit $C$ where the Paulis indicated by $F$ have been added to the corresponding spacetime locations.
  We say a fault $F$ is \emph{trivial} if $C^F = C$ and \emph{detectable} if $C^F = 0$.
\end{definition}

We need to specify not just what the possible faults in a circuit are, but also how likely each fault is.
\begin{definition}[Noise Models and Circuit-Level Noise]
  Let $C$ be a Clifford circuit with spacetime locations $\mathcal{L}$.
  Then a \emph{noise model} on $C$ is a function $\mathcal{N} : \mathcal{P}^{|\mathcal{L}|} \to \mathbb{N} \cup \{\infty\}$ that assigns each fault a \emph{weight}, where we require $\mathcal{N}(I) = 0$ (the trivial fault has weight 0).
  \emph{Circuit-level noise} is the noise model that assigns each idling fault, gate fault, preparation fault and measurement fault weight 1.
  The weight of all other faults is the minimum number of the previously mentioned basic faults required to generate them.
\end{definition}

Now, we can define the main problem we are interested in:
\begin{definition}[Fault-Tolerant State Preparation]
  \label{def:ft-state-prep}
  Let $C$ be a Clifford circuit that prepares a state. Let $w \in \mathbb{N} \cup \{\infty\}$. We say $C$ is \emph{$w$-fault-tolerant} under some noise model $\mathcal{N}$, if, for any undetectable fault $F$ on $C$ with $\mathcal{N}(F) < w$, there exists a fault $F'$ that only acts on the outputs of $C$ such that $C^F = C^{F'}$ and $\mathcal{N}(F') \leq \mathcal{N}(F)$. We refer to $\infty$-fault-tolerant state preparations as \emph{fault-tolerant state preparation}s.
\end{definition}

\begin{remark}
  \label{rem:ft-indexing}
  Here, $w$ should be interpreted as a cutoff on fault weight, in the same spirit as code distance: a code has distance at least $d$ when every Pauli of weight less than $d$ either acts trivially on the codespace or is detected.
\end{remark}


Informally, a circuit is fault-tolerant if any fault occurring somewhere in the circuit is either detectable or can be propagated to the outputs of the circuit in such a way that its weight does not increase.


As we focus on CSS states, a common assumption~\cite{forlivesi2025flag, strikis2026highperformancesyndromeextractioncircuits,weilandt2026synthesisfaulttolerantstatepreparation,Poor2026ultralow}, is to restrict the analysis to the CSS noise model.
\begin{definition}[CSS Submodel]
  Let $C$ be a Clifford circuit with a noise model $\mathcal{N}$.
  Then the \emph{CSS submodel} of $\mathcal{N}$ is defined as $\mathcal{N}_{CSS}(F) = \max(\mathcal{N}(F|_X), \mathcal{N}(F|_Z))$ where $F|_X$ is the restriction of $F$ to its $X$ component, i.e.\@ dropping all the $Z$ faults and making $Y$ faults into $X$ faults and where $F|_Z$ is the analogous restriction to $Z$ faults.
\end{definition}

Analysing the fault tolerance of circuits under the CSS noise submodel significantly simplifies the reasoning and does not affect the circuit distance for the circuits we consider in this paper.

\begin{definition}[Circuit Distance]
  Let $C$ be a Clifford circuit with a noise model $\mathcal{N}$.
  Then the distance of $C$ under $\mathcal{N}$, $d_{\mathcal{N}}(C)$, is defined as the minimum weight of all non-trivial, undetectable faults.
\end{definition}

From this, we can motivate the restriction to CSS submodels:
\begin{proposition}\label{prop:CSS-distance-equivalence}
  Let $C$ be a circuit built from $\ket{0}, \ket{+}$, single-qubit Pauli $X$ and $Z$ rotations, CNOTs, and $Z$- and $X$-basis measurements. Let $\mathcal{N}$ be a noise model on $C$ such that
  $\mathcal{N}(F|_X), \mathcal{N}(F|_Z) \leq \mathcal{N}(F)$ for every fault $F$.
  Then the distance of $C$ under $\mathcal{N}$ is the same as the distance of $C$ under $\mathcal{N}_{CSS}$, i.e.\@ $d_{\mathcal{N}}(C) = d_{\mathcal{N}_{CSS}}(C)$.
\end{proposition}
\begin{proof}
  $d_{\mathcal{N}}(C) \geq d_{\mathcal{N}_{CSS}}(C)$: By assumption,
  $\mathcal{N}(F) \geq \max(\mathcal{N}(F|_X), \mathcal{N}(F|_Z)) = \mathcal{N}_{CSS}(F)$
  for every fault $F$. 
  Minimizing over non-trivial undetectable faults gives
  $d_{\mathcal{N}}(C) \geq d_{\mathcal{N}_{CSS}}(C)$.

  $d_{\mathcal{N}}(C)\leq d_{\mathcal{N}_{CSS}}(C)$: Let $F$ be a non-trivial undetectable fault with
  $\mathcal{N}_{CSS}(F) = d_{\mathcal{N}_{CSS}}(C)$. 
  As $C$ is built from $\ket{0}, \ket{+}$, single-qubit Pauli $X$ and $Z$ rotations, CNOTs, and $Z$- and $X$-basis measurements, the detectability of $X$ and $Z$ faults is independent, meaning that $F|_X$ and $F|_Z$ are undetectable.
  Furthermore, at least one of them, say $F|_P$,
  is non-trivial. Hence
  \[
    d_{\mathcal{N}}(C) \leq \mathcal{N}(F|_P)
    \leq \max(\mathcal{N}(F|_X), \mathcal{N}(F|_Z))
    = \mathcal{N}_{CSS}(F) = d_{\mathcal{N}_{CSS}}(C).
  \]
  \\[-20pt]\qedhere
\end{proof}
This means that if CSS states are fault-tolerant under CSS noise, in the right context, they can be regarded as generally fault-tolerant.
If used in a larger protocol with non-CSS gates, the protocol's non-CSS distance may be lower than the CSS distance.
However, when handled carefully, such as by adding syndrome extraction, analysis under CSS noise is sufficient, even in the presence of non-CSS gates.
Therefore, throughout the rest of this work, we focus on fault-tolerant CSS state preparation under the CSS restriction of circuit-level noise.




\subsection{CSS States in the ZX-Calculus}
\label{subsec:zx-representation}

For the task of circuit synthesis, we use the ZX-calculus, a graphical language for representing and reasoning about quantum computations. For the CNOT-based CSS-state preparation circuits considered here, the Pauli fragment of the ZX-calculus suffices to represent both the ideal operations and Pauli faults. For a complete introduction, see~\cite{KissingerWetering2024Book}.

\begin{definition}[Pauli and Phase-Free ZX-Diagrams]
  \label{def:zx-diagram}
  A \emph{Pauli ZX-diagram} is a graph consisting of $Z$- and $X$-spiders with a phase $k\pi$, with $k\in\{0,1\}$, defined as:
  \begin{align*}
    \textit{$Z$-spider:} \qquad
    \input{z-spider.tikz}
    \quad &\coloneqq \quad
    \ket{0}^{\otimes n}\! \bra{0}^{\otimes m} + e^{i k\pi} \ket{1}^{\otimes n}\! \bra{1}^{\otimes m}\\[8pt]
    \textit{$X$-spider:} \qquad
    \input{x-spider.tikz}
    \quad &\coloneqq \quad
    \ket{+}^{\otimes n}\! \bra{+}^{\otimes m} + e^{i k\pi} \ket{-}^{\otimes n}\! \bra{-}^{\otimes m}
  \end{align*}
  Sequential composition corresponds to composition of linear maps, and parallel composition corresponds to their tensor product.
  A diagram is \emph{phase-free} if every spider has phase zero.
\end{definition}

We will often refer to spiders as vertices and treat ZX-diagrams as undirected graphs.
Spiders that have an edge not connected to another spider, but directly to an input or output are called \emph{boundary vertices}, and the corresponding edges are called \emph{boundary edges}. Hence, a \emph{boundary spider} is incident to a boundary edge while an \emph{internal spider} is not.
The \emph{arity} (or degree) of a spider is its total number of incident legs.
If a diagram has no inputs it represents a quantum state and we call it a \emph{state diagram}.

The $Z$-spiders and $X$-spiders with $0$ or $\pi$ phases and their compositions represent circuits consisting of $\ket{0}$ and $\ket{+}$ preparations, single-qubit Pauli $X$ and $Z$ rotations, CNOT gates, and $X$- and $Z$-basis measurements; see \Cref{fig:translation-guide}.
Beyond representing computations, a key feature of ZX-diagrams is that we can manipulate them with a small set of elementary rewrites while preserving the linear map they represent.
We identify diagrams up to non-zero scalars.
Some relevant rewrites are shown in \Cref{fig:zx-axioms}.

\begin{figure}[hbt]
  \captionsetup[subfigure]{justification=Justified,singlelinecheck=false,margin=0pt,aboveskip=0pt}
  \begin{subfigure}[b]{0.44\textwidth}
    \centering
    \[\input{translation-guide.tikz}\]
    \caption{}\label{fig:translation-guide}
  \end{subfigure}
  \hfill
  \begin{subfigure}[b]{0.55\textwidth}
    \centering
    \[\input{zx-axioms.tikz}\]
    \caption{}\label{fig:zx-axioms}
  \end{subfigure}
  \caption{ZX overview. (a) Translation between quantum circuits and ZX-diagrams. The $k$ parameter in the representation of $Z$- and $X$-measurements corresponds to the measurement outcome. Note that a $Z$-spider is used to represent an $X$-type measurement and state preparation, and vice versa for an $X$-spider. (b) A selection of ZX-rewrites. $k, k_1,k_2\in\{0,1\}$.
  }
  \label{fig:circuit-to-zx}
\end{figure}

With the translation table of \Cref{fig:translation-guide} we can represent any CNOT circuit involving ancillas and measurements as a phase-free ZX-diagram and hence represent any circuit implementing a CSS state as a ZX-diagram. There is also a canonical way to represent any CSS stabilizer tableau directly as a phase-free ZX-diagram.

\begin{definition}[$X$- and $Z$-Generator Forms~\cite{kissinger2022phasefreezxdiagramscss}]
  \label{def:generator-form}
  \label{def:XZ-nf}
  Let $\ket{\psi}$ be a CSS state on $n$ qubits with a generating set of stabilizers $S^X_1, \dots, S^X_{n_X}, S^Z_1, \dots, S^Z_{n_Z}$.
  The \emph{$X$-generator form} of $\psi$ consists of a column of $n_X$ internal phase-free $Z$-spiders and a column of $n$ boundary phase-free $X$-spiders.
  Each of the $X$-spiders is connected to precisely one output.
  Furthermore, the $i$-th internal spider is connected to its $j$-th output spider if and only if $S^X_i$ acts non-trivially on the $j$-th qubit.
  The \emph{$Z$-generator form} is defined analogously, with the roles of $Z$- and $X$-spiders interchanged.
\end{definition}

We can, for example, represent the $\ket{\overline{0}}$ state of the Steane code in the $X$- and $Z$-generator forms:
\begin{example}[Steane-Code Diagrams]
  \label{ex:steane-generator-form}
  The stabilizers of the $\ket{\overline{0}}$ state in the $\llbracket 7, 1, 3 \rrbracket$ Steane code can be generated by the following Paulis: 
  \[
    X_1X_2X_3X_4 \qquad X_2X_3X_5X_6 \qquad X_3X_4X_6X_7 \qquad
    Z_1Z_2Z_3Z_4 \qquad Z_2Z_3Z_5Z_6 \qquad Z_3Z_4Z_6Z_7 \qquad Z_5Z_6Z_7
  \]
  where the first six stabilizers are the stabilizers of the code and the last stabilizer is specific to the logical state. 
  In terms of the $X$- and $Z$-generator forms this corresponds to:
  \begin{equation}
    \ket{\overline{0}}
    \overset{X\text{-generator diagram}}{\longmapsto} \scalebox{.75}{\input{figures/steane-zero-faulty.tikz}}
    \qquad \qquad \qquad
    \ket{\overline{0}}
    \overset{Z\text{-generator diagram}}{\longmapsto} \scalebox{.75}{\input{figures/steane-zero-dual.tikz}}
  \end{equation}
\end{example}

\begin{remark}[Terminology]
  The $X$-generator and $Z$-generator forms are called the $Z$-$X$ and $X$-$Z$ normal forms in~\cite[Definition~3.1]{kissinger2022phasefreezxdiagramscss}, respectively. We use an alternative name to prevent confusion with the term ZX-calculus.
\end{remark}

Any phase-free state diagram can be rewritten into either generator form~\cite{kissinger2022phasefreezxdiagramscss} and so they act as a normal form.
The diagrams of the $X$- and $Z$-generator forms of some CSS stabilizer state $\psi$ represent exactly the state $\psi$ as a ZX-diagram~\cite{kissinger2022phasefreezxdiagramscss}.
As the $X$- and $Z$-generator forms are dependent on a choice of stabilizer generators, this representation is non-unique.
However, for any choice of generators the resulting diagram encodes the same state and we can rewrite any two such diagrams into each other using ZX-rewrites.

\subsection{Fault-Equivalent ZX Rewrites}
\label{subsec:fault-equivalence}

The key methodology underlying this work builds on~\cite{rodatz2025fault}: the idea of starting with an idealized description of the computation one would like to implement and transforming it into an implementable circuit while preserving the specification's behaviour under noise.
As we use the ZX-calculus for this, we first have to introduce noise and noise models for ZX-diagrams:
\begin{definition}[Faults on ZX-diagrams~\cite{rodatz2025fault}]
  Let $D$ be a Pauli ZX-diagram. The spacetime locations $\mathcal{L}$ of $D$ are the edges $E$ between the spiders.
  Then a fault $F \in \mathcal{P}^{|E|}$ on $D$ assigns each edge an element in the single-qubit Pauli group (modulo phase).
  We write $D^F$ as the noisy diagram $D$ under $F$ obtained by adding the Pauli rotations indicated by $F$ to $D$. 
  We say a fault $F$ is \emph{trivial} if $D^F = D$ and \emph{detectable} if $D^F = 0$. 
\end{definition}

Now that we have a notion of faults, we can define a noise model.
As ZX-diagrams do not have an inherent time ordering, we define a noise model that is slightly different from the circuit-level noise model. We also introduce the notion of \emph{idealization}, where we state that certain edges in a diagram are not allowed to produce a fault. This is not possible in practice as we cannot ensure that no faults occur on certain qubits.
However, idealization serves as a useful mathematical tool to specify the behaviour of a circuit we would like to have.

\begin{definition}[Edge-flip Noise, Idealized Edges]
  Let $D$ be a ZX-diagram and $F \in \mathcal{P}^{|E|}$ be a fault on $D$.
  Then under the \emph{edge-flip noise model} $\mathcal{N}(F) = |F|$ is the Pauli weight of $F$.
  Furthermore, we allow the \emph{idealization of edges} as fault-free (meaning no faults allowed to occur on the edge) which we represent in the diagram by drawing those edges in purple.
  If a fault $F$ acts non-trivially on an idealized edge, then we set $\mathcal{N}(F) = \infty$, so that any finite-weight fault only has support on unidealized edges. We will also refer to unidealized edges as \emph{noisy} edges.
\end{definition}

Idealization allows us to define an ``idealized specification'' of a circuit or state where every internal edge is idealized.
A circuit is then fault-tolerant if and only if each fault on it corresponds to a fault on its idealization.
To make this precise, we define the notion of fault equivalence.

The key perspective shift proposed by~\cite{rodatz2025fault} is not to ask whether a particular circuit is fault-tolerant but rather whether two circuits are ``equally fault-tolerant'':
\begin{definition}[Fault Equivalence]
  \label{def:fault-equivalence}
  Let $D_1$ and $D_2$ be non-zero ZX-diagrams under respective noise models $\mathcal{N}_1, \mathcal{N}_2$.
  Let $w \in \mathbb{N} \cup \{\infty\}$.
  We say $D_1$ under $\mathcal{N}_1$ is $w$-fault-equivalent to $D_2$ under $\mathcal{N}_2$, if for every undetectable fault $F_1$ on $D_1$ with $\mathcal{N}_1(F_1) < w$, there exists an $F_2$ on $D_2$ such that $D_1^{F_1} = D_2^{F_2}$ and $\mathcal{N}_2(F_2) \leq \mathcal{N}_1(F_1)$ and vice versa every undetectable fault $F_2$ on $D_2$ has a corresponding fault on $D_1$.
  When this is the case, we write $D_1\wFaultEq{w}D_2$.
  When $w = \infty$, we say that the two diagrams are \emph{fault-equivalent}, written $D_1 \FaultEq D_2$.
\end{definition}
Just like our definition of fault tolerance, fault equivalence is defined with respect to a noise model. 
Therefore, we can consider the fault equivalence of ZX-diagrams under edge-flip noise but also, for example, under CSS edge-flip noise.

Fault equivalence has nice properties: 
\begin{proposition}[Compositionality and Transitivity~\cite{rodatz2025fault}]
  \label{prop:css-fe-compositionality}
  We note $w$-fault equivalence is transitive and is preserved by sequential and parallel composition of diagrams with compatible boundaries and noise specifications.
  Consequently, a $w$-fault-equivalent rewrite may be performed inside a larger diagram, and a sequence of such rewrites preserves $w$-fault equivalence.
\end{proposition}
On top of these, the ZX-calculus with only fault-equivalent rewrites is also complete~\cite{rusch2025completeness}, meaning that it can derive any possible equality.

Using fault equivalence, we can recapture common definitions in the literature, such as fault-tolerant state preparation (see~\Cref{def:ft-state-prep}):
\begin{proposition}
  \label{prop:ft-prep-fault-equivalence}
  Let $D$ be a ZX-diagram that prepares some state $\ket{\psi}$ under noise model $\mathcal{N}$.
  $D$ $w$-fault-tolerantly prepares $\ket{\psi}$ according to \Cref{def:ft-state-prep} if and only if it is $w$-fault-equivalent to $D$ under a noise model $\mathcal{N}'$ where the only allowed faults are output qubit flips and the weight of a fault is the number of output qubits that are flipped. 
\end{proposition}
\begin{proof}
  The fault equivalence of a diagram $D$ to the idealized specification requires that any fault $F$ of weight $w' < w$ on $D$ is either detectable or has a corresponding fault $F'$ of at most the same weight. 
  As the only faults allowed on the specification are the single-qubit output flips, this means that $F$ must either be detectable or propagate to at most $w'$ many output qubit flips.
\end{proof}




This leads to a new way of synthesizing fault-tolerant state preparation circuits: if one starts with the idealized, fault-free state and transforms it using local, ``fault-equivalence-preserving'' transformations, then by construction the final circuit will be fault-equivalent to the idealized specification and hence fault-tolerant by construction.

However, not all rewrites preserve fault equivalence. 
For example, ordinary spider unfusion need not preserve fault equivalence: 
\begin{lemma}
An ordinary unfusion is not fault-equivalent,
\[
\input{figures/non-fe-rewrite.tikz}
\]
\label{lem:non-fe-rewrites}
\end{lemma}
\begin{proof}
An internal $X$-type edge flip will be pushed to the output edges using the $\pi$-copy rule, blowing up the weight count.
\[
\input{figures/non-fe-rewrite-proof.tikz}
\]
  \\[-20pt]\qedhere
\end{proof}

Rewrites that preserve fault equivalence are called \emph{fault-equivalent rewrites}. 
Some elementary rewrites that are fault-equivalent are shown in \Cref{fig:elementary-fe-rewrites}~\cite{rodatz2025fault}.

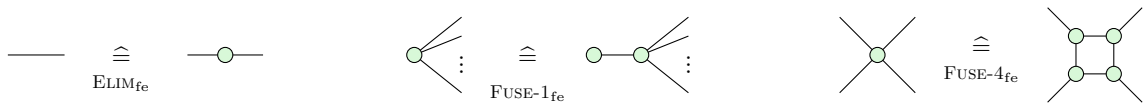
\begin{figure}[H]
    \centering
    
    \[\scalebox{1}{\input{figures/fe-rewrites1.tikz}}\]
    \caption{The rewrites together with their colour duals are fault-equivalent~\cite{rodatz2025fault}.}
  \label{fig:elementary-fe-rewrites}
\end{figure}

We can now frame the problem of fault-tolerant synthesis of CSS states as follows. We start with the $X$- or $Z$-generator form (\Cref{def:XZ-nf}) of the CSS state we want to synthesize. We treat all the internal wires as idealized, and hence this diagram represents the idealized fault-tolerant implementation. We then rewrite this diagram to something circuit-like using rewrites that are fault-equivalent under edge-flip noise.
While in general the quantum circuits in \Cref{fig:translation-guide} under circuit-level noise are not fault-equivalent to the corresponding ZX-diagrams under edge-flip noise, if we use the CSS submodel of circuit-level noise, they are~\cite[Section~7.2.2]{rodatz2025fault}. This means that once we have rewritten the diagram to something circuit-like, we can interpret it as an actual circuit and we are still guaranteed to have the same degree of fault tolerance.

\subsection{Fault-Tolerant Spider Expansions}
\label{subsec:cat-preliminaries}

To transform the idealized ZX-diagram describing the CSS state from \Cref{def:XZ-nf} into a circuit, the biggest obstacle is decomposing the high-arity spiders into arity-3 spiders that can be readily interpreted as CNOT gates, and to do this in a fault-tolerant way.

As a single spider is a \CAT state, this corresponds to finding fault-tolerant implementations of \CAT states. This problem was tackled in our previous work, SpiderCat~\cite{spidercat2026}. Since we will need to change some parts of the procedure from SpiderCat in this paper, we describe here at a high level how this works.
SpiderCat starts with a single phase-free $Z$-spider $Z_n$ with $n$ output legs, representing the \CAT state.
This specification is fault-equivalently expanded into a diagram of three-legged $Z$-spiders, from which a circuit is extracted.

\begin{definition}[$Z$-Graphs and Fault-Tolerant Expansions]
  \label{def:spider-expansion}
  A \emph{$Z$-graph} is a phase-free state diagram consisting solely of three-legged $Z$-spiders.
  A $Z$-graph $D$ with $n$ labelled outputs is a \emph{$w$-fault-tolerant expansion} of $Z_n$ if all its edges are noisy and $D \wFaultEq{w} Z_n$.
  $X$-graphs and fault-tolerant expansions of $X$-spiders are defined analogously by interchanging the colours.
\end{definition}

By \Cref{prop:css-fe-compositionality}, a fault-tolerant expansion can replace a spider inside a larger diagram, rather than only serving as an isolated \CAT-state preparation.

\begin{definition}[Marked-Graph Representation]
  \label{def:marked-cat-graph}
  A \emph{marked graph} consists of a $3$-regular graph $G$ with marked positions on its edges, drawn as dashes crossing the edges.
  Its \emph{expanded graph} $\widetilde G$ is obtained by replacing every mark by a degree-two vertex. 
  The \emph{corresponding $Z$-graph} of $G$ is obtained by replacing each vertex of $\widetilde G$ by a phase-free $Z$-spider and attaching one output edge to each degree-two vertex, so that each spider has three legs.
\end{definition}

See \Cref{fig:cat-extraction-overview} for the difference between marked graphs, expanded graphs and $Z$-graphs.
SpiderCat derived necessary and sufficient conditions for a marked graph to be a $w$-fault-tolerant expansion. We will just take this as given in this paper and assume that we have access to appropriate marked graphs satisfying the necessary conditions.

\begin{remark}[Convention on Notation for Fault-Tolerant Circuits]
\label{rem:spidercat-convention}
In the SpiderCat paper~\cite{spidercat2026}, the definition of a $t$-fault-tolerant state differs by 1 from the one we adopt in this paper, where we would use the value $t+1$ instead of $t$.
As an example, a circuit that this paper refers to as a 4-fault-tolerant 20-qubit \CAT state preparation would be called a 3-fault-tolerant 20-qubit \CAT state preparation by the convention in SpiderCat.
\end{remark}

Given a $Z$-graph we still need to interpret this as a circuit. To do this, we assign to each spider a space and time location (i.e.~on which qubit it occurs and after which operations it occurs). We do this by finding an appropriate spanning tree or forest on the graph.

\begin{definition}[Spanning Trees and Forests]
  \label{def:spanning-forest}
  A \emph{spanning tree} of a graph $G$ is a connected, acyclic subgraph containing every vertex of $G$. A \emph{spanning forest} satisfies the same conditions except that it may be disconnected, with each connected component forming a tree. A \emph{rooted tree} has a designated root vertex. In a rooted tree we define the \emph{parent} for each non-root to be its unique neighbour on the path to the root. A \emph{rooted forest} has one root in each component.
\end{definition}

Given a $Z$-graph with a rooted spanning tree or forest, we turn it into a circuit (`extract' a circuit) in the following way.
Pick a root and place it at the start of the circuit as an ancilla prepared in the $\ket{+}$ state. Pick one child and place it on the same qubit line. For each other child, place it on a new qubit line. Now do the same with all the children: pick one vertex to stay on the same qubit and let all the other unplaced children be placed on new qubit lines. If a spider is a boundary spider, we instead place its output wire on the same qubit and all its unplaced children on other qubit wires. When all spiders are placed, we have laid out the diagram in a way where we can easily turn it into a circuit; see \Cref{fig:cat-extraction-overview}.

To finish turning it into a circuit, we introduce a $\ket{0}$ preparation and a CNOT whenever a spider is put onto a new qubit line and use a CNOT and $Z$-basis measurement post-selected to $\ket{0}$ when a qubit line ends.
In both cases, the target of the CNOT is on the line that is starting or ending.
See~\cite{spidercat2026} for more details on why this is correct.


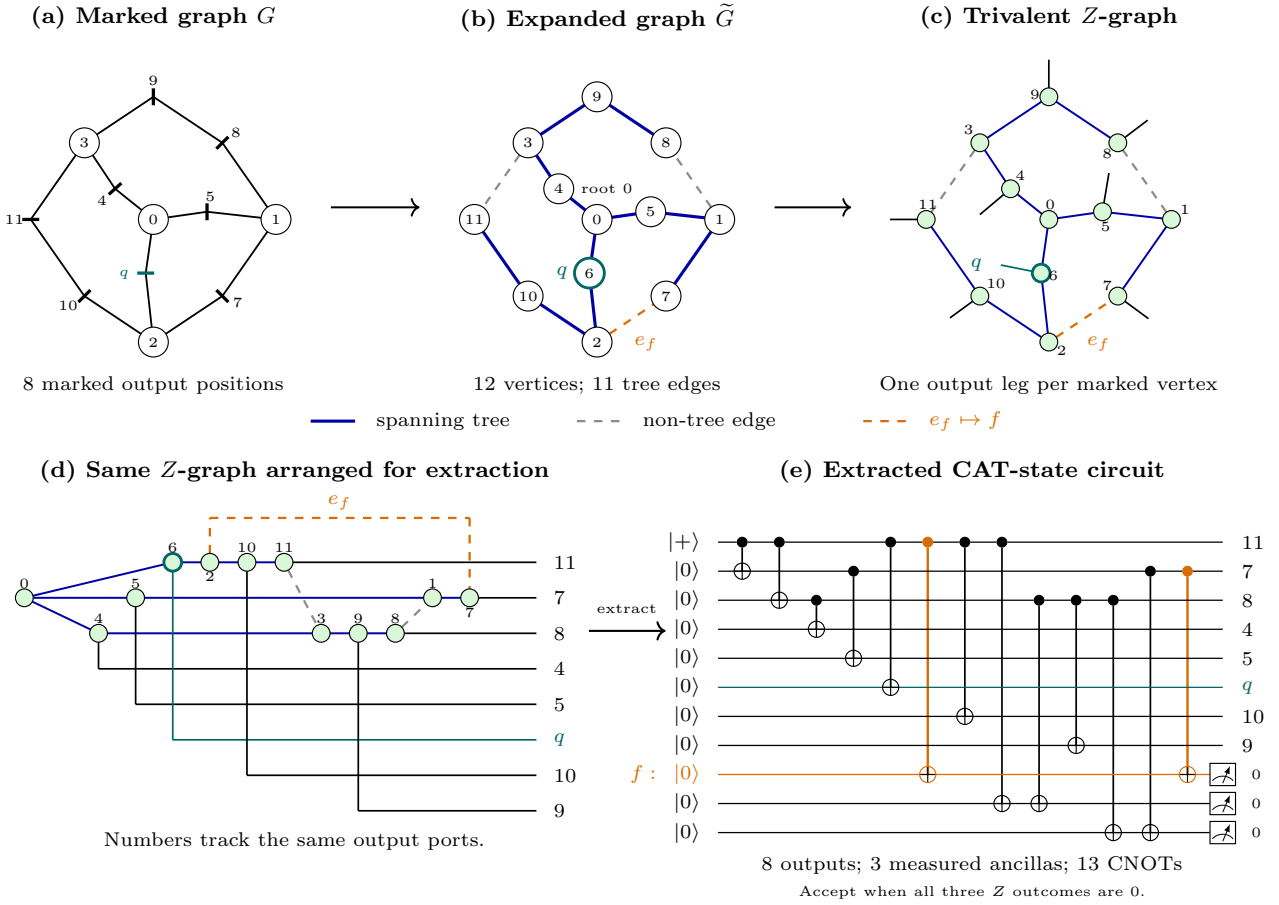
\begin{figure}[hbt]
  \centering
  \resizebox{\linewidth}{!}{\input{figures/cat-extraction-overview.tikz}}
  \caption{
    Circuit extraction from a fault-tolerant $Z$-spider expansion.
    (a) A graph with eight marked output positions.
    (b) Its expanded graph with a spanning tree rooted at vertex $0$.
    (c) The trivalent $Z$-graph obtained by attaching an output leg to each marked vertex.
    (d) The same $Z$-graph arranged for extraction.
    (e) The extracted \CAT-state preparation circuit, with three verification ancillas.
    The mark $q$ (vertex $6$, teal) becomes output qubit $q$, while the non-tree edge
    $e_f=(2,7)$ (orange) is implemented by the two CNOTs involving ancilla $f$.
    Numbers identify corresponding vertices and output ports, not spider phases.
    All three measurement outcomes are post-selected to $0$.
  }
  \label{fig:cat-extraction-overview}
\end{figure}

The SpiderCat paper~\cite{spidercat2026} shows that the circuits obtained from this procedure above can achieve provably minimal CNOT count for fault-tolerant \CAT state preparation. This minimality is achieved when the rooted spanning tree is chosen such that its root is an (internal) degree-3 vertex and all the leaves are degree two in $\widetilde G$. In this paper we will need the spanning trees to satisfy additional properties that will be explained later, but these in practice do not influence the optimality of individual \CAT states.

To construct fault-tolerant implementations of $X$-spiders, we do the procedure above but interchange the role of the $Z$- and $X$-basis and the target and control of the CNOTs.

\section{CSS State Synthesis: The Theory}
\label{sec:the-theory}

In this section, we establish the theoretical foundation of our synthesis method.
Following the \emph{Fault Tolerance by Construction} method, we can separate our synthesis into two steps: (1) starting from an idealized description of our CSS state as a ZX-diagram, we use fault-equivalent ZX-rewrites to unidealize the edges and decompose the large-weight spiders, to bring it into a form from which we can (2) extract a fault-equivalent quantum circuit. Each of these steps has a lot of freedom in how we choose to do certain operations, which has a great impact on the final performance. We explain the implementation details in \Cref{sec:the-pipeline}.

\subsection{Step 1: ZX Transformations}\label{sec:ZX-transformations}

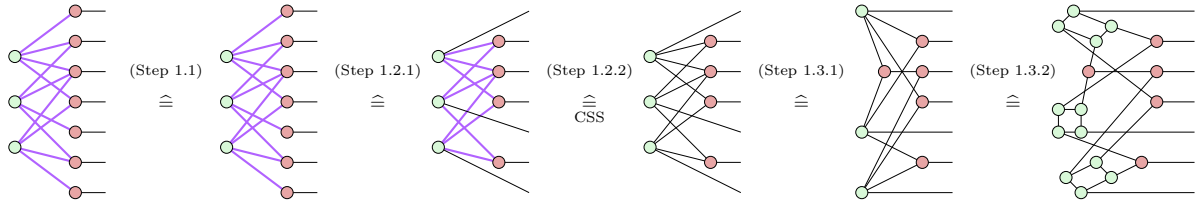
\begin{figure}[tb]
  \begin{equation*}
    \scalebox{.8}{\input{figures/steane-zero-full-fig.tikz}}
  \end{equation*}
  \caption{
    Transformation steps towards an extractable ZX-diagram for fault-tolerant state preparation. Step 1.1 brings the diagram into a row-reduced normal form,
    then, Step 1.2 unidealizes the internal edges, and Step 1.3 unfuses the spiders along safe hook errors (Step 1.3.1) and using SpiderCat~\cite{spidercat2026} (Step 1.3.2).
  }
  \label{fig:steane-full}
\end{figure}

Our starting point is an idealized specification of the CSS state we wish to synthesize as a circuit given as a $Z$- or $X$-generator form state diagram, which we transform until we have a ZX-diagram from which we can extract our final circuit.
\Cref{fig:steane-full} showcases the individual steps of this procedure: Step 1.1 brings the diagram into a special normal form, Step 1.2 unidealizes the edges and Step 1.3 decomposes the spiders into spiders of weight at most three. 

\subsubsection{Step 1.1: Bringing the Diagram to Row-Reduced Normal Form}
\label{subsec:bipartite}

To bring our ZX-diagram closer to an implementable circuit, we have to unidealize the edges. 
For this, we first bring our ZX-diagram into a more specific normal form than the standard $X$- and $Z$-generator forms outlined in the preliminaries. 
Recall that the $X$- and $Z$-generator forms of \Cref{def:XZ-nf} are dependent on the choice of basis of stabilizers, so that we have freedom in choosing these generators.
\begin{definition}[$X$- and $Z$-RREFs]
    Let $\ket{\psi}$ be a CSS state on $n$ qubits with a generating set of purely $X$- and purely $Z$-stabilizers $S^X_1, \dots, S^X_{n_X}, S^Z_1, \dots, S^Z_{n_Z}$. We say this choice of generators is in \emph{$X$-generator row-reduced echelon form} ($X$-RREF) if every $X$-stabilizer generator $S^X_i$ has at least one qubit that it acts upon that no other $X$-stabilizer acts upon. We define $Z$-RREF analogously. As these are independent, a set of generators that is both in $X$-RREF and $Z$-RREF is said to be in RREF.
\end{definition}

\begin{proposition}
  For any CSS state, we can find a generating set of stabilizers that are in RREF.
\end{proposition}
\begin{proof}
  Start with any stabilizer tableau and bring it into RREF through row-reduction.
\end{proof}

Just like the $X$- and $Z$-generator forms, the $X$- and $Z$-RREFs are not usually unique. As we will describe later, we will use this lack of uniqueness as an opportunity to reduce the resource count by minimizing the size of the spiders involved.

\subsubsection{Step 1.2: Fault-Equivalent Edge Unidealization}
\label{subsec:edge-elimination}
In order to unidealize edges, we introduce two fault-equivalent rewrites:
\begin{proposition}
  \label{prop:semi-ideal-id}
  If a spider has 1 idealized leg and 1 noisy leg, the spider and its two legs can be replaced by a single noisy leg:
  \begin{equation*}
    \input{figures/semi-idealized-id.tikz}
  \end{equation*}
\end{proposition}
\begin{proof}
  Neither side has any noisy internal edges, and therefore fault equivalence trivially follows from~\cite[Proposition B.4]{rodatz2025fault}.
\end{proof}

If our diagram is in $X$-RREF, then we can use this rewrite to ensure that all spiders are connected to a boundary:
Recall that in $X$-generator form we have an internal $Z$-spider for each $X$-stabilizer. The RREF condition ensures that then for each internal $Z$-spider there is at least one $X$-spider on the boundary that it is connected to that is not connected to any other internal spider. This boundary spider then has degree 2, so that the rewrite above applies.
This removes that boundary spider and makes the previously internal $Z$-spider directly connected to an output.
After doing this with all internal $Z$-spiders, we are left with a diagram that only has boundary spiders, with the only remaining idealized edges connecting the $Z$-spiders to the $X$-spiders, meaning that our diagram is in a bipartite form; see \Cref{fig:steane-full}, Step 1.2.1 for an example of this.

We can now get rid of the remaining idealized edges using the following rewrite: 
\begin{proposition}
  \label{prop:unidealize}

  An idealized edge can be unidealized under CSS-noise if the edge is between a $Z$- and an $X$-spider, and both spiders have at least 1 unidealized edge.
  That is, if $a, c > 0$, the following fault equivalence holds:
  \begin{equation*}
    \input{figures/johns-rule.tikz}
  \end{equation*}
\end{proposition}
\begin{proof}
    The left-hand side only has noisy boundary edges, and so each fault on it trivially corresponds to a fault on the right-hand side. For the other direction, note that for CSS fault equivalence it suffices to take pure $X$- and pure $Z$-type faults from internal edges and find equivalent faults of at most the same weight on noisy boundary edges without changing their type~\cite{rodatz2025fault}.
  The following two equations show this for the two elementary edge flips, $X$ and $Z$:
  \[
    \input{figures/colour-change-int-proof2.tikz}
  \]
  \\[-20pt]\qedhere
\end{proof}

\begin{remark}
  In the ordinary edge-flip model, a single internal $Y$ fault has weight $1$, but its boundary representation requires an $X$ and a $Z$ fault on distinct edges, giving a total weight of $2$.
  Under the CSS edge-flip convention, both fault patterns have CSS weight $1$.
    As such, this rewrite is fault-equivalent under CSS edge-flip noise, but not under normal edge-flip noise. 
    We indicate this by adding ``CSS'' under the fault equivalence sign. As discussed in \Cref{sec:CSS-states}, for CSS state preparation it suffices to consider CSS noise.
\end{remark}

Applying these rewrites in Step 1.2.2, we can now unidealize all idealized edges as demonstrated in \Cref{fig:steane-full}.

\subsubsection{Step 1.3: Spider Decompositions}
\label{subsec:spider-decompositions}

Once the bipartite ZX-diagram is free of idealized edges, the remaining graph consists of physical, but high-degree spiders, which cannot natively be realized using two-qubit gates.
In order to obtain a ZX-diagram from which a circuit can be extracted, each high-degree spider must be fault-equivalently decomposed into a network of degree-3 spiders.

To do this, we use two different families of rewrites: unfusion along a strictly safe hook error, and fully fault-equivalent unfusion using SpiderCat~\cite{spidercat2026}.
As a fully fault-equivalent unfusion is relatively expensive in terms of CNOT gates that will eventually be required, we first try to unfuse the spiders as much as possible by finding and using the presence of strictly safe hook errors.

Strictly safe hook errors are single internal faults that propagate to a multi-qubit Pauli fault on the boundary, but are equivalent up to a stabilizer to a physical fault of weight at most 1.
This equivalence arises because the logical state we are preparing is stabilized not just by the code's stabilizer group, but also by the stabilizers of the logical itself. When we have determined a hook error is strictly safe we can unfuse a spider along that direction:
\[
  \input{safe-hook-split.tikz}
\]
We present more details on the definition of safe hook errors and how we can efficiently determine all possible safe hook errors for splitting a spider in \Cref{apdx:safe-hook-errors}.

For the running example, in Step 1.3.1, we unfuse a weight-4 spider along the only strictly safe hook error admitted by the Steane code.
While a $Z$ fault on this new internal edge would propagate to the boundary as $Z_1 Z_5 (\sim Z_3 Z_7)$, multiplying this fault by the logical operator $Z_L = Z_1 Z_2 Z_5 (\sim Z_2 Z_3 Z_7)$ (which stabilizes $\ket{\overline 0}$) leaves the single-qubit fault $Z_2$.

Finally, as Step 1.3.2, we can apply the SpiderCat~\cite{spidercat2026} decompositions to the remaining high-weight spiders.
These decompositions come with provably optimal $Z$-spider counts, which translate to optimal CNOT counts in the extracted circuit.

\subsection{Step 2: Fault-Equivalent Circuit Extraction}\label{sec:circuit-extraction}

Once a completely unidealized ZX-diagram containing just arity-3 spiders is obtained, the only thing that remains is extracting a quantum circuit fault-equivalently.
However, extracting a circuit that uses the minimal resources to implement a general ZX-diagram is likely to be computationally hard, as the problem is closely related to finding minimal path-covers on the ZX-diagram with additional constraints.

A naive approach of doing circuit extraction is to directly use the \CAT states and implement the connections between them as Bell-state measurements:
\begin{equation}
  \scalebox{.75}{\input{figures/steane-zero-naive-se.tikz}}
\end{equation}
However, this approach requires a large number of ancilla qubits and operations, and it does not use the fact that the legs between a $Z$- and an $X$-spider can often be extracted directly as CNOT gates (see \Cref{fig:translation-guide}).
More precisely, in the approach above each internal edge in the bipartite graph is extracted as 3 CNOT gates with additional state preparations and measurements: one state preparation and CNOT gate per \CAT state output, with an additional CNOT and two measurements for the Bell-state measurement.

In contrast, if each edge between a $Z$- and an $X$-spider were extracted as a single CNOT, our circuit extraction procedure would be optimal for this type of graph:
\begin{proposition}
  \label{prop:min-cnot}
  Let $D$ be a ZX-diagram that $t$-fault-tolerantly implements a bipartite graph state $G$ with vertex and edge sets $V$ and $E$, and let $C(n, t)$ be the minimum number of CNOT gates for an isolated $t$-fault-tolerant \CAT state preparation circuit of size $n$.
  For a modular circuit extraction where each spider $v \in V$ is independently instantiated as a \CAT state and interconnected strictly along the bipartite edges $E$, the minimum number of CNOT gates required is:
  \[
    \sum_{v \in V} C(\deg(v), t) - |E|
  \]
\end{proposition}
\begin{proof}
  Extracting a circuit from a ZX-diagram can be viewed as finding a path cover on the diagram.
  We can exactly quantify the cost of extraction from the cost of implementing uncovered edges: $Z-X$ edges cost 1 CNOT, while $Z-Z$ or $X-X$ edges require a parity measurement costing 2 CNOTs.
  Because $G$ is bipartite, extracting all $|E|$ inter-spider connections directly as 1-cost $Z-X$ CNOTs is the global optimum.

  For each spider $v$, an isolated $t$-fault-tolerant \CAT state requires $C(\deg(v), t)$ CNOTs.
  In the best case, for each edge we can combine the two fan-out CNOTs required to create the qubits on each side of the edge into a single CNOT.
  This results in a saving of $|E|$ CNOTs. This yields the desired bound.
\end{proof}

\begin{remark}
    The above proposition assumes that each spider is implemented purely $t$-fault-tolerantly.
    If we also take unfusions along safe hook errors into account, this proposition does not directly apply as the unfusion is not a rewrite that is $t$-fault-equivalent in an arbitrary context.
    However, by only unfusing along hook errors when it decreases the total cost, the above number provides an upper bound on the required resources.
\end{remark}

\subsubsection{Efficient Circuit Extraction Using Well-Ordered \CAT States}

To optimize the resource overhead of circuit extraction, we want to achieve the CNOT cost of \Cref{prop:min-cnot} instead of the naive cost of extracting all connections between $Z$- and $X$-spiders as Bell measurements.
To achieve this, we rely on a special type of fault-tolerant \CAT states that have additional circuit structure: 
\begin{definition}[Well-Ordered \CAT State]
  \label{def:well-ordered-cat}
  We say a circuit implementing an $n$-qubit \CAT state is \emph{well-ordered} if for $n-1$ of its output qubits, the last operation on a given qubit $q$ is a CNOT gate and this CNOT gate is either:
  \begin{enumerate}
    \item the only CNOT gate applied to $q$, or,
    \item a CNOT between a flag qubit and $q$, provided it is the first CNOT applied to that flag.
  \end{enumerate}
  The single remaining output qubit, which is exempt from these conditions, is referred to as the \emph{primary output qubit}.
\end{definition}


By restricting our compilation pipeline to use well-ordered \CAT states, we can now reduce the CNOT count (we will see in the next subsection how we can construct well-ordered \CAT states).
Thus, in the decomposition of the high-weight spiders in Step 1.3, we use well-ordered \CAT states and ensure that the one edge of each \CAT state that is a boundary leg is the primary output qubit. Then all the other legs of the spiders in the well-ordered \CAT state are connected to other \CAT states.

By definition, the last operation on each of these other legs of the \CAT states is either the only CNOT gate applied to it or is part of a flag. It turns out that in both cases we can efficiently implement a connection to another \CAT state.
First, if the qubit only has a single CNOT on it:
\begin{lemma}[Direct CNOT Condition]
  Suppose we have a \CAT state with an output qubit $q$ which has only one CNOT gate on that wire and that this qubit is part of a Bell measurement. 
  Then we can implement this Bell measurement at no additional cost of entangling gates.
  \label{prop:cat-reduction-1}
\end{lemma}
\begin{proof}
  As the last CNOT on the qubit $q$ is also the only entangling gate, up to single-qubit unitaries, qubit $q$ consists of a state preparation and the CNOT that connects it to some other qubit inside the \CAT state. We may furthermore assume that if we are connecting a $Z$-type \CAT with an $X$-type \CAT, that the CNOT in the Bell measurement has its control on the $Z$-side and the target on the $X$-side. We can then simplify as follows: 
  \begin{equation*}
    \label{eq:cat-part-1}
    \scalebox{.85}{\input{figures/cat-part-1.tikz}}
  \end{equation*}
  \\[-20pt]\qedhere
\end{proof}

Second, if the last CNOT is the first CNOT on a flag qubit (a qubit which is prepared, has two CNOTs applied to it, and is then measured):
\begin{lemma}[Flag-Delegation Condition]
  \label{prop:cat-reduction-2}
  Suppose we have a \CAT state with an output qubit $q$ on which the last CNOT gate is the first CNOT on some flag qubit $f$, and that qubit $q$ is part of a Bell measurement.
  Then we can implement this Bell measurement at no additional cost of entangling gates. 
\end{lemma}
\begin{proof}
    As the last CNOT on qubit $q$ prior to the Bell measurement connects it to a flag qubit, we can simplify as follows:
  \begin{equation*}
    \label{eq:cat-part-2}
    \scalebox{.85}{\input{figures/cat-part-2-fe.tikz}}
  \end{equation*}
  \\[-20pt]\qedhere
\end{proof}

The results for the $X$-type \CAT states are the same but with the colours interchanged.

\begin{remark}
    In the second condition, we require from the well-ordered \CAT state that the last CNOT from $q$ is the first CNOT on the flag. 
    This is in order to preserve temporal ordering. 
    If the rewrite is performed with the reversed order (so that the last CNOT is also the last CNOT on the flag), the resulting ZX-diagram (shown below) may produce an acausal internal edge connecting spiders that cannot be placed in the same time step.
    Such an edge cannot be extracted as a single physical CNOT gate.
    While it could be implemented using an additional ancilla and two CNOT gates, doing so introduces unnecessary resource overhead and defeats the purpose of the optimization.
  \begin{equation*}
    \scalebox{.85}{\input{figures/cat-part-counter-example-2.tikz}}
  \end{equation*}
\end{remark}

\subsubsection{Constructing Well-Ordered \CAT States}\label{subsec:well-ordered-cat-states}

As the two results above show, by using well-ordered \CAT states we can connect together $Z$-type \CAT states and $X$-type \CAT states with Bell measurements at no additional cost in entangling gates. It remains to show how we can actually obtain such well-ordered \CAT states.

While the methods described in SpiderCat~\cite{spidercat2026} result in fault-tolerant \CAT states that require a provably minimum number of CNOT gates, the methods have no restrictions on the ordering of the CNOTs applied to the flag qubits, and thus can produce \CAT circuits that are not well-ordered.
Here, we show how to use a modified version of SpiderCat to produce CNOT-minimal fault-tolerant \CAT states that are also well-ordered.
In particular, we propose a modification to extraction that enforces the applicability of the flag-delegation condition (i.e.~that the last CNOT on an output qubit corresponds to the \emph{first} CNOT on a flag) whenever the direct CNOT condition does not apply.

As explained in \Cref{subsec:cat-preliminaries}, SpiderCat operates on 3-regular graphs and spanning trees of these graphs; see for example \Cref{fig:cat-extraction-overview}.
We extend this pipeline with a Directed Acyclic Graph (DAG) whose topological generations dictate the order in which nodes are extracted; see \Cref{fig:dag-cat} for the full procedure. 

\begin{figure}[hbt]
\scalebox{0.58}{\input{figures/well_ordered_extraction.tikz}}
  \caption{
    Modified \CAT state extraction using a DAG:
    (1) take the original graph $G$ and spanning tree $F$ specifying the fault-tolerant \CAT state,
    (2) construct the directed graph $D$ corresponding to the \CAT state augmented with artificial CNOT ordering conditions which are marked with dashed red edges,
    (3) resolve a DAG $D'$ by removing a single artificial CNOT ordering edge from $D$,
    (4-5) rearrange the initial diagram $G$ according to the topological generations of $D'$,
    (6) introduce identity spiders and unfuse state preparation and measurements so that the diagram is circuit-like,
    (7) directly extract a well-ordered \CAT state from $G$.
  }
  \label{fig:dag-cat}
\end{figure}
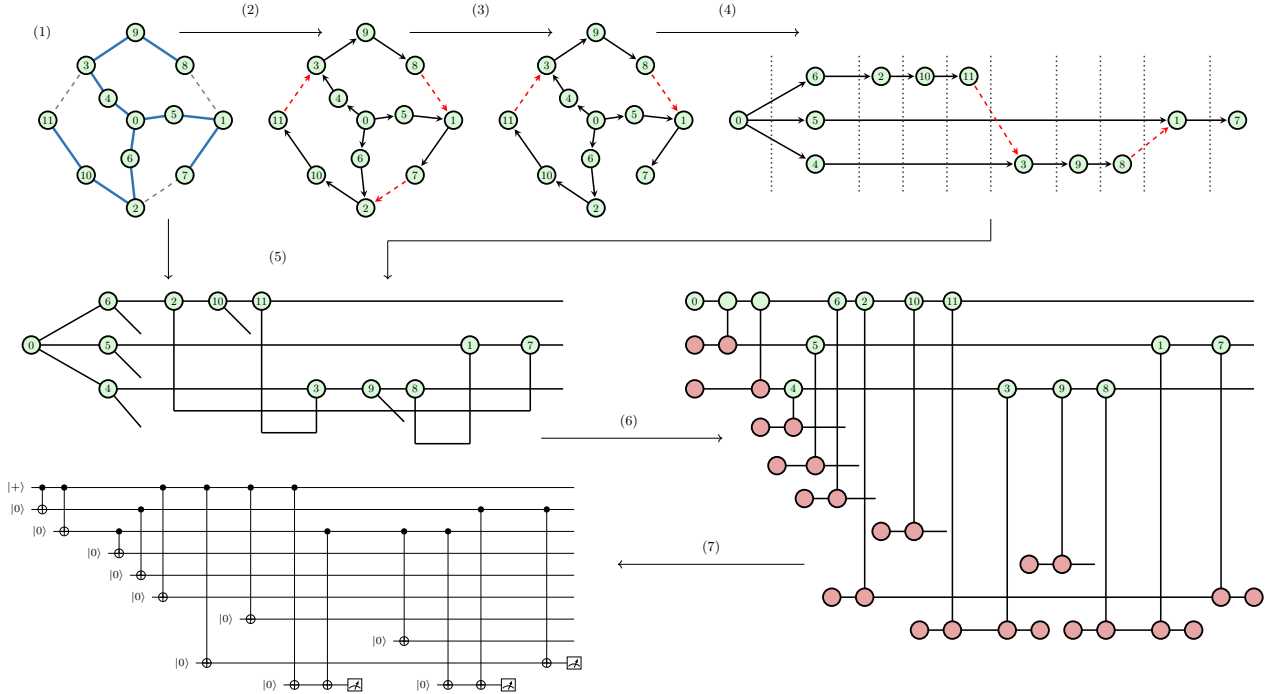

We start with the original graph and spanning tree from SpiderCat. The edges not covered by the tree will correspond to flag qubits in the extracted circuit. We need to decide an order on how these flags are resolved, which is arbitrary in the original SpiderCat procedure. 

We turn the spanning tree into a DAG by starting at the root and directing the edges towards the children. Then if the tree ends in a degree 2 vertex (a boundary spider), it will have an edge not part of the original tree (a flag). We consider a directed edge from this boundary spider to its neighbour and add it to the DAG, which ensures that in the extracted circuit, this boundary spider applies its CNOT to the flag first. This works as long as its neighbour is not also a degree 2 vertex, so in the generation of the spanning tree we make sure this is the case. This then results in the DAG shown in the second diagram in \Cref{fig:dag-cat}, with the newly added directed edges shown in red. Note that if a degree 2 vertex is not a leaf, then its boundary wire has only a single CNOT on it, so that the well-ordered condition is automatically satisfied.

This graph with the artificial directed edges will generally not be a DAG, because the newly inserted edges can form a cycle.
However, most of the time this can be resolved by the removal of exactly one artificial directed edge, which still results in a well-ordered \CAT state since the primary output qubit is exempt from these conditions; see the third diagram of \Cref{fig:dag-cat}. In case there are multiple independent cycles, we generate a different spanning tree candidate and try again. If we still do not succeed after a number of tries (which is very rare), we generate a new candidate marked graph and try again. In practice we can always find a suitable DAG for any marked graph supplied from SpiderCat.

Now that we have a suitable DAG, we can arrange the spiders according to the topological generations of the DAG and turn the diagram into a circuit, as is described in the remaining diagrams of \Cref{fig:dag-cat}. Note that in this example, the second qubit output, corresponding to spider 7 is the primary output qubit, as it ends up connected to the final CNOT of a flag.

In the case where we first split a high-degree spider into multiple parts using a strictly safe hook error, we get two or more spiders that each will independently be decomposed using SpiderCat and which individually can be made well-ordered. However, since only one edge of this composite \CAT state is a boundary of the diagram, all the other primary output qubits must not connect to a boundary so that combined \CAT is still well-ordered.  We accomplish this by using the left-over primary output qubits to cover the internal edges between the decomposed spider components: suppose we have a spider with $n$ legs, which gets decomposed into three parts containing $a+1$, $b+2$ and $c+1$ legs, where $a+b+c=n$. Then we have one leg from $a$ going to $b$, and another leg of $b$ going to $c$. If we suppose the $c$ part is directly connected to an output, then its primary output qubit can go there. We can use $a$'s primary output leg to cover the connection to $b$, and $b$'s primary output leg to cover the connection to $c$. The full $a+b+c$ legged \CAT state then still has one overall primary output qubit, meaning it is well-ordered.

Using this construction for producing well-ordered \CAT states, we are now guaranteed that our circuit extraction achieves the optimal gate count suggested by \Cref{prop:min-cnot}.

\section{CSS State Synthesis: The Implementation}
\label{sec:the-pipeline}

To realize the theoretical procedure outlined in \Cref{sec:the-theory}, we introduce a compilation pipeline we call SpiderCSS\@.
The pipeline is divided into two phases: the offline preparation of well-ordered \CAT states, and an active compiler that synthesizes the final state.
Given a CSS state, the active compiler finds a representation of an $X$-RREF state with optimized resource requirements, queries data structures for the appropriately sized well-ordered \CAT states, and finally orders the spiders for the extraction of a fault-tolerant circuit.
Since the extraction is ordered by the DAGs of the well-ordered \CAT states, no acausal loops are formed between spiders and the final circuit meets the minimum CNOT bound described in \Cref{prop:min-cnot}.
The rest of this section details the exact methods used in our implementation of the pipeline, which has been made available at~\cite{spiderstateGitHub}.

\subsection{Offline Generation of Well-Ordered \CAT States}

To generate a library of well-ordered \CAT states, we start with the marked graphs identified by SpiderCat.
For parameter regimes where SpiderCat provides provably optimal graphs, we use those to guarantee minimum CNOT counts.
In case SpiderCat is unable to provide an optimal \CAT state, we fall back to one that is optimal for a larger size, and remove the appropriate number of outputs.
While this means that for some sizes of $t$-fault-tolerant \CAT states it is impossible to obtain an optimal \CAT state, one can be obtained with a small overhead and the fault tolerance of the method is still guaranteed.

For each \CAT state we then identify a suitable spanning tree.
This needs to satisfy the requirements outlined in \Cref{subsec:cat-preliminaries}:
the root is an internal degree-3 vertex, the leaves are degree 2 in $\widetilde G$.
To construct spiders with certain desirable properties, we require the tree to satisfy some additional properties.

\begin{definition}[Minimum Diameter Spanning Forest]
  \label{def:mdsf}
  The \emph{diameter} of a tree is the maximum number of edges on a path between two vertices, with diameter zero for a single vertex.
  For a fixed number of connected components $k$, a \emph{minimum diameter spanning forest} (MDSF) of $G$ minimizes the maximum component diameter among spanning forests of $G$ with exactly $k$ components. For $k=1$, this gives a \emph{minimum diameter spanning tree} (MDST).
\end{definition}

The diameter of the spanning tree corresponds roughly to how deep the resulting extracted circuit is.
We are of course interested in minimizing depth to reduce environmental noise, but this will also help later in finding a suitable DAG for ensuring the well-ordered condition.
In particular, we would like flag qubits to exist for the shortest amount of time, which means we would like edges not selected for the tree (which correspond to flag qubit lines) to correspond to vertices placed at roughly the same distance from the root.
Having a lower diameter generally makes this happen more frequently, giving us more flexibility to arrange the final diagram into a well-ordered circuit.

In addition to SpiderCat's condition that leaves of the spanning tree are vertices of degree 2, so that each leaf carries a \CAT-state output~\cite[Definition~4.2]{spidercat2026}, we now additionally require tree leaves to be pairwise non-adjacent in $G$ to help with the ordering needed for flag delegation.

While there are polynomial-time algorithms for finding a minimum diameter spanning tree, our additional constraints make the problem challenging to solve algorithmically, and therefore, we use simulated annealing.
The optimization begins with a valid initial spanning forest generated via a greedy heuristic.
The objective function we choose to minimize using simulated annealing is the maximum diameter among all trees in the forest, supplemented by a large numerical penalty if any two leaves within the same tree share an edge in the original graph (enforcing the leaf adjacency constraint).
At each mutation step, we ``swap'' edges: re-inserting a previously removed edge from the forest and taking away another.
We strictly maintain the required number of components and the conditions on the root and leaves.
At each step we probabilistically accept worse states using the standard exponential distribution to escape local minima as the temperature cools.

Having selected a suitable spanning tree, we construct the directed graph as described in \Cref{subsec:well-ordered-cat-states}.
To resolve cycles in the directed graph, we iterate over its edges encoding the CNOT constraint, and check if removing it results in a DAG\@.
In the unlikely case we do not succeed, we rerun the simulated annealing process and try again with a different spanning tree.

\subsection{Optimization of the Generators of the CSS State}

As an input to our CSS state preparation, we have to give a parity matrix describing the stabilizer generators in RREF.
However, picking different pivot columns while doing Gaussian elimination can have an impact on the final CNOT count of the circuit.
For this reason we construct a cost function that estimates the total CNOT cost of the final circuit given a particular RREF matrix, and then we iterate over different possibilities of the RREF matrix to minimize the cost.

To calculate the cost of a particular circuit implementation, we count the number of flag qubits required to $t$-fault-tolerantly prepare a \CAT state of degree $n$ using SpiderCat, which we write as $f(n, t)$ (note the convention in~\Cref{rem:spidercat-convention}).
The CNOT cost to implement a single isolated spider of degree $n$ is then bounded by $C_{\text{spider}}(n, t) = n - 1 + 2f(n, t)$, as each flag requires 2 CNOTs and each of $n-1$ qubits needs a fan-out CNOT.

As explained in \Cref{sec:ZX-transformations}, a set of CSS generators in RREF corresponds to a bipartite phase-free ZX-diagram that we can then turn into a fault-tolerant circuit.
Given a parity matrix $M$ in RREF, we can estimate the cost of implementing this circuit as follows.
Let $c$ and $r$ denote the number of columns and rows of $M$.
The total CNOT cost of the resulting circuit is then:
\begin{equation}
  C_{\text{total}}(M, t) = \sum_{\substack{i=1 \\ w_i > 1}}^c 2f(w_i + 1, t) + \sum_{j=1}^r \left[ w_j - 1 + 2f(w_j, t) \right],
\end{equation}
where $w_i$ and $w_j$ are the column and row Hamming weights.
In the first sum we only count those columns with weight at least 2.
The CNOTs connecting the $Z$- and $X$-spiders of our ZX-diagram only need to be counted as fan-outs by either the $Z$- or $X$-spiders, which is why only one sum includes the fan-out contribution $w-1$.

Since finding the globally optimal basis over $\mathbb{F}_2$ can be computationally hard and to naturally handle rank-deficient parity matrices containing redundant stabilizers, we use a Monte Carlo optimization heuristic based on Gaussian elimination.
In each iteration, we randomly permute the columns of $M$ and compute its RREF over $\mathbb{F}_2$.
This procedure results in a maximal set of linearly independent rows while embedding an identity matrix into the pivot columns, guaranteeing the resulting matrix $M_{\text{new}}$ is of the required standard form.
The columns are then un-permuted to restore their original physical qubit mappings.
We evaluate $C_{\text{total}}(M_{\text{new}}, t)$ for each valid basis generated this way and retain the matrix that minimizes the total logical CNOT cost.

\subsection{Generation of the DAG}\label{subsec:ordering-external-edges}

Once a good parity matrix $M_{\text{best}}$ is selected, the compiler queries the offline library of well-ordered \CAT states to instantiate the necessary $Z$- and $X$-spiders.
Each of these spiders comes with a DAG denoting time-ordering of the spiders.

We could view the procedure now as described in \Cref{sec:the-theory}, where we have now successfully decomposed the ZX-diagram into 3-ary spiders and all that remains is to turn it into a circuit.
However, it is useful to instead view the procedure more dynamically, where we construct the desired diagram step by step by connecting together the initially disjoint \CAT states with the connectivity described by $M_{\text{best}}$.
Because each $Z$-type and $X$-type \CAT state is fully symmetric in its qubits, we have the freedom to choose which of the wires we connect together.

Our initial state is described by the following two data structures:
\begin{enumerate}
  \item The \textbf{Global Graph ($G_{\text{global}}$)}, representing the ZX-diagram we have constructed so far, initialized to be the disjoint union of all individual decomposed \CAT states of appropriate weight and type retrieved from the offline library.
  \item The \textbf{Global Dependency Graph ($D_{\text{global}}$)}, a DAG initially containing only the disjoint internal causal dependencies of each decomposed \CAT state.
\end{enumerate}
What remains then is to add edges to $G_{\text{global}}$ connecting the individual \CAT states according to the bipartite connectivity described in $M_{\text{best}}$ while updating the global dependency graph $D_{\text{global}}$ appropriately.

Since each individual \CAT state is already expanded into a specific internal ordering according to its DAG, arbitrarily connecting edges between \CAT states can introduce acausal cycles into the global dependency graph.
While digraphs with cycles can still be used for extraction, it requires additional ancillas and CNOTs to resolve these cycles, so we only allow DAG-based time orderings to match the lower bound.
To prevent the formation of acausal cycles and ensure a valid final DAG, we define \Cref{alg:dag-construction} based on selecting a strict total ordering of bipartite connections (which we refer to as external edges as they are external to individual spiders).
For the running example, \Cref{fig:building-G-and-D} shows steps of the construction of these data structures.

\begin{algorithm}[H]
Given an arbitrary strict total ordering of the external connections specified in $M_{\text{best}}$, proceed as follows:
\begin{enumerate}
    \item Iterate through each connection $(a, b)$ according to the total ordering, where $a$ and $b$ denote particular \CAT states.
    \item Identify the earliest (according to the DAG) unassigned free legs $u(a)$ and $u(b)$ on the \CAT states $a$ and $b$.
    \item In $G_{\mathrm{global}}$, insert an undirected edge connecting $u(a)$ and $u(b)$.
    \item For a directed graph $D$, let $P_D(v) = \{w \in V(D) \mid (w \to v) \in E(D)\}$ denote the immediate causal predecessors (the parents) of a vertex $v$.
    Then in the global dependency graph $D_{\mathrm{global}}$, insert directed edges from all $w \in P_{D_{\mathrm{global}}}(u(a))$ to $u(b)$, and conversely from all $x \in P_{D_{\mathrm{global}}}(u(b))$ to $u(a)$.
    These connections guarantee that all the predecessors of $u(a)$ and $u(b)$ are placed before the CNOT between $u(a)$ and $u(b)$ can be placed.
    \item Mark the legs $u(a)$ and $u(b)$ as assigned and proceed to the next connection.
\end{enumerate}
\caption{Finalizing $G_{\text{global}}$ and $D_{\text{global}}$}
\label{alg:dag-construction}
\end{algorithm}

\begin{lemma}
    \label{lem:dag-preservation}
    \Cref{alg:dag-construction} guarantees that $D_{\mathrm{global}}$ remains a DAG at every step.
\end{lemma}
\begin{proof}
    Suppose that prior to assigning a connection $(a,b)$, $D_{\mathrm{global}}$ is acyclic.
    Then because $u(a)$ and $u(b)$ are strictly the earliest available outputs in their respective local \CAT expansions, any connections going outside of the \CAT state must be from prior outputs that were already placed, which must have occurred before $u(a)$ and $u(b)$, preventing a closed loop.
\end{proof}

Notably, \Cref{lem:dag-preservation} has no restrictions on the ordering of external connections and still produces a valid DAG.
As the order in which external edges are processed changes the particular nodes $u(a)$ and $u(b)$ that are being connected, the structure of $D_{\mathrm{global}}$ also varies with different permutations of edges.
This provides us with a degree of freedom which can result in digraphs that translate to better depth or circuit width.
We exploit this freedom by developing different heuristics for the orderings, aiming to optimize for different properties of the resulting circuits.

We implemented a number of dynamic heuristics that pick an edge based on the state of $D_{\mathrm{global}}$ during the execution of \Cref{alg:dag-construction}:
\begin{itemize}
  \item \emph{Earliest Start} minimizes immediate depth by greedily prioritizing edges between spiders whose available legs appear earliest in the current schedule.
  \item Critical Path (\emph{Crit.\@ Path}) aims to minimize the overall circuit depth by prioritizing connections that lie on the longest remaining paths.
  \item \emph{Active Spider} seeks to minimize the maximum number of simultaneously active spiders (and thereby circuit width) by prioritizing connections that reuse already active spiders over those that would introduce new ones.
\end{itemize}
Finally, we also implemented a static heuristic that aims to optimize over ordering of external edges.
Sequence Distance (\emph{Seq.\@ Dist.\@}) uses simulated annealing to minimize the average distance between consecutive uses of the same spider, aiming to reduce qubit lifetime.

\subsection{Extracting the Final Circuit}

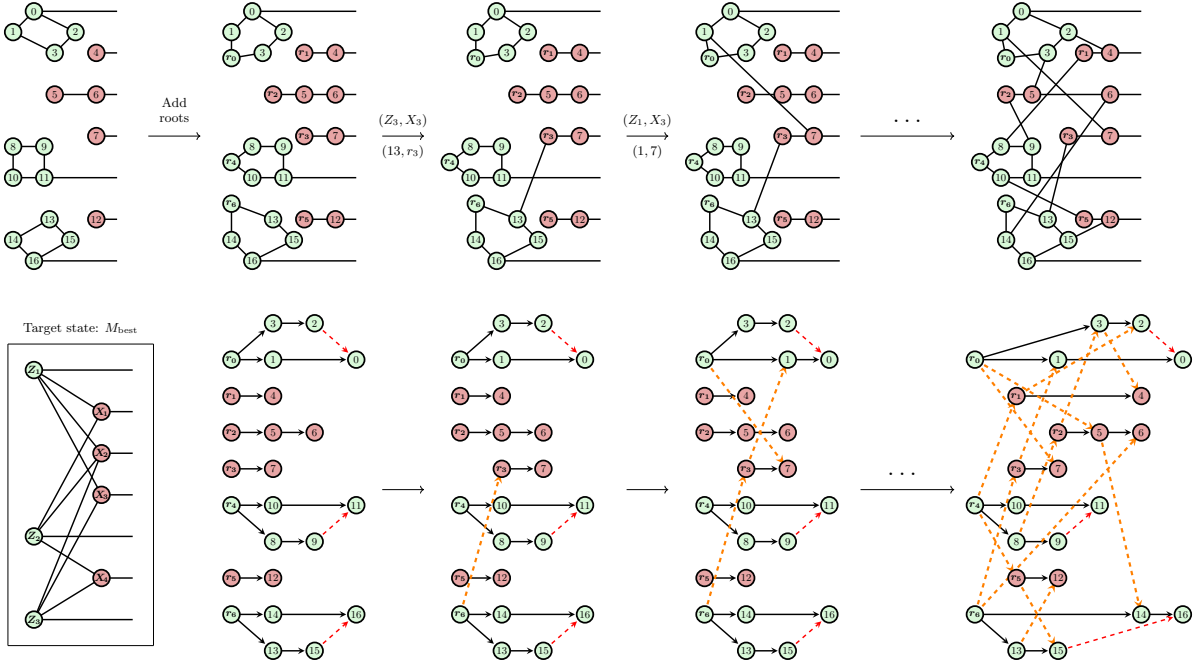
\begin{figure}[!b]
  \[
      \scalebox{0.55}{\input{building-G-and-D.tikz}}
  \]
  \caption{Building the data structures $G_{\text{global}}$ and $D_{\text{global}}$ according to the optimized matrix $M_{\text{best}}$. In the second row, each edge $u\rightarrow v$ annotates a causal dependency between nodes $u$ and $v$. That is, $u$ must be scheduled before $v$. There are three types of arrows: Solid black arrows denote dependencies inherited from the spanning tree produced by SpiderCat~\cite{spidercat2026}. Dashed red arrows indicate additional flag-ordering constraints within a \CAT gadget. Dashed orange arrows record additional constraints introduced when connecting different \CAT gadgets. We start from fault-tolerantly expanded spiders with missing external edges and gradually add edges.
  In this example, we first add an edge between the expanded spiders $Z_3$ and $X_3$ and therefore we pick $13$ and $r_3$ to connect in $G_{\text{global}}$.
  Then we add the edge $r_6 \to r_3$ into $D_{\text{global}}$ as $P_{D_{\mathrm{global}}}(u(r_3)) = \emptyset$ and $P_{D_{\mathrm{global}}}(u(13)) = \{r_6\}$.
  These additional edges are drawn dashed in orange.
  We proceed with $Z_1$ and $X_3$, adding the edge $(1, 7)$ to $G_{\text{global}}$ and also $r_3 \to 1$ and $r_0 \to 7$ to $D_{\text{global}}$, and continue until the graphs are fully populated.
  }
  \label{fig:building-G-and-D}
\end{figure}
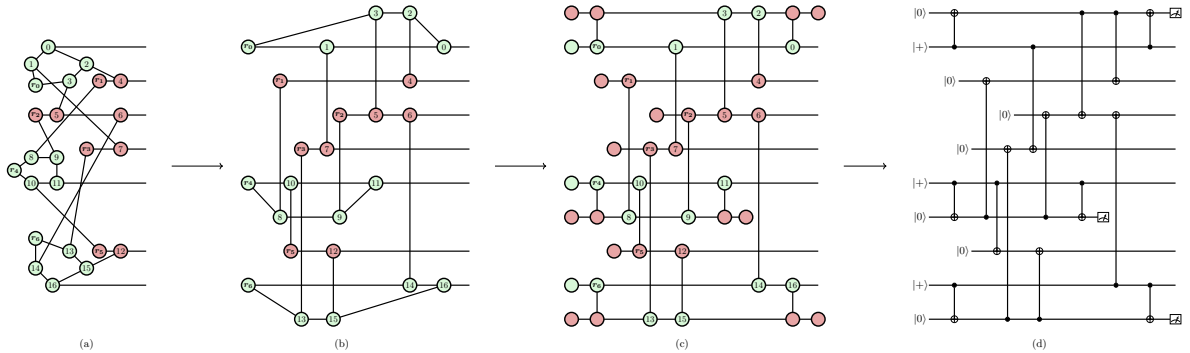
\begin{figure}[!b]
    \[
        \scalebox{0.45}{\input{full-extraction.tikz}}
    \]
    \caption{Fault-equivalent circuit extraction for preparing the logical $\ket{0}$ state of the Steane code. We first take $G_{\text{global}}$ and rearrange it according to the topological generation of $D_{\text{global}}$. $G_{\text{global}}$ and $D_{\text{global}}$ are shown in \Cref{fig:building-G-and-D}.
    Then insert identity spiders and unfuse state preparations and measurements so the the diagram is fully circuit like.
    Finally, we extract the circuit.}
    \label{fig:steane-full-extraction}
\end{figure}

The final phase of the pipeline turns the complete global connectivity graph into a fault-tolerant quantum circuit.
To do this, we take the final ZX-diagram $G_{\mathrm{global}}$ implementing our initial state $M$, and we rearrange the spiders on the plane according to $D_{\mathrm{global}}$.
In particular, we fix the primary qubit lines of all expanded spiders as the directed path starting from the roots and ending in the boundaries.
We rearrange $G_{\mathrm{global}}$ according to $D_{\mathrm{global}}$ by setting the X-coordinate of each spider $s$ according to which topological generation of $D_{\mathrm{global}}$ we can find $s$.
We set the Y-coordinate of all spiders on the primary qubit lines as the index of the corresponding qubit.
The remaining spiders get Y-coordinates based on the path they lie on.
After this, we apply the identity rule and the single spider unfusion rule to turn the diagram into a circuit-like diagram.
From there, a circuit can directly and fault-equivalently be extracted using the translation guide of \Cref{fig:translation-guide}.
As an example of this procedure, the extraction of the Steane code is shown in \Cref{fig:steane-full-extraction}.


\begin{remark}
    Our tree-based extraction procedure can be viewed in terms of measurement-circuit flow (MC flow), introduced for distance-preserving circuit extraction~\cite[Section~4]{Rodatz2026floquetifying}. MC flow builds on causal flow, which Danos and Kashefi introduced to ensure deterministic measurement-based quantum computation~\cite{DanosKashefi2006Determinism}. In the path formulation of causal flow, distinct paths cannot share vertices. MC flow relaxes this restriction: paths may share spiders but not edges, with ordering conditions that accommodate these shared vertices. Paths can therefore start or end at spiders on other paths, as in the branching and flag interactions above. In our construction, the spanning forest guides the data-qubit paths, with a choice at each branch of which continuation stays on the same qubit. The dependency DAG orders the interactions between paths. The compiler turns these choices into a circuit using the fault-equivalent extraction rules described above.
\end{remark}

\section{Benchmarking}
\label{sec:simulation}

\begin{figure}[b]
    \centering
    \begin{subfigure}[b]{0.49\linewidth}
        \includegraphics[width=\linewidth]{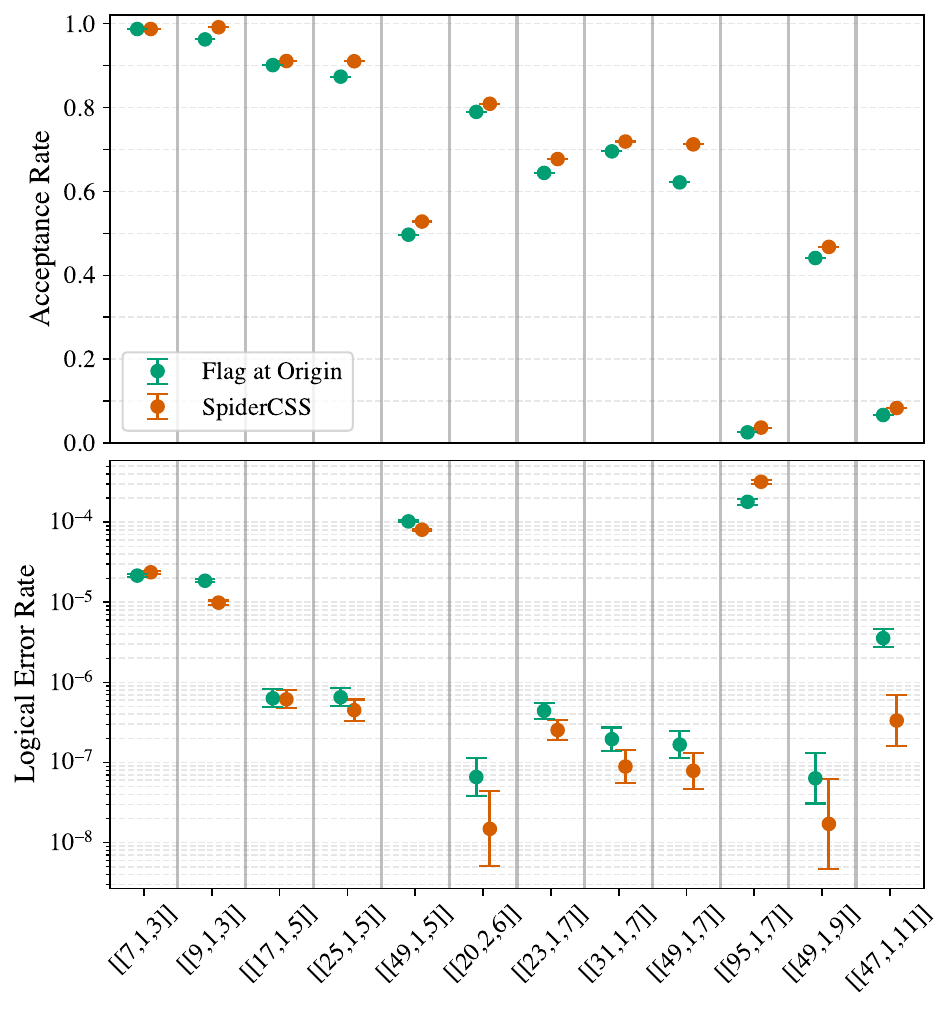}
        \caption{Acceptance rate (higher is better) and logical error rate (lower is better) of fault-tolerant CSS state preparation circuits by SpiderCSS (orange) and FaO (green) for the benchmarked codes. Each point shows the acceptance and logical error rate obtained from Monte Carlo simulation, with an error bar indicating $95\%$ Wilson confidence intervals.}
        \label{fig:LEU}
    \end{subfigure}
    \hfill
    \begin{subfigure}[b]{0.49\linewidth}
        \includegraphics[width=\linewidth]{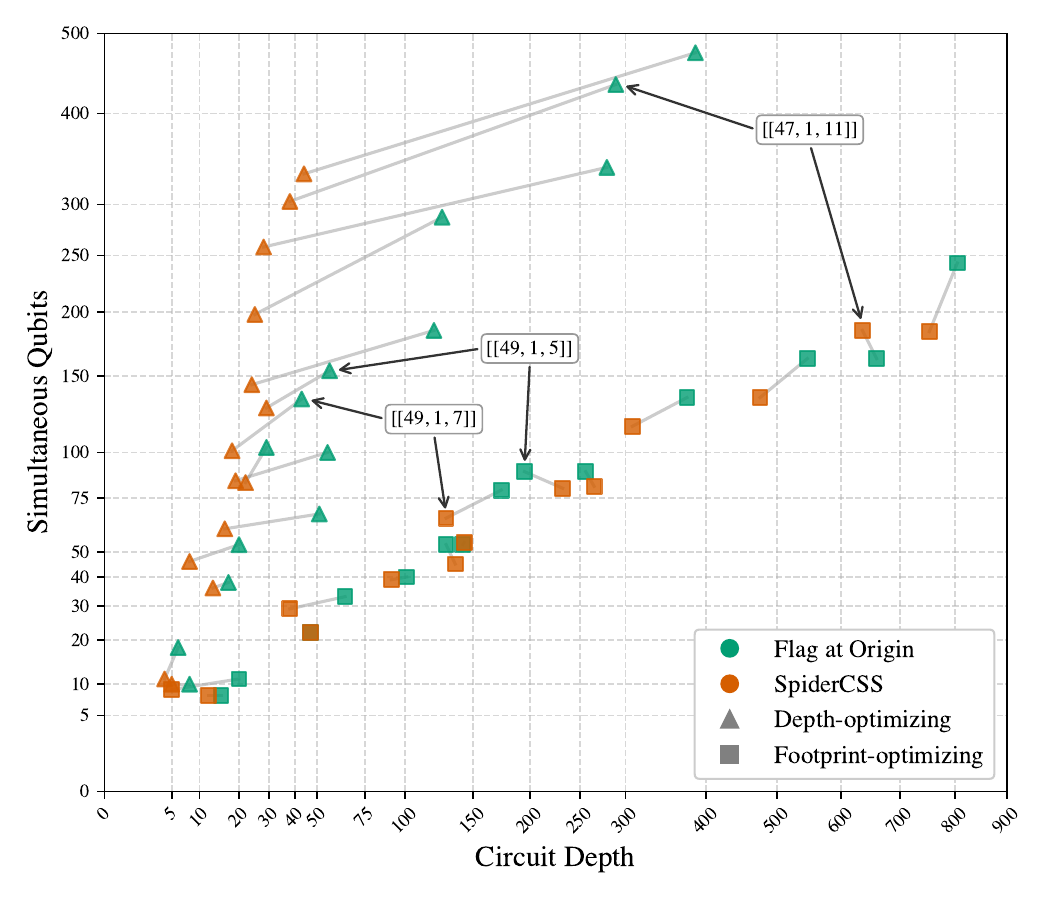}
        \caption{Circuit depth and peak number of simultaneously active qubits (lower is better in both cases) for fault-tolerant CSS-state preparation using SpiderCSS (orange) and FaO (green). A triangle denotes the strategy prioritising depth while a square denotes maximal qubit reuse, which prioritises lowering qubit count. For each code, a gray line connects FaO reference points to corresponding SpiderCSS data points, obtained by prioritising circuit depth or simultaneous qubit count. Several data points corresponding to the same code are marked to show the impact of the different qubit reuse strategies. }
        \label{subfig:depth-qubit-resources}
    \end{subfigure}
    \caption{A comparison of SpiderCSS and Flag-at-Origin (FaO) across all benchmarked CSS states in terms of logical error rates, acceptance rates, and space-time resource overhead.}
    \label{fig:comparison2}
\end{figure}

\input{sim_table}

To validate the performance of the circuits produced by SpiderCSS, we performed Monte Carlo simulations using Stim~\cite{Gidney2021stimfaststabilizer}.
In order to aid comparison, we adopt the noise model of Flag at Origin~\cite{forlivesi2025flag}:
the simulations were done under circuit-level depolarizing noise ($\pTwo = \pOne = \pSPAM = 10^{-3}$ and $\pMem = 10^{-5}$) with depolarizing memory error applied after each CNOT gate and data qubits measured perfectly after the state preparation is fully complete.

For each code we sample from a circuit with no qubit reuse applied, and for fair comparison, we apply an exact scheduler based on CP-SAT that finds the schedule of operations minimizing total qubit lifetime.
The results are averaged over $10^{9}$ samples for codes $d < 6$, and over $5\times10^{9}$ samples for codes $d \geq 6$. For each sample we check whether all the flag measurements got the correct outcome, and we discard the shot if not. The proportion of shots we keep after this post-selection we call the \emph{acceptance rate}.
Each accepted sample was decoded using code-capacity lookup tables and samples were discarded if their syndromes indicated $d / 2$ or more faults.
We re-benchmarked the circuits of FaO as well for a fair comparison under the exact same implementation of the noise model, sampling, and decoding.

For each code, we construct two circuits with different qubit reuse targets.
When the reuse target is simultaneous number of qubits, we apply the greedy qubit reuse algorithm of \citet{deCorss2023QubitReuse}, which can increase the depth of the circuits.
When the reuse target is depth, we perform a minimal greedy qubit reuse where we reuse free, measured qubits, and the depth remains unchanged.

The benchmark results are presented for each code in \Cref{tab:sim_results}.
\Cref{fig:LEU} shows the comparison between SpiderCSS and Flag at Origin for acceptance rate and logical error rate (the proportion of wrongly decoded faults amongst all shots that were accepted) while \Cref{subfig:depth-qubit-resources} shows for each code its circuit depth and maximum number of simultaneously active qubits during the circuit (i.e.\@ the width of the circuit) under the two reuse modes.
\Cref{fig:comparison} summarizes this data in a different way.

The results demonstrate that the SpiderCSS pipeline scales across a range of CSS codes, from small distance-$3$ examples to larger distance-$15$ codes.
On average, we reduce logical error rate by $33.7\%$, increase acceptance rate by $9.5\%$, and depending on the chosen qubit reuse strategy reduce depth by $13.1\%$ or $60.3\%$, and the number of simultaneous qubits by $8.9\%$ or $20.2\%$ compared to FaO circuits.

A clear pattern is the trade-off between qubit count and depth which is determined by the reuse strategy:
aggressive reuse reduces the peak number of simultaneous qubits but significantly increases depth, with no clear winner in error rates under the noise model used.
Compared to Flag at Origin~\cite{forlivesi2025flag}, the most striking difference is our reduced depth, and the slightly increased qubit count.
For most codes, the logical error rate of the SpiderCSS circuit is lower, while the acceptance rates of many codes are better.

\section{Conclusion}
\label{sec:conclusion}
We presented a new pipeline for constructing highly-optimized fault-tolerant state preparations of codewords in CSS codes.
This required combining together optimizations in finding suitable stabilizer generators, optimal fault-tolerant \CAT state gadgets, and improvements in routing the resulting ZX-diagram to a locally optimal circuit.
Compared to the previous state of the art for scalable fault-tolerant state preparation, Flag at Origin~\cite{forlivesi2025flag}, our method scales to larger codes and produces states using lower resource requirements (CNOT count, qubit count and depth) while producing better logical error rates and post-selection acceptance rates.

Importantly, these results are achieved without any expensive, exponentially scaling computation.
All parts of the pipeline run in polynomial time.
This is made possible by the Fault Tolerance by Construction framework~\cite{rodatz2025fault}, which guarantees fault tolerance without any expensive checking of fault propagation.

As our approach is fully modular, certain aspects of our method may be reused elsewhere.
For example, we can use well-ordered \CAT states to get a circuit for Shor-style syndrome measurement, which then can be optimized using \Cref{prop:cat-reduction-1,prop:cat-reduction-2} to obtain more efficient syndrome measurement circuits.
Another part of our work that may be reused in contexts other than CSS state preparation is our polynomial-time characterization of strictly safe hook errors in \Cref{apdx:safe-hook-errors}.

Additionally, improvements to different modular aspects of this work can also be applied directly.
Naturally, if larger and more fault-tolerant \CAT states can be generated by SpiderCat, they can be used directly in our framework to synthesize even larger states.
However, we have seen that as the size of the state increases, so does the rejection rate, requiring the development of decoding based on observed flags, which we leave for future work.
Furthermore, our current pipeline uses only strictly safe hook errors due to their relatively simple characterization and compositionality.
It is straightforward to also make use of safe hook errors that fall outside of this category.
However, finding a maximum set of safe hook errors is likely computationally hard.

The circuits produced by our pipeline can be used for Steane error correction~\cite{Steane1997ActiveStabilization} if the prepared state is the logical $\ket{0}$ or $\ket{+}$.
When the prepared state is a logical Bell state, our circuits can also be used for Knill error correction~\cite{knillQuantumComputing2005} .
These circuits may then be further optimized using fault-equivalent ZX-rewrites that may produce lower-overhead syndrome measurement protocols~\cite{rodatz2025fault,Poor2026ultralow}.

In upcoming work we generalize the results to work for arbitrary stabilizer codes.
This means the assumption of CSS noise is no longer valid, which invalidates the edge unidealization rule we used here.
Instead we use more involved rules which result in circuits that are larger but which work for arbitrary noise and thus support the preparation of general stabilizer states.
We leave extensions beyond the Clifford regime to future work.

Ultimately, we have demonstrated a modular approach to building fault-tolerant CSS-state preparation circuits at scale.
The generated circuits outperform other polynomial-time methods in logical error rate, acceptance rate, depth, and qubit count.
We believe that scalable methods that generate efficient circuits will be a key component in the realization of large-scale fault-tolerant quantum computing.

\paragraph*{Acknowledgements}
We appreciate helpful discussions with David Amaro.
The diagrams in this paper were typeset using TikZiT~\cite{tikz}.
JvdW wishes to acknowledge support from an NWO Veni grant.
ABK and BP are supported by the Engineering and Physical Sciences Research Council grant number EP/Z002230/1, ``(De)constructing quantum software (DeQS)''.
BR and RY acknowledge support from Simon Harrison via the Wolfson Harrison UK Research Council Quantum Foundation Scholarship.

\paragraph*{AI Disclosure}
All novel theoretical results in this paper were developed entirely by the authors, without AI assistance. The manuscript was written entirely by the authors, where AI assistance is limited to the activities described below. 
OpenAI GPT-5.6 and GPT-6, accessed through ChatGPT since June 2026, were used for editing suggestions. Claude Opus-5.5 was used for typo, grammar, and consistency checking. Gemini 3.1 Pro and 3.8 Flash accessed through Antigravity were used for assistance with the implementation of the codebase and the simulations~\cite{spiderstateGitHub}.
The authors take full responsibility for the correctness of the results.

\clearpage

\bibliography{spiderstate}

\appendix
\renewcommand{\theHequation}{\theequation}

\section{Safe Hook Errors}
\label{apdx:safe-hook-errors}

In fault-tolerant quantum computing, a hook error occurs when a single physical fault during circuit execution cascades into a multi-qubit error on the data block.
While arbitrary hook errors reduce fault tolerance by lowering the weight of uncorrectable errors, a subset of these, that we call \emph{strictly safe hook errors}, do not damage the fault tolerance.
While such hook errors spread to more qubits, we can multiply these by a stabilizer of the state to turn them back into single-qubit faults.

To make this more formal, suppose some operation in our circuit produces a multi-qubit fault $E$ on the output qubits. Let $\mathcal{S}$ be the stabilizer group of the stabilizer state we are preparing which contains both the stabilizers of the CSS code as well as the stabilizers of the prepared state.
We then say an error $E$ is \emph{strictly safe} if there exists a stabilizer $S\in \mathcal{S}$ such that $E\oplus S$ (denoting the product of the two Paulis) has support on just 1 qubit. 
Since Pauli errors compose linearly, combinations of strictly safe hook errors sum to an effective fault whose weight is bounded by the number of original faults, preserving fault tolerance under composition.

\begin{remark}
  Note that there are safe hook errors beyond the ones we describe here.
  In particular, certain hook errors can have a minimum weight representative of weight 2 on the boundary, but under perfect decoding it is sufficient to count this as weight 1.
  For example, diagonal hook errors of the surface code are well-known to be harmless~\cite{dennis2002TopologicalQuantumMemory}.
  The problem with this type of safe hook errors is that they are not composable: if we wanted a maximal set of safe hook errors every combination of these hook errors would individually need to be verified against a perfect decoder.
  We use the term ``strictly safe'' to emphasize the fact that the types of errors we work with trivially remain safe under composition.
\end{remark}

We want to use strictly safe hook errors in order to further optimize the decomposition of high-arity spiders into arity-3 spiders.
We could just decompose the whole high-arity spider using SpiderCat, which will guarantee that there are no hook errors, not even strictly safe ones.
Instead, we can check if unfusing the spider along some edge, which groups the outputs of the spider into two sets, would give a strictly safe hook error.
If we identify such a choice of subsets of the outputs, then we decompose the spider into these subsets first before applying SpiderCat to decompose the remaining high-arity spiders.

To see the benefit of this, observe in \Cref{fig:elementary-fe-rewrites} that to decompose an arity-4 spider fault-tolerantly we would usually require 4 arity-3 spiders that make a loop.
The naive decomposition in \Cref{lem:non-fe-rewrites} into just two arity-3 spiders is not fault-tolerant, because a single $X$-fault on the internal edge can spread to two faults on the output.
But suppose that these 2 resulting $X$-faults propagate to output wires where we have a weight-3 $X$-stabilizer.
We could then combine the fault with this stabilizer to reduce the hook error to just a weight 1 $X$-fault.
In that case we do not have to do the expensive cycle decomposition and just use the decomposition with two spiders instead.

\subsection{Finding Strictly Safe Hook Errors}
It remains to show how we find strictly safe hook errors.
Since our spiders always produce hook errors that are purely $X$- or $Z$-type we can just consider one type of fault.
For concreteness, let us assume we are dealing with a $Z$-spider, and hence $X$-type hook errors.
Let $Q$ be the set of output qubits of the state, let $G\subseteq Q$ denote the qubits of the spider that we want to decompose, and let $E\subseteq G$ be a candidate hook error produced by splitting the spider into subsets with wires acting on $E$ and $G\setminus E$, so $E$ denotes a set of $X$-faults on all the qubits of $E$.

The requirement for $E$ to be strictly safe 
implies that there exists some $X$-stabilizer $S \in \mathcal{S}$ and a physical $X$ error $e$ of weight $\le 1$ such that
\begin{equation}
    E = S \oplus e.
\end{equation}
Note that as $E\subseteq G$, we must have that $S\oplus e$ has no support on the qubits outside of $G$.
Let $I=Q\setminus G$ denote the remaining qubits. Then $E_I$, the fault restricted to the qubits in $I$, is trivial.
Hence we must have $S_I = e_I$.
That is, either the support of $S$ outside of $G$ is trivial, or it is equal to a weight 1 fault corresponding to $e$ if $e$ happens to have support outside of $G$.

Since $E$ was an arbitrary subset of $G$ we see that there is no further restriction on $S\oplus e$ as long as the condition $S_I = e_I$ is satisfied.
That is, once we have found a stabilizer $S$ and a weight-at-most-1 fault $e$ such that $S_I=e_I$, then $S\oplus e$ defines a subset $E\subseteq G$ that is a strictly safe hook error.

We can describe finding such $S$ and $e$ as a linear system of equations.
Let $M$ denote the binary generator matrix of the $X$-stabilizers and let $M_I$ be its restriction to the columns in $I$.
Any $X$-stabilizer in $\mathcal{S}$ can then be described as a linear combination of rows of $M$, which we represent by $x^TM$ for some vector $x$ describing the linear combination.
The condition $S_I=e_I$ then becomes the linear equation $x^T M_I = e_I$.

We can then first solve the homogeneous case where $e_I=0$ (meaning that the possible weight 1 fault has overlap with the qubits of the spider $G$).
This means finding those $x$ belonging to the left null-space of $M_I$, which requires Gaussian elimination of $M_I$, taking $\mathcal{O}(n^3)$ time.
For the $x$ in the null-space the space $V_0 = \{ x^T M \mid x^T M_I = 0 \}$ then describes the possible strictly safe hook errors on $G$ that exactly correspond to stabilizers.
For every $e\in G$, we can also take any $S\in V_0$ and consider $S\oplus e$ as a valid strictly safe hook error.

Other potential strictly safe hook errors correspond to where $e_I = 1$, i.e.~where the hook error has support 1 outside of $G$.
For this we consider for each qubit $q\in I$ the weight 1 fault $e$ having support at $q$, and we solve the system $x^T M_I = e_I$.
Since $M_I$ has already been reduced via Gaussian elimination, solving for a specific $e_I$ can be done cheaply via back-substitution taking $\mathcal{O}(n^2)$ time.
For a specific $e_I$, this gives a single non-homogeneous solution $x_n$ which can be combined with any $x_h$ that is a solution to the homogeneous system to get $S=(x_n+x_h)^T M$.
Then $E=S\oplus e$ is a strictly safe hook error.

Checking all $|I|$ possible outside faults takes $|I| \times \mathcal{O}(n^2) = \mathcal{O}(n^3)$ total time.
This covers all possible ways we can construct strictly safe hook errors, and we see that this can be done in time $\mathcal{O}(n^3)$.

\subsection{Universal vs. Partial Safety}

When a spider is split into two components along a strictly safe hook error, we refer to it as a partition. We will for instance denote decomposing a degree-6 spider into two pieces of size 3 as a $[3,3]$ partition.
When we are investigating which partitions correspond to strictly safe hook errors, there is a useful distinction that we can make.
We say a partition is \textbf{partially safe} when only particular settings are safe (i.e.\@ some legs must connect to particular spiders to be safe).
This property means that \Cref{alg:dag-construction} and \Cref{lem:dag-preservation} only hold if the absolute ordering of the external edges respects the order demanded by the partition.
For example in a $[3,3]$ partition, the first 3 edges must precede the second 3, or vice versa.
While this only places easily resolvable additional constraints on the absolute ordering, we still prefer when we have full freedom for optimization.
This motivates defining \textbf{universally safe} partitions, in which case any partition of the spider legs yields a safe hook error.

A spider admits universally safe $[k, n-k]$ partitions, if it admits all possible $\binom{|G|}{k}$ partitions of $k$ legs, which can be generated from the above linear equations.
If the spider is not universally safe, we generate all partially safe non-trivial hook errors.
From this, we also compute if a partition contains nested splits (e.g.\@ splitting a degree-8 spider into $[3, 2, 3]$) and save them.
From all possible partially safe hook errors, we select one that minimizes the smallest partition and then maximizes the number of partitions.

\end{document}








\title{CSSCat: Scalable Fault-Tolerant CSS State Preparation\\[0.4em]
  \normalfont\large Supporting definitions and proofs}

\keywords{QEC, FTQC, Fault Tolerance, State Preparation, Fault-tolerant State Preparation}

\maketitle

In what follows, we use $\semantics{C}$ to denote the matrix representation of the linear map implemented by a circuit or ZX diagram $C$. Pauli operators and semantic equalities are understood modulo global phase. For ZX diagrams, semantic equalities are understood up to a
nonzero scalar.

\section{Circuit-level noise}

\begin{definition}[Circuit-level Pauli noise model]
\label{def:notes-circuit-noise}
Let $C$ be an ideal circuit with a finite set $L$ of circuit locations: gates, state preparations, measurements, and specified idle steps. At each location $\ell$, specify a set $\mathcal E_\ell$ of allowed local errors, including $I$ for no fault, with the following insertion rules, as illustrated in \Cref{fig:notes-fault-locations}.
\begin{itemize}
  \item As shown in \Cref{fig:gate-fault}, at an $n$-qubit gate, allow any element of $\mathcal P_n$, immediately after the ideal gate. Note that a multi-qubit gate results in one fault location. 
  \item As shown in \Cref{fig:state-fault}, at a state preparation or idle step, allow Pauli errors on its outputs.
  \item As shown in \Cref{fig:measurement-fault}, at a nondestructive measurement, allow a Pauli error placed both immediately before and after the measurement, with the same Pauli on both sides. This Pauli pair flips the recorded outcome and is counted as one fault.
  \item As shown in \Cref{fig:measurement-fault2}, at a destructive measurement, allow an outcome flip. This is represented by an anticommuting Pauli immediately before the measurement.
\end{itemize}

A fault pattern chooses one allowed error at each location, where $\prod$ denotes the cartesian product of sets of allowed errors at all fault locations. Write $C^F$ for the resulting faulty circuit.
\[
  F=(F_\ell)_{\ell\in L}\in\mathcal F_C
  :=\prod_{\ell\in L}\mathcal E_\ell.
\]
\end{definition}

\begin{figure}[hbt]
  \captionsetup[subfigure]{justification=raggedright,singlelinecheck=false,margin=0pt,aboveskip=0pt}
  \begin{subfigure}[b]{0.24\textwidth}
    \centering
    \[\tikzfig{figures/gate-fault}\]
    \caption{Fault at a multi-qubit gate.}\label{fig:gate-fault}
  \end{subfigure}\hfill
  \begin{subfigure}[b]{0.24\textwidth}
    \centering
    \[\tikzfig{figures/state-fault}\]
    \caption{Fault at a state preparation.}\label{fig:state-fault}
  \end{subfigure}\hfill
  \begin{subfigure}[b]{0.24\textwidth}
    \centering
    \[\tikzfig{figures/measurement-fault}\]
    \caption{Fault at a nondestructive measurement.}\label{fig:measurement-fault}
  \end{subfigure}\hfill
  \begin{subfigure}[b]{0.2\textwidth}
    \centering
    \[\tikzfig{figures/measurement-fault2}\]
    \caption{Fault at a destructive measurement.}\label{fig:measurement-fault2}
  \end{subfigure}
  \caption{
    Fault locations in the circuit-level noise model. Fault-free components are drawn in purple and the inserted Pauli error $\epsilon$ in black.}
  \label{fig:notes-fault-locations}
\end{figure}

\begin{definition}[Circuit fault weight]
\label{def:notes-circuit-weight}
The number of faulty locations in $F$ is
\[
  n(F):=\bigl|\{\ell\in L:F_\ell\neq I\}\bigr|.
\]
Using the minimum-representative convention, define
\[
  \operatorname{wt}_C(F)
  :=\min_{\substack{G\in\mathcal F_C\\\semantics{C^G}=\semantics{C^F}}}n(G).
\]
\end{definition}

\begin{remark}[Locations and wires]
A wire identifies a qubit, whereas a circuit location identifies an operation occurrence. One location can involve several qubits, and one wire can participate in several locations. 
\end{remark}

\section{Fault equivalence}

For a postselected circuit $C$ with $\semantics{C}\neq0$, a fault $F$ is \emph{detectable} if $\semantics{C^F}=0$, and \emph{undetectable} otherwise. The circuit comparisons below assume the circuit-level Pauli noise model specified in \Cref{def:notes-circuit-noise}.

\begin{definition}[Pauli fault equivalence]
\label{def:pauli-circuit-fault-equivalence}
Let $C$ and $D$ be circuits with the same ideal semantics, each equipped
with a circuit-level Pauli noise model. They are \emph{Pauli
fault-equivalent} if every allowed undetectable fault $F$ on $C$ has an
allowed fault $G$ on $D$ satisfying
\begin{equation}
  \semantics{C^F}=\semantics{D^G},
  \qquad
  \operatorname{wt}_D(G)\leq\operatorname{wt}_C(F)
  \label{eq:matching-condition}
\end{equation}
and the same condition holds with $C$ and $D$ interchanged. We call \eqref{eq:matching-condition} the \emph{matching condition} from $C$ to $D$.
\end{definition}

\begin{definition}[CSS fault weight and fault equivalence]
\label{def:css-circuit-fault-equivalence}
For a Pauli fault pattern $F$ on a circuit $C$, let $F_X$ and $F_Z$ be its
pure $X$- and $Z$-type components at the original locations, so that
$F=F_XF_Z$. For $P\in\{X,Z\}$, let
$\operatorname{wt}_{C,P}(F_P)$ be the minimum number of faulty locations
among pure $P$-type fault patterns on $C$ with the same faulty semantics
as $F_P$. Define
\[
  \operatorname{wt}_{C,\mathrm{CSS}}(F)
  :=\max\bigl\{\operatorname{wt}_{C,X}(F_X),
               \operatorname{wt}_{C,Z}(F_Z)\bigr\}.
\]
Circuits $C$ and $D$ with the same ideal semantics are \emph{CSS
fault-equivalent} if every allowed undetectable fault $F$ on $C$ has an
allowed fault $G$ on $D$ such that
\begin{equation}
  \semantics{C^F}=\semantics{D^G},
  \qquad
  \operatorname{wt}_{D,\mathrm{CSS}}(G)
  \leq\operatorname{wt}_{C,\mathrm{CSS}}(F),
  \label{eq:css-matching-condition}
\end{equation}
and the same condition holds with $C$ and $D$ interchanged.
\end{definition}

When the circuit is clear from the context, we omit its subscript and write $\operatorname{wt}(F_P)$ for $\operatorname{wt}_{C,P}(F_P)$ and $\operatorname{wt}_{\mathrm{CSS}}(F)$ for $\operatorname{wt}_{C,\mathrm{CSS}}(F)$.

\begin{lemma}[Minimum representatives suffice]
\label{lem:minimum-representatives-suffice}
To verify Pauli fault equivalence of two circuits or diagrams with the same ideal
semantics, it suffices to consider the minimum-weight representative of each
class of undetectable faults, in both directions. Two faults belong to the
same class when they produce the same faulty semantics.
\end{lemma}

\begin{proof}
Let $F$ be an undetectable fault on $C$, and let $G$ be the chosen minimum
representative of its class. Then $\semantics{C^G}=\semantics{C^F}$ and
$\operatorname{wt}_C(G)=\operatorname{wt}_C(F)$. By assumption, $G$ has a
matching fault $H$ on $D$ with
\[
  \semantics{D^H}=\semantics{C^G}=\semantics{C^F},
  \qquad
  \operatorname{wt}_D(H)\leq\operatorname{wt}_C(G)
  =\operatorname{wt}_C(F).
\]
Thus, the matching condition also works for $F$. Interchanging $C$ and $D$ proves the converse.
\end{proof}

\begin{remark}[Separate component bounds]
CSS noise analysis is useful when separate bounds on $X$ and $Z$ errors suffice for the task, particularly when the output is protected by a CSS code. What matters is then whether an error is correctable, not whether its total Pauli weight is small. If a CSS code corrects up to $t$ errors of each type, ideal CSS recovery corrects $X^aZ^b$ whenever $|a|\leq t$ and $|b|\leq t$. Here $a,b$ are binary strings, $|a|$ denotes the Hamming weight of $a$, and $X^a=\bigotimes_jX^{a_j}$, with $Z^b$ defined similarly. This holds even when the two components act on different qubits, so that the total Pauli weight can reach $2t$: the code's $Z$-type checks detect $X^a$ and its $X$-type checks detect $Z^b$, independently of each other. For example, given a distance $d=3$ Steane code, $X_1Z_2$ is correctable: although its total weight is two, each component has weight one, which is bounded by $t=\lfloor \frac{d-1}{2}\rfloor =1$.
\end{remark}

\begin{lemma}[Pauli fault equivalence of non-entangling Clifford circuits]
\label{lem:non-entangling-pauli-fault-equivalence}
Let $C$ and $D$ be non-entangling Clifford circuits, that is, circuits composed of single-qubit Clifford gates, idle steps, state preparations $\lvert 0\rangle$ and $\lvert +\rangle$, and destructive postselected measurements $\langle 0\rvert$ and $\langle +\rvert$, without classical feed-forward, such that every output wire is the output of some location. Under the circuit-level Pauli noise model, if $C$ and $D$ have the same nonzero ideal semantics, then they are Pauli fault-equivalent.
\end{lemma}

\begin{proof}
Take an undetectable fault $F$ on $C$. By \Cref{lem:minimum-representatives-suffice}, we may assume that $F$ is a minimum representative of its class, so that $n(F)=\operatorname{wt}_C(F)$. Push its errors to the input or output boundary. Since propagating a Pauli through a single-qubit Clifford gate does not increase the number of insertions, neither does pushing $F$ to the boundary and combining errors there. On a closed preparation–measurement path, a Pauli fault that passes postselection leaves the semantics unchanged.

Hence, there exists a boundary fault pattern $B$ such that
\[
  \semantics{C^B}=\semantics{C^F},
  \qquad
  n(B)\leq n(F)=\operatorname{wt}_C(F).
\]
Realise the same boundary errors on $D$ as a fault $G$, moving each error along its wire to the nearest location, which exists by assumption. This is needed because, in the circuit-level Pauli noise model of \Cref{def:notes-circuit-noise}, errors can only be inserted at specific circuit locations. Since $\semantics{C}=\semantics{D}$,
\[
  \semantics{D^G}=\semantics{C^B}=\semantics{C^F},
  \qquad
  \operatorname{wt}_D(G)\leq n(G)\leq n(B)\leq\operatorname{wt}_C(F).
\]
This satisfies the matching condition from $C$ to $D$. Interchanging $C$ and
$D$ proves the converse.
\end{proof}

\begin{example}[Pauli equivalence need not imply CSS equivalence]
\label{ex:pauli-not-css-fault-equivalence}
Consider the two semantically equivalent circuits $C$ and $D$, as shown in \Cref{fig:pauli-not-css-circuits}. They are Pauli fault-equivalent by \Cref{lem:non-entangling-pauli-fault-equivalence}.

\begin{figure}[H]
  \captionsetup[subfigure]{justification=raggedright,singlelinecheck=false,margin=0pt,aboveskip=0pt}
  \begin{subfigure}[b]{0.38\textwidth}
    \centering
    \[\tikzfig{figures/fault-equivalence-example}\]
    \caption{The circuits $C$ and $D$, and $\semantics{C}=\semantics{D}$.}
    \label{fig:pauli-not-css-circuits}
  \end{subfigure}\hfill
  \begin{subfigure}[b]{0.6\textwidth}
    \centering
    \[\tikzfig{figures/fault-equivalence-fault-location}\]
    \caption{Fault locations of $C$ and $D$, $\epsilon_i\in \mathcal{P}_1$.}
    \label{fig:pauli-not-css-fault-locations}
  \end{subfigure}
  \caption{$\epsilon_i$ denotes the Pauli error inserted at one fault location. The dashed lines separating the time steps of $C$ and $D$.}
  \label{fig:pauli-not-css-fault-equivalence}
\end{figure}

In \Cref{fig:pauli-not-css-fault-locations}, we show that they are not CSS fault-equivalent by giving a counterexample. On $D$, take $\epsilon_3 = X$, $\epsilon_4 = Z$, and $\epsilon_5 = \epsilon_6 = I$ as the inserted faults. This pattern has CSS fault weight one, whereas producing the same behaviour on $C$ requires inserting a fault where $\epsilon_1 = \epsilon_2 = X$, which has CSS fault weight two. Since $\operatorname{wt}_{\mathrm{CSS}}(F) > \operatorname{wt}_{\mathrm{CSS}}(G)$, $C$ and $D$ are not CSS fault-equivalent.

\[\tikzfig{figures/fault-equivalence-fault-inserted}\]
\end{example}

\begin{example}
\label{ex:pauli-not-css-hadamard-pairs}
The same distinction holds if $D$ has a Hadamard pair on each wire, as shown in \Cref{fig:pauli-not-css-circuits2}. By \Cref{lem:non-entangling-pauli-fault-equivalence}, $C$ and $D$ are Pauli fault-equivalent.

\begin{figure}[H]
  \captionsetup[subfigure]{justification=raggedright,singlelinecheck=false,margin=0pt,aboveskip=0pt}
  \begin{subfigure}[b]{0.38\textwidth}
    \centering
    \[\tikzfig{figures/fault-equivalence-example2}\]
    \caption{The circuits $C$ and $D$, and $\semantics{C}=\semantics{D}$.}
    \label{fig:pauli-not-css-circuits2}
  \end{subfigure}\hfill
  \begin{subfigure}[b]{0.6\textwidth}
    \centering
    \[\tikzfig{figures/fault-equivalence-fault-location2}\]
    \caption{Fault locations of $C$ and $D$, $\epsilon_i\in \mathcal{P}_1$.}
    \label{fig:pauli-not-css-fault-locations2}
  \end{subfigure}
  \caption{$\epsilon_i$ denotes the Pauli error inserted at one fault location. The dashed lines separating the time steps of $C$ and $D$ are shown.}
  \label{fig:pauli-not-css-hadamard-pairs}
\end{figure}

In \Cref{fig:pauli-not-css-fault-locations2}, we show that they are not CSS fault-equivalent by giving a counterexample. On $D$, take $\epsilon_5 = X$, $\epsilon_4 = Z$, and $\epsilon_3 = \epsilon_6 = I$ as the inserted faults. This pattern has CSS fault weight one, whereas producing the same behaviour on $C$ requires inserting a fault where $\epsilon_1 = \epsilon_2 = X$, which has CSS fault weight two. Since $\operatorname{wt}_{\mathrm{CSS}}(F) > \operatorname{wt}_{\mathrm{CSS}}(G)$, $C$ and $D$ are not CSS fault-equivalent.

\[\tikzfig{figures/fault-equivalence-fault-inserted2}\]
\end{example}

\subsection{CNOT-based circuits}

Here, a \emph{CNOT-based circuit} uses CNOT gates, idle steps, $\lvert0\rangle$ and $\lvert+\rangle$ preparations, and destructive postselected measurements $\langle0\rvert$ and $\langle+\rvert$, without classical feed-forward. A \emph{circuit-like phase-free ZX diagram} is the ZX diagram of such a circuit, as in \Cref{prop:circuit-zx-css-equivalent}. Its faults are the edge flips discussed in \Cref{sec:edge-flip}. For a circuit $C$ with $n$ inputs, define its \emph{circuit state}
\begin{equation}
  \lvert C\rangle
  :=\sum_{x\in\{0,1\}^n}\lvert x\rangle\otimes C\lvert x\rangle.
  \label{eq:circuit-state}
\end{equation}

\begin{equation}
\tikzfig{figures/circuit-state}
  \label{eq:circuit-state-vis}
\end{equation}
As shown in \Cref{eq:circuit-state-vis}, $\lvert C\rangle$ is obtained by bending the $n$ input wires of $C$ into outputs. Each bent wire is the unnormalised Bell state $\lvert00\rangle+\lvert11\rangle$. Thus, to read $\lvert C\rangle$, start by preparing $n$ copies of $\lvert00\rangle+\lvert11\rangle$, and then apply $I\otimes C$, where $I$ acts on the first qubit of each Bell pair and $C$ acts on the second. Up to a permutation of qubits, $(\lvert00\rangle+\lvert11\rangle)^{\otimes n}=\sum_{x}\lvert x\rangle\otimes\lvert x\rangle$, so this recovers \Cref{eq:circuit-state}.

\begin{lemma}[CSS circuit states]
\label{lem:css-circuit-state}
If $C$ is a CNOT-based circuit with nonzero semantics, then
$\lvert C\rangle$ has a stabiliser generating set consisting of purely $X$- and $Z$-type operators.
\end{lemma}
\begin{proof}
To find stabiliser generators of $\lvert C\rangle$, it suffices to track how the stabiliser generators of the initial state propagate through $I\otimes C$: if $S$ stabilises $\lvert\psi\rangle$, then $USU^\dagger$ stabilises $U\lvert\psi\rangle$. By \Cref{eq:circuit-state-vis}, up to a permutation of qubit wires, the initial state is a tensor product of $n$ Bell pairs $\lvert00\rangle+\lvert11\rangle$, each with the pure-type stabiliser generators $X\otimes X$ and $Z\otimes Z$.

Recall that a CNOT with control $c$ and target $t$ propagates Paulis as
\begin{equation}
  X_c\mapsto X_cX_t,\qquad X_t\mapsto X_t,\qquad
  Z_c\mapsto Z_c,\qquad Z_t\mapsto Z_cZ_t.
  \label{eq:CNOT-prop}
\end{equation}
By \Cref{eq:CNOT-prop}, CNOT gates map $X$-type Paulis to $X$-type Paulis and $Z$-type Paulis to $Z$-type Paulis, so the type of each generator remains unchanged. Likewise, preparations and postselections also keep the stabiliser types unchanged.

Hence, after propagation through $I\otimes C$, every stabiliser generator is still of purely $X$- or $Z$-type. Since the resulting state is nonzero, it has the claimed generating set.
\end{proof}

\begin{lemma}[Map-state correspondence]
\label{lem:circuit-state-correspondence}
Let $C$ and $D$ be linear maps with the same input and output spaces.
Then $\semantics{C}=\semantics{D}$ if and only if $\lvert C\rangle=\lvert D\rangle$.
\end{lemma}

\begin{proof}
The forward implication follows from the definition given by \eqref{eq:circuit-state}. Conversely, suppose $\ket{C} = \ket{D}$. Applying $\langle x\rvert\otimes I$ extracts $C\lvert x\rangle$ from $\lvert C\rangle$, since
\[
\left(\bra{x}\otimes I\right)\ket{C}= \left(\bra{x}\otimes I\right)\left(\sum_{x\in\{0,1\}^n}\lvert x\rangle\otimes C\lvert x\rangle\right) = C\ket{x}.
\]

Equality of the circuit states therefore gives $C\lvert x\rangle=D\lvert x\rangle$ for every computational basis state, and hence $\semantics{C}=\semantics{D}$.
\end{proof}

The next lemma lets us analyse the $X$- and $Z$-components of a fault separately. Suppose $\semantics{C}=\semantics{D}$, and let $F=F_XF_Z$ be an undetectable fault on $C$. First, detectability is preserved when passing to components: the $X$- and $Z$-components of $F$ are undetectable, and a fault on $D$ is undetectable exactly when both of its components are. Second, faulty semantics separate into components. If we can find a pure $X$-type fault $G_X$ and a pure $Z$-type fault $G_Z$ on $D$ such that $\semantics{D^{G_X}}=\semantics{C^{F_X}}$ and $\semantics{D^{G_Z}}=\semantics{C^{F_Z}}$, then the fault $G=G_XG_Z$ satisfies $\semantics{D^G}=\semantics{C^F}$. Conversely, if a fault $G$ on $D$ has the same faulty semantics as $F$, then their $X$-components have the same faulty semantics, and so do their $Z$-components.

\begin{lemma}[Separation of the $X$ and $Z$ components]
\label{lem:css-component-matching}
Let $C$ and $D$ each be a CNOT-based circuit or a circuit-like phase-free
ZX diagram, with the same nonzero ideal semantics. Let $F$ be an undetectable Pauli fault pattern on $C$ and $G$ a Pauli fault pattern on $D$. Write $F=F_XF_Z$ and $G=G_XG_Z$, where the pure $X$- and $Z$-type components are taken at the original circuit locations or diagram edges, respectively. Then:
\begin{enumerate}
  \item $F_X$ and $F_Z$ are undetectable. Moreover, $G$ is undetectable
  if and only if both $G_X$ and $G_Z$ are undetectable.
  \item $\semantics{C^F}=\semantics{D^G}$ if and only if
  $\semantics{C^{F_X}}=\semantics{D^{G_X}}$ and
  $\semantics{C^{F_Z}}=\semantics{D^{G_Z}}$.
\end{enumerate}
\end{lemma}

\begin{proof}
By \Cref{eq:CNOT-prop}, propagation through a CNOT gate preserves Pauli type.
The identities in \Cref{fig:cnot-zx-pauli-propagation} show that propagation in a circuit-like phase-free ZX diagram also preserves Pauli type. In particular, a flip on the middle edge of a CNOT diagram equals a Pauli of the same type on its outputs. Hence, we may treat a circuit-like phase-free ZX diagram as its circuit representation.

\begin{figure}[H]
  \centering
  \tikzfig{figures/cnot-zx-pauli-propagation}
  \caption{Pauli propagation in the ZX representation of CNOT preserves
  $X$- and $Z$-type Pauli faults.}
  \label{fig:cnot-zx-pauli-propagation}
\end{figure}

To prove the first claim, we only need the postselections already present in $C$ and $D$. We call a deterministic postselection, that is, one that the ideal circuit passes with certainty, an \emph{acceptance check}. A $Z$-type acceptance check is a deterministic postselection on $\langle0\rvert$ and an $X$-type acceptance check is a deterministic postselection on $\langle+\rvert$. A Pauli fault is detectable exactly when it flips the outcome of an acceptance check. Since propagation preserves Pauli type, an $X$-type fault can flip only $Z$-type acceptance checks, and a $Z$-type fault can flip only $X$-type acceptance checks.

In particular, adding a $Z$-type component to a an $X$-type fault does not change the outcomes of the $Z$-type acceptance checks. Suppose, towards a contradiction that $F_X$ is detectable. Then so is $F$. Therefore, when $F$ is undetectable, $F_X$ must also be undetectable. The same argument applies to $F_Z$ and to the components of $G$. Conversely, suppose $G_X$ and $G_Z$ are both undetectable. Then $G_X$ flips no $Z$-type acceptance check and $G_Z$ flips no $X$-type acceptance check. Since adding the other component does not change these outcomes, $G=G_XG_Z$ flips no acceptance check and is undetectable.

To prove the second claim, we turn circuits into states as in \Cref{eq:circuit-state-vis}. By \Cref{lem:circuit-state-correspondence}, comparing the faulty circuit maps is equivalent to comparing their circuit states. By \Cref{lem:css-circuit-state} and the assumption that $\semantics{C}=\semantics{D}$, $\ket{C}=\ket{D}$ is a CSS state. Let $S_X$ and $S_Z$ be its purely $X$- and $Z$-type stabiliser generators. Track these generators through a faulty circuit as in the proof of \Cref{lem:css-circuit-state}. Since propagation preserves Pauli type, $X$-type faults flip only the signs of $S_Z$, and $Z$-type faults only those of $S_X$. Thus, the faulty circuit state of an undetectable fault is stabilised by $S_X$ and $S_Z$ up to signs. Moreover, an undetectable Pauli fault does not change the norm of the circuit state. Hence, two states are equal exactly when all their generator signs match.

\begin{description}
  \item[$(\Rightarrow)$] Suppose $\semantics{C^F}=\semantics{D^G}$. Since $F$ is undetectable, both sides are nonzero. This implies that $G$ and all four components $F_X$, $F_Z$, $G_X$, and $G_Z$ are undetectable. By \Cref{lem:circuit-state-correspondence}, $\ket{C^F}=\ket{D^G}$. Thus, they have the same signs on $S_X$ and $S_Z$. Removing $F_Z$ and $G_Z$ preserves the $Z$-generator signs, so $\ket{C^{F_X}}$ and $\ket{D^{G_X}}$ have identical $Z$-generator signs. They also have identical $X$-generator signs, since $X$-type faults commute with $S_X$. This implies that $\ket{C^{F_X}} = \ket{D^{G_X}}$. By \Cref{lem:circuit-state-correspondence}, $\semantics{C^{F_X}} = \semantics{D^{G_X}}$. Interchanging $X$ and $Z$ gives the corresponding equality for the $Z$ components.
  \item[$(\Leftarrow)$] Suppose $\semantics{C^{F_X}}=\semantics{D^{G_X}}$ and $\semantics{C^{F_Z}}=\semantics{D^{G_Z}}$. Since $F_X$ and $F_Z$ are undetectable, so are $G_X$ and $G_Z$. This implies that $G=G_XG_Z$ is undetectable. By \Cref{lem:circuit-state-correspondence}, $\ket{C^{F_X}} = \ket{D^{G_X}}$. This implies that the two states have matching $Z$-generator signs. Since adding the respective $Z$ components preserves these signs, $\ket{C^F}$ and $\ket{D^G}$ have matching $Z$-generator signs. Likewise, $\ket{C^F}$ and $\ket{D^G}$ have matching $X$-generator signs. This implies that $\ket{C^F} = \ket{D^G}$. By \Cref{lem:circuit-state-correspondence}, $\semantics{C^F} = \semantics{D^G}$.
\end{description}
\end{proof}

\begin{lemma}[Comparison of Pauli and CSS fault weights]
\label{lem:css-pauli-weight-comparison}
For a CNOT-based circuit $C$ with nonzero ideal semantics and an
undetectable Pauli fault pattern $F$,
\[
  \operatorname{wt}_{C,\mathrm{CSS}}(F)
  \leq\operatorname{wt}_C(F)
  \leq2\operatorname{wt}_{C,\mathrm{CSS}}(F).
\]
\end{lemma}

\begin{proof}
Choose a minimum-weight Pauli representative $G$ of $F$.
Applying \Cref{lem:css-component-matching} with $D=C$, we find that $G_X$ and $G_Z$ represent $F_X$ and $F_Z$, respectively.
Thus, for $P\in\{X,Z\}$,
\[
  \operatorname{wt}_{C,P}(F_P)
  \leq n(G_P)\leq n(G)=\operatorname{wt}_C(F).
\]
By \cref{def:css-circuit-fault-equivalence}, $\operatorname{wt}_{C,\mathrm{CSS}}(F) =\max\bigl\{\operatorname{wt}_{C,P}(F_P);P\in \{X,Z\}\bigr\}\leq \operatorname{wt}_C(F)$, which proves the first inequality.

Conversely, choose minimum-weight pure representatives $G_X$ and $G_Z$ of $F_X$ and $F_Z$, respectively. Thus, $\semantics{C^{G_X}} = \semantics{C^{F_X}}$ and $\semantics{C^{G_Z}} = \semantics{C^{F_Z}}$. Moreover, $n(G_X) = \operatorname{wt}(F_X)$ and $n(G_Z) = \operatorname{wt}(F_Z)$. By \Cref{lem:css-component-matching}, their product
$G=G_XG_Z$ has the same faulty semantics as $F$. By \cref{def:css-circuit-fault-equivalence},

\begin{align*}
  \operatorname{wt}_C(F)
  &\leq n(G_X)+n(G_Z)=\operatorname{wt}(F_X)+\operatorname{wt}(F_Z)\leq 2\max\{\operatorname{wt}(F_X),\operatorname{wt}(F_Z)\} =2\operatorname{wt}_{C,\mathrm{CSS}}(F).
\end{align*}
\end{proof}


\begin{proposition}[Pauli fault equivalence implies CSS fault equivalence]
\label{prop:pauli-implies-css-fault-equivalence}
Let $C$ and $D$ be circuits composed of CNOT gates, idle steps,
preparations of $\lvert0\rangle$ and $\lvert+\rangle$, and destructive
postselected measurements $\langle0\rvert$ and $\langle+\rvert$, without
classical feed-forward. Nondestructive measurements of purely $X$- or
$Z$-type Pauli operators are allowed through implementations using
ancillas and these operations. Under the circuit-level noise model, if
$C$ and $D$ have the same nonzero ideal semantics and are Pauli
fault-equivalent, then they are CSS fault-equivalent.
\end{proposition}

\begin{proof}
We first prove the matching condition from $C$ to $D$. Let $F$ be an
undetectable Pauli fault on $C$, and write $F=F_XF_Z$, with
the components taken at the original fault locations. By
\Cref{lem:css-component-matching}, both components are undetectable.
Pauli fault equivalence means that there exists a Pauli fault $H$ on $D$ such that
\[
  \semantics{D^H}=\semantics{C^{F_X}},
  \qquad \operatorname{wt}(H)\leq\operatorname{wt}(F_X).
\]

Since $F_X$ is undetectable, $H$ is undetectable. Applying the forward direction of the second part of \Cref{lem:css-component-matching} to this equality gives

\[
  \semantics{D^{H_X}}
  =\semantics{C^{(F_X)_X}}
  =\semantics{C^{F_X}}.
\]

Let $G_X=H_X$. Since
$\operatorname{wt}_{\mathrm{CSS}}(H)=\max\{\operatorname{wt}(G_X),\operatorname{wt}(H_Z)\}$
and $H$ is undetectable, \Cref{lem:css-pauli-weight-comparison} gives
\[
  \operatorname{wt}(G_X)
  \leq\operatorname{wt}_{\mathrm{CSS}}(H)
  \leq\operatorname{wt}(H)
  \leq\operatorname{wt}(F_X).
\]

Therefore, there exists a purely $X$-type fault $G_X$ on $D$ satisfying $\semantics{D^{G_X}}=\semantics{C^{F_X}}$ and $\operatorname{wt}(G_X)\leq\operatorname{wt}(F_X)$. Similarly, there is a purely $Z$-type fault $G_Z$ on $D$ satisfying $\semantics{D^{G_Z}}=\semantics{C^{F_Z}}$ and $\operatorname{wt}(G_Z)\leq\operatorname{wt}(F_Z)$.

Let $G=G_XG_Z$. The reverse direction of the second part of \Cref{lem:css-component-matching} yields $\semantics{D^G}=\semantics{C^F}$, and
\begin{align*}
  \operatorname{wt}_{\mathrm{CSS}}(G)=\max\{\operatorname{wt}(G_X),\operatorname{wt}(G_Z)\}\leq\max\{\operatorname{wt}(F_X),\operatorname{wt}(F_Z)\}=\operatorname{wt}_{\mathrm{CSS}}(F).
\end{align*}
Interchanging $C$ and $D$ proves the converse matching condition.
\end{proof}

\section{Edge-flip noise on ZX diagrams}
\label{sec:edge-flip}

\begin{definition}[Pauli edge-flip noise model]
\label{def:pauli-edge-flip-noise}
Each edge of a ZX diagram $D$ is a fault location. A fault pattern $F$ assigns
a Pauli operator $F_e\in\{I,X,Y,Z\}$ to each edge $e$, where $I$ means no
fault. The faulty diagram $D^F$ is obtained by inserting these operators on
the corresponding edges. The number of flipped edges is
\[
  n(F):=\bigl|\{e\in E(D):F_e\neq I\}\bigr|.
\]
Using the minimum-representative convention, define the fault weight by
\[
  \operatorname{wt}_D(F)
  :=\min_{\semantics{D^G}=\semantics{D^F}} n(G),
\]
where $G$ ranges over Pauli fault patterns on $D$. Thus, $\operatorname{wt}_D(F)$ is the minimum number of
edge flips needed to produce the same faulty semantics as $F$.
\end{definition}

\begin{definition}[CSS edge-flip noise model]
\label{def:css-edge-flip-noise}
Each edge of a ZX diagram $D$ is a fault location. For a fixed
$P\in\{X,Z\}$, a pure $P$-type fault pattern $F$ assigns $F_e\in\{I,P\}$
to each edge $e$, where $I$ means no fault. Its number of flipped edges
$n(F)$ is as in \Cref{def:pauli-edge-flip-noise}. Define its fault weight by
\[
  \operatorname{wt}_{D,P}(F)
  :=\min_{\substack{G_e\in\{I,P\}\ \text{for all }e\in E(D)\\
                   \semantics{D^G}=\semantics{D^F}}} n(G).
\]
Since $P$ is fixed throughout the minimisation, one type of fault is counted at a time: for
$P=X$, $n(G)$ counts only $X$-type edge flips.

For a general Pauli fault pattern $F$, let $F_X$ and $F_Z$ be its pure $X$-
and pure $Z$-type components on the original edges, obtained using $Y=iXZ$.
Thus, $F=F_XF_Z$. Define its \emph{CSS fault weight} by
\[
  \operatorname{wt}_{D,\mathrm{CSS}}(F)
  :=\max\bigl\{\operatorname{wt}_{D,X}(F_X),
               \operatorname{wt}_{D,Z}(F_Z)\bigr\}.
\]
\end{definition}

Pauli and CSS fault equivalence between a circuit and a ZX diagram are
defined as in \Cref{def:pauli-circuit-fault-equivalence} and
\Cref{def:css-circuit-fault-equivalence}, using these noise models on the
diagram.

\begin{lemma}
\label{lem:cnot-zx-not-pauli-equivalent}
Let $C$ be a two-qubit circuit consisting of a single CNOT gate and let $D$ be its
ZX diagram. Consider $C$ under the circuit-level Pauli noise model and $D$
under the Pauli edge-flip noise model. Then $C$ and $D$ are not fault-equivalent under their respective Pauli noise models.
\[
\tikzfig{figures/CNOT-map-not-FE}
\]
\end{lemma}

\begin{proof}
As shown in \Cref{fig:cnot-fault-locations}, there is only one fault location in $C$, where a fault $\epsilon_1\in\mathcal P_2$
is inserted after the CNOT gate. There are five fault locations in $D$, each corresponds to one of its five edges. For $2\leq j\leq6$, each edge flip $\epsilon_j\in\mathcal P_1$. 

\begin{figure}[H]
  \centering
  \tikzfig{figures/cnot-fault-locations}
  \caption{Fault locations of the CNOT circuit $C$ and its ZX diagram $D$.
  The fault $\epsilon_1\in\mathcal P_2$ acts after the CNOT, and each
  $\epsilon_j\in\mathcal P_1$, for $j=2,\dots,6$, acts on one edge of $D$.}
  \label{fig:cnot-fault-locations}
\end{figure}

Since there are no measurements in $C$ nor $D$, every fault on $C$ or $D$ is undetectable. To show $C$ is not fault-bounded by $D$, it suffices to find a fault on $C$ that cannot be matched by any fault on $D$ with the same or lower weight. Let $F = X_1Z_2$. Then $\operatorname{wt}(F) = 1$. We show that realising the same faulty semantics in $D$ requires a fault of weight greater than one.

To start, consider all fault patterns $G$ on $D$ such that $n(G)= 1$. Its effect on the outputs can be determined by propagating the single edge flip through the ZX diagram. By the identities in \Cref{fig:cnot-zx-pauli-propagation},

\begin{alignat*}{3}
  \text{control input}\;\; (\epsilon_2):\quad X &\mapsto X_1X_2, &\qquad Z &\mapsto Z_1, &\qquad Y &\mapsto Y_1X_2,\\
  \text{target input} \;\;(\epsilon_3):\quad X &\mapsto X_2, &\qquad Z &\mapsto Z_1Z_2, &\qquad Y &\mapsto Z_1Y_2,\\
  \text{middle} \;\;(\epsilon_4):\quad X &\mapsto X_2, &\qquad Z &\mapsto Z_1, &\qquad Y &\mapsto Z_1X_2,\\
  \text{control output} \;\;(\epsilon_5):\quad X &\mapsto X_1, &\qquad Z &\mapsto Z_1, &\qquad Y &\mapsto Y_1,\\
  \text{target output} \;\;(\epsilon_6):\quad X &\mapsto X_2, &\qquad Z &\mapsto Z_2, &\qquad Y &\mapsto Y_2.
\end{alignat*}

Since $X_1Z_2$ does not appear in the list above, no fault $G$ on $D$ with $n(G)\leq1$ satisfies $\semantics{D^G}=\semantics{C^F}$, so the matching condition from $C$ to $D$ fails. Therefore, $C$ and $D$ are not fault-equivalent under their respective Pauli noise models.
\end{proof}

\begin{remark}
$D$ is fault-bounded by $C$. Take any fault $G$ on $D$. Using the identities
  above, push all inserted Paulis to the outputs and combine them into a
  two-qubit Pauli $P$. The circuit-level noise model allows a faulty CNOT to
  introduce any two-qubit Pauli on its outputs, and this is counted as a
  weight-one fault. Hence, when $P\neq I$, the fault $F=P$ on $C$ satisfies
  $\semantics{D^G}=\semantics{C^F}$ and
  $\operatorname{wt}(F)=1\leq\operatorname{wt}(G)$, since $G$ acts
  nontrivially. When $P=I$, choose $F=I$; then
  $\semantics{D^G}=\semantics{C^F}$ and
  $\operatorname{wt}(F)=0=\operatorname{wt}(G)$.
\end{remark}

\begin{figure}[H]
  \centering
  \tikzfig{figures/circuit-zx-map}
  \caption{A local translation between circuit-like Pauli ZX diagrams
  and CNOT-based circuits.}
  \label{fig:circuit-zx-map}
\end{figure}

\begin{proposition}
\label{prop:circuit-zx-css-equivalent}
Let $D$ be a circuit-like phase-free ZX diagram, and let $C$ be its circuit
representation, obtained by mapping each CNOT diagram to a CNOT, each
one-legged $Z$- and $X$-spider on an input to $\lvert+\rangle$ and
$\lvert0\rangle$, and each one-legged $Z$- and $X$-spider on an output to
$\langle+\rvert$ and $\langle0\rvert$, respectively, as shown in
\Cref{fig:circuit-zx-map}.
Consider $C$ under the circuit-level CSS noise model and $D$ under the CSS
edge-flip noise model. Then $C$ and $D$ are CSS fault-equivalent. Here, a
bare circuit wire is counted as one noisy location, and the measurements in
$C$ are destructive.
\end{proposition}

\begin{remark}[Intuition]
In each direction, we take minimum-weight representatives of the $X$- and
$Z$-components of a fault, translate each component separately by case
distinction on preparations, destructive measurements, idle wires, and
CNOTs, and preserve its faulty semantics without increasing the fault count.
Recombining the components with \Cref{lem:css-component-matching} then
preserves the faulty semantics without increasing the CSS weight. In the
direction from $D$ to $C$, the edge flips at a CNOT diagram are pushed to its
outputs and combined into at most one gate fault.
\figtodo{notes p.~25: the local pieces of both case distinctions, including
$C^F$ for a CNOT and $D^G$ for its ZX diagram.}
\end{remark}

\begin{proof}
\Cref{fig:fault-map} shows how single faults at preparations, measurements,
and idle wires translate into edge flips; these are Cases~1--3 below.

\begin{figure}[H]
  \centering
  \tikzfig{figures/fault-map}
  \caption{Local translation of single $X$ and $Z$ faults at preparations,
  destructive measurements, and idle wires into edge flips. Each fault $F$
  with $n(F)=1$ corresponds to an edge flip $G$ with $n(G)=1$.}
  \label{fig:fault-map}
\end{figure}

\begin{description}
  \item[(From $C$ to $D$)] Let $F$ be an undetectable fault on $C$. Write
  $F=F_XF_Z$ at the original fault locations. Replace $F_X$ and $F_Z$ by
  minimum-weight representatives such that
  $n(F_P)=\operatorname{wt}_{C,P}(F_P)$ for $P\in\{X,Z\}$. Translate each
  component separately into edge flips on $D$, by case distinction on the
  local circuit-to-ZX translation.
  \begin{description}
    \item[Case 1 (preparation)] For a preparation fault $F\in\mathcal P_1$,
    let $G=F$ on the corresponding output edge in $D$. This preserves the
    faulty semantics and uses one edge flip per preparation fault, so
    $n(G)=n(F)$.
    \item[Case 2 (measurement)] Represent a measurement outcome flip by
    $F=Z$ for an $X$-basis measurement, or $F=X$ for a $Z$-basis
    measurement. Let $G=F$ on the corresponding incoming edge in $D$. This
    preserves the faulty semantics and uses one edge flip per measurement
    fault, so $n(G)=n(F)$.
    \item[Case 3 (idle wire)] For an idling fault $F\in\mathcal P_1$, let
    $G=F$ on the corresponding wire edge in $D$. This preserves the faulty
    semantics and uses one edge flip per idling fault, so $n(G)=n(F)$.
    \item[Case 4 (CNOT)] Let $F\in\mathcal P_2$ be a CNOT fault. If $F$ acts
    nontrivially on one qubit, let $G=F$ on the corresponding output edge in
    $D$. If $F$ acts nontrivially on both qubits, then $F$ is $X_1X_2$ or
    $Z_1Z_2$, which is represented by $X$ on the control input edge or $Z$
    on the target input edge, respectively. In both cases,
    $\semantics{C^F}=\semantics{D^G}$ and $n(G)=n(F)$.
    \figtodo{notes p.~23: the four single-qubit CNOT faults on output edges,
    and $X_1X_2$, $Z_1Z_2$ on input edges.}
  \end{description}
  Therefore, each faulty circuit location requires at most one edge flip per
  $X$- and $Z$-component. When several locations map to the same edge, their
  flips combine, which does not increase the count. Hence, for each
  $P\in\{X,Z\}$, $\semantics{D^{G_P}}=\semantics{C^{F_P}}$ and
  \[
    \operatorname{wt}_{D,P}(G_P)\leq n(G_P)\leq n(F_P)
    =\operatorname{wt}_{C,P}(F_P).
  \]
  Take $G=G_XG_Z$. By \Cref{lem:css-component-matching},
  $\semantics{D^G}=\semantics{C^F}$. By \Cref{def:css-edge-flip-noise},
  \[
    \operatorname{wt}_{D,\mathrm{CSS}}(G)
    =\max_{P\in\{X,Z\}}\operatorname{wt}_{D,P}(G_P)
    \leq\max_{P\in\{X,Z\}}\operatorname{wt}_{C,P}(F_P)
    =\operatorname{wt}_{C,\mathrm{CSS}}(F).
  \]
  \item[(From $D$ to $C$)] We show similarly that each edge flip requires at
  most one circuit location per $X$- and $Z$-component. Let $G$ be an
  undetectable fault on $D$. Write $G=G_XG_Z$ at the original edges. Replace
  $G_X$ and $G_Z$ by minimum-weight representatives such that
  $n(G_P)=\operatorname{wt}_{D,P}(G_P)$ for $P\in\{X,Z\}$. Translate each
  component separately into a circuit fault on $C$, by case distinction on
  the local ZX-to-circuit translation.
  \begin{description}
    \item[Cases 1--3] When a faulty edge lies immediately after a state
    preparation, immediately before a destructive measurement, or along an
    idle wire, there is a circuit fault producing the same faulty semantics
    with no greater fault count.
    \item[Case 4 (CNOT)] By the propagation identities in the proof of
    \Cref{lem:cnot-zx-not-pauli-equivalent}, any $P$-type flips on the five
    edges of a CNOT diagram can be pushed to its outputs. If the resulting
    $P$-type Pauli is nontrivial, it is a weight-one CNOT fault.
    \figtodo{notes p.~24: $D^G$ for a CNOT diagram with flips
    $\mathcal E_2,\dots,\mathcal E_6$ on its five edges.}
  \end{description}
  We translate each flipped edge only once. Hence, for each
  $P\in\{X,Z\}$, the translated fault $F_P$ satisfies
  $\semantics{C^{F_P}}=\semantics{D^{G_P}}$ and
  \[
    \operatorname{wt}_{C,P}(F_P)\leq n(F_P)\leq n(G_P)
    =\operatorname{wt}_{D,P}(G_P).
  \]
  Take $F=F_XF_Z$. By \Cref{lem:css-component-matching}, with the roles of
  $C$ and $D$ interchanged, $\semantics{C^F}=\semantics{D^G}$. By
  \Cref{def:css-circuit-fault-equivalence},
  \[
    \operatorname{wt}_{C,\mathrm{CSS}}(F)
    =\max_{P\in\{X,Z\}}\operatorname{wt}_{C,P}(F_P)
    \leq\max_{P\in\{X,Z\}}\operatorname{wt}_{D,P}(G_P)
    =\operatorname{wt}_{D,\mathrm{CSS}}(G).
  \]
\end{description}
\end{proof}

\begin{example}
Inside a larger circuit, subsequent operations may make even a
state-changing fault harmless. Consider an $X$ fault after $\lvert0\rangle$,
followed by an $X$-basis measurement postselected on $\langle+\rvert$. The postselection gives the same amplitude for both states, so the fault
$F=X$ is harmless for the whole circuit. Hence, $\operatorname{wt}(F)=0$
even though $n(F)=1$.


\[
\tikzfig{figures/edge-case}
\]

\end{example}

%% file: macros/floquet.tikzstyles
\tikzstyle{Z bold phase dot}=[minimum size=5mm, font={\small\boldmath}, shape=rectangle, rounded corners=2mm, inner sep=0.1mm, outer sep=-2mm, scale=0.8, tikzit shape=circle, draw=black, fill=zxGreen, tikzit draw=blue, tikzit fill={rgb,255: red,216; green,248; blue,216}, line width=1.2pt]
\tikzstyle{Z dot}=[inner sep=0mm, minimum size=2mm, shape=circle, draw=black, fill=zxGreen, tikzit fill={rgb,255: red,216; green,248; blue,216}, outer sep=-0.5mm]
\tikzstyle{Z phase dot}=[draw=black, fill=zxGreen, shape=rectangle, minimum size=4.5mm, rounded corners=1.8mm, inner sep=0.5mm, outer sep=-0.5mm, scale=0.8, tikzit shape=circle, tikzit fill={rgb,255: red,216; green,248; blue,216}, font={\footnotesize\boldmath}]
\tikzstyle{X bold phase dot}=[minimum size=5mm, font={\small\boldmath}, shape=rectangle, rounded corners=2mm, inner sep=0.1mm, outer sep=-2mm, scale=0.8, tikzit shape=circle, draw=black, fill=zxRed, tikzit draw=blue, tikzit fill={rgb,255: red,221; green,165; blue,165}, line width=1.2pt]
\tikzstyle{X dot}=[shape=circle, draw=black, fill=zxRed, tikzit fill={rgb,255: red,221; green,165; blue,165}, inner sep=0 mm, minimum size=2 mm, outer sep=-0.5mm]
\tikzstyle{X phase dot}=[Z phase dot, draw=black, fill=zxRed, tikzit fill={rgb,255: red,221; green,165; blue,165}]
\tikzstyle{hadamard}=[fill=zxHad, draw=black, shape=rectangle, inner sep=0.6mm, minimum height=1.5mm, minimum width=1.5mm, tikzit fill=yellow, font={\scriptsize}]
\tikzstyle{box}=[draw=black, shape=rectangle, fill=white, minimum size=1em, inner sep=0.2em, scale=0.85, font={\scriptsize}, outer sep=-0.5mm]
\tikzstyle{black dot}=[fill=black, draw=black, shape=circle, inner sep=1pt, outer sep=-3pt]
\tikzstyle{white dot}=[fill=white, draw=black, shape=circle, inner sep=1.5pt, outer sep=-3pt]
\tikzstyle{sLabel}=[font={\scriptsize}, tikzit draw=black, auto]
\tikzstyle{floatingLabel}=[rectangle, fill=white, inner sep=0pt, font={\scriptsize}, label distance=1mm, fill opacity=.5, text opacity=1]
\tikzstyle{not}=[draw=black, circle, addcross, minimum size=2mm, outer sep=-3pt, inner sep=0mm]
\tikzstyle{meter}=[draw=black, fill=white, shape=rectangle, addmeter]
\tikzstyle{Z dot thick}=[inner sep=0mm, minimum size=2mm, shape=circle, draw=black, fill=zxGreen, tikzit fill={rgb,255: red,216; green,248; blue,216}, outer sep=-0.5mm, line width=1pt]
\tikzstyle{X dot thick}=[inner sep=0mm, minimum size=2mm, shape=circle, draw=black, fill=zxRed, tikzit fill={rgb,255: red,221; green,165; blue,165}, outer sep=-0.5mm, line width=1pt]
\tikzstyle{fault-location}=[fill=white, draw=black, shape=circle, minimum size=2mm, inner sep=0mm, outer sep=-0.5 mm, regular polygon, regular polygon sides=8]
\tikzstyle{fault-location-faulty}=[fill=white, draw={rgb,255: red,191; green,0; blue,64}, shape=circle, minimum size=2mm, inner sep=0mm, outer sep=-0.5 mm, regular polygon, regular polygon sides=8, minimum size=3mm]
\tikzstyle{new style 0}=[fill=white, draw=black, shape=circle]
\tikzstyle{cnot ctrl}=[fill=black, draw=black, shape=circle, inner sep=0pt, minimum width=1.2mm, tikzit category=circuit]

\tikzstyle{dashed-line}=[-, style=dashed, draw={rgb,255: red,128; green,128; blue,128}, line width=1pt]
\tikzstyle{lightdots}=[-, style=dotted, draw={rgb,255: red,64; green,64; blue,64}, line width=1pt]
\tikzstyle{hadamard edge}=[-, style=dashed, draw=blue]
\tikzstyle{X Web}=[-, preaction={line width=1mm, draw=zxDarkRed}, tikzit draw=red]
\tikzstyle{Z Web}=[-, preaction={line width=1.5mm, draw=zxDarkGreen}, tikzit draw=green]
\tikzstyle{error}=[-, draw=red, thick]
\tikzstyle{XZ Web}=[-, preaction={line width=1.8mm, draw=zxDarkGreen}, preaction={line width=1mm, draw=zxDarkRed}, tikzit draw=blue]
\tikzstyle{braceedge}=[-, decorate, decoration={brace, amplitude=2mm, raise=-1mm}]
\tikzstyle{arrow}=[->]
\tikzstyle{fault-free}=[-, draw={rgb,255: red,177; green,98; blue,255}, line width=1pt]
\tikzstyle{fault-free arrow}=[->, draw={rgb,255: red,177; green,98; blue,255}, line width=1pt]
\tikzstyle{new edge style 0}=[-, fill={rgb,255: red,211; green,211; blue,211}, draw=none]
\tikzstyle{XerrorPropagation}=[-, draw=red, line width=1pt]
\tikzstyle{XerrorPropagationArrow}=[->, draw=red, line width=1pt]
\tikzstyle{ZerrorPropagation}=[-, draw=green, line width=1pt]
\tikzstyle{ZerrorPropagationArrow}=[->, draw=green, line width=1pt]
\tikzstyle{fault-accounting-component}=[-, style=dashed, draw={rgb,255: red,0; green,128; blue,128}]
\tikzstyle{new edge style 2}=[-, style=dashed, draw={rgb,255: red,0; green,128; blue,128}]
\tikzstyle{blue edge}=[-, line width=1.6pt, draw={rgb,255: red,46; green,117; blue,182}]
\tikzstyle{graph arrow}=[->, >=stealth, shorten >=2mm, line width=1pt]
\tikzstyle{graph line}=[-, >=stealth, line width=1pt]
\tikzstyle{red dashed arrow}=[->, >=stealth, style=dashed, shorten >=2mm, draw={rgb,255: red,255; green,0; blue,0}, line width=1pt]
\tikzstyle{cnot_edge_dir}=[->, >=stealth, style=dashed, line width=1.5pt, draw=orange, shorten >=1mm]

%% file: figures/representation.tikz
\begin{tikzpicture}
	\begin{pgfonlayer}{nodelayer}
		\node [style=X dot] (17) at (-8.5, 2) {};
		\node [style=X dot] (18) at (-8.5, 0.5) {};
		\node [style=X dot] (19) at (-8.5, -1) {};
		\node [style=Z dot] (20) at (-6.5, 2.5) {};
		\node [style=Z dot] (21) at (-6.5, 1.5) {};
		\node [style=Z dot] (22) at (-6.5, 0.5) {};
		\node [style=Z dot] (23) at (-6.5, -0.5) {};
		\node [style=Z dot] (24) at (-6.5, -1.5) {};
		\node [style=Z dot] (25) at (-6.5, -2.5) {};
		\node [style=none] (27) at (-5.5, 2.5) {};
		\node [style=none] (28) at (-5.5, 1.5) {};
		\node [style=none] (29) at (-5.5, 0.5) {};
		\node [style=none] (30) at (-5.5, -0.5) {};
		\node [style=none] (31) at (-5.5, -1.5) {};
		\node [style=none] (32) at (-5.5, -2.5) {};
		\node [style=none] (35) at (-3.5, 0) {$\FaultEq$};
		\node [style=X dot] (36) at (-8.5, -2.5) {};
		\node [style=X dot] (37) at (-1.25, 2) {};
		\node [style=X dot] (38) at (-1.25, 0.5) {};
		\node [style=X dot] (39) at (-1.25, -1) {};
		\node [style=Z dot] (40) at (0.75, 2.5) {};
		\node [style=Z dot] (41) at (0.75, 1.5) {};
		\node [style=Z dot] (42) at (0.75, 0.5) {};
		\node [style=Z dot] (43) at (0.75, -0.5) {};
		\node [style=Z dot] (44) at (0.75, -1.5) {};
		\node [style=Z dot] (45) at (0.75, -2.5) {};
		\node [style=none] (46) at (1.75, 2.5) {};
		\node [style=none] (47) at (1.75, 1.5) {};
		\node [style=none] (48) at (1.75, 0.5) {};
		\node [style=none] (49) at (1.75, -0.5) {};
		\node [style=none] (50) at (1.75, -1.5) {};
		\node [style=none] (51) at (1.75, -2.5) {};
		\node [style=X dot] (52) at (-1.25, -2.25) {};
	\end{pgfonlayer}
	\begin{pgfonlayer}{edgelayer}
		\draw [style=fault-free] (17) to (20);
		\draw (32.center) to (25);
		\draw (31.center) to (24);
		\draw (30.center) to (23);
		\draw (29.center) to (22);
		\draw (28.center) to (21);
		\draw (27.center) to (20);
		\draw [style=fault-free] (18) to (20);
		\draw [style=fault-free] (19) to (21);
		\draw [style=fault-free, in=195, out=15] (37) to (40);
		\draw (51.center) to (45);
		\draw (50.center) to (44);
		\draw (49.center) to (43);
		\draw (48.center) to (42);
		\draw (47.center) to (41);
		\draw (46.center) to (40);
		\draw [style=fault-free] (38) to (41);
		\draw [style=fault-free] (52) to (43);
		\draw [style=fault-free] (39) to (42);
		\draw [style=fault-free] (37) to (44);
		\draw [style=fault-free] (38) to (44);
		\draw [style=fault-free] (39) to (44);
		\draw [style=fault-free] (52) to (44);
		\draw [style=fault-free] (37) to (45);
		\draw [style=fault-free] (38) to (45);
		\draw [style=fault-free] (39) to (45);
		\draw [style=fault-free] (45) to (52);
		\draw [style=fault-free] (17) to (24);
		\draw [style=fault-free] (17) to (25);
		\draw [style=fault-free] (18) to (21);
		\draw [style=fault-free] (19) to (22);
		\draw [style=fault-free] (36) to (23);
		\draw [style=fault-free] (22) to (36);
	\end{pgfonlayer}
\end{tikzpicture}

%% file: figures/rewrites.tikz
\begin{tikzpicture}
	\begin{pgfonlayer}{nodelayer}
		\node [style=X dot] (0) at (-15.25, 2) {};
		\node [style=X dot] (1) at (-15.25, 0.5) {};
		\node [style=X dot] (2) at (-15.25, -1) {};
		\node [style=Z dot] (3) at (-13.25, 2.5) {};
		\node [style=Z dot] (4) at (-13.25, 1.5) {};
		\node [style=Z dot] (5) at (-13.25, 0.5) {};
		\node [style=Z dot] (6) at (-13.25, -0.5) {};
		\node [style=Z dot] (7) at (-13.25, -1.5) {};
		\node [style=Z dot] (8) at (-13.25, -2.5) {};
		\node [style=none] (9) at (-12.25, 2.5) {};
		\node [style=none] (10) at (-12.25, 1.5) {};
		\node [style=none] (11) at (-12.25, 0.5) {};
		\node [style=none] (12) at (-12.25, -0.5) {};
		\node [style=none] (13) at (-12.25, -1.5) {};
		\node [style=none] (14) at (-12.25, -2.5) {};
		\node [style=X dot] (15) at (-15.25, -2.25) {};
		\node [style=none] (16) at (3, 0) {$\FaultEq$};
		\node [style=X dot] (17) at (6.25, 2.5) {};
		\node [style=X dot] (18) at (6.25, 1.5) {};
		\node [style=X dot] (19) at (6.25, 0.5) {};
		\node [style=Z dot] (20) at (6.25, -1.5) {};
		\node [style=Z dot] (21) at (6.25, -2.5) {};
		\node [style=none] (22) at (7.25, 2.5) {};
		\node [style=none] (23) at (7.25, 1.5) {};
		\node [style=none] (24) at (7.25, 0.5) {};
		\node [style=none] (25) at (7.25, -0.5) {};
		\node [style=none] (26) at (7.25, -1.5) {};
		\node [style=none] (27) at (7.25, -2.5) {};
		\node [style=X dot] (28) at (6.25, -0.5) {};
		\node [style=none] (29) at (-10.5, 0) {$\FaultEq$};
		\node [style=X dot] (30) at (-8.5, 2) {};
		\node [style=X dot] (31) at (-8.5, 0.5) {};
		\node [style=X dot] (32) at (-8.5, -1) {};
		\node [style=Z dot] (33) at (-6.5, 2.5) {};
		\node [style=Z dot] (34) at (-6.5, 1.5) {};
		\node [style=Z dot] (35) at (-6.5, 0.5) {};
		\node [style=Z dot] (36) at (-6.5, -0.5) {};
		\node [style=Z dot] (37) at (-6.5, -1.5) {};
		\node [style=Z dot] (38) at (-6.5, -2.5) {};
		\node [style=none] (39) at (-5.5, 2.5) {};
		\node [style=none] (40) at (-5.5, 1.5) {};
		\node [style=none] (41) at (-5.5, 0.5) {};
		\node [style=none] (42) at (-5.5, -0.5) {};
		\node [style=none] (43) at (-5.5, -1.5) {};
		\node [style=none] (44) at (-5.5, -2.5) {};
		\node [style=X dot] (45) at (-8.5, -2.25) {};
		\node [style=none] (46) at (-3.75, 0) {$\FaultEq$};
		\node [style=X dot] (47) at (-1.75, 2) {};
		\node [style=X dot] (48) at (-1.75, 0.5) {};
		\node [style=X dot] (49) at (-1.75, -1) {};
		\node [style=none] (50) at (0.25, 2.5) {};
		\node [style=none] (51) at (0.25, 1.5) {};
		\node [style=none] (52) at (0.25, 0.5) {};
		\node [style=none] (53) at (0.25, -0.5) {};
		\node [style=Z dot] (54) at (0.25, -1.5) {};
		\node [style=Z dot] (55) at (0.25, -2.5) {};
		\node [style=none] (56) at (1.25, 2.5) {};
		\node [style=none] (57) at (1.25, 1.5) {};
		\node [style=none] (58) at (1.25, 0.5) {};
		\node [style=none] (59) at (1.25, -0.5) {};
		\node [style=none] (60) at (1.25, -1.5) {};
		\node [style=none] (61) at (1.25, -2.5) {};
		\node [style=X dot] (62) at (-1.75, -2.25) {};
		\node [style=none] (63) at (8.75, 0) {$\FaultEq$};
		\node [style=X dot] (64) at (12, 2.5) {};
		\node [style=X dot] (65) at (12, 1.5) {};
		\node [style=X dot] (66) at (12, 0.5) {};
		\node [style=Z dot] (67) at (12, -1.5) {};
		\node [style=Z dot] (68) at (12, -2.5) {};
		\node [style=none] (69) at (13, 2.5) {};
		\node [style=none] (70) at (13, 1.5) {};
		\node [style=none] (71) at (13, 0.5) {};
		\node [style=none] (72) at (13, -0.5) {};
		\node [style=none] (73) at (13, -1.5) {};
		\node [style=none] (74) at (13, -2.5) {};
		\node [style=X dot] (75) at (12, -0.5) {};
	\end{pgfonlayer}
	\begin{pgfonlayer}{edgelayer}
		\draw [style=fault-free, in=195, out=15] (0) to (3);
		\draw (14.center) to (8);
		\draw (13.center) to (7);
		\draw (12.center) to (6);
		\draw (11.center) to (5);
		\draw (10.center) to (4);
		\draw (9.center) to (3);
		\draw [style=fault-free] (1) to (4);
		\draw [style=fault-free] (15) to (6);
		\draw [style=fault-free] (2) to (5);
		\draw [style=fault-free] (0) to (7);
		\draw [style=fault-free] (1) to (7);
		\draw [style=fault-free] (2) to (7);
		\draw [style=fault-free] (15) to (7);
		\draw [style=fault-free] (0) to (8);
		\draw [style=fault-free] (1) to (8);
		\draw [style=fault-free] (2) to (8);
		\draw [style=fault-free] (8) to (15);
		\draw (27.center) to (21);
		\draw (26.center) to (20);
		\draw [style=fault-free, bend right=60] (17) to (20);
		\draw [style=fault-free, bend right=45, looseness=1.25] (18) to (20);
		\draw [style=fault-free, bend right=45] (19) to (20);
		\draw [style=fault-free] (28) to (20);
		\draw [style=fault-free, bend right=75, looseness=1.25] (17) to (21);
		\draw [style=fault-free, bend right=60, looseness=1.25] (18) to (21);
		\draw [style=fault-free, bend right=45, looseness=1.50] (19) to (21);
		\draw [style=fault-free, bend left=60, looseness=1.25] (21) to (28);
		\draw (22.center) to (17);
		\draw (23.center) to (18);
		\draw (25.center) to (28);
		\draw (24.center) to (19);
		\draw [in=195, out=15] (30) to (33);
		\draw (44.center) to (38);
		\draw (43.center) to (37);
		\draw (42.center) to (36);
		\draw (41.center) to (35);
		\draw (40.center) to (34);
		\draw (39.center) to (33);
		\draw (31) to (34);
		\draw (45) to (36);
		\draw (32) to (35);
		\draw [style=fault-free] (30) to (37);
		\draw [style=fault-free] (31) to (37);
		\draw [style=fault-free] (32) to (37);
		\draw [style=fault-free] (45) to (37);
		\draw [style=fault-free] (30) to (38);
		\draw [style=fault-free] (31) to (38);
		\draw [style=fault-free] (32) to (38);
		\draw [style=fault-free] (38) to (45);
		\draw [in=195, out=15] (47) to (50.center);
		\draw (61.center) to (55);
		\draw (60.center) to (54);
		\draw (59.center) to (53.center);
		\draw (58.center) to (52.center);
		\draw (57.center) to (51.center);
		\draw (56.center) to (50.center);
		\draw (48) to (51.center);
		\draw (62) to (53.center);
		\draw (49) to (52.center);
		\draw [style=fault-free] (47) to (54);
		\draw [style=fault-free] (48) to (54);
		\draw [style=fault-free] (49) to (54);
		\draw [style=fault-free] (62) to (54);
		\draw [style=fault-free] (47) to (55);
		\draw [style=fault-free] (48) to (55);
		\draw [style=fault-free] (49) to (55);
		\draw [style=fault-free] (55) to (62);
		\draw (74.center) to (68);
		\draw (73.center) to (67);
		\draw [bend right=60] (64) to (67);
		\draw [bend right=45, looseness=1.25] (65) to (67);
		\draw [bend right=45] (66) to (67);
		\draw (75) to (67);
		\draw [bend right=75, looseness=1.25] (64) to (68);
		\draw [bend right=60, looseness=1.25] (65) to (68);
		\draw [bend right=45, looseness=1.50] (66) to (68);
		\draw [bend left=60, looseness=1.25] (68) to (75);
		\draw (69.center) to (64);
		\draw (70.center) to (65);
		\draw (72.center) to (75);
		\draw (71.center) to (66);
	\end{pgfonlayer}
\end{tikzpicture}

%% file: figures/decompose.tikz
\begin{tikzpicture}
	\begin{pgfonlayer}{nodelayer}
		\node [style=none] (0) at (-2.5, 0) {$\FaultEq$};
		\node [style=X dot] (1) at (-5.75, 2.5) {};
		\node [style=X dot] (2) at (-5.75, 1.5) {};
		\node [style=X dot] (3) at (-5.75, 0.5) {};
		\node [style=Z dot] (4) at (-5.75, -1.5) {};
		\node [style=Z dot] (5) at (-5.75, -2.5) {};
		\node [style=none] (6) at (-4.75, 2.5) {};
		\node [style=none] (7) at (-4.75, 1.5) {};
		\node [style=none] (8) at (-4.75, 0.5) {};
		\node [style=none] (9) at (-4.75, -0.5) {};
		\node [style=none] (10) at (-4.75, -1.5) {};
		\node [style=none] (11) at (-4.75, -2.5) {};
		\node [style=X dot] (12) at (-5.75, -0.5) {};
		\node [style=X dot] (13) at (2.75, 4) {};
		\node [style=X dot] (14) at (2.75, 3) {};
		\node [style=X dot] (15) at (2.75, 2) {};
		\node [style=none] (18) at (4.25, 4) {};
		\node [style=none] (19) at (4.25, 3) {};
		\node [style=none] (20) at (4.25, 2) {};
		\node [style=none] (21) at (4.25, 1) {};
		\node [style=none] (22) at (4.25, -1.25) {};
		\node [style=none] (23) at (4.25, -3.25) {};
		\node [style=X dot] (24) at (2.75, 1) {};
		\node [style=Z dot] (25) at (1.25, -0.5) {};
		\node [style=Z dot] (26) at (2, 0) {};
		\node [style=Z dot] (27) at (1.5, -1.25) {};
		\node [style=Z dot] (28) at (2.75, -0.5) {};
		\node [style=Z dot] (29) at (2.5, -1.25) {};
		\node [style=Z dot] (30) at (1.25, -2.5) {};
		\node [style=Z dot] (31) at (2, -2) {};
		\node [style=Z dot] (32) at (1.5, -3.25) {};
		\node [style=Z dot] (33) at (2.75, -2.5) {};
		\node [style=Z dot] (34) at (2.5, -3.25) {};
	\end{pgfonlayer}
	\begin{pgfonlayer}{edgelayer}
		\draw (11.center) to (5);
		\draw (10.center) to (4);
		\draw [bend right=60] (1) to (4);
		\draw [bend right=45, looseness=1.25] (2) to (4);
		\draw [bend right=45] (3) to (4);
		\draw (12) to (4);
		\draw [bend right=75, looseness=1.25] (1) to (5);
		\draw [bend right=60, looseness=1.25] (2) to (5);
		\draw [bend right=45, looseness=1.50] (3) to (5);
		\draw [bend left=60, looseness=1.25] (5) to (12);
		\draw (6.center) to (1);
		\draw (7.center) to (2);
		\draw (9.center) to (12);
		\draw (8.center) to (3);
		\draw (18.center) to (13);
		\draw (19.center) to (14);
		\draw (21.center) to (24);
		\draw (20.center) to (15);
		\draw (25) to (27);
		\draw (28) to (26);
		\draw (26) to (25);
		\draw (27) to (29);
		\draw (29) to (28);
		\draw (28) to (24);
		\draw (22.center) to (29);
		\draw [bend left=45] (26) to (15);
		\draw [bend left, looseness=0.75] (25) to (14);
		\draw [bend left=60, looseness=1.25] (27) to (13);
		\draw (30) to (32);
		\draw (33) to (31);
		\draw (31) to (30);
		\draw (32) to (34);
		\draw (34) to (33);
		\draw (34) to (23.center);
		\draw [bend right, looseness=1.25] (33) to (24);
		\draw (31) to (15);
		\draw [bend left=45, looseness=1.25] (30) to (14);
		\draw [bend right=75, looseness=1.25] (13) to (32);
	\end{pgfonlayer}
\end{tikzpicture}

%% file: figures/layout.tikz
\begin{tikzpicture}
	\begin{pgfonlayer}{nodelayer}
		\node [style=Z bold phase dot] (21) at (17, -2) {9};
		\node [style=Z bold phase dot] (22) at (12.5, 0) {7};
		\node [style=Z bold phase dot] (23) at (8.5, 0) {5};
		\node [style=Z bold phase dot] (24) at (12.5, -2) {8};
		\node [style=X bold phase dot] (25) at (17, 8) {1};
		\node [style=X bold phase dot] (26) at (12.5, 4) {3};
		\node [style=X bold phase dot] (27) at (17, 2) {4};
		\node [style=X bold phase dot] (28) at (12.5, 6) {2};
		\node [style=Z bold phase dot] (29) at (6, -2) {$r5$};
		\node [style=X bold phase dot] (30) at (12.5, 8) {$r1$};
		\node [style=X bold phase dot] (31) at (8.5, 4) {$r3$};
		\node [style=X bold phase dot] (32) at (12.5, 2) {$r4$};
		\node [style=X bold phase dot] (33) at (8.5, 6) {$r2$};
		\node [style=Z bold phase dot] (34) at (8.5, -2) {6};
		\node [style=Z bold phase dot] (35) at (12.5, -4) {10};
		\node [style=Z bold phase dot] (36) at (17, -4) {12};
		\node [style=Z bold phase dot] (37) at (21, -6) {14};
		\node [style=Z bold phase dot] (38) at (12.5, -6) {11};
		\node [style=Z bold phase dot] (39) at (17, -6) {13};
		\node [style=Z bold phase dot] (40) at (8.5, -6) {$r6$};
		\node [style=Z bold phase dot] (42) at (-14.75, -2) {9};
		\node [style=Z bold phase dot] (43) at (-13.75, -0.75) {7};
		\node [style=Z bold phase dot] (44) at (-14.75, 0.5) {5};
		\node [style=Z bold phase dot] (45) at (-16.25, -2) {8};
		\node [style=X bold phase dot] (46) at (-10.25, 8) {1};
		\node [style=X bold phase dot] (47) at (-10.25, 4) {3};
		\node [style=X bold phase dot] (48) at (-10.25, 2) {4};
		\node [style=X bold phase dot] (49) at (-10.25, 6) {2};
		\node [style=none] (50) at (-8.75, 8) {};
		\node [style=none] (51) at (-8.75, 6) {};
		\node [style=none] (52) at (-8.75, 4) {};
		\node [style=none] (53) at (-8.75, -2) {};
		\node [style=none] (54) at (-8.75, 2) {};
		\node [style=Z bold phase dot] (55) at (-16.25, 0.5) {$r5$};
		\node [style=X bold phase dot] (56) at (-12.25, 8) {$r1$};
		\node [style=X bold phase dot] (57) at (-12.25, 4) {$r3$};
		\node [style=X bold phase dot] (58) at (-12.25, 2) {$r4$};
		\node [style=X bold phase dot] (59) at (-12.25, 6) {$r2$};
		\node [style=Z bold phase dot] (60) at (-17.25, -0.75) {6};
		\node [style=Z bold phase dot] (61) at (-14.75, -3.5) {10};
		\node [style=Z bold phase dot] (62) at (-13.75, -4.75) {12};
		\node [style=Z bold phase dot] (63) at (-14.75, -6) {14};
		\node [style=Z bold phase dot] (64) at (-17.25, -4.75) {11};
		\node [style=Z bold phase dot] (65) at (-16.25, -6) {13};
		\node [style=Z bold phase dot] (66) at (-16.25, -3.5) {$r6$};
		\node [style=none] (67) at (-8.75, -6) {};
	\end{pgfonlayer}
	\begin{pgfonlayer}{edgelayer}
		\draw [style=graph arrow] (29) to (34);
		\draw [style=graph arrow] (34) to (24);
		\draw [style=graph arrow] (24) to (21);
		\draw [style=graph arrow] (40) to (38);
		\draw [style=graph arrow] (38) to (39);
		\draw [style=graph arrow] (39) to (37);
		\draw [style=graph arrow] (40) to (35);
		\draw [style=graph arrow] (35) to (36);
		\draw [style=graph arrow] (29) to (23);
		\draw [style=graph arrow] (23) to (22);
		\draw [style=red dashed arrow] (22) to (21);
		\draw [style=red dashed arrow] (36) to (37);
		\draw [style=graph arrow] (30) to (25);
		\draw [style=graph arrow] (33) to (28);
		\draw [style=graph arrow] (31) to (26);
		\draw [style=graph arrow] (32) to (27);
		\draw [style={cnot_edge_dir}] (29) to (31);
		\draw [style={cnot_edge_dir}] (40) to (26);
		\draw [style={cnot_edge_dir}] (29) to (33);
		\draw [style={cnot_edge_dir}] (40) to (28);
		\draw [style={cnot_edge_dir}] (33) to (38);
		\draw [style={cnot_edge_dir}] (31) to (35);
		\draw [style={cnot_edge_dir}] (23) to (32);
		\draw [style={cnot_edge_dir}] (32) to (36);
		\draw [style={cnot_edge_dir}] (35) to (27);
		\draw [style={cnot_edge_dir}] (34) to (30);
		\draw [style={cnot_edge_dir}] (30) to (39);
		\draw [style={cnot_edge_dir}] (38) to (25);
		\draw [style=graph line] (48) to (54.center);
		\draw [style=graph line] (53.center) to (42);
		\draw [style=graph line] (42) to (43);
		\draw [style=graph line] (43) to (44);
		\draw [style=graph line] (45) to (42);
		\draw [style=graph line] (47) to (52.center);
		\draw [style=graph line] (51.center) to (49);
		\draw [style=graph line] (46) to (50.center);
		\draw [style=graph line] (57) to (47);
		\draw [style=graph line] (58) to (48);
		\draw [style=graph line] (46) to (56);
		\draw [style=graph line] (55) to (44);
		\draw [style=graph line] (59) to (49);
		\draw [style=graph line] (60) to (45);
		\draw [style=graph line] (55) to (60);
		\draw [style=graph line] (61) to (62);
		\draw [style=graph line] (62) to (63);
		\draw [style=graph line] (65) to (63);
		\draw [style=graph line] (67.center) to (63);
		\draw [style=graph line] (65) to (64);
		\draw [style=graph line] (64) to (66);
		\draw [style=graph line] (66) to (61);
		\draw [style=graph line] (44) to (57);
		\draw [style=graph line, bend right=15] (61) to (47);
		\draw [style=graph line, in=180, out=90] (60) to (59);
		\draw [style=graph line, in=-165, out=135, looseness=1.25] (64) to (49);
		\draw [style=graph line] (43) to (58);
		\draw [style=graph line, bend right=15] (62) to (48);
		\draw [style=graph line, in=180, out=165, looseness=1.25] (45) to (56);
		\draw [style=graph line, in=-165, out=165, looseness=1.50] (65) to (46);
	\end{pgfonlayer}
\end{tikzpicture}

%% file: figures/sarah-figure.tikz
\begin{tikzpicture}
	\begin{pgfonlayer}{nodelayer}
		\node [style=Z bold phase dot] (41) at (1.5, -3) {9};
		\node [style=Z bold phase dot] (42) at (2.5, -1.75) {7};
		\node [style=Z bold phase dot] (43) at (1.5, -0.5) {5};
		\node [style=Z bold phase dot] (44) at (0, -3) {8};
		\node [style=X bold phase dot] (45) at (6, 7) {1};
		\node [style=X bold phase dot] (46) at (6, 3) {3};
		\node [style=X bold phase dot] (47) at (6, 1) {4};
		\node [style=X bold phase dot] (48) at (6, 5) {2};
		\node [style=none] (51) at (7.5, 7) {};
		\node [style=none] (52) at (7.5, 5) {};
		\node [style=none] (53) at (7.5, 3) {};
		\node [style=none] (55) at (7.5, -3) {};
		\node [style=none] (56) at (7.5, 1) {};
		\node [style=Z bold phase dot] (59) at (0, -0.5) {$r5$};
		\node [style=X bold phase dot] (60) at (4, 7) {$r1$};
		\node [style=X bold phase dot] (61) at (4, 3) {$r3$};
		\node [style=X bold phase dot] (62) at (4, 1) {$r4$};
		\node [style=X bold phase dot] (63) at (4, 5) {$r2$};
		\node [style=Z bold phase dot] (88) at (-1, -1.75) {6};
		\node [style=Z bold phase dot] (89) at (1.5, -4.5) {10};
		\node [style=Z bold phase dot] (90) at (2.5, -5.75) {12};
		\node [style=Z bold phase dot] (91) at (1.5, -7) {14};
		\node [style=Z bold phase dot] (92) at (-1, -5.75) {11};
		\node [style=Z bold phase dot] (94) at (0, -7) {13};
		\node [style=Z bold phase dot] (95) at (0, -4.5) {$r6$};
		\node [style=none] (96) at (7.5, -7) {};
		\node [style=Z bold phase dot] (122) at (21.75, -3) {9};
		\node [style=Z bold phase dot] (123) at (19.25, -1) {7};
		\node [style=Z bold phase dot] (124) at (13.75, -1) {5};
		\node [style=Z bold phase dot] (125) at (17.75, -3) {8};
		\node [style=X bold phase dot] (126) at (22.5, 7) {1};
		\node [style=X bold phase dot] (127) at (17, 3) {3};
		\node [style=X bold phase dot] (128) at (23.25, 1) {4};
		\node [style=X bold phase dot] (129) at (18.5, 5) {2};
		\node [style=none] (130) at (27.25, 7) {};
		\node [style=none] (131) at (27.25, 5) {};
		\node [style=none] (132) at (27.25, 3) {};
		\node [style=none] (133) at (27.25, -3) {};
		\node [style=none] (134) at (27.25, 1) {};
		\node [style=Z bold phase dot] (135) at (11.75, -3) {$r5$};
		\node [style=X bold phase dot] (136) at (17.75, 7) {$r1$};
		\node [style=X bold phase dot] (137) at (13.75, 3) {$r3$};
		\node [style=X bold phase dot] (138) at (19.25, 1) {$r4$};
		\node [style=X bold phase dot] (139) at (14.5, 5) {$r2$};
		\node [style=Z bold phase dot] (140) at (14.5, -3) {6};
		\node [style=Z bold phase dot] (141) at (17, -5) {10};
		\node [style=Z bold phase dot] (142) at (23.25, -5) {12};
		\node [style=Z bold phase dot] (143) at (25.75, -7) {14};
		\node [style=Z bold phase dot] (144) at (18.5, -7) {11};
		\node [style=Z bold phase dot] (145) at (22.5, -7) {13};
		\node [style=Z bold phase dot] (146) at (14.5, -7) {$r6$};
		\node [style=none] (147) at (27.25, -7) {};
		\node [style=Z bold phase dot] (148) at (42, -3) {9};
		\node [style=Z bold phase dot] (149) at (39.5, -1) {7};
		\node [style=Z bold phase dot] (150) at (34, -1) {5};
		\node [style=Z bold phase dot] (151) at (38, -3) {8};
		\node [style=X bold phase dot] (152) at (42.75, 7) {1};
		\node [style=X bold phase dot] (153) at (37.25, 3) {3};
		\node [style=X bold phase dot] (154) at (43.5, 1) {4};
		\node [style=X bold phase dot] (155) at (38.75, 5) {2};
		\node [style=none] (156) at (47.5, 7) {};
		\node [style=none] (157) at (47.5, 5) {};
		\node [style=none] (158) at (47.5, 3) {};
		\node [style=none] (159) at (47.5, -3) {};
		\node [style=none] (160) at (47.5, 1) {};
		\node [style=Z bold phase dot] (161) at (32, -3) {$r5$};
		\node [style=X bold phase dot] (162) at (38, 7) {$r1$};
		\node [style=X bold phase dot] (163) at (34, 3) {$r3$};
		\node [style=X bold phase dot] (164) at (39.5, 1) {$r4$};
		\node [style=X bold phase dot] (165) at (34.75, 5) {$r2$};
		\node [style=Z bold phase dot] (166) at (34.75, -3) {6};
		\node [style=Z bold phase dot] (167) at (37.25, -5) {10};
		\node [style=Z bold phase dot] (168) at (43.5, -5) {12};
		\node [style=Z bold phase dot] (169) at (46, -7) {14};
		\node [style=Z bold phase dot] (170) at (38.75, -7) {11};
		\node [style=Z bold phase dot] (171) at (42.75, -7) {13};
		\node [style=Z bold phase dot] (172) at (34.75, -7) {$r6$};
		\node [style=none] (173) at (47.5, -7) {};
		\node [style=X bold phase dot] (174) at (46, -5) {};
		\node [style=X bold phase dot] (175) at (47.5, -5) {};
		\node [style=X bold phase dot] (176) at (34.75, -5) {};
		\node [style=X bold phase dot] (177) at (33.25, -5) {};
		\node [style=X bold phase dot] (178) at (32, -1) {};
		\node [style=X bold phase dot] (179) at (30.5, -1) {};
		\node [style=X bold phase dot] (180) at (36.5, 7) {};
		\node [style=X bold phase dot] (181) at (33.5, 5) {};
		\node [style=X bold phase dot] (182) at (32.5, 3) {};
		\node [style=Z bold phase dot] (183) at (30.5, -3) {};
		\node [style=Z bold phase dot] (184) at (33.25, -7) {};
		\node [style=X bold phase dot] (185) at (36.5, 1) {};
		\node [style=X bold phase dot] (186) at (42, -1) {};
		\node [style=X bold phase dot] (187) at (44.25, -1) {};
		\node [style=none] (192) at (9, 7) {};
		\node [style=none] (193) at (15.5, 7) {};
		\node [style=none, align=center] (195) at (12, 8.75) {Rearrange the diagram accoding to \\ the topological generations of the DAG};
		\node [style=none] (196) at (28.75, 7) {};
		\node [style=none] (197) at (34.25, 7) {};
		\node [style=none] (198) at (31.5, 7.75) {Make circuit--like};
		\node [style=cnot ctrl] (199) at (62.5, -3) {};
		\node [style=cnot ctrl] (200) at (60, -1) {};
		\node [style=cnot ctrl] (201) at (54.5, -1) {};
		\node [style=cnot ctrl] (202) at (58.5, -3) {};
		\node [style=not] (203) at (63.25, 7) {};
		\node [style=not] (204) at (57.75, 3) {};
		\node [style=not] (205) at (64, 1) {};
		\node [style=not] (206) at (59.25, 5) {};
		\node [style=none] (207) at (68, 7) {};
		\node [style=none] (208) at (68, 5) {};
		\node [style=none] (209) at (68, 3) {};
		\node [style=none] (210) at (68, -3) {};
		\node [style=none] (211) at (68, 1) {};
		\node [style=cnot ctrl] (212) at (52.5, -3) {};
		\node [style=not] (213) at (58.5, 7) {};
		\node [style=not] (214) at (54.5, 3) {};
		\node [style=not] (215) at (60, 1) {};
		\node [style=not] (216) at (55.25, 5) {};
		\node [style=cnot ctrl] (217) at (55.25, -3) {};
		\node [style=cnot ctrl] (218) at (57.75, -5) {};
		\node [style=cnot ctrl] (219) at (64, -5) {};
		\node [style=cnot ctrl] (220) at (66.5, -7) {};
		\node [style=cnot ctrl] (221) at (59.25, -7) {};
		\node [style=cnot ctrl] (222) at (63.25, -7) {};
		\node [style=cnot ctrl] (223) at (55.25, -7) {};
		\node [style=none] (224) at (68, -7) {};
		\node [style=not] (225) at (66.5, -5) {};
		\node [style=meter] (226) at (68, -5) {};
		\node [style=not] (227) at (55.25, -5) {};
		\node [style=none] (228) at (53.75, -5) {};
		\node [style=not] (229) at (52.5, -1) {};
		\node [style=none] (230) at (51, -1) {};
		\node [style=none] (231) at (57, 7) {};
		\node [style=none] (232) at (54, 5) {};
		\node [style=none] (233) at (53, 3) {};
		\node [style=none] (234) at (51, -3) {};
		\node [style=none] (235) at (53.75, -7) {};
		\node [style=none] (236) at (57, 1) {};
		\node [style=not] (237) at (62.5, -1) {};
		\node [style=meter] (238) at (64.75, -1) {};
		\node [style=none] (239) at (49, 7) {};
		\node [style=none] (240) at (54.75, 7) {};
		\node [style=none] (241) at (52, 7.75) {Directly extract};
		\node [style=none] (242) at (53.25, -5) {$\ket{0}$};
		\node [style=none] (243) at (50.5, -1) {$\ket{0}$};
		\node [style=none] (244) at (56.5, 7) {$\ket{0}$};
		\node [style=none] (245) at (53.5, 5) {$\ket{0}$};
		\node [style=none] (246) at (52.5, 3) {$\ket{0}$};
		\node [style=none] (247) at (56.5, 1) {$\ket{0}$};
		\node [style=none] (248) at (50.5, -3) {$\ket{+}$};
		\node [style=none] (249) at (53.25, -7) {$\ket{+}$};
		\node [style=none] (250) at (49, 0) {$\FaultEq$};
		\node [style=none] (251) at (28.75, 0) {$\FaultEq$};
		\node [style=none] (252) at (9.75, 0) {$\FaultEq$};
	\end{pgfonlayer}
	\begin{pgfonlayer}{edgelayer}
		\draw [style=graph line] (47) to (56.center);
		\draw [style=graph line] (55.center) to (41);
		\draw [style=graph line] (41) to (42);
		\draw [style=graph line] (42) to (43);
		\draw [style=graph line] (44) to (41);
		\draw [style=graph line] (46) to (53.center);
		\draw [style=graph line] (52.center) to (48);
		\draw [style=graph line] (45) to (51.center);
		\draw [style=graph line] (61) to (46);
		\draw [style=graph line] (62) to (47);
		\draw [style=graph line] (45) to (60);
		\draw [style=graph line] (59) to (43);
		\draw [style=graph line] (63) to (48);
		\draw [style=graph line] (88) to (44);
		\draw [style=graph line] (59) to (88);
		\draw [style=graph line] (89) to (90);
		\draw [style=graph line] (90) to (91);
		\draw [style=graph line] (94) to (91);
		\draw [style=graph line] (96.center) to (91);
		\draw [style=graph line] (94) to (92);
		\draw [style=graph line] (92) to (95);
		\draw [style=graph line] (95) to (89);
		\draw [style=graph line] (43) to (61);
		\draw [style=graph line, bend right=15] (89) to (46);
		\draw [style=graph line, in=180, out=90] (88) to (63);
		\draw [style=graph line, in=-165, out=135, looseness=1.25] (92) to (48);
		\draw [style=graph line] (42) to (62);
		\draw [style=graph line, bend right=15] (90) to (47);
		\draw [style=graph line, in=180, out=165, looseness=1.25] (44) to (60);
		\draw [style=graph line, in=-165, out=165, looseness=1.50] (94) to (45);
		\draw [style=graph line] (128) to (134.center);
		\draw [style=graph line] (133.center) to (122);
		\draw [style=graph line] (122) to (123);
		\draw [style=graph line] (123) to (124);
		\draw [style=graph line] (125) to (122);
		\draw [style=graph line] (127) to (132.center);
		\draw [style=graph line] (131.center) to (129);
		\draw [style=graph line] (126) to (130.center);
		\draw [style=graph line] (137) to (127);
		\draw [style=graph line] (138) to (128);
		\draw [style=graph line] (126) to (136);
		\draw [style=graph line] (135) to (124);
		\draw [style=graph line] (139) to (129);
		\draw [style=graph line] (140) to (125);
		\draw [style=graph line] (135) to (140);
		\draw [style=graph line] (141) to (142);
		\draw [style=graph line] (142) to (143);
		\draw [style=graph line] (145) to (143);
		\draw [style=graph line] (147.center) to (143);
		\draw [style=graph line] (145) to (144);
		\draw [style=graph line] (144) to (146);
		\draw [style=graph line] (146) to (141);
		\draw [style=graph line] (124) to (137);
		\draw [style=graph line] (141) to (127);
		\draw [style=graph line] (140) to (139);
		\draw [style=graph line] (144) to (129);
		\draw [style=graph line] (123) to (138);
		\draw [style=graph line] (142) to (128);
		\draw [style=graph line] (125) to (136);
		\draw [style=graph line] (145) to (126);
		\draw [style=graph line] (154) to (160.center);
		\draw [style=graph line] (159.center) to (148);
		\draw [style=graph line] (149) to (150);
		\draw [style=graph line] (151) to (148);
		\draw [style=graph line] (153) to (158.center);
		\draw [style=graph line] (157.center) to (155);
		\draw [style=graph line] (152) to (156.center);
		\draw [style=graph line] (163) to (153);
		\draw [style=graph line] (164) to (154);
		\draw [style=graph line] (152) to (162);
		\draw [style=graph line] (165) to (155);
		\draw [style=graph line] (166) to (151);
		\draw [style=graph line] (161) to (166);
		\draw [style=graph line] (167) to (168);
		\draw [style=graph line] (171) to (169);
		\draw [style=graph line] (173.center) to (169);
		\draw [style=graph line] (171) to (170);
		\draw [style=graph line] (170) to (172);
		\draw [style=graph line] (150) to (163);
		\draw [style=graph line] (167) to (153);
		\draw [style=graph line] (166) to (165);
		\draw [style=graph line] (170) to (155);
		\draw [style=graph line] (149) to (164);
		\draw [style=graph line] (168) to (154);
		\draw [style=graph line] (151) to (162);
		\draw [style=graph line] (171) to (152);
		\draw [style=graph line] (175) to (174);
		\draw [style=graph line] (174) to (169);
		\draw [style=graph line] (174) to (168);
		\draw [style=graph line] (176) to (172);
		\draw [style=graph line] (176) to (167);
		\draw [style=graph line] (177) to (176);
		\draw [style=graph line] (178) to (161);
		\draw [style=graph line] (178) to (150);
		\draw [style=graph line] (179) to (178);
		\draw [style=graph line] (180) to (162);
		\draw [style=graph line] (181) to (165);
		\draw [style=graph line] (182) to (163);
		\draw [style=graph line] (183) to (161);
		\draw [style=graph line] (184) to (172);
		\draw [style=graph line] (185) to (164);
		\draw [style=graph line] (186) to (148);
		\draw [style=graph line] (186) to (149);
		\draw [style=graph line] (187) to (186);
		\draw [style=arrow] (192.center) to (193.center);
		\draw [style=arrow] (196.center) to (197.center);
		\draw [style=graph line] (205) to (211.center);
		\draw [style=graph line] (210.center) to (199);
		\draw [style=graph line] (200) to (201);
		\draw [style=graph line] (202) to (199);
		\draw [style=graph line] (204) to (209.center);
		\draw [style=graph line] (208.center) to (206);
		\draw [style=graph line] (203) to (207.center);
		\draw [style=graph line] (214) to (204);
		\draw [style=graph line] (215) to (205);
		\draw [style=graph line] (203) to (213);
		\draw [style=graph line] (216) to (206);
		\draw [style=graph line] (217) to (202);
		\draw [style=graph line] (212) to (217);
		\draw [style=graph line] (218) to (219);
		\draw [style=graph line] (222) to (220);
		\draw [style=graph line] (224.center) to (220);
		\draw [style=graph line] (222) to (221);
		\draw [style=graph line] (221) to (223);
		\draw [style=graph line] (201) to (214);
		\draw [style=graph line] (218) to (204);
		\draw [style=graph line] (217) to (216);
		\draw [style=graph line] (221) to (206);
		\draw [style=graph line] (200) to (215);
		\draw [style=graph line] (219) to (205);
		\draw [style=graph line] (202) to (213);
		\draw [style=graph line] (222) to (203);
		\draw [style=graph line] (226) to (225);
		\draw [style=graph line] (225) to (220);
		\draw [style=graph line] (225) to (219);
		\draw [style=graph line] (227) to (223);
		\draw [style=graph line] (227) to (218);
		\draw [style=graph line] (228.center) to (227);
		\draw [style=graph line] (229) to (212);
		\draw [style=graph line] (229) to (201);
		\draw [style=graph line] (230.center) to (229);
		\draw [style=graph line] (231.center) to (213);
		\draw [style=graph line] (232.center) to (216);
		\draw [style=graph line] (233.center) to (214);
		\draw [style=graph line] (234.center) to (212);
		\draw [style=graph line] (235.center) to (223);
		\draw [style=graph line] (236.center) to (215);
		\draw [style=graph line] (237) to (199);
		\draw [style=graph line] (237) to (200);
		\draw [style=graph line] (238) to (237);
		\draw [style=arrow] (239.center) to (240.center);
	\end{pgfonlayer}
\end{tikzpicture}

%% file: figures/z-spider.tikz
\begin{tikzpicture}
	\begin{pgfonlayer}{nodelayer}
		\node [style={Z phase dot}] (0) at (0, 0) {$k\pi$};
		\node [style=none] (1) at (1.5, 0.75) {};
		\node [style=none] (2) at (-1.5, 0.75) {};
		\node [style=none] (3) at (-1.5, -0.75) {};
		\node [style=none] (4) at (1.5, -0.75) {};
		\node [style=none] (9) at (1.5, 0) {{\footnotesize $n$}};
		\node [style=none] (10) at (-1.5, 0) {{\footnotesize $m$}};
		\node [style=none] (13) at (1, 0.25) {$\vdots$};
		\node [style=none] (14) at (-1, 0.25) {$\vdots$};
	\end{pgfonlayer}
	\begin{pgfonlayer}{edgelayer}
		\draw [in=-135, out=0, looseness=0.75] (3.center) to (0);
		\draw [in=180, out=-45, looseness=0.75] (0) to (4.center);
		\draw [in=180, out=45, looseness=0.75] (0) to (1.center);
		\draw [in=135, out=0, looseness=0.75] (2.center) to (0);
	\end{pgfonlayer}
\end{tikzpicture}

%% file: figures/x-spider.tikz
\begin{tikzpicture}
	\begin{pgfonlayer}{nodelayer}
		\node [style=X phase dot] (0) at (0, 0) {$k\pi$};
		\node [style=none] (1) at (1.5, 0.75) {};
		\node [style=none] (2) at (-1.5, 0.75) {};
		\node [style=none] (3) at (-1.5, -0.75) {};
		\node [style=none] (4) at (1.5, -0.75) {};
		\node [style=none] (9) at (1.5, 0) {{\footnotesize $n$}};
		\node [style=none] (10) at (-1.5, 0) {{\footnotesize $m$}};
		\node [style=none] (13) at (1, 0.25) {$\vdots$};
		\node [style=none] (14) at (-1, 0.25) {$\vdots$};
	\end{pgfonlayer}
	\begin{pgfonlayer}{edgelayer}
		\draw [in=-135, out=0, looseness=0.75] (3.center) to (0);
		\draw [in=180, out=-45, looseness=0.75] (0) to (4.center);
		\draw [in=180, out=45, looseness=0.75] (0) to (1.center);
		\draw [in=135, out=0, looseness=0.75] (2.center) to (0);
	\end{pgfonlayer}
\end{tikzpicture}

%% file: figures/translation-guide.tikz
\begin{tikzpicture}
	\begin{pgfonlayer}{nodelayer}
		\node [style=none] (183) at (-5.5, 1.25) {$\ket{0}$};
		\node [style=none] (184) at (-3.5, 1.25) {$\longmapsto$};
		\node [style=X dot] (185) at (-2, 1.25) {};
		\node [style=none] (186) at (-1, 1.25) {};
		\node [style=none] (187) at (-5.5, 0) {$\ket{+}$};
		\node [style=none] (188) at (-3.5, 0) {$\longmapsto$};
		\node [style=Z dot] (189) at (-2, 0) {};
		\node [style=none] (190) at (-1, 0) {};
		\node [style=none] (201) at (-6.25, -1.5) {};
		\node [style=none] (202) at (-3.5, -1.5) {$\longmapsto$};
		\node [style=Z phase dot] (203) at (-1, -1.5) {$k\pi$};
		\node [style=none] (204) at (-2, -1.5) {};
		\node [style=none] (215) at (1.25, -1.5) {};
		\node [style=none] (216) at (4, -1.5) {$\longmapsto$};
		\node [style=X phase dot] (217) at (6.5, -1.5) {$k\pi$};
		\node [style=none] (218) at (5.5, -1.5) {};
		\node [style=none] (222) at (2.5, 1.25) {};
		\node [style=black dot] (223) at (1.75, 1.25) {};
		\node [style=none] (224) at (1, 1.25) {};
		\node [style=none] (225) at (1, 0.25) {};
		\node [style=none] (226) at (2.5, 0.25) {};
		\node [style=none] (229) at (7, 1.25) {};
		\node [style=Z dot] (230) at (6.25, 1.25) {};
		\node [style=none] (231) at (5.5, 1.25) {};
		\node [style=none] (232) at (5.5, 0.25) {};
		\node [style=none] (233) at (7, 0.25) {};
		\node [style=none] (235) at (4, 0.75) {$\longmapsto$};
		\node [style=X dot] (236) at (6.25, 0.25) {};
		\node [style=not] (237) at (1.75, 0.25) {};
		\node [style=meter] (238) at (2.25, -1.5) {};
		\node [style=sLabel] (239) at (2.125, -1.375) {\tiny $Z$};
		\node [style=meter] (240) at (-5.125, -1.5) {};
		\node [style=sLabel] (241) at (-5.25, -1.375) {\tiny $X$};
		\node [style=none] (243) at (1.25, -3) {};
		\node [style=none] (244) at (4, -3) {$\longmapsto$};
		\node [style=Z phase dot] (245) at (6, -3) {$\pi$};
		\node [style=none] (246) at (5.25, -3) {};
		\node [style=box] (247) at (2, -3) {};
		\node [style=sLabel] (248) at (2, -3) {\tiny $Z$};
		\node [style=none] (249) at (6.75, -3) {};
		\node [style=none] (250) at (2.75, -3) {};
		\node [style=none] (251) at (-6.25, -3) {};
		\node [style=none] (252) at (-3.5, -3) {$\longmapsto$};
		\node [style=X phase dot] (253) at (-1.5, -3) {$\pi$};
		\node [style=none] (254) at (-2.25, -3) {};
		\node [style=box] (255) at (-5.5, -3) {};
		\node [style=sLabel] (256) at (-5.5, -3) {\tiny $X$};
		\node [style=none] (257) at (-0.75, -3) {};
		\node [style=none] (258) at (-4.75, -3) {};
	\end{pgfonlayer}
	\begin{pgfonlayer}{edgelayer}
		\draw (186.center) to (185);
		\draw (190.center) to (189);
		\draw (204.center) to (203);
		\draw (218.center) to (217);
		\draw (223) to (224.center);
		\draw (222.center) to (223);
		\draw (230) to (231.center);
		\draw (229.center) to (230);
		\draw (225.center) to (237);
		\draw (237) to (226.center);
		\draw (237) to (223);
		\draw (232.center) to (236);
		\draw (236) to (233.center);
		\draw (236) to (230);
		\draw (238) to (215.center);
		\draw (240) to (201.center);
		\draw (246.center) to (245);
		\draw (247) to (243.center);
		\draw (247) to (250.center);
		\draw (245) to (249.center);
		\draw (254.center) to (253);
		\draw (255) to (251.center);
		\draw (255) to (258.center);
		\draw (253) to (257.center);
	\end{pgfonlayer}
\end{tikzpicture}

%% file: figures/zx-axioms.tikz
\begin{tikzpicture}
	\begin{pgfonlayer}{nodelayer}
		\node [style=Z phase dot] (1) at (-9.25, 1.125) {$k_1 \pi$};
		\node [style=none] (2) at (-10.25, 1.75) {};
		\node [style=none] (3) at (-10.25, 0.5) {};
		\node [style=none] (4) at (-7.75, 1.75) {};
		\node [style=none] (5) at (-7.75, 0.5) {};
		\node [style=none] (6) at (-7.5, 1.25) {$\vdots$};
		\node [style=none] (7) at (-10, 1.25) {$\vdots$};
		\node [style=Z phase dot] (8) at (-8.25, -1.125) {$k_2 \pi$};
		\node [style=none] (9) at (-9.75, -0.5) {};
		\node [style=none] (10) at (-9.75, -1.75) {};
		\node [style=none] (11) at (-7.25, -0.5) {};
		\node [style=none] (12) at (-7.25, -1.75) {};
		\node [style=none] (13) at (-7.5, -1) {$\vdots$};
		\node [style=none] (14) at (-10, -1) {$\vdots$};
		\node [style=Z phase dot] (15) at (-3, 0) {$(k_1 + k_2) \pi$};
		\node [style=none] (16) at (-1.5, 1.75) {};
		\node [style=none] (17) at (-1.5, 0.5) {};
		\node [style=none] (18) at (-1.5, -0.5) {};
		\node [style=none] (19) at (-1.5, -1.75) {};
		\node [style=none] (20) at (-4.5, 1.75) {};
		\node [style=none] (21) at (-4.5, 0.5) {};
		\node [style=none] (22) at (-4.5, -0.5) {};
		\node [style=none] (23) at (-4.5, -1.75) {};
		\node [style=none] (24) at (-1.75, 1.25) {$\vdots$};
		\node [style=none] (25) at (-1.75, -1) {$\vdots$};
		\node [style=none] (26) at (-4.25, 1.25) {$\vdots$};
		\node [style=none] (27) at (-4.25, -1) {$\vdots$};
		\node [style=none] (28) at (-10.25, -0.5) {};
		\node [style=none] (29) at (-10.25, -1.75) {};
		\node [style=none] (30) at (-7.25, 1.75) {};
		\node [style=none] (31) at (-7.25, 0.5) {};
		\node [style=X phase dot] (67) at (0.5, -1.25) {$\pi$};
		\node [style=Z phase dot] (68) at (1.5, -1.25) {$k \pi$};
		\node [style=none] (70) at (2.25, -0.5) {};
		\node [style=none] (71) at (2.25, -2) {};
		\node [style=none] (72) at (-0.25, -1.25) {};
		\node [style=Z phase dot] (73) at (6, -1.25) {$k \pi$};
		\node [style=X phase dot] (74) at (7, -0.5) {$\pi$};
		\node [style=X phase dot] (75) at (7, -2) {$\pi$};
		\node [style=none] (76) at (5, -1.25) {};
		\node [style=none] (77) at (7.75, -0.5) {};
		\node [style=none] (78) at (7.75, -2) {};
		\node [style=none] (88) at (0.75, 1.5) {};
		\node [style=none] (89) at (2.75, 1.5) {};
		\node [style=none] (90) at (4.75, 1.5) {};
		\node [style=none] (91) at (6.75, 1.5) {};
		\node [style=Z dot] (95) at (1.75, 1.5) {};
		\node [style=none] (177) at (2.5, -2) {};
		\node [style=none] (178) at (2.5, -0.5) {};
		\node [style=none] (182) at (3.75, 1.75) {$\xlongequal[]{\scriptsize\textsc{id}}$};
		\node [style=none] (185) at (3.75, -1.25) {$\xlongequal[\textsc{copy}]{\scriptsize\pi}$};
		\node [style=none] (186) at (-5.75, 0.25) {$\xlongequal[]{\scriptsize\textsc{fusion}}$};
	\end{pgfonlayer}
	\begin{pgfonlayer}{edgelayer}
		\draw (1)
			 to [in=180, out=30] (4.center)
			 to (30.center);
		\draw (1)
			 to [in=180, out=-30] (5.center)
			 to (31.center);
		\draw [in=360, out=120] (1) to (2.center);
		\draw [in=0, out=-120] (1) to (3.center);
		\draw [in=180, out=60] (8) to (11.center);
		\draw [in=180, out=-60] (8) to (12.center);
		\draw (28.center)
			 to (9.center)
			 to [in=150, out=360] (8);
		\draw (8)
			 to [in=0, out=-150] (10.center)
			 to (29.center);
		\draw [in=180, out=30] (15) to (17.center);
		\draw [in=180, out=60] (15) to (16.center);
		\draw [in=180, out=-30] (15) to (18.center);
		\draw [in=180, out=-60] (15) to (19.center);
		\draw [in=0, out=120] (15) to (20.center);
		\draw [in=0, out=150] (15) to (21.center);
		\draw [in=0, out=-150] (15) to (22.center);
		\draw [in=0, out=-120] (15) to (23.center);
		\draw [in=135, out=-45] (1) to (8);
		\draw (68) to (67);
		\draw [in=180, out=60] (68) to (70.center);
		\draw [in=180, out=-60] (68) to (71.center);
		\draw (72.center) to (67);
		\draw [in=180, out=60] (73) to (74);
		\draw [in=180, out=-60] (73) to (75);
		\draw (75) to (78.center);
		\draw (77.center) to (74);
		\draw (76.center) to (73);
		\draw (88.center) to (95);
		\draw (95) to (89.center);
		\draw (90.center) to (91.center);
		\draw (177.center) to (71.center);
		\draw (178.center) to (70.center);
	\end{pgfonlayer}
\end{tikzpicture}

%% file: figures/steane-zero-faulty.tikz
\begin{tikzpicture}
	\begin{pgfonlayer}{nodelayer}
		\node [style=Z dot] (0) at (-1, 1.5) {};
		\node [style=Z dot] (1) at (-1, 0) {};
		\node [style=Z dot] (2) at (-1, -1.5) {};
		\node [style=X dot] (3) at (1, 3) {};
		\node [style=X dot] (4) at (1, 2) {};
		\node [style=X dot] (5) at (1, 1) {};
		\node [style=X dot] (6) at (1, 0) {};
		\node [style=X dot] (7) at (1, -1) {};
		\node [style=X dot] (8) at (1, -2) {};
		\node [style=X dot] (9) at (1, -3) {};
		\node [style=none] (10) at (2.25, 3) {};
		\node [style=none] (11) at (2.25, 2) {};
		\node [style=none] (12) at (2.25, 1) {};
		\node [style=none] (13) at (2.25, 0) {};
		\node [style=none] (14) at (2.25, -1) {};
		\node [style=none] (15) at (2.25, -2) {};
		\node [style=none] (16) at (2.25, -3) {};
	\end{pgfonlayer}
	\begin{pgfonlayer}{edgelayer}
		\draw (0) to (3);
		\draw (0) to (4);
		\draw (0) to (5);
		\draw (0) to (6);
		\draw (1) to (4);
		\draw (1) to (5);
		\draw (1) to (7);
		\draw (1) to (8);
		\draw (2) to (5);
		\draw (2) to (6);
		\draw (2) to (8);
		\draw (2) to (9);
		\draw (16.center) to (9);
		\draw (15.center) to (8);
		\draw (14.center) to (7);
		\draw (13.center) to (6);
		\draw (12.center) to (5);
		\draw (11.center) to (4);
		\draw (10.center) to (3);
	\end{pgfonlayer}
\end{tikzpicture}

%% file: figures/steane-zero-dual.tikz
\begin{tikzpicture}
	\begin{pgfonlayer}{nodelayer}
		\node [style=X dot] (17) at (-1, -2.25) {};
		\node [style=X dot] (0) at (-1, 2.25) {};
		\node [style=X dot] (1) at (-1, .75) {};
		\node [style=X dot] (2) at (-1, -.75) {};
		\node [style=Z dot] (3) at (1, 3) {};
		\node [style=Z dot] (4) at (1, 2) {};
		\node [style=Z dot] (5) at (1, 1) {};
		\node [style=Z dot] (6) at (1, 0) {};
		\node [style=Z dot] (7) at (1, -1) {};
		\node [style=Z dot] (8) at (1, -2) {};
		\node [style=Z dot] (9) at (1, -3) {};
		\node [style=none] (10) at (2.25, 3) {};
		\node [style=none] (11) at (2.25, 2) {};
		\node [style=none] (12) at (2.25, 1) {};
		\node [style=none] (13) at (2.25, 0) {};
		\node [style=none] (14) at (2.25, -1) {};
		\node [style=none] (15) at (2.25, -2) {};
		\node [style=none] (16) at (2.25, -3) {};
	\end{pgfonlayer}
	\begin{pgfonlayer}{edgelayer}
		\draw (17) to (7);
		\draw (17) to (8);
		\draw (17) to (9);
		\draw (0) to (3);
		\draw (0) to (4);
		\draw (0) to (5);
		\draw (0) to (6);
		\draw (1) to (4);
		\draw (1) to (5);
		\draw (1) to (7);
		\draw (1) to (8);
		\draw (2) to (5);
		\draw (2) to (6);
		\draw (2) to (8);
		\draw (2) to (9);
		\draw (16.center) to (9);
		\draw (15.center) to (8);
		\draw (14.center) to (7);
		\draw (13.center) to (6);
		\draw (12.center) to (5);
		\draw (11.center) to (4);
		\draw (10.center) to (3);
	\end{pgfonlayer}
\end{tikzpicture}

%% file: figures/non-fe-rewrite.tikz
\begin{tikzpicture}
	\begin{pgfonlayer}{nodelayer}
		\node [style=Z dot] (22) at (-2.75, 0.75) {};
		\node [style=none] (23) at (-1.5, 1.5) {};
		\node [style=none] (24) at (-1.5, 0) {};
		\node [style=none] (26) at (-4, 1.5) {};
		\node [style=none] (27) at (-4, 0) {};
		\node [style=none] (29) at (0, 0.75) {$\NotFaultEq$};
		\node [style=Z dot] (32) at (3.75, 0.75) {};
		\node [style=none] (33) at (4.75, 1.5) {};
		\node [style=none] (34) at (4.75, 0) {};
		\node [style=Z dot] (35) at (2.5, 0.75) {};
		\node [style=none] (36) at (1.5, 1.5) {};
		\node [style=none] (37) at (1.5, 0) {};
	\end{pgfonlayer}
	\begin{pgfonlayer}{edgelayer}
		\draw (22) to (23.center);
		\draw (22) to (24.center);
		\draw (32) to (33.center);
		\draw (32) to (34.center);
		\draw (35) to (32);
		\draw (36.center) to (35);
		\draw (35) to (37.center);
		\draw (22) to (26.center);
		\draw (22) to (27.center);
	\end{pgfonlayer}
\end{tikzpicture}

%% file: figures/non-fe-rewrite-proof.tikz
\begin{tikzpicture}
	\begin{pgfonlayer}{nodelayer}
		\node [style=Z dot] (0) at (-2.75, 0) {};
		\node [style=none] (1) at (-1.75, 1) {};
		\node [style=none] (2) at (-1.75, -1) {};
		\node [style=Z dot] (3) at (-4.75, 0) {};
		\node [style=none] (4) at (-5.75, 1) {};
		\node [style=none] (5) at (-5.75, -1) {};
		\node [style=X dot] (6) at (-3.75, 0) {$\pi$};
		\node [style=none] (7) at (0, 0) {$\xlongequal[\textsc{copy}]{\pi}$};
		\node [style=Z dot] (8) at (3.75, 0) {};
		\node [style=none] (9) at (5.25, 1.25) {};
		\node [style=none] (10) at (5.25, -1.25) {};
		\node [style=Z dot] (11) at (2.75, 0) {};
		\node [style=none] (12) at (1.75, 1) {};
		\node [style=none] (13) at (1.75, -1) {};
		\node [style=X dot] (14) at (4.5, 0.65) {$\pi$};
		\node [style=X dot] (15) at (4.5, -0.7) {$\pi$};
	\end{pgfonlayer}
	\begin{pgfonlayer}{edgelayer}
		\draw (0) to (1.center);
		\draw (0) to (2.center);
		\draw (3) to (0);
		\draw (4.center) to (3);
		\draw (3) to (5.center);
		\draw (8) to (9.center);
		\draw (8) to (10.center);
		\draw (11) to (8);
		\draw (12.center) to (11);
		\draw (11) to (13.center);
	\end{pgfonlayer}
\end{tikzpicture}

%% file: figures/fe-rewrites1.tikz
\begin{tikzpicture}
	\begin{pgfonlayer}{nodelayer}
		\node [style=none] (0) at (-13.75, 0) {};
		\node [style=none] (1) at (-12.25, 0) {};
		\node [style=none] (2) at (-10.75, 0) {$\FaultEq$};
		\node [style=none] (3) at (-9, 0) {};
		\node [style=none] (4) at (-7, 0) {};
		\node [style=Z dot] (5) at (-8, 0) {};
		\node [style=none] (6) at (-10.75, -0.75) {\scriptsize$\textsc{Elim}_{\text{fe}}$};
		\node [style=Z dot] (7) at (-3, 0) {};
		\node [style=none] (8) at (-1.75, 1) {};
		\node [style=none] (9) at (-1.75, 0.5) {};
		\node [style=none] (10) at (-1.75, -1) {};
		\node [style=none, rotate=90] (11) at (-1.75, -0.25) {...};
		\node [style=none] (12) at (0, 0) {$\FaultEq$};
		\node [style=Z dot] (13) at (3, 0) {};
		\node [style=none] (14) at (4.25, 1) {};
		\node [style=none] (15) at (4.25, 0.5) {};
		\node [style=none] (16) at (4.25, -1) {};
		\node [style=none, rotate=90] (17) at (4.25, -0.25) {...};
		\node [style=Z dot] (18) at (1.75, 0) {};
		\node [style=none] (19) at (0, -1) {\scriptsize $\textsc{Fuse-1}_{\text{fe}}$};
		\node [style=Z dot] (20) at (9.25, 0) {};
		\node [style=Z dot] (21) at (14.5, 0.5) {};
		\node [style=Z dot] (22) at (14.5, -0.5) {};
		\node [style=Z dot] (23) at (15.5, 0.5) {};
		\node [style=Z dot] (24) at (15.5, -0.5) {};
		\node [style=none] (25) at (10.25, 1) {};
		\node [style=none] (26) at (8.25, 1) {};
		\node [style=none] (27) at (10.25, -1) {};
		\node [style=none] (28) at (8.25, -1) {};
		\node [style=none] (29) at (13.75, 1.25) {};
		\node [style=none] (30) at (16.25, 1.25) {};
		\node [style=none] (31) at (13.75, -1.25) {};
		\node [style=none] (32) at (16.25, -1.25) {};
		\node [style=none] (33) at (12, 0.25) {$\FaultEq$};
		\node [style=none] (34) at (12, -0.5) {\scriptsize $\textsc{Fuse-4}_{\text{fe}}$};
	\end{pgfonlayer}
	\begin{pgfonlayer}{edgelayer}
		\draw (0.center) to (1.center);
		\draw (3.center) to (5);
		\draw (5) to (4.center);
		\draw (7) to (8.center);
		\draw (7) to (9.center);
		\draw (7) to (10.center);
		\draw (13) to (14.center);
		\draw (13) to (15.center);
		\draw (13) to (16.center);
		\draw (18) to (13);
		\draw (20) to (26.center);
		\draw (25.center) to (20);
		\draw (20) to (27.center);
		\draw (28.center) to (20);
		\draw (21) to (23);
		\draw (23) to (24);
		\draw (24) to (22);
		\draw (22) to (21);
		\draw (21) to (29.center);
		\draw (23) to (30.center);
		\draw (24) to (32.center);
		\draw (22) to (31.center);
	\end{pgfonlayer}
\end{tikzpicture}

%% file: figures/cat-extraction-overview.tikz
\begin{tikzpicture}
	\begin{pgfonlayer}{nodelayer}
		\node [style=sLabel, font={\scriptsize}, font={\footnotesize}] (0) at (2.15, 2.5) {\textbf{(a) Marked graph $G$}};
		\node [style=white dot, minimum size=3.7mm, inner sep=0pt, outer sep=0pt, font={\tiny}] (1) at (2.15, 0) {$0$};
		\node [style=white dot, minimum size=3.7mm, inner sep=0pt, outer sep=0pt, font={\tiny}] (2) at (3.67, 0) {$1$};
		\node [style=white dot, minimum size=3.7mm, inner sep=0pt, outer sep=0pt, font={\tiny}] (3) at (2.15, -1.52) {$2$};
		\node [style=white dot, minimum size=3.7mm, inner sep=0pt, outer sep=0pt, font={\tiny}] (4) at (1.295, 0.95) {$3$};
		\node [style=none, outer sep=0pt] (5) at (1.675, 0.38) {};
		\node [style=none, outer sep=0pt] (6) at (2.815, 0.095) {};
		\node [style=none, outer sep=0pt] (7) at (2.055, -0.665) {};
		\node [style=none, outer sep=0pt] (8) at (3.005, -0.95) {};
		\node [style=none, outer sep=0pt] (9) at (3.005, 0.95) {};
		\node [style=none, outer sep=0pt] (10) at (2.15, 1.52) {};
		\node [style=none, outer sep=0pt] (11) at (1.295, -0.95) {};
		\node [style=none, outer sep=0pt] (12) at (0.63, 0) {};
		\node [style=none, outer sep=0pt] (13) at (1.753, 0.4502) {};
		\node [style=none, outer sep=0pt] (14) at (1.597, 0.3098) {};
		\node [style=sLabel, font={\scriptsize}, font={\tiny}] (15) at (1.515, 0.24) {$4$};
		\node [style=none, outer sep=0pt] (16) at (2.815, -0.01) {};
		\node [style=none, outer sep=0pt] (17) at (2.815, 0.2) {};
		\node [style=sLabel, font={\scriptsize}, font={\tiny}] (18) at (2.855, 0.285) {$5$};
		\node [style=none, outer sep=0pt] (19) at (1.95, -0.665) {};
		\node [style=none, outer sep=0pt] (20) at (2.16, -0.665) {};
		\node [style=sLabel, font={\scriptsize}, font={\tiny}, text={teal!80!black}] (21) at (1.785, -0.685) {$q$};
		\node [style=none, outer sep=0pt] (22) at (2.9308, -0.8758) {};
		\node [style=none, outer sep=0pt] (23) at (3.0792, -1.0242) {};
		\node [style=sLabel, font={\scriptsize}, font={\tiny}] (24) at (3.195, -1.03) {$7$};
		\node [style=none, outer sep=0pt] (25) at (2.15, 1.415) {};
		\node [style=none, outer sep=0pt] (26) at (2.15, 1.625) {};
		\node [style=sLabel, font={\scriptsize}, font={\tiny}] (27) at (2.15, 1.72) {$9$};
		\node [style=none, outer sep=0pt] (28) at (2.9308, 0.8758) {};
		\node [style=none, outer sep=0pt] (29) at (3.0792, 1.0242) {};
		\node [style=sLabel, font={\scriptsize}, font={\tiny}] (30) at (3.175, 1.08) {$8$};
		\node [style=none, outer sep=0pt] (31) at (0.525, 0) {};
		\node [style=none, outer sep=0pt] (32) at (0.735, 0) {};
		\node [style=sLabel, font={\scriptsize}, font={\tiny}] (33) at (0.43, 0) {$11$};
		\node [style=none, outer sep=0pt] (34) at (1.2208, -1.0242) {};
		\node [style=none, outer sep=0pt] (35) at (1.3692, -0.8758) {};
		\node [style=sLabel, font={\scriptsize}, font={\tiny}] (36) at (1.095, -1.07) {$10$};
		\node [style=sLabel, font={\scriptsize}] (37) at (2.15, -2.05) {8 marked output positions};
		\node [style=sLabel, font={\scriptsize}, font={\footnotesize}] (38) at (7.65, 2.5) {\textbf{(b) Expanded graph $\widetilde G$}};
		\node [style=white dot, minimum size=3.7mm, inner sep=0pt, outer sep=0pt, font={\tiny}] (39) at (7.65, 0) {$0$};
		\node [style=white dot, minimum size=3.7mm, inner sep=0pt, outer sep=0pt, font={\tiny}] (40) at (9.17, 0) {$1$};
		\node [style=white dot, minimum size=3.7mm, inner sep=0pt, outer sep=0pt, font={\tiny}] (41) at (7.65, -1.52) {$2$};
		\node [style=white dot, minimum size=3.7mm, inner sep=0pt, outer sep=0pt, font={\tiny}] (42) at (6.795, 0.95) {$3$};
		\node [style=white dot, minimum size=3.7mm, inner sep=0pt, outer sep=0pt, font={\tiny}] (43) at (7.175, 0.38) {$4$};
		\node [style=white dot, minimum size=3.7mm, inner sep=0pt, outer sep=0pt, font={\tiny}] (44) at (8.315, 0.095) {$5$};
		\node [style=white dot, minimum size=3.7mm, inner sep=0pt, outer sep=0pt, font={\tiny}, draw={teal!80!black}, line width=1pt] (45) at (7.555, -0.665) {$6$};
		\node [style=white dot, minimum size=3.7mm, inner sep=0pt, outer sep=0pt, font={\tiny}] (46) at (8.505, -0.95) {$7$};
		\node [style=white dot, minimum size=3.7mm, inner sep=0pt, outer sep=0pt, font={\tiny}] (47) at (8.505, 0.95) {$8$};
		\node [style=white dot, minimum size=3.7mm, inner sep=0pt, outer sep=0pt, font={\tiny}] (48) at (7.65, 1.52) {$9$};
		\node [style=white dot, minimum size=3.7mm, inner sep=0pt, outer sep=0pt, font={\tiny}] (49) at (6.795, -0.95) {$10$};
		\node [style=white dot, minimum size=3.7mm, inner sep=0pt, outer sep=0pt, font={\tiny}] (50) at (6.13, 0) {$11$};
		\node [style=sLabel, font={\scriptsize}, font={\tiny}] (51) at (7.77, 0.39) {root 0};
		\node [style=sLabel, font={\scriptsize}, text={teal!80!black}] (52) at (7.215, -0.665) {$q$};
		\node [style=sLabel, font={\scriptsize}, text={orange!85!black}] (53) at (8.26, -1.54) {$e_f$};
		\node [style=sLabel, font={\scriptsize}] (54) at (7.65, -2.05) {12 vertices; 11 tree edges};
		\node [style=sLabel, font={\scriptsize}, font={\footnotesize}] (55) at (13.25, 2.5) {\textbf{(c) Trivalent $Z$-graph}};
		\node [style=Z dot, minimum size=2.2mm, outer sep=0pt] (56) at (13.25, 0) {};
		\node [style=Z dot, minimum size=2.2mm, outer sep=0pt] (57) at (14.77, 0) {};
		\node [style=Z dot, minimum size=2.2mm, outer sep=0pt] (58) at (13.25, -1.52) {};
		\node [style=Z dot, minimum size=2.2mm, outer sep=0pt] (59) at (12.395, 0.95) {};
		\node [style=Z dot, minimum size=2.2mm, outer sep=0pt] (60) at (12.775, 0.38) {};
		\node [style=Z dot, minimum size=2.2mm, outer sep=0pt] (61) at (13.915, 0.095) {};
		\node [style=Z dot, minimum size=2.2mm, outer sep=0pt, draw={teal!80!black}, line width=1pt] (62) at (13.155, -0.665) {};
		\node [style=Z dot, minimum size=2.2mm, outer sep=0pt] (63) at (14.105, -0.95) {};
		\node [style=Z dot, minimum size=2.2mm, outer sep=0pt] (64) at (14.105, 0.95) {};
		\node [style=Z dot, minimum size=2.2mm, outer sep=0pt] (65) at (13.25, 1.52) {};
		\node [style=Z dot, minimum size=2.2mm, outer sep=0pt] (66) at (12.395, -0.95) {};
		\node [style=Z dot, minimum size=2.2mm, outer sep=0pt] (67) at (11.73, 0) {};
		\node [style=sLabel, font={\scriptsize}, font={\tiny}] (68) at (13.27, 0.18) {$0$};
		\node [style=sLabel, font={\scriptsize}, font={\tiny}] (69) at (14.94, 0.12) {$1$};
		\node [style=sLabel, font={\scriptsize}, font={\tiny}] (70) at (13.41, -1.62) {$2$};
		\node [style=sLabel, font={\scriptsize}, font={\tiny}] (71) at (12.255, 1.11) {$3$};
		\node [style=sLabel, font={\scriptsize}, font={\tiny}] (72) at (12.895, 0.54) {$4$};
		\node [style=sLabel, font={\scriptsize}, font={\tiny}] (73) at (13.955, -0.085) {$5$};
		\node [style=sLabel, font={\scriptsize}, font={\tiny}] (74) at (13.305, -0.695) {$6$};
		\node [style=sLabel, font={\scriptsize}, font={\tiny}] (75) at (13.975, -0.82) {$7$};
		\node [style=sLabel, font={\scriptsize}, font={\tiny}] (76) at (13.975, 0.82) {$8$};
		\node [style=sLabel, font={\scriptsize}, font={\tiny}] (77) at (13.08, 1.55) {$9$};
		\node [style=sLabel, font={\scriptsize}, font={\tiny}] (78) at (12.595, -0.79) {$10$};
		\node [style=sLabel, font={\scriptsize}, font={\tiny}] (79) at (11.75, 0.18) {$11$};
		\node [style=none, outer sep=0pt] (80) at (12.395, 0.06) {};
		\node [style=none, outer sep=0pt] (81) at (13.995, 0.575) {};
		\node [style=none, outer sep=0pt] (82) at (12.655, -0.565) {};
		\node [style=sLabel, font={\scriptsize}, text={teal!80!black}, anchor=east] (83) at (12.535, -0.565) {$q$};
		\node [style=none, outer sep=0pt] (84) at (14.485, -1.23) {};
		\node [style=none, outer sep=0pt] (85) at (14.475, 1.23) {};
		\node [style=none, outer sep=0pt] (86) at (13.25, 1.98) {};
		\node [style=none, outer sep=0pt] (87) at (12.025, -1.22) {};
		\node [style=none, outer sep=0pt] (88) at (11.28, 0) {};
		\node [style=sLabel, font={\scriptsize}, text={orange!85!black}] (89) at (13.87, -1.54) {$e_f$};
		\node [style=sLabel, font={\scriptsize}] (90) at (13.25, -2.05) {One output leg per marked vertex};
		\node [style=none, outer sep=0pt] (91) at (4.35, 0.15) {};
		\node [style=none, outer sep=0pt] (92) at (5.45, 0.15) {};
		\node [style=none, outer sep=0pt] (93) at (9.85, 0.15) {};
		\node [style=none, outer sep=0pt] (94) at (10.8, 0.15) {};
		\node [style=none, outer sep=0pt] (95) at (4.1, -2.5) {};
		\node [style=none, outer sep=0pt] (96) at (4.65, -2.5) {};
		\node [style=sLabel, font={\scriptsize}, anchor=west] (97) at (4.8, -2.5) {spanning tree};
		\node [style=none, outer sep=0pt] (98) at (7.4, -2.5) {};
		\node [style=none, outer sep=0pt] (99) at (7.95, -2.5) {};
		\node [style=sLabel, font={\scriptsize}, anchor=west] (100) at (8.1, -2.5) {non-tree edge};
		\node [style=none, outer sep=0pt] (101) at (10.95, -2.5) {};
		\node [style=none, outer sep=0pt] (102) at (11.5, -2.5) {};
		\node [style=sLabel, font={\scriptsize}, anchor=west, text={orange!85!black}] (103) at (11.65, -2.5) {$e_f\mapsto f$};
		\node [style=sLabel, font={\scriptsize}, font={\footnotesize}] (104) at (3.9, -3.12) {\textbf{(d) Same $Z$-graph arranged for extraction}};
		\node [style=Z dot, minimum size=2.2mm, outer sep=0pt] (105) at (0.55, -4.69) {};
		\node [style=Z dot, minimum size=2.2mm, outer sep=0pt] (106) at (1.47, -5.13) {};
		\node [style=Z dot, minimum size=2.2mm, outer sep=0pt] (107) at (1.93, -4.69) {};
		\node [style=Z dot, minimum size=2.2mm, outer sep=0pt, draw={teal!80!black}, line width=1pt] (108) at (2.39, -4.25) {};
		\node [style=Z dot, minimum size=2.2mm, outer sep=0pt] (109) at (2.85, -4.25) {};
		\node [style=Z dot, minimum size=2.2mm, outer sep=0pt] (110) at (3.31, -4.25) {};
		\node [style=Z dot, minimum size=2.2mm, outer sep=0pt] (111) at (3.77, -4.25) {};
		\node [style=Z dot, minimum size=2.2mm, outer sep=0pt] (112) at (4.23, -5.13) {};
		\node [style=Z dot, minimum size=2.2mm, outer sep=0pt] (113) at (4.69, -5.13) {};
		\node [style=Z dot, minimum size=2.2mm, outer sep=0pt] (114) at (5.15, -5.13) {};
		\node [style=Z dot, minimum size=2.2mm, outer sep=0pt] (115) at (5.61, -4.69) {};
		\node [style=Z dot, minimum size=2.2mm, outer sep=0pt] (116) at (6.07, -4.69) {};
		\node [style=sLabel, font={\scriptsize}, font={\tiny}] (117) at (0.55, -4.51) {$0$};
		\node [style=sLabel, font={\scriptsize}, font={\tiny}] (118) at (1.47, -4.95) {$4$};
		\node [style=sLabel, font={\scriptsize}, font={\tiny}] (119) at (1.93, -4.51) {$5$};
		\node [style=sLabel, font={\scriptsize}, font={\tiny}] (120) at (2.39, -4.07) {$6$};
		\node [style=sLabel, font={\scriptsize}, font={\tiny}] (121) at (2.85, -4.43) {$2$};
		\node [style=sLabel, font={\scriptsize}, font={\tiny}] (122) at (3.31, -4.07) {$10$};
		\node [style=sLabel, font={\scriptsize}, font={\tiny}] (123) at (3.77, -4.07) {$11$};
		\node [style=sLabel, font={\scriptsize}, font={\tiny}] (124) at (4.23, -4.95) {$3$};
		\node [style=sLabel, font={\scriptsize}, font={\tiny}] (125) at (4.69, -4.95) {$9$};
		\node [style=sLabel, font={\scriptsize}, font={\tiny}] (126) at (5.15, -4.95) {$8$};
		\node [style=sLabel, font={\scriptsize}, font={\tiny}] (127) at (5.61, -4.51) {$1$};
		\node [style=sLabel, font={\scriptsize}, font={\tiny}] (128) at (6.07, -4.87) {$7$};
		\node [style=none, outer sep=0pt] (129) at (2.85, -3.7) {};
		\node [style=none, outer sep=0pt] (130) at (6.07, -3.7) {};
		\node [style=sLabel, font={\scriptsize}, text={orange!85!black}] (131) at (4.46, -3.51) {$e_f$};
		\node [style=none, outer sep=0pt] (132) at (6.898, -4.25) {};
		\node [style=sLabel, font={\scriptsize}, anchor=west] (133) at (6.998, -4.25) {$11$};
		\node [style=none, outer sep=0pt] (134) at (6.898, -4.69) {};
		\node [style=sLabel, font={\scriptsize}, anchor=west] (135) at (6.998, -4.69) {$7$};
		\node [style=none, outer sep=0pt] (136) at (6.898, -5.13) {};
		\node [style=sLabel, font={\scriptsize}, anchor=west] (137) at (6.998, -5.13) {$8$};
		\node [style=none, outer sep=0pt] (138) at (6.898, -5.57) {};
		\node [style=none, outer sep=0pt] (139) at (1.47, -5.57) {};
		\node [style=sLabel, font={\scriptsize}, anchor=west] (140) at (6.998, -5.57) {$4$};
		\node [style=none, outer sep=0pt] (141) at (6.898, -6.01) {};
		\node [style=none, outer sep=0pt] (142) at (1.93, -6.01) {};
		\node [style=sLabel, font={\scriptsize}, anchor=west] (143) at (6.998, -6.01) {$5$};
		\node [style=none, outer sep=0pt] (144) at (6.898, -6.45) {};
		\node [style=none, outer sep=0pt] (145) at (2.39, -6.45) {};
		\node [style=sLabel, font={\scriptsize}, anchor=west, text={teal!80!black}] (146) at (6.998, -6.45) {$q$};
		\node [style=none, outer sep=0pt] (147) at (6.898, -6.89) {};
		\node [style=none, outer sep=0pt] (148) at (3.31, -6.89) {};
		\node [style=sLabel, font={\scriptsize}, anchor=west] (149) at (6.998, -6.89) {$10$};
		\node [style=none, outer sep=0pt] (150) at (6.898, -7.33) {};
		\node [style=none, outer sep=0pt] (151) at (4.69, -7.33) {};
		\node [style=sLabel, font={\scriptsize}, anchor=west] (152) at (6.998, -7.33) {$9$};
		\node [style=sLabel, font={\scriptsize}] (153) at (3.9, -7.7) {Numbers track the same output ports.};
		\node [style=none, outer sep=0pt] (154) at (7.55, -5.1) {};
		\node [style=none, outer sep=0pt] (155) at (8.5, -5.1) {};
		\node [style=sLabel, font={\scriptsize}, font={\tiny}] (156) at (8.02, -4.83) {extract};
		\node [style=sLabel, font={\scriptsize}, font={\footnotesize}] (157) at (12.3, -3.12) {\textbf{(e) Extracted CAT-state circuit}};
		\node [style=sLabel, font={\scriptsize}, anchor=east] (158) at (9.05, -4) {$\ket{+}$};
		\node [style=none, outer sep=0pt] (159) at (9.15, -4) {};
		\node [style=sLabel, font={\scriptsize}, anchor=east] (160) at (9.05, -4.36) {$\ket{0}$};
		\node [style=none, outer sep=0pt] (161) at (9.15, -4.36) {};
		\node [style=sLabel, font={\scriptsize}, anchor=east] (162) at (9.05, -4.72) {$\ket{0}$};
		\node [style=none, outer sep=0pt] (163) at (9.15, -4.72) {};
		\node [style=sLabel, font={\scriptsize}, anchor=east] (164) at (9.05, -5.08) {$\ket{0}$};
		\node [style=none, outer sep=0pt] (165) at (9.15, -5.08) {};
		\node [style=sLabel, font={\scriptsize}, anchor=east] (166) at (9.05, -5.44) {$\ket{0}$};
		\node [style=none, outer sep=0pt] (167) at (9.15, -5.44) {};
		\node [style=sLabel, font={\scriptsize}, anchor=east] (168) at (9.05, -5.8) {$\ket{0}$};
		\node [style=none, outer sep=0pt] (169) at (9.15, -5.8) {};
		\node [style=sLabel, font={\scriptsize}, anchor=east] (170) at (9.05, -6.16) {$\ket{0}$};
		\node [style=none, outer sep=0pt] (171) at (9.15, -6.16) {};
		\node [style=sLabel, font={\scriptsize}, anchor=east] (172) at (9.05, -6.52) {$\ket{0}$};
		\node [style=none, outer sep=0pt] (173) at (9.15, -6.52) {};
		\node [style=sLabel, font={\scriptsize}, anchor=east, text={orange!85!black}] (174) at (9.05, -6.88) {$\ket{0}$};
		\node [style=none, outer sep=0pt] (175) at (9.15, -6.88) {};
		\node [style=sLabel, font={\scriptsize}, anchor=east] (176) at (9.05, -7.24) {$\ket{0}$};
		\node [style=none, outer sep=0pt] (177) at (9.15, -7.24) {};
		\node [style=sLabel, font={\scriptsize}, anchor=east] (178) at (9.05, -7.6) {$\ket{0}$};
		\node [style=none, outer sep=0pt] (179) at (9.15, -7.6) {};
		\node [style=cnot ctrl] (180) at (9.449, -4) {};
		\node [style=not] (181) at (9.449, -4.36) {};
		\node [style=cnot ctrl] (182) at (9.909, -4) {};
		\node [style=not] (183) at (9.909, -4.72) {};
		\node [style=cnot ctrl] (184) at (10.369, -4.72) {};
		\node [style=not] (185) at (10.369, -5.08) {};
		\node [style=cnot ctrl] (186) at (10.829, -4.36) {};
		\node [style=not] (187) at (10.829, -5.44) {};
		\node [style=cnot ctrl] (188) at (11.289, -4) {};
		\node [style=not] (189) at (11.289, -5.8) {};
		\node [style=cnot ctrl, draw={orange!85!black}, fill={orange!85!black}] (190) at (11.749, -4) {};
		\node [style=not, draw={orange!85!black}] (191) at (11.749, -6.88) {};
		\node [style=cnot ctrl] (192) at (12.209, -4) {};
		\node [style=not] (193) at (12.209, -6.16) {};
		\node [style=cnot ctrl] (194) at (12.669, -4) {};
		\node [style=not] (195) at (12.669, -7.24) {};
		\node [style=cnot ctrl] (196) at (13.129, -4.72) {};
		\node [style=not] (197) at (13.129, -7.24) {};
		\node [style=cnot ctrl] (198) at (13.589, -4.72) {};
		\node [style=not] (199) at (13.589, -6.52) {};
		\node [style=cnot ctrl] (200) at (14.049, -4.72) {};
		\node [style=not] (201) at (14.049, -7.6) {};
		\node [style=cnot ctrl] (202) at (14.509, -4.36) {};
		\node [style=not] (203) at (14.509, -7.6) {};
		\node [style=cnot ctrl, draw={orange!85!black}, fill={orange!85!black}] (204) at (14.969, -4.36) {};
		\node [style=not, draw={orange!85!black}] (205) at (14.969, -6.88) {};
		\node [style=none, outer sep=0pt] (206) at (15.406, -4) {};
		\node [style=sLabel, font={\scriptsize}, anchor=west] (207) at (15.526, -4) {$11$};
		\node [style=none, outer sep=0pt] (208) at (15.406, -4.36) {};
		\node [style=sLabel, font={\scriptsize}, anchor=west] (209) at (15.526, -4.36) {$7$};
		\node [style=none, outer sep=0pt] (210) at (15.406, -4.72) {};
		\node [style=sLabel, font={\scriptsize}, anchor=west] (211) at (15.526, -4.72) {$8$};
		\node [style=none, outer sep=0pt] (212) at (15.406, -5.08) {};
		\node [style=sLabel, font={\scriptsize}, anchor=west] (213) at (15.526, -5.08) {$4$};
		\node [style=none, outer sep=0pt] (214) at (15.406, -5.44) {};
		\node [style=sLabel, font={\scriptsize}, anchor=west] (215) at (15.526, -5.44) {$5$};
		\node [style=none, outer sep=0pt] (216) at (15.406, -5.8) {};
		\node [style=sLabel, font={\scriptsize}, anchor=west, text={teal!80!black}] (217) at (15.526, -5.8) {$q$};
		\node [style=none, outer sep=0pt] (218) at (15.406, -6.16) {};
		\node [style=sLabel, font={\scriptsize}, anchor=west] (219) at (15.526, -6.16) {$10$};
		\node [style=none, outer sep=0pt] (220) at (15.406, -6.52) {};
		\node [style=sLabel, font={\scriptsize}, anchor=west] (221) at (15.526, -6.52) {$9$};
		\node [style=meter] (222) at (15.406, -6.88) {};
		\node [style=sLabel, font={\scriptsize}, font={\tiny}, anchor=west] (224) at (15.646, -6.88) {$0$};
		\node [style=sLabel, font={\scriptsize}, text={orange!85!black}, anchor=east] (225) at (8.5, -6.88) {$f:$};
		\node [style=meter] (226) at (15.406, -7.24) {};
		\node [style=sLabel, font={\scriptsize}, font={\tiny}, anchor=west] (228) at (15.646, -7.24) {$0$};
		\node [style=meter] (229) at (15.406, -7.6) {};
		\node [style=sLabel, font={\scriptsize}, font={\tiny}, anchor=west] (231) at (15.646, -7.6) {$0$};
		\node [style=sLabel, font={\scriptsize}] (232) at (12.3, -8) {8 outputs; 3 measured ancillas; 13 CNOTs};
		\node [style=sLabel, font={\scriptsize}, font={\tiny}] (233) at (12.3, -8.34) {Accept when all three $Z$ outcomes are $0$.};
	\end{pgfonlayer}
	\begin{pgfonlayer}{edgelayer}
		\draw [line width=0.65pt] (1) to (5.center);
		\draw [line width=0.65pt] (5.center) to (4);
		\draw [line width=1.15pt] (13.center) to (14.center);
		\draw [line width=0.65pt] (1) to (6.center);
		\draw [line width=0.65pt] (6.center) to (2);
		\draw [line width=1.15pt] (16.center) to (17.center);
		\draw [line width=0.65pt] (1) to (7.center);
		\draw [line width=0.65pt] (7.center) to (3);
		\draw [line width=1.15pt, draw={teal!80!black}] (19.center) to (20.center);
		\draw [line width=0.65pt] (2) to (8.center);
		\draw [line width=0.65pt] (8.center) to (3);
		\draw [line width=1.15pt] (22.center) to (23.center);
		\draw [line width=0.65pt] (4) to (10.center);
		\draw [line width=0.65pt] (10.center) to (9.center);
		\draw [line width=0.65pt] (9.center) to (2);
		\draw [line width=1.15pt] (25.center) to (26.center);
		\draw [line width=1.15pt] (28.center) to (29.center);
		\draw [line width=0.65pt] (4) to (12.center);
		\draw [line width=0.65pt] (12.center) to (11.center);
		\draw [line width=0.65pt] (11.center) to (3);
		\draw [line width=1.15pt] (31.center) to (32.center);
		\draw [line width=1.15pt] (34.center) to (35.center);
		\draw [draw={blue!65!black}, line width=1.1pt] (39) to (43);
		\draw [draw={blue!65!black}, line width=1.1pt] (43) to (42);
		\draw [draw={blue!65!black}, line width=1.1pt] (42) to (48);
		\draw [draw={blue!65!black}, line width=1.1pt] (48) to (47);
		\draw [draw={blue!65!black}, line width=1.1pt] (39) to (44);
		\draw [draw={blue!65!black}, line width=1.1pt] (44) to (40);
		\draw [draw={blue!65!black}, line width=1.1pt] (40) to (46);
		\draw [draw={blue!65!black}, line width=1.1pt] (39) to (45);
		\draw [draw={blue!65!black}, line width=1.1pt] (45) to (41);
		\draw [draw={blue!65!black}, line width=1.1pt] (41) to (49);
		\draw [draw={blue!65!black}, line width=1.1pt] (49) to (50);
		\draw [dashed, line width=0.8pt, draw={black!45}] (42) to (50);
		\draw [dashed, line width=0.8pt, draw={black!45}] (40) to (47);
		\draw [dashed, line width=0.8pt, draw={orange!85!black}] (41) to (46);
		\draw [draw={blue!65!black}, line width=0.7pt] (56) to (60);
		\draw [draw={blue!65!black}, line width=0.7pt] (60) to (59);
		\draw [draw={blue!65!black}, line width=0.7pt] (59) to (65);
		\draw [draw={blue!65!black}, line width=0.7pt] (65) to (64);
		\draw [draw={blue!65!black}, line width=0.7pt] (56) to (61);
		\draw [draw={blue!65!black}, line width=0.7pt] (61) to (57);
		\draw [draw={blue!65!black}, line width=0.7pt] (57) to (63);
		\draw [draw={blue!65!black}, line width=0.7pt] (56) to (62);
		\draw [draw={blue!65!black}, line width=0.7pt] (62) to (58);
		\draw [draw={blue!65!black}, line width=0.7pt] (58) to (66);
		\draw [draw={blue!65!black}, line width=0.7pt] (66) to (67);
		\draw [dashed, line width=0.8pt, draw={black!45}] (59) to (67);
		\draw [dashed, line width=0.8pt, draw={black!45}] (57) to (64);
		\draw [dashed, line width=0.8pt, draw={orange!85!black}] (58) to (63);
		\draw [line width=0.6pt] (60) to (80.center);
		\draw [line width=0.6pt] (61) to (81.center);
		\draw [line width=0.6pt, draw={teal!80!black}] (62) to (82.center);
		\draw [line width=0.6pt] (63) to (84.center);
		\draw [line width=0.6pt] (64) to (85.center);
		\draw [line width=0.6pt] (65) to (86.center);
		\draw [line width=0.6pt] (66) to (87.center);
		\draw [line width=0.6pt] (67) to (88.center);
		\draw [style=arrow, line width=0.7pt] (91.center) to (92.center);
		\draw [style=arrow, line width=0.7pt] (93.center) to (94.center);
		\draw [draw={blue!65!black}, line width=1.1pt] (95.center) to (96.center);
		\draw [dashed, draw={black!45}, line width=.8pt] (98.center) to (99.center);
		\draw [dashed, draw={orange!85!black}, line width=.8pt] (101.center) to (102.center);
		\draw [draw={blue!65!black}, line width=.7pt] (105) to (106);
		\draw [draw={blue!65!black}, line width=.7pt] (106) to (112);
		\draw [draw={blue!65!black}, line width=.7pt] (112) to (113);
		\draw [draw={blue!65!black}, line width=.7pt] (113) to (114);
		\draw [draw={blue!65!black}, line width=.7pt] (105) to (107);
		\draw [draw={blue!65!black}, line width=.7pt] (107) to (115);
		\draw [draw={blue!65!black}, line width=.7pt] (115) to (116);
		\draw [draw={blue!65!black}, line width=.7pt] (105) to (108);
		\draw [draw={blue!65!black}, line width=.7pt] (108) to (109);
		\draw [draw={blue!65!black}, line width=.7pt] (109) to (110);
		\draw [draw={blue!65!black}, line width=.7pt] (110) to (111);
		\draw [dashed, draw={orange!85!black}, line width=.8pt] (109) to (129.center);
		\draw [dashed, draw={orange!85!black}, line width=.8pt] (129.center) to (130.center);
		\draw [dashed, draw={orange!85!black}, line width=.8pt] (130.center) to (116);
		\draw [dashed, draw={black!45}, line width=.8pt] (111) to (112);
		\draw [dashed, draw={black!45}, line width=.8pt] (114) to (115);
		\draw [line width=.6pt] (111) to (132.center);
		\draw [line width=.6pt] (116) to (134.center);
		\draw [line width=.6pt] (114) to (136.center);
		\draw [line width=.6pt] (106) to (139.center);
		\draw [line width=.6pt] (139.center) to (138.center);
		\draw [line width=.6pt] (107) to (142.center);
		\draw [line width=.6pt] (142.center) to (141.center);
		\draw [line width=.6pt, draw={teal!80!black}] (108) to (145.center);
		\draw [line width=.6pt, draw={teal!80!black}] (145.center) to (144.center);
		\draw [line width=.6pt] (110) to (148.center);
		\draw [line width=.6pt] (148.center) to (147.center);
		\draw [line width=.6pt] (113) to (151.center);
		\draw [line width=.6pt] (151.center) to (150.center);
		\draw [style=arrow, line width=.75pt] (154.center) to (155.center);
		\draw [line width=.6pt] (180) to (181);
		\draw [line width=.6pt] (182) to (183);
		\draw [line width=.6pt] (184) to (185);
		\draw [line width=.6pt] (186) to (187);
		\draw [line width=.6pt] (188) to (189);
		\draw [draw={orange!85!black}, line width=.85pt] (190) to (191);
		\draw [line width=.6pt] (192) to (193);
		\draw [line width=.6pt] (194) to (195);
		\draw [line width=.6pt] (196) to (197);
		\draw [line width=.6pt] (198) to (199);
		\draw [line width=.6pt] (200) to (201);
		\draw [line width=.6pt] (202) to (203);
		\draw [draw={orange!85!black}, line width=.85pt] (204) to (205);
		\draw [line width=.45pt] (159.center) to (180);
		\draw [line width=.45pt] (180) to (182);
		\draw [line width=.45pt] (182) to (188);
		\draw [line width=.45pt] (188) to (190);
		\draw [line width=.45pt] (190) to (192);
		\draw [line width=.45pt] (192) to (194);
		\draw [line width=.45pt] (194) to (206.center);
		\draw [line width=.45pt] (161.center) to (181);
		\draw [line width=.45pt] (181) to (186);
		\draw [line width=.45pt] (186) to (202);
		\draw [line width=.45pt] (202) to (204);
		\draw [line width=.45pt] (204) to (208.center);
		\draw [line width=.45pt] (163.center) to (183);
		\draw [line width=.45pt] (183) to (184);
		\draw [line width=.45pt] (184) to (196);
		\draw [line width=.45pt] (196) to (198);
		\draw [line width=.45pt] (198) to (200);
		\draw [line width=.45pt] (200) to (210.center);
		\draw [line width=.45pt] (165.center) to (185);
		\draw [line width=.45pt] (185) to (212.center);
		\draw [line width=.45pt] (167.center) to (187);
		\draw [line width=.45pt] (187) to (214.center);
		\draw [line width=.45pt, draw={teal!80!black}] (169.center) to (189);
		\draw [line width=.45pt, draw={teal!80!black}] (189) to (216.center);
		\draw [line width=.45pt] (171.center) to (193);
		\draw [line width=.45pt] (193) to (218.center);
		\draw [line width=.45pt] (173.center) to (199);
		\draw [line width=.45pt] (199) to (220.center);
		\draw [line width=.45pt, draw={orange!85!black}] (175.center) to (191);
		\draw [line width=.45pt, draw={orange!85!black}] (191) to (205);
		\draw [line width=.45pt, draw={orange!85!black}] (205) to (222);
		\draw [line width=.45pt] (177.center) to (195);
		\draw [line width=.45pt] (195) to (197);
		\draw [line width=.45pt] (197) to (226);
		\draw [line width=.45pt] (179.center) to (201);
		\draw [line width=.45pt] (201) to (203);
		\draw [line width=.45pt] (203) to (229);
	\end{pgfonlayer}
\end{tikzpicture}

%% file: figures/steane-zero-full-fig.tikz
\begin{tikzpicture}
	\begin{pgfonlayer}{nodelayer}
		\node [style=Z dot] (0) at (-12, 1.5) {};
		\node [style=Z dot] (1) at (-12, 0) {};
		\node [style=Z dot] (2) at (-12, -1.5) {};
		\node [style=X dot] (3) at (-10, 3) {};
		\node [style=X dot] (4) at (-10, 2) {};
		\node [style=X dot] (5) at (-10, 1) {};
		\node [style=X dot] (6) at (-10, 0) {};
		\node [style=X dot] (7) at (-10, -1) {};
		\node [style=X dot] (8) at (-10, -2) {};
		\node [style=X dot] (9) at (-10, -3) {};
		\node [style=none] (10) at (-9, 3) {};
		\node [style=none] (11) at (-9, 2) {};
		\node [style=none] (12) at (-9, 1) {};
		\node [style=none] (13) at (-9, 0) {};
		\node [style=none] (14) at (-9, -1) {};
		\node [style=none] (15) at (-9, -2) {};
		\node [style=none] (16) at (-9, -3) {};
		\node [style=none] (17) at (0, 0) {$\FaultEq$};
		\node [style=none] (46) at (7, 0) {$\FaultEq$};
		\node [style=sLabel] (47) at (7, -0.5) {CSS};
		\node [style=none] (49) at (21, 0) {$\FaultEq$};
		\node [style=X dot] (50) at (25.75, 2) {};
		\node [style=X dot] (51) at (25.75, 0) {};
		\node [style=X dot] (52) at (25.25, -2) {};
		\node [style=none] (53) at (27, 3) {};
		\node [style=none] (54) at (27, 2) {};
		\node [style=none] (55) at (27, 1) {};
		\node [style=none] (56) at (27, 0) {};
		\node [style=none] (57) at (27, -1) {};
		\node [style=none] (58) at (27, -2) {};
		\node [style=none] (59) at (27, -3) {};
		\node [style=Z dot] (60) at (22.5, 2.5) {};
		\node [style=Z dot] (61) at (23, 3) {};
		\node [style=Z dot] (62) at (23.75, 2) {};
		\node [style=Z dot] (63) at (24.25, 2.5) {};
		\node [style=Z dot] (64) at (22.75, -2.5) {};
		\node [style=Z dot] (65) at (23.25, -3) {};
		\node [style=Z dot] (66) at (23.75, -2) {};
		\node [style=Z dot] (67) at (24.25, -2.5) {};
		\node [style=Z dot] (68) at (22.5, -1) {};
		\node [style=Z dot] (69) at (23.25, -1) {};
		\node [style=Z dot] (70) at (23.25, -0.25) {};
		\node [style=Z dot] (71) at (22.5, -0.25) {};
		\node [style=X dot] (73) at (25.75, 1) {};
		\node [style=X dot] (74) at (23.5, 1) {};
		\node [style=sLabel] (76) at (0, 1) {(Step 1.2.1)};
		\node [style=sLabel] (77) at (7, 1) {(Step 1.2.2)};
		\node [style=sLabel] (78) at (21, 1) {(Step 1.3.2)};
		\node [style=Z dot] (79) at (16, 3) {};
		\node [style=Z dot] (80) at (16, -1) {};
		\node [style=Z dot] (81) at (16, -3) {};
		\node [style=X dot] (82) at (18, 2) {};
		\node [style=X dot] (83) at (16.75, 1) {};
		\node [style=X dot] (84) at (18, 0) {};
		\node [style=X dot] (85) at (18, -2) {};
		\node [style=none] (86) at (19, 3) {};
		\node [style=none] (87) at (19, 2) {};
		\node [style=none] (88) at (19, 1) {};
		\node [style=none] (89) at (19, 0) {};
		\node [style=none] (90) at (19, -1) {};
		\node [style=none] (91) at (19, -2) {};
		\node [style=none] (92) at (19, -3) {};
		\node [style=none] (93) at (14, 0) {$\FaultEq$};
		\node [style=sLabel] (94) at (14, 1) {(Step 1.3.1)};
		\node [style=X dot] (95) at (18, 1) {};
		\node [style=none] (110) at (-7, 0) {$\FaultEq$};
		\node [style=Z dot] (111) at (-5, 1.5) {};
		\node [style=Z dot] (112) at (-5, 0) {};
		\node [style=Z dot] (113) at (-5, -1.5) {};
		\node [style=X dot] (114) at (-3, 2) {};
		\node [style=X dot] (115) at (-3, 1) {};
		\node [style=X dot] (116) at (-3, 0) {};
		\node [style=X dot] (117) at (-3, -2) {};
		\node [style=none] (118) at (-2, 3) {};
		\node [style=none] (119) at (-2, 2) {};
		\node [style=none] (120) at (-2, 1) {};
		\node [style=none] (121) at (-2, 0) {};
		\node [style=none] (122) at (-2, -1) {};
		\node [style=none] (123) at (-2, -2) {};
		\node [style=none] (124) at (-2, -3) {};
		\node [style=sLabel] (125) at (-7, 1) {(Step 1.1)};
		\node [style=X dot] (126) at (-3, 3) {};
		\node [style=X dot] (127) at (-3, -1) {};
		\node [style=X dot] (128) at (-3, -3) {};
		\node [style=Z dot] (129) at (2, 1.5) {};
		\node [style=Z dot] (130) at (2, 0) {};
		\node [style=Z dot] (131) at (2, -1.5) {};
		\node [style=X dot] (132) at (4, 2) {};
		\node [style=X dot] (133) at (4, 1) {};
		\node [style=X dot] (134) at (4, 0) {};
		\node [style=X dot] (135) at (4, -2) {};
		\node [style=none] (136) at (5, 3) {};
		\node [style=none] (137) at (5, 2) {};
		\node [style=none] (138) at (5, 1) {};
		\node [style=none] (139) at (5, 0) {};
		\node [style=none] (140) at (5, -1) {};
		\node [style=none] (141) at (5, -2) {};
		\node [style=none] (142) at (5, -3) {};
		\node [style=Z dot] (143) at (9, 1.5) {};
		\node [style=Z dot] (144) at (9, 0) {};
		\node [style=Z dot] (145) at (9, -1.5) {};
		\node [style=X dot] (146) at (11, 2) {};
		\node [style=X dot] (147) at (11, 1) {};
		\node [style=X dot] (148) at (11, 0) {};
		\node [style=X dot] (149) at (11, -2) {};
		\node [style=none] (150) at (12, 3) {};
		\node [style=none] (151) at (12, 2) {};
		\node [style=none] (152) at (12, 1) {};
		\node [style=none] (153) at (12, 0) {};
		\node [style=none] (154) at (12, -1) {};
		\node [style=none] (155) at (12, -2) {};
		\node [style=none] (156) at (12, -3) {};
	\end{pgfonlayer}
	\begin{pgfonlayer}{edgelayer}
		\draw [style=fault-free] (0) to (3);
		\draw [style=fault-free] (0) to (4);
		\draw [style=fault-free] (0) to (5);
		\draw [style=fault-free] (0) to (6);
		\draw [style=fault-free] (1) to (4);
		\draw [style=fault-free] (1) to (5);
		\draw [style=fault-free] (1) to (7);
		\draw [style=fault-free] (1) to (8);
		\draw [style=fault-free] (2) to (5);
		\draw [style=fault-free] (2) to (6);
		\draw [style=fault-free] (2) to (8);
		\draw [style=fault-free] (2) to (9);
		\draw (16.center) to (9);
		\draw (15.center) to (8);
		\draw (14.center) to (7);
		\draw (13.center) to (6);
		\draw (12.center) to (5);
		\draw (11.center) to (4);
		\draw (10.center) to (3);
		\draw (58.center) to (52);
		\draw (56.center) to (51);
		\draw (54.center) to (50);
		\draw (61) to (60);
		\draw (60) to (62);
		\draw (62) to (63);
		\draw (63) to (61);
		\draw (63) to (50);
		\draw (60) to (51);
		\draw (61) to (53.center);
		\draw (65) to (64);
		\draw (64) to (66);
		\draw (66) to (67);
		\draw (67) to (65);
		\draw (65) to (59.center);
		\draw (67) to (52);
		\draw (66) to (51);
		\draw (71) to (68);
		\draw (68) to (69);
		\draw (69) to (70);
		\draw (70) to (71);
		\draw (69) to (57.center);
		\draw (68) to (52);
		\draw (71) to (50);
		\draw (73) to (74);
		\draw (73) to (55.center);
		\draw (74) to (62);
		\draw (79) to (82);
		\draw (79) to (83);
		\draw (79) to (84);
		\draw (80) to (82);
		\draw (80) to (83);
		\draw (80) to (85);
		\draw (81) to (84);
		\draw (81) to (85);
		\draw (91.center) to (85);
		\draw (89.center) to (84);
		\draw (87.center) to (82);
		\draw (79) to (86.center);
		\draw (90.center) to (80);
		\draw (92.center) to (81);
		\draw (83) to (95);
		\draw (95) to (88.center);
		\draw (95) to (81);
		\draw (74) to (70);
		\draw (73) to (64);
		\draw [style=fault-free] (111) to (114);
		\draw [style=fault-free] (111) to (115);
		\draw [style=fault-free] (111) to (116);
		\draw [style=fault-free] (112) to (114);
		\draw [style=fault-free] (112) to (115);
		\draw [style=fault-free] (112) to (117);
		\draw [style=fault-free] (113) to (115);
		\draw [style=fault-free] (113) to (116);
		\draw [style=fault-free] (113) to (117);
		\draw (123.center) to (117);
		\draw (121.center) to (116);
		\draw (120.center) to (115);
		\draw (119.center) to (114);
		\draw [style=fault-free] (126) to (111);
		\draw [style=fault-free] (112) to (127);
		\draw [style=fault-free] (113) to (128);
		\draw (127) to (122.center);
		\draw (126) to (118.center);
		\draw (128) to (124.center);
		\draw [style=fault-free] (129) to (132);
		\draw [style=fault-free] (129) to (133);
		\draw [style=fault-free] (129) to (134);
		\draw [style=fault-free] (130) to (132);
		\draw [style=fault-free] (130) to (133);
		\draw [style=fault-free] (130) to (135);
		\draw [style=fault-free] (131) to (133);
		\draw [style=fault-free] (131) to (134);
		\draw [style=fault-free] (131) to (135);
		\draw (141.center) to (135);
		\draw (139.center) to (134);
		\draw (138.center) to (133);
		\draw (137.center) to (132);
		\draw (129) to (136.center);
		\draw (131) to (142.center);
		\draw (130) to (140.center);
		\draw (143) to (146);
		\draw (143) to (147);
		\draw (143) to (148);
		\draw (144) to (146);
		\draw (144) to (147);
		\draw (144) to (149);
		\draw (145) to (147);
		\draw (145) to (148);
		\draw (145) to (149);
		\draw (155.center) to (149);
		\draw (153.center) to (148);
		\draw (152.center) to (147);
		\draw (151.center) to (146);
		\draw (143) to (150.center);
		\draw (145) to (156.center);
		\draw (144) to (154.center);
	\end{pgfonlayer}
\end{tikzpicture}

%% file: figures/semi-idealized-id.tikz
\begin{tikzpicture}
	\begin{pgfonlayer}{nodelayer}
		\node [style=X dot] (0) at (-2.75, 0) {};
		\node [style=none] (7) at (-4.25, 0) {};
		\node [style=none] (14) at (0, 0) {$\FaultEq$};
		\node [style=none] (29) at (-1.5, 0) {};
		\node [style=none] (31) at (1.5, 0) {};
		\node [style=none] (38) at (3.5, 0) {};
	\end{pgfonlayer}
	\begin{pgfonlayer}{edgelayer}
		\draw [style=fault-free] (7.center) to (0);
		\draw (29.center) to (0);
		\draw (38.center) to (31.center);
	\end{pgfonlayer}
\end{tikzpicture}

%% file: figures/johns-rule.tikz
\begin{tikzpicture}
	\begin{pgfonlayer}{nodelayer}
		\node [style=X dot] (1) at (-3, 0) {};
		\node [style=none] (2) at (-2.5, 1.25) {};
		\node [style=none] (3) at (-1.75, 0.5) {};
		\node [style=none] (4) at (-1.75, -0.5) {};
		\node [style=none] (5) at (-2.5, -1.25) {};
		\node [style=none, rotate=-40] (6) at (-2.25, 0.75) {$\ldots$};
		\node [style=none, rotate=40] (8) at (-2.25, -0.75) {$\ldots$};
		\node [style=Z dot] (9) at (-5, 0) {};
		\node [style=none] (10) at (-5.5, 1.25) {};
		\node [style=none] (11) at (-6.25, 0.5) {};
		\node [style=none] (12) at (-6.25, -0.5) {};
		\node [style=none] (13) at (-5.5, -1.25) {};
		\node [style=none, rotate=-40] (14) at (-5.75, -0.75) {$\ldots$};
		\node [style=none, rotate=40] (15) at (-5.75, 0.75) {$\ldots$};
		\node [style=none] (16) at (0, 0) {$\FaultEq$};
		\node [style=X dot] (17) at (5, 0) {};
		\node [style=none] (18) at (5.5, 1.25) {};
		\node [style=none] (19) at (6.25, 0.5) {};
		\node [style=none] (20) at (6.25, -0.5) {};
		\node [style=none] (21) at (5.5, -1.25) {};
		\node [style=none, rotate=-40] (22) at (5.75, 0.75) {$\ldots$};
		\node [style=none, rotate=40] (23) at (5.75, -0.75) {$\ldots$};
		\node [style=Z dot] (24) at (3, 0) {};
		\node [style=none] (25) at (2.5, 1.25) {};
		\node [style=none] (26) at (1.75, 0.5) {};
		\node [style=none] (27) at (1.75, -0.5) {};
		\node [style=none] (28) at (2.5, -1.25) {};
		\node [style=none, rotate=-40] (29) at (2.25, -0.75) {$\ldots$};
		\node [style=none, rotate=40] (30) at (2.25, 0.75) {$\ldots$};
		\node [style=sLabel] (31) at (0, -0.5) {CSS};
		\node [style=sLabel] (32) at (-6, 1.125) {$a$};
		\node [style=sLabel] (33) at (-1.875, 1.125) {$c$};
		\node [style=sLabel] (34) at (-6, -1) {$b$};
		\node [style=sLabel] (35) at (-2, -1) {$d$};
		\node [style=sLabel] (36) at (2, 1.125) {$a$};
		\node [style=sLabel] (37) at (6.125, 1.125) {$c$};
		\node [style=sLabel] (38) at (2, -1) {$b$};
		\node [style=sLabel] (39) at (6, -1) {$d$};
	\end{pgfonlayer}
	\begin{pgfonlayer}{edgelayer}
		\draw (2.center) to (1);
		\draw (1) to (3.center);
		\draw [style=fault-free] (4.center) to (1);
		\draw [style=fault-free] (1) to (5.center);
		\draw (10.center) to (9);
		\draw (9) to (11.center);
		\draw [style=fault-free] (12.center) to (9);
		\draw [style=fault-free] (9) to (13.center);
		\draw [style=fault-free] (9) to (1);
		\draw (18.center) to (17);
		\draw (17) to (19.center);
		\draw [style=fault-free] (20.center) to (17);
		\draw [style=fault-free] (17) to (21.center);
		\draw (25.center) to (24);
		\draw (24) to (26.center);
		\draw [style=fault-free] (27.center) to (24);
		\draw [style=fault-free] (24) to (28.center);
		\draw (24) to (17);
	\end{pgfonlayer}
\end{tikzpicture}

%% file: figures/colour-change-int-proof2.tikz
\begin{tikzpicture}
	\begin{pgfonlayer}{nodelayer}
		\node [style=X phase dot] (54) at (-12.25, 0) {$\pi$};
		\node [style=X dot] (137) at (-11.25, 0) {};
		\node [style=none] (138) at (-10.75, 1.25) {};
		\node [style=none] (139) at (-10, 0.5) {};
		\node [style=none] (140) at (-10, -0.5) {};
		\node [style=none] (141) at (-10.75, -1.25) {};
		\node [style=none, rotate=-40] (142) at (-10.5, 0.75) {$\ldots$};
		\node [style=none, rotate=40] (143) at (-10.5, -0.75) {$\ldots$};
		\node [style=Z dot] (144) at (-13.25, 0) {};
		\node [style=none] (145) at (-13.75, 1.25) {};
		\node [style=none] (146) at (-14.5, 0.5) {};
		\node [style=none] (147) at (-14.5, -0.5) {};
		\node [style=none] (148) at (-13.75, -1.25) {};
		\node [style=none, rotate=-40] (149) at (-14, -0.75) {$\ldots$};
		\node [style=none, rotate=40] (150) at (-14, 0.75) {$\ldots$};
		\node [style=none] (151) at (-8.5, 0) {$=$};
		\node [style=X dot] (152) at (-3.75, 0) {};
		\node [style=none] (153) at (-3, 2.25) {};
		\node [style=none] (154) at (-2.5, 0.5) {};
		\node [style=none] (155) at (-2.5, -0.5) {};
		\node [style=none] (156) at (-3.25, -1.25) {};
		\node [style=none, rotate=-40] (157) at (-3, 0.75) {$\ldots$};
		\node [style=none, rotate=40] (158) at (-3, -0.75) {$\ldots$};
		\node [style=Z dot] (159) at (-5.75, 0) {};
		\node [style=none] (160) at (-6.25, 1.25) {};
		\node [style=none] (161) at (-7, 0.5) {};
		\node [style=none] (162) at (-7, -0.5) {};
		\node [style=none] (163) at (-6.25, -1.25) {};
		\node [style=none, rotate=-40] (164) at (-6.5, -0.75) {$\ldots$};
		\node [style=none, rotate=40] (165) at (-6.5, 0.75) {$\ldots$};
		\node [style=Z phase dot] (166) at (4.25, 0) {$\pi$};
		\node [style=X dot] (167) at (5.25, 0) {};
		\node [style=none] (168) at (5.75, 1.25) {};
		\node [style=none] (169) at (6.5, 0.5) {};
		\node [style=none] (170) at (6.5, -0.5) {};
		\node [style=none] (171) at (5.75, -1.25) {};
		\node [style=none, rotate=-40] (172) at (6, 0.75) {$\ldots$};
		\node [style=none, rotate=40] (173) at (6, -0.75) {$\ldots$};
		\node [style=Z dot] (174) at (3.25, 0) {};
		\node [style=none] (175) at (2.75, 1.25) {};
		\node [style=none] (176) at (2, 0.5) {};
		\node [style=none] (177) at (2, -0.5) {};
		\node [style=none] (178) at (2.75, -1.25) {};
		\node [style=none, rotate=-40] (179) at (2.5, -0.75) {$\ldots$};
		\node [style=none, rotate=40] (180) at (2.5, 0.75) {$\ldots$};
		\node [style=none] (181) at (8, 0) {$=$};
		\node [style=X dot] (182) at (12.75, 0) {};
		\node [style=none] (183) at (13.25, 1.25) {};
		\node [style=none] (184) at (14, 0.5) {};
		\node [style=none] (185) at (14, -0.5) {};
		\node [style=none] (186) at (13.25, -1.25) {};
		\node [style=none, rotate=-40] (187) at (13.5, 0.75) {$\ldots$};
		\node [style=none, rotate=40] (188) at (13.5, -0.75) {$\ldots$};
		\node [style=Z dot] (189) at (10.75, 0) {};
		\node [style=none] (190) at (10, 2.25) {};
		\node [style=none] (191) at (9, 0.75) {};
		\node [style=none] (192) at (9.5, -0.5) {};
		\node [style=none] (193) at (10.25, -1.25) {};
		\node [style=none, rotate=-40] (194) at (10, -0.75) {$\ldots$};
		\node [style=none, rotate=40] (195) at (10, 0.75) {$\ldots$};
		\node [style=Z phase dot] (196) at (10.25, 1.5) {$\pi$};
		\node [style=X phase dot] (197) at (-3.25, 1.5) {$\pi$};
	\end{pgfonlayer}
	\begin{pgfonlayer}{edgelayer}
		\draw (138.center) to (137);
		\draw (137) to (139.center);
		\draw [style=fault-free] (140.center) to (137);
		\draw [style=fault-free] (137) to (141.center);
		\draw (145.center) to (144);
		\draw (144) to (146.center);
		\draw [style=fault-free] (147.center) to (144);
		\draw [style=fault-free] (144) to (148.center);
		\draw (144) to (137);
		\draw (153.center) to (152);
		\draw (152) to (154.center);
		\draw [style=fault-free] (155.center) to (152);
		\draw [style=fault-free] (152) to (156.center);
		\draw (160.center) to (159);
		\draw (159) to (161.center);
		\draw [style=fault-free] (162.center) to (159);
		\draw [style=fault-free] (159) to (163.center);
		\draw (159) to (152);
		\draw (168.center) to (167);
		\draw (167) to (169.center);
		\draw [style=fault-free] (170.center) to (167);
		\draw [style=fault-free] (167) to (171.center);
		\draw (175.center) to (174);
		\draw (174) to (176.center);
		\draw [style=fault-free] (177.center) to (174);
		\draw [style=fault-free] (174) to (178.center);
		\draw (174) to (167);
		\draw (183.center) to (182);
		\draw (182) to (184.center);
		\draw [style=fault-free] (185.center) to (182);
		\draw [style=fault-free] (182) to (186.center);
		\draw (190.center) to (189);
		\draw (189) to (191.center);
		\draw [style=fault-free] (192.center) to (189);
		\draw [style=fault-free] (189) to (193.center);
		\draw (189) to (182);
	\end{pgfonlayer}
\end{tikzpicture}

%% file: figures/safe-hook-split.tikz
\begin{tikzpicture}
	\begin{pgfonlayer}{nodelayer}
		\node [style=Z dot] (0) at (-9, 0) {};
		\node [style=none] (1) at (-10, 1) {};
		\node [style=none] (4) at (-10, -1) {};
		\node [style=Z dot] (7) at (-7, 0) {};
		\node [style=none] (8) at (-6, 1) {};
		\node [style=none] (11) at (-6, -1) {};
		\node [style=none, rotate=90] (12) at (-6, 0) {$\ldots$};
		\node [style=none, rotate=90] (13) at (-10, 0) {$\ldots$};
		\node [style=none] (14) at (-5.75, -1.125) {};
		\node [style=none] (15) at (-5.75, 1.125) {};
		\node [style=sLabel, align=center] (18) at (-3.75, 0) {Strictly safe \\ hook error};
		\node [style=none] (19) at (-11.5, 0) {$\FaultEq$};
		\node [style=none] (21) at (-13, 1) {};
		\node [style=none] (22) at (-13, -1) {};
		\node [style=Z dot] (23) at (-14, 0) {};
		\node [style=none] (24) at (-15, 1) {};
		\node [style=none] (25) at (-15, -1) {};
		\node [style=none, rotate=90] (26) at (-15, 0) {$\ldots$};
		\node [style=none, rotate=90] (27) at (-13, 0) {$\ldots$};
		\node [style=X dot] (28) at (10, 0) {};
		\node [style=none] (29) at (9, 1) {};
		\node [style=none] (30) at (9, -1) {};
		\node [style=X dot] (31) at (12, 0) {};
		\node [style=none] (32) at (13, 1) {};
		\node [style=none] (33) at (13, -1) {};
		\node [style=none, rotate=90] (34) at (13, 0) {$\ldots$};
		\node [style=none, rotate=90] (35) at (9, 0) {$\ldots$};
		\node [style=none] (36) at (13.25, -1.125) {};
		\node [style=none] (37) at (13.25, 1.125) {};
		\node [style=sLabel, align=center] (38) at (15.25, 0) {Strictly safe \\ hook error};
		\node [style=none] (39) at (7.5, 0) {$\FaultEq$};
		\node [style=none] (40) at (6, 1) {};
		\node [style=none] (41) at (6, -1) {};
		\node [style=X dot] (42) at (5, 0) {};
		\node [style=none] (43) at (4, 1) {};
		\node [style=none] (44) at (4, -1) {};
		\node [style=none, rotate=90] (45) at (4, 0) {$\ldots$};
		\node [style=none, rotate=90] (46) at (6, 0) {$\ldots$};
	\end{pgfonlayer}
	\begin{pgfonlayer}{edgelayer}
		\draw (1.center) to (0);
		\draw (0) to (4.center);
		\draw (8.center) to (7);
		\draw (7) to (11.center);
		\draw (7) to (0);
		\draw [style=braceedge] (15.center) to (14.center);
		\draw (24.center) to (23);
		\draw (23) to (25.center);
		\draw (21.center) to (23);
		\draw (23) to (22.center);
		\draw (29.center) to (28);
		\draw (28) to (30.center);
		\draw (32.center) to (31);
		\draw (31) to (33.center);
		\draw (31) to (28);
		\draw [style=braceedge] (37.center) to (36.center);
		\draw (43.center) to (42);
		\draw (42) to (44.center);
		\draw (40.center) to (42);
		\draw (42) to (41.center);
	\end{pgfonlayer}
\end{tikzpicture}

%% file: figures/steane-zero-naive-se.tikz
\begin{tikzpicture}
	\begin{pgfonlayer}{nodelayer}
		\node [style=X dot] (26) at (-12, 2) {};
		\node [style=X dot] (27) at (-12, 0) {};
		\node [style=X dot] (28) at (-12.5, -2) {};
		\node [style=none] (29) at (-10.25, 3) {};
		\node [style=none] (30) at (-10.25, 2) {};
		\node [style=none] (31) at (-10.25, 1) {};
		\node [style=none] (32) at (-10.25, 0) {};
		\node [style=none] (33) at (-10.25, -1) {};
		\node [style=none] (34) at (-10.25, -2) {};
		\node [style=none] (35) at (-10.25, -3) {};
		\node [style=Z dot] (36) at (-15.25, 2.5) {};
		\node [style=Z dot] (37) at (-14.75, 3) {};
		\node [style=Z dot] (38) at (-14, 2) {};
		\node [style=Z dot] (39) at (-13.5, 2.5) {};
		\node [style=Z dot] (40) at (-15, -2.5) {};
		\node [style=Z dot] (41) at (-14.5, -3) {};
		\node [style=Z dot] (42) at (-14, -2) {};
		\node [style=Z dot] (43) at (-13.5, -2.5) {};
		\node [style=Z dot] (44) at (-15.25, -1) {};
		\node [style=Z dot] (45) at (-14.5, -1) {};
		\node [style=Z dot] (46) at (-14.5, -0.25) {};
		\node [style=Z dot] (47) at (-15.25, -0.25) {};
		\node [style=X dot] (50) at (-14.25, 1) {};
		\node [style=X dot] (51) at (-12, 1) {};
		\node [style=none] (52) at (-8.75, 0) {$\FaultEq$};
		\node [style=X dot] (100) at (-4.25, 2) {};
		\node [style=X dot] (101) at (-3, 0) {};
		\node [style=X dot] (102) at (-4.25, -2) {};
		\node [style=none] (103) at (-1.25, 3) {};
		\node [style=none] (104) at (-1.25, 2) {};
		\node [style=none] (105) at (-1.25, 1) {};
		\node [style=none] (106) at (-1.25, 0) {};
		\node [style=none] (107) at (-1.25, -1) {};
		\node [style=none] (108) at (-1.25, -2) {};
		\node [style=none] (109) at (-1.25, -3) {};
		\node [style=Z dot] (110) at (-7.25, 3.5) {};
		\node [style=Z dot] (111) at (-6.75, 3) {};
		\node [style=Z dot] (112) at (-6, 4) {};
		\node [style=Z dot] (113) at (-5.5, 3.5) {};
		\node [style=Z dot] (114) at (-7.5, -3.5) {};
		\node [style=Z dot] (115) at (-7, -3) {};
		\node [style=Z dot] (116) at (-6.25, -4) {};
		\node [style=Z dot] (117) at (-5.75, -3.5) {};
		\node [style=Z dot] (118) at (-6.75, -1.5) {};
		\node [style=Z dot] (119) at (-6, -1) {};
		\node [style=Z dot] (120) at (-6.5, -0.5) {};
		\node [style=Z dot] (121) at (-7.25, -0.5) {};
		\node [style=X dot] (122) at (-7.125, 1.775) {};
		\node [style=X dot] (125) at (-7.125, 1) {};
		\node [style=none] (126) at (0.25, 0) {$\FaultEq$};
		\node [style=none] (127) at (-2.25, 5) {};
		\node [style=none] (128) at (-4.75, 4.25) {};
		\node [style=none] (129) at (-3.5, 3.5) {};
		\node [style=none] (130) at (-3.5, 2.5) {};
		\node [style=none] (131) at (-3.5, 1.5) {};
		\node [style=none] (132) at (-2.25, 0.5) {};
		\node [style=none] (133) at (-2.25, -0.5) {};
		\node [style=none] (134) at (-3.5, -1.5) {};
		\node [style=none] (135) at (-3.5, -2.5) {};
		\node [style=none] (138) at (-4.75, -4.5) {};
		\node [style=none] (139) at (-3.5, -4) {};
		\node [style=none] (140) at (-2.25, -3.5) {};
		\node [style=none] (141) at (-3.5, -0.5) {};
		\node [style=none] (142) at (-5.25, -0.5) {};
		\node [style=none] (143) at (-3.5, 0.5) {};
		\node [style=none] (144) at (-5.25, 1.75) {};
		\node [style=none] (145) at (-4.75, 0.5) {};
		\node [style=none] (146) at (-4.75, 2.25) {};
		\node [style=none] (147) at (-7.5, 5) {};
		\node [style=none] (148) at (-7.5, 2.75) {};
		\node [style=none] (149) at (-5.25, 2.75) {};
		\node [style=none] (150) at (-5.25, 5) {};
		\node [style=none] (151) at (-7.75, -2.75) {};
		\node [style=none] (152) at (-7.75, -4.75) {};
		\node [style=none] (153) at (-5.5, -4.75) {};
		\node [style=none] (154) at (-5.5, -2.75) {};
		\node [style=none] (155) at (-7.5, 0) {};
		\node [style=none] (156) at (-7.5, -1.75) {};
		\node [style=none] (157) at (-5.5, -1.75) {};
		\node [style=none] (158) at (-5.5, 0) {};
		\node [style=none] (163) at (-4.5, 2.75) {};
		\node [style=none] (164) at (-4.5, 1.25) {};
		\node [style=none] (165) at (-3.75, 1.25) {};
		\node [style=none] (166) at (-3.75, 2.75) {};
		\node [style=none] (167) at (-3.25, 0.75) {};
		\node [style=none] (168) at (-3.25, -0.75) {};
		\node [style=none] (169) at (-2.5, -0.75) {};
		\node [style=none] (170) at (-2.5, 0.75) {};
		\node [style=none] (171) at (-4.5, -1.25) {};
		\node [style=none] (172) at (-4.5, -2.75) {};
		\node [style=none] (173) at (-3.75, -2.75) {};
		\node [style=none] (174) at (-3.75, -1.25) {};
		\node [style=none] (178) at (14.5, 3.25) {};
		\node [style=none] (179) at (14.5, 2.25) {};
		\node [style=none] (180) at (14.5, 1) {};
		\node [style=none] (181) at (14.5, -0.25) {};
		\node [style=none] (182) at (14.5, -1.25) {};
		\node [style=none] (183) at (14.5, -2.25) {};
		\node [style=none] (184) at (14.5, -3.25) {};
		\node [style=Z dot] (202) at (12.75, 4.75) {};
		\node [style=Z dot] (203) at (5.25, 4.25) {};
		\node [style=Z dot] (204) at (9, 3.75) {};
		\node [style=X dot] (205) at (9, 2.75) {};
		\node [style=X dot] (206) at (9, 1.75) {};
		\node [style=X dot] (207) at (12.75, 0.25) {};
		\node [style=X dot] (208) at (12.75, -0.75) {};
		\node [style=X dot] (209) at (9, -1.75) {};
		\node [style=X dot] (210) at (9, -2.75) {};
		\node [style=Z dot] (211) at (5.25, -4.75) {};
		\node [style=Z dot] (212) at (9, -4.25) {};
		\node [style=Z dot] (213) at (12.75, -3.75) {};
		\node [style=Z dot] (214) at (9, -0.75) {};
		\node [style=Z dot] (215) at (4.75, -1) {};
		\node [style=Z dot] (216) at (9, 0) {};
		\node [style=X dot] (217) at (4.75, 2.5) {};
		\node [style=X dot] (218) at (5.25, 0.25) {};
		\node [style=X dot] (219) at (5.25, 1.75) {};
		\node [style=none] (220) at (2, 5) {};
		\node [style=none] (221) at (2, 3) {};
		\node [style=none] (222) at (4, 3) {};
		\node [style=none] (223) at (4, 5) {};
		\node [style=Z dot] (248) at (6, 4.25) {};
		\node [style=Z dot] (249) at (13.5, 4.75) {};
		\node [style=Z dot] (250) at (9.75, 3.75) {};
		\node [style=X dot] (251) at (6, 0.25) {};
		\node [style=Z dot] (252) at (9.75, -4.25) {};
		\node [style=Z dot] (253) at (13.5, -3.75) {};
		\node [style=Z dot] (254) at (6, -4.75) {};
		\node [style=X dot] (255) at (9.75, -2.75) {};
		\node [style=X dot] (256) at (9.75, -1.75) {};
		\node [style=X dot] (257) at (9.75, 1.75) {};
		\node [style=X dot] (258) at (9.75, 2.75) {};
		\node [style=X dot] (259) at (6, 1.75) {};
		\node [style=Z dot] (260) at (9.75, -0.75) {};
		\node [style=Z dot] (261) at (9.75, 0) {};
		\node [style=X dot] (262) at (13.5, 0.25) {};
		\node [style=X dot] (263) at (13.5, -0.75) {};
		\node [style=X dot] (264) at (5.75, 2.5) {};
		\node [style=Z dot] (265) at (5.75, -1) {};
		\node [style=sLabel] (272) at (3, 4.25) {Z-CAT};
		\node [style=none] (273) at (4, 4.75) {};
		\node [style=none] (274) at (4, 4.25) {};
		\node [style=none] (275) at (4, 3.75) {};
		\node [style=none] (276) at (4, 3.25) {};
		\node [style=none] (277) at (2, -3) {};
		\node [style=none] (278) at (2, -5) {};
		\node [style=none] (279) at (4, -5) {};
		\node [style=none] (280) at (4, -3) {};
		\node [style=none] (282) at (4, -3.25) {};
		\node [style=none] (283) at (4, -3.75) {};
		\node [style=none] (284) at (4, -4.25) {};
		\node [style=none] (285) at (4, -4.75) {};
		\node [style=none] (286) at (2, -0.25) {};
		\node [style=none] (287) at (2, -2.25) {};
		\node [style=none] (288) at (4, -2.25) {};
		\node [style=none] (289) at (4, -0.25) {};
		\node [style=none] (291) at (4, -0.5) {};
		\node [style=none] (292) at (4, -1) {};
		\node [style=none] (293) at (4, -1.5) {};
		\node [style=none] (294) at (4, -2) {};
		\node [style=none] (295) at (2, 2.75) {};
		\node [style=none] (298) at (4, 2.75) {};
		\node [style=none] (300) at (4, 2.5) {};
		\node [style=none] (301) at (4, 1.75) {};
		\node [style=none] (302) at (4, 1) {};
		\node [style=none] (303) at (4, 0.25) {};
		\node [style=none] (304) at (6.5, -1.5) {};
		\node [style=none] (305) at (6.5, -3) {};
		\node [style=none] (306) at (8.5, -3) {};
		\node [style=none] (307) at (8.5, -1.5) {};
		\node [style=none] (310) at (8.5, -1.75) {};
		\node [style=none] (311) at (8.5, -2.25) {};
		\node [style=none] (312) at (8.5, -2.75) {};
		\node [style=none] (313) at (10.25, 0.5) {};
		\node [style=none] (314) at (10.25, -1) {};
		\node [style=none] (315) at (12.25, -1) {};
		\node [style=none] (316) at (12.25, 0.5) {};
		\node [style=none] (318) at (12.25, 0.25) {};
		\node [style=none] (319) at (12.25, -0.25) {};
		\node [style=none] (320) at (12.25, -0.75) {};
		\node [style=none] (321) at (6.5, 3) {};
		\node [style=none] (322) at (6.5, 1.5) {};
		\node [style=none] (323) at (8.5, 1.5) {};
		\node [style=none] (324) at (8.5, 3) {};
		\node [style=none] (326) at (8.5, 2.75) {};
		\node [style=none] (327) at (8.5, 2.25) {};
		\node [style=none] (328) at (8.5, 1.75) {};
		\node [style=none] (329) at (6.5, -1.25) {};
		\node [style=sLabel] (330) at (3, 3.625) {Circuit};
		\node [style=sLabel] (331) at (3, -0.875) {Z-CAT};
		\node [style=sLabel] (332) at (3, -1.625) {Circuit};
		\node [style=sLabel] (333) at (3, -3.625) {Z-CAT};
		\node [style=sLabel] (334) at (3, -4.375) {Circuit};
		\node [style=sLabel] (337) at (7.5, 2.625) {X-CAT};
		\node [style=sLabel] (338) at (7.5, 1.875) {Circuit};
		\node [style=sLabel] (339) at (11.25, 0.125) {X-CAT};
		\node [style=sLabel] (340) at (11.25, -0.625) {Circuit};
		\node [style=sLabel] (341) at (7.5, -1.875) {X-CAT};
		\node [style=sLabel] (342) at (7.5, -2.625) {Circuit};
		\node [style=none] (343) at (-5.5, -0.25) {};
		\node [style=none] (344) at (-7.5, 2.275) {};
		\node [style=none] (345) at (-7.5, 0.525) {};
		\node [style=none] (346) at (-6.5, 0.525) {};
		\node [style=none] (347) at (-6.5, 2.275) {};
		\node [style=none] (354) at (2, 2.75) {};
		\node [style=none] (355) at (2, 0) {};
		\node [style=none] (356) at (4, 0) {};
		\node [style=none] (357) at (4, 2.75) {};
		\node [style=sLabel] (358) at (3, 1.875) {X-CAT};
		\node [style=sLabel] (359) at (3, 1.125) {Circuit};
	\end{pgfonlayer}
	\begin{pgfonlayer}{edgelayer}
		\draw (34.center) to (28);
		\draw (32.center) to (27);
		\draw (30.center) to (26);
		\draw (37) to (36);
		\draw (36) to (38);
		\draw (38) to (39);
		\draw (39) to (37);
		\draw (39) to (26);
		\draw (36) to (27);
		\draw (37) to (29.center);
		\draw (41) to (40);
		\draw (40) to (42);
		\draw (42) to (43);
		\draw (43) to (41);
		\draw (41) to (35.center);
		\draw (43) to (28);
		\draw (42) to (27);
		\draw (47) to (44);
		\draw (44) to (45);
		\draw (45) to (46);
		\draw (46) to (47);
		\draw (45) to (33.center);
		\draw (44) to (28);
		\draw (47) to (26);
		\draw (50) to (38);
		\draw (51) to (40);
		\draw (108.center) to (102);
		\draw (106.center) to (101);
		\draw (104.center) to (100);
		\draw (111) to (110);
		\draw (110) to (112);
		\draw (112) to (113);
		\draw (113) to (111);
		\draw (111) to (103.center);
		\draw (115) to (114);
		\draw (114) to (116);
		\draw (116) to (117);
		\draw (117) to (115);
		\draw (115) to (109.center);
		\draw (121) to (118);
		\draw (118) to (119);
		\draw (119) to (120);
		\draw (120) to (121);
		\draw (119) to (107.center);
		\draw [in=-60, out=180] (138.center) to (114);
		\draw [in=180, out=60] (110) to (127.center);
		\draw (127.center) to (132.center);
		\draw [in=30, out=-180] (132.center) to (101);
		\draw (146.center) to (128.center);
		\draw [in=15, out=-180] (128.center) to (112);
		\draw (113) to (129.center);
		\draw (129.center) to (130.center);
		\draw [in=30, out=-180] (130.center) to (100);
		\draw (143.center) to (131.center);
		\draw [in=-30, out=-180] (131.center) to (100);
		\draw (122) to (144.center);
		\draw (144.center) to (142.center);
		\draw (142.center) to (120);
		\draw [in=-45, out=-180, looseness=0.50] (145.center) to (125);
		\draw (138.center) to (145.center);
		\draw [in=-180, out=0] (118) to (141.center);
		\draw (141.center) to (134.center);
		\draw [in=30, out=180] (134.center) to (102);
		\draw [in=-180, out=-30] (101) to (133.center);
		\draw (133.center) to (140.center);
		\draw (140.center) to (117);
		\draw [in=180, out=-30] (102) to (135.center);
		\draw (135.center) to (139.center);
		\draw (139.center) to (116);
		\draw [style=dashed-line] (149.center)
			 to (150.center)
			 to (147.center)
			 to (148.center)
			 to cycle;
		\draw [style=dashed-line] (153.center)
			 to (154.center)
			 to (151.center)
			 to (152.center)
			 to cycle;
		\draw [style=dashed-line] (157.center)
			 to (158.center)
			 to (155.center)
			 to (156.center)
			 to cycle;
		\draw [style=dashed-line] (165.center)
			 to (166.center)
			 to (163.center)
			 to (164.center)
			 to cycle;
		\draw [style=dashed-line] (169.center)
			 to (170.center)
			 to (167.center)
			 to (168.center)
			 to cycle;
		\draw [style=dashed-line] (173.center)
			 to (174.center)
			 to (171.center)
			 to (172.center)
			 to cycle;
		\draw (202) to (207);
		\draw (219) to (203);
		\draw (204) to (205);
		\draw (216) to (206);
		\draw (217) to (215);
		\draw (211) to (218);
		\draw (214) to (209);
		\draw (208) to (213);
		\draw (210) to (212);
		\draw (222.center)
			 to (223.center)
			 to (220.center)
			 to (221.center)
			 to cycle;
		\draw (249) to (202);
		\draw (248) to (203);
		\draw (250) to (204);
		\draw (251) to (218);
		\draw (252) to (212);
		\draw (253) to (213);
		\draw (254) to (211);
		\draw (255) to (210);
		\draw (256) to (209);
		\draw (257) to (206);
		\draw (258) to (205);
		\draw (259) to (219);
		\draw (261) to (216);
		\draw (260) to (214);
		\draw (262) to (207);
		\draw (263) to (208);
		\draw (264) to (217);
		\draw (265) to (215);
		\draw (273.center) to (202);
		\draw (274.center) to (203);
		\draw (275.center) to (204);
		\draw (276.center) to (178.center);
		\draw (279.center)
			 to (280.center)
			 to (277.center)
			 to (278.center)
			 to cycle;
		\draw (282.center) to (184.center);
		\draw (283.center) to (213);
		\draw (284.center) to (212);
		\draw (285.center) to (211);
		\draw (288.center) to (289.center);
		\draw (289.center) to (286.center);
		\draw (286.center) to (287.center);
		\draw (287.center) to (288.center);
		\draw [in=180, out=0] (291.center) to (216);
		\draw (292.center) to (215);
		\draw [in=-180, out=0, looseness=0.50] (294.center) to (214);
		\draw [in=-180, out=0] (300.center) to (217);
		\draw [in=180, out=0] (301.center) to (219);
		\draw [in=180, out=0] (302.center) to (180.center);
		\draw [in=-180, out=0] (303.center) to (218);
		\draw (306.center) to (307.center);
		\draw (307.center) to (304.center);
		\draw (304.center) to (305.center);
		\draw (305.center) to (306.center);
		\draw (310.center) to (209);
		\draw (311.center) to (183.center);
		\draw (312.center) to (210);
		\draw (315.center) to (316.center);
		\draw (316.center) to (313.center);
		\draw (313.center) to (314.center);
		\draw (314.center) to (315.center);
		\draw (322.center)
			 to (323.center)
			 to (324.center)
			 to (321.center)
			 to cycle;
		\draw (318.center) to (207);
		\draw (319.center) to (181.center);
		\draw (320.center) to (208);
		\draw (326.center) to (205);
		\draw (327.center) to (179.center);
		\draw (328.center) to (206);
		\draw (329.center) to (182.center);
		\draw [in=180, out=0] (293.center) to (329.center);
		\draw [in=180, out=60, looseness=0.75] (121) to (343.center);
		\draw [in=-180, out=0, looseness=1.50] (343.center) to (143.center);
		\draw (50) to (51);
		\draw (50) to (46);
		\draw (51) to (31.center);
		\draw [in=-180, out=45, looseness=0.50] (122) to (146.center);
		\draw [in=180, out=0] (125) to (105.center);
		\draw [style=dashed-line] (346.center)
			 to (347.center)
			 to (344.center)
			 to (345.center)
			 to cycle;
		\draw (355.center)
			 to (356.center)
			 to (357.center)
			 to (354.center)
			 to cycle;
		\draw (122) to (125);
	\end{pgfonlayer}
\end{tikzpicture}

%% file: figures/cat-part-1.tikz
\begin{tikzpicture}
	\begin{pgfonlayer}{nodelayer}
		\node [style=none] (0) at (-6.75, 0.5) {};
		\node [style=none] (1) at (-6.75, 3) {};
		\node [style=none] (2) at (-2.75, 3) {};
		\node [style=none] (3) at (-2.75, 0.5) {};
		\node [style=none] (4) at (-3.75, 0.5) {};
		\node [style=none] (5) at (-5.75, 0.5) {};
		\node [style=none] (6) at (-5.75, 2) {};
		\node [style=none] (7) at (-3.75, 2) {};
		\node [style=Z dot] (8) at (-4.75, 1) {};
		\node [style=none] (9) at (-3.75, 1) {};
		\node [style=none] (10) at (-5.75, 1) {};
		\node [style=X dot] (11) at (-4.75, 0) {};
		\node [style=X dot] (12) at (-5.5, 0) {};
		\node [style=Z dot] (13) at (-2.25, 0) {};
		\node [style=none] (14) at (-1.75, 0.75) {};
		\node [style=none] (15) at (-2.75, 0.75) {};
		\node [style=none] (16) at (-1.75, 2.75) {};
		\node [style=none] (17) at (-2.75, 2.75) {};
		\node [style=none, rotate=90] (18) at (-2.25, 1.75) {$\ldots$};
		\node [style=sLabel] (19) at (-4.75, 2.5) {CAT-PART};
		\node [style=none] (37) at (-6.75, -1.25) {};
		\node [style=Z dot] (40) at (-1.5, 0) {};
		\node [style=X dot] (41) at (-2.25, -1.25) {};
		\node [style=none] (42) at (-1.5, -1.25) {};
		\node [style=none] (43) at (0, 0) {$\FaultEq$};
		\node [style=none] (44) at (1.5, 0.5) {};
		\node [style=none] (45) at (1.5, 3) {};
		\node [style=none] (46) at (5.5, 3) {};
		\node [style=none] (47) at (5.5, 0.5) {};
		\node [style=none] (48) at (4.25, 0.5) {};
		\node [style=none] (49) at (2.75, 0.5) {};
		\node [style=none] (50) at (2.75, 1.75) {};
		\node [style=none] (51) at (4.25, 1.75) {};
		\node [style=Z dot] (52) at (3.5, 1.25) {};
		\node [style=none] (53) at (4.25, 1.25) {};
		\node [style=none] (54) at (2.75, 1.25) {};
		\node [style=none] (58) at (6.25, 0.75) {};
		\node [style=none] (59) at (5.5, 0.75) {};
		\node [style=none] (60) at (6.25, 2.75) {};
		\node [style=none] (61) at (5.5, 2.75) {};
		\node [style=none, rotate=90] (62) at (6, 1.75) {$\ldots$};
		\node [style=sLabel] (63) at (3.5, 2.375) {CAT-PART};
		\node [style=none] (71) at (1.5, -1.25) {};
		\node [style=none] (76) at (3.5, 0.5) {};
		\node [style=none] (77) at (4, 0) {};
		\node [style=none] (78) at (5.5, -0.5) {};
		\node [style=none] (79) at (5, 0) {};
		\node [style=X dot] (80) at (5.5, -1.25) {};
		\node [style=none] (81) at (6.25, -1.25) {};
		\node [style=none] (82) at (7.75, 0) {$\FaultEq$};
		\node [style=none] (83) at (9.25, 0.5) {};
		\node [style=none] (84) at (9.25, 3) {};
		\node [style=none] (85) at (13.25, 3) {};
		\node [style=none] (86) at (13.25, 0.5) {};
		\node [style=none] (87) at (12.25, 0.5) {};
		\node [style=none] (88) at (10.25, 0.5) {};
		\node [style=none] (89) at (10.25, 2) {};
		\node [style=none] (90) at (12.25, 2) {};
		\node [style=Z dot] (91) at (11.25, 1) {};
		\node [style=none] (92) at (12.25, 1) {};
		\node [style=none] (93) at (10.25, 1) {};
		\node [style=none] (94) at (14, 0.75) {};
		\node [style=none] (95) at (13.25, 0.75) {};
		\node [style=none] (96) at (14, 2.75) {};
		\node [style=none] (97) at (13.25, 2.75) {};
		\node [style=none, rotate=90] (98) at (13.75, 1.75) {$\ldots$};
		\node [style=sLabel] (99) at (11.25, 2.5) {CAT-PART};
		\node [style=none] (107) at (9.25, -1.25) {};
		\node [style=X dot] (114) at (11.25, -1.25) {};
		\node [style=none] (115) at (14, -1.25) {};
		\node [style=none] (116) at (-13.25, -0.25) {};
		\node [style=none] (117) at (-13.25, 3) {};
		\node [style=none] (118) at (-11.25, 3) {};
		\node [style=none] (119) at (-11.25, -0.25) {};
		\node [style=Z dot] (129) at (-10.5, 0) {};
		\node [style=none] (130) at (-10.25, 0.75) {};
		\node [style=none] (131) at (-11.25, 0.75) {};
		\node [style=none] (132) at (-10.25, 2.75) {};
		\node [style=none] (133) at (-11.25, 2.75) {};
		\node [style=none, rotate=90] (134) at (-10.75, 1.75) {$\ldots$};
		\node [style=sLabel] (135) at (-12.25, 1.375) {CAT};
		\node [style=none] (142) at (-13.25, -1.25) {};
		\node [style=Z dot] (145) at (-9.75, 0) {};
		\node [style=X dot] (146) at (-10.5, -1.25) {};
		\node [style=none] (147) at (-9.75, -1.25) {};
		\node [style=none] (148) at (-8.25, 0) {$\eqqcolon$};
		\node [style=none] (149) at (-11.25, 0) {};
		\node [style=sLabel] (150) at (-2.75, -0.375) {qubit $q$};
		\node [style=none] (151) at (-5.875, 0.375) {};
		\node [style=none] (152) at (-5.875, -0.375) {};
		\node [style=none] (153) at (-1.125, -0.625) {};
		\node [style=none] (154) at (-1.125, 0.375) {};
		\node [style=sLabel] (159) at (-5.75, -0.95) {qubit $b$};
		\node [style=none] (160) at (-3.75, -0.375) {};
		\node [style=none] (161) at (-3.75, -0.625) {};
		\node [style=none] (162) at (-1.25, -1.625) {};
		\node [style=none] (163) at (-1.25, -0.875) {};
		\node [style=none] (164) at (-7, -0.625) {};
		\node [style=none] (165) at (-7, -1.625) {};
		\node [style=none] (166) at (-4.625, -0.875) {};
		\node [style=none] (167) at (-4.625, -0.625) {};
		\node [style=none] (168) at (-5.75, 1.75) {};
		\node [style=none] (169) at (-5.75, 0.625) {};
		\node [style=none] (170) at (-3.75, 0.625) {};
		\node [style=none] (171) at (-3.75, 1.75) {};
		\node [style=none] (172) at (-4.375, 0.625) {};
		\node [style=none] (173) at (-4.875, 0.625) {};
		\node [style=sLabel] (174) at (-4.75, 1.425) {qubit $a$};
		\node [style=none] (184) at (10.25, 1.75) {};
		\node [style=none] (186) at (12.25, 0.625) {};
		\node [style=none] (187) at (12.25, 1.75) {};
		\node [style=none] (189) at (10.25, 0.625) {};
		\node [style=sLabel] (190) at (11.25, 1.425) {qubit $a$};
		\node [style=sLabel] (194) at (13, -0.95) {qubit $b$};
		\node [style=none] (195) at (9, -1.625) {};
		\node [style=none] (196) at (9, -0.875) {};
		\node [style=none] (197) at (14.25, -0.625) {};
		\node [style=none] (198) at (14.25, -1.625) {};
		\node [style=none] (199) at (11.875, -0.875) {};
		\node [style=none] (200) at (11.875, -0.625) {};
	\end{pgfonlayer}
	\begin{pgfonlayer}{edgelayer}
		\draw (2.center) to (3.center);
		\draw (3.center) to (4.center);
		\draw (7.center) to (6.center);
		\draw (4.center) to (7.center);
		\draw (6.center) to (5.center);
		\draw (5.center) to (0.center);
		\draw (0.center) to (1.center);
		\draw (1.center) to (2.center);
		\draw (10.center) to (8);
		\draw (8) to (9.center);
		\draw (8) to (11);
		\draw (12) to (11);
		\draw (13) to (11);
		\draw (14.center) to (15.center);
		\draw (17.center) to (16.center);
		\draw (42.center) to (41);
		\draw (41) to (13);
		\draw (13) to (40);
		\draw (46.center) to (47.center);
		\draw (47.center) to (48.center);
		\draw (51.center) to (50.center);
		\draw (48.center) to (51.center);
		\draw (50.center) to (49.center);
		\draw (49.center) to (44.center);
		\draw (44.center) to (45.center);
		\draw (45.center) to (46.center);
		\draw (54.center) to (52);
		\draw (52) to (53.center);
		\draw (58.center) to (59.center);
		\draw (61.center) to (60.center);
		\draw [bend right=315, looseness=1.50] (77.center) to (76.center);
		\draw (76.center) to (52);
		\draw [bend right=315, looseness=1.50] (79.center) to (78.center);
		\draw (80) to (78.center);
		\draw (81.center) to (80);
		\draw (79.center) to (77.center);
		\draw (85.center) to (86.center);
		\draw (86.center) to (87.center);
		\draw (90.center) to (89.center);
		\draw (87.center) to (90.center);
		\draw (89.center) to (88.center);
		\draw (88.center) to (83.center);
		\draw (83.center) to (84.center);
		\draw (84.center) to (85.center);
		\draw (93.center) to (91);
		\draw (91) to (92.center);
		\draw (94.center) to (95.center);
		\draw (97.center) to (96.center);
		\draw (115.center) to (114);
		\draw (91) to (114);
		\draw (114) to (107.center);
		\draw (80) to (71.center);
		\draw (41) to (37.center);
		\draw (118.center) to (119.center);
		\draw (116.center) to (117.center);
		\draw (117.center) to (118.center);
		\draw (130.center) to (131.center);
		\draw (133.center) to (132.center);
		\draw (147.center) to (146);
		\draw (146) to (129);
		\draw (129) to (145);
		\draw (146) to (142.center);
		\draw (116.center) to (119.center);
		\draw (149.center) to (129);
		\draw [fill=gray, opacity=0.2] (154.center)
			 to (153.center)
			 to (161.center)
			 to (160.center)
			 to (152.center)
			 to (151.center)
			 to cycle;
		\draw [fill=gray, opacity=0.2] (165.center)
			 to (164.center)
			 to (167.center)
			 to (166.center)
			 to (163.center)
			 to (162.center)
			 to cycle;
		\draw [fill=gray, opacity=0.2] (171.center)
			 to (170.center)
			 to (173.center)
			 to (172.center)
			 to (169.center)
			 to (168.center)
			 to cycle;
		\draw [fill=gray, opacity=0.2] (187.center)
			 to (186.center)
			 to (189.center)
			 to (184.center)
			 to cycle;
		\draw [fill=gray, opacity=0.2] (198.center)
			 to (197.center)
			 to (200.center)
			 to (199.center)
			 to (196.center)
			 to (195.center)
			 to cycle;
	\end{pgfonlayer}
\end{tikzpicture}

%% file: figures/cat-part-2-fe.tikz
\begin{tikzpicture}
	\begin{pgfonlayer}{nodelayer}
		\node [style=none] (0) at (-13.5, 1) {};
		\node [style=none] (1) at (-13.5, 3) {};
		\node [style=none] (2) at (-2.5, 3) {};
		\node [style=none] (3) at (-2.5, 0.5) {};
		\node [style=none] (4) at (-11, 0.5) {};
		\node [style=none] (5) at (-12.5, 1) {};
		\node [style=none] (6) at (-12.5, 2.075) {};
		\node [style=none] (7) at (-11, 2.075) {};
		\node [style=Z dot] (8) at (-11.75, 1.5) {};
		\node [style=none] (9) at (-11, 1.5) {};
		\node [style=none] (10) at (-12.5, 1.5) {};
		\node [style=X dot] (11) at (-11.75, -0.25) {};
		\node [style=X dot] (12) at (-12.75, -0.25) {};
		\node [style=none] (13) at (-9.75, -0.25) {};
		\node [style=none] (14) at (-1.5, 0.75) {};
		\node [style=none] (15) at (-2.5, 0.75) {};
		\node [style=none] (16) at (-1.5, 2.75) {};
		\node [style=none] (17) at (-2.5, 2.75) {};
		\node [style=none, rotate=90] (18) at (-2, 1.75) {$\ldots$};
		\node [style=sLabel] (19) at (-8.5, 1.75) {CAT-PART};
		\node [style=Z dot] (20) at (-10.5, -0.25) {};
		\node [style=none] (23) at (-10.5, 0.5) {};
		\node [style=none] (24) at (-7.75, -0.25) {};
		\node [style=none] (25) at (-8.75, -0.25) {$\ldots$};
		\node [style=none] (26) at (-5, 0.5) {};
		\node [style=none] (27) at (-5, 2.075) {};
		\node [style=none] (28) at (-3, 2.075) {};
		\node [style=none] (29) at (-3, 0.5) {};
		\node [style=Z dot] (30) at (-4, 1) {};
		\node [style=none] (32) at (-3, 1) {};
		\node [style=none] (33) at (-5, 1) {};
		\node [style=Z dot] (35) at (-6, -0.25) {};
		\node [style=Z dot] (36) at (-2.5, -0.25) {};
		\node [style=Z dot] (37) at (-7, -0.25) {};
		\node [style=none] (38) at (-7, 0.5) {};
		\node [style=X dot] (39) at (-4, -1.25) {};
		\node [style=X dot] (40) at (-6, -1.25) {};
		\node [style=X dot] (41) at (-6.75, -1.25) {};
		\node [style=X dot] (42) at (-3.25, -1.25) {};
		\node [style=Z dot] (56) at (-1.75, -0.25) {};
		\node [style=none] (115) at (-0.25, -0.25) {$\FaultEq$};
		\node [style=X dot] (134) at (-2.5, -2.5) {};
		\node [style=none] (135) at (-1.75, -2.5) {};
		\node [style=none] (142) at (-13.5, -2.5) {};
		\node [style=none] (145) at (1, 1) {};
		\node [style=none] (146) at (1, 3) {};
		\node [style=none] (147) at (11.25, 3) {};
		\node [style=none] (148) at (11.25, 0.5) {};
		\node [style=none] (149) at (3.5, 0.5) {};
		\node [style=none] (150) at (2, 1) {};
		\node [style=none] (151) at (2, 2) {};
		\node [style=none] (152) at (3.5, 2) {};
		\node [style=Z dot] (153) at (2.75, 1.5) {};
		\node [style=none] (154) at (3.5, 1.5) {};
		\node [style=none] (155) at (2, 1.5) {};
		\node [style=X dot] (156) at (2.75, -0.25) {};
		\node [style=X dot] (157) at (1.75, -0.25) {};
		\node [style=none] (158) at (4.75, -0.25) {};
		\node [style=none] (159) at (12.25, 0.75) {};
		\node [style=none] (160) at (11.25, 0.75) {};
		\node [style=none] (161) at (12.25, 2.75) {};
		\node [style=none] (162) at (11.25, 2.75) {};
		\node [style=none, rotate=90] (163) at (11.75, 1.75) {$\ldots$};
		\node [style=sLabel] (164) at (6, 1.75) {CAT-PART};
		\node [style=Z dot] (165) at (4, -0.25) {};
		\node [style=none] (166) at (4, 0.5) {};
		\node [style=none] (167) at (6.75, -0.25) {};
		\node [style=none] (168) at (5.75, -0.25) {$\ldots$};
		\node [style=none] (169) at (8.25, 0.5) {};
		\node [style=none] (170) at (8.25, 2) {};
		\node [style=none] (171) at (10.25, 2) {};
		\node [style=none] (172) at (10.25, 0.5) {};
		\node [style=Z dot] (173) at (9.25, 1) {};
		\node [style=none] (174) at (10.25, 1) {};
		\node [style=none] (175) at (8.25, 1) {};
		\node [style=Z dot] (176) at (8.25, -0.25) {};
		\node [style=Z dot] (178) at (7.5, -0.25) {};
		\node [style=none] (179) at (7.5, 0.5) {};
		\node [style=X dot] (185) at (8.25, -2.25) {};
		\node [style=none] (186) at (12, -2.25) {};
		\node [style=none] (193) at (1, -2.25) {};
		\node [style=X dot] (196) at (9.25, -0.25) {};
		\node [style=X dot] (197) at (10, -0.25) {};
		\node [style=none] (198) at (-19.5, -0.5) {};
		\node [style=none] (199) at (-19.5, 3) {};
		\node [style=none] (200) at (-17.5, 3) {};
		\node [style=none] (201) at (-17.5, -0.5) {};
		\node [style=Z dot] (202) at (-16.75, -0.25) {};
		\node [style=none] (203) at (-16.5, 0.75) {};
		\node [style=none] (204) at (-17.5, 0.75) {};
		\node [style=none] (205) at (-16.5, 2.75) {};
		\node [style=none] (206) at (-17.5, 2.75) {};
		\node [style=none, rotate=90] (207) at (-17, 1.75) {$\ldots$};
		\node [style=sLabel] (208) at (-18.5, 1.375) {CAT};
		\node [style=none] (215) at (-19.5, -2.25) {};
		\node [style=Z dot] (218) at (-16, -0.25) {};
		\node [style=X dot] (219) at (-16.75, -2.25) {};
		\node [style=none] (220) at (-16, -2.25) {};
		\node [style=none] (221) at (-14.75, -0.25) {$\eqqcolon$};
		\node [style=none] (222) at (-17.5, -0.25) {};
		\node [style=none] (226) at (-2.875, -0.875) {};
		\node [style=none] (227) at (-2.875, -2) {};
		\node [style=none] (228) at (-7.125, -2) {};
		\node [style=none] (229) at (-7.125, -0.875) {};
		\node [style=sLabel, inner sep=1pt] (233) at (-5, -1.75) {qubit $f$};
		\node [style=none] (234) at (-1.375, 0.125) {};
		\node [style=none] (235) at (-1.375, -0.625) {};
		\node [style=none] (236) at (-13.125, -1) {};
		\node [style=none] (237) at (-13.125, 0.125) {};
		\node [style=sLabel, inner sep=1pt] (238) at (-11.875, -0.75) {qubit $q$};
		\node [style=none] (239) at (-10.75, -1) {};
		\node [style=none] (240) at (-10.75, -0.625) {};
		\node [style=none] (245) at (-1.5, -2.75) {};
		\node [style=none] (246) at (-11.25, -1.775) {};
		\node [style=none] (247) at (-13.75, -1.775) {};
		\node [style=none] (248) at (-13.75, -2.75) {};
		\node [style=sLabel, inner sep=1pt] (249) at (-12.5, -2.125) {qubit $b$};
		\node [style=none] (250) at (-11.25, -2.25) {};
		\node [style=none] (251) at (-1.5, -2.25) {};
		\node [style=none] (256) at (-3, 0.625) {};
		\node [style=none] (257) at (-3, 1.825) {};
		\node [style=none] (258) at (-5, 1.825) {};
		\node [style=none] (259) at (-5, 0.625) {};
		\node [style=sLabel, inner sep=1pt] (260) at (-4, 1.5) {qubit $a$};
		\node [style=none] (288) at (10.375, 0.125) {};
		\node [style=none] (289) at (10.375, -0.625) {};
		\node [style=none] (290) at (1.375, -1) {};
		\node [style=none] (291) at (1.375, 0.125) {};
		\node [style=sLabel, inner sep=1pt] (292) at (2.625, -0.75) {qubit $q$};
		\node [style=none] (293) at (3.75, -1) {};
		\node [style=none] (294) at (3.75, -0.625) {};
		\node [style=none] (295) at (12.25, -2.5) {};
		\node [style=none] (296) at (3.25, -1.525) {};
		\node [style=none] (297) at (0.75, -1.525) {};
		\node [style=none] (298) at (0.75, -2.5) {};
		\node [style=sLabel, inner sep=1pt] (299) at (2, -1.875) {qubit $b$};
		\node [style=none] (300) at (3.25, -2) {};
		\node [style=none] (301) at (12.25, -2) {};
		\node [style=none] (302) at (10.25, 0.625) {};
		\node [style=none] (303) at (10.25, 1.825) {};
		\node [style=none] (304) at (8.25, 1.825) {};
		\node [style=none] (305) at (8.25, 0.625) {};
		\node [style=sLabel, inner sep=1pt] (306) at (9.25, 1.5) {qubit $a$};
	\end{pgfonlayer}
	\begin{pgfonlayer}{edgelayer}
		\draw (2.center) to (3.center);
		\draw (7.center) to (6.center);
		\draw (4.center) to (7.center);
		\draw (6.center) to (5.center);
		\draw (5.center) to (0.center);
		\draw (0.center) to (1.center);
		\draw (1.center) to (2.center);
		\draw (10.center) to (8);
		\draw (8) to (9.center);
		\draw (8) to (11);
		\draw (12) to (11);
		\draw (13.center) to (11);
		\draw (14.center) to (15.center);
		\draw (17.center) to (16.center);
		\draw (23.center) to (20);
		\draw (3.center) to (29.center);
		\draw (29.center) to (28.center);
		\draw (28.center) to (27.center);
		\draw (27.center) to (26.center);
		\draw (26.center) to (4.center);
		\draw (33.center) to (30);
		\draw (35) to (24.center);
		\draw (36) to (35);
		\draw (38.center) to (37);
		\draw (30) to (39);
		\draw (40) to (35);
		\draw (56) to (36);
		\draw (32.center) to (30);
		\draw (40) to (41);
		\draw (40) to (39);
		\draw (39) to (42);
		\draw (135.center) to (134);
		\draw (142.center) to (134);
		\draw (134) to (36);
		\draw (147.center) to (148.center);
		\draw (152.center) to (151.center);
		\draw (149.center) to (152.center);
		\draw (151.center) to (150.center);
		\draw (150.center) to (145.center);
		\draw (145.center) to (146.center);
		\draw (146.center) to (147.center);
		\draw (155.center) to (153);
		\draw (153) to (154.center);
		\draw (153) to (156);
		\draw (157) to (156);
		\draw (158.center) to (156);
		\draw (159.center) to (160.center);
		\draw (162.center) to (161.center);
		\draw (166.center) to (165);
		\draw (148.center) to (172.center);
		\draw (172.center) to (171.center);
		\draw (171.center) to (170.center);
		\draw (170.center) to (169.center);
		\draw (169.center) to (149.center);
		\draw (175.center) to (173);
		\draw (176) to (167.center);
		\draw (179.center) to (178);
		\draw (174.center) to (173);
		\draw (186.center) to (185);
		\draw (193.center) to (185);
		\draw (176) to (185);
		\draw (196) to (173);
		\draw (197) to (196);
		\draw (196) to (176);
		\draw (200.center) to (201.center);
		\draw (198.center) to (199.center);
		\draw (199.center) to (200.center);
		\draw (203.center) to (204.center);
		\draw (206.center) to (205.center);
		\draw (220.center) to (219);
		\draw (219) to (202);
		\draw (202) to (218);
		\draw (219) to (215.center);
		\draw (198.center) to (201.center);
		\draw (222.center) to (202);
		\draw [fill=gray, opacity=0.2] (229.center)
			 to (226.center)
			 to (227.center)
			 to (228.center)
			 to cycle;
		\draw [fill=gray, opacity=0.2] (240.center)
			 to (235.center)
			 to (234.center)
			 to [in=0, out=180] (237.center)
			 to (236.center)
			 to (239.center);
		\draw [fill=gray, opacity=0.2] (247.center)
			 to (248.center)
			 to (245.center)
			 to (251.center)
			 to (250.center)
			 to (246.center)
			 to cycle;
		\draw [fill=gray, opacity=0.2] (259.center)
			 to (256.center)
			 to (257.center)
			 to (258.center)
			 to cycle;
		\draw [fill=gray, opacity=0.2] (294.center)
			 to (289.center)
			 to (288.center)
			 to [in=0, out=180] (291.center)
			 to (290.center)
			 to (293.center);
		\draw [fill=gray, opacity=0.2] (297.center)
			 to (298.center)
			 to (295.center)
			 to (301.center)
			 to (300.center)
			 to (296.center)
			 to cycle;
		\draw [fill=gray, opacity=0.2] (305.center)
			 to (302.center)
			 to (303.center)
			 to (304.center)
			 to cycle;
	\end{pgfonlayer}
\end{tikzpicture}

%% file: figures/cat-part-counter-example-2.tikz
\begin{tikzpicture}
	\begin{pgfonlayer}{nodelayer}
		\node [style=none] (0) at (0, -0.75) {};
		\node [style=none] (1) at (0, 1.25) {};
		\node [style=none] (2) at (11.75, 1.25) {};
		\node [style=none] (3) at (11.75, -1.25) {};
		\node [style=none] (4) at (2.5, -1.25) {};
		\node [style=none] (5) at (1, -0.75) {};
		\node [style=none] (6) at (1, 0.25) {};
		\node [style=none] (7) at (2.5, 0.25) {};
		\node [style=Z dot] (8) at (1.75, -0.25) {};
		\node [style=none] (9) at (2.5, -0.25) {};
		\node [style=none] (10) at (1, -0.25) {};
		\node [style=X dot] (11) at (1.75, -2) {};
		\node [style=X dot] (12) at (0.5, -2) {};
		\node [style=none] (13) at (3.75, -2) {};
		\node [style=none] (14) at (12.75, -1) {};
		\node [style=none] (15) at (11.75, -1) {};
		\node [style=none] (16) at (12.75, 1) {};
		\node [style=none] (17) at (11.75, 1) {};
		\node [style=none, rotate=90] (18) at (12.25, 0) {$\ldots$};
		\node [style=sLabel] (19) at (6, 0.5) {CAT-PART};
		\node [style=Z dot] (21) at (3, -2) {};
		\node [style=none] (22) at (3, -1.25) {};
		\node [style=none] (25) at (7.5, -2) {$\ldots$};
		\node [style=none] (26) at (5.25, -1.25) {};
		\node [style=none] (27) at (5.25, -0.25) {};
		\node [style=none] (28) at (6.75, -0.25) {};
		\node [style=none] (29) at (6.75, -1.25) {};
		\node [style=Z dot] (30) at (6, -0.75) {};
		\node [style=none] (32) at (6.75, -0.75) {};
		\node [style=none] (33) at (5.25, -0.75) {};
		\node [style=Z dot] (35) at (10, -2) {};
		\node [style=Z dot] (36) at (11.75, -2) {};
		\node [style=X dot] (39) at (6, -3) {};
		\node [style=X dot] (40) at (10, -3) {};
		\node [style=X dot] (41) at (5, -3) {};
		\node [style=X dot] (42) at (11, -3) {};
		\node [style=none] (48) at (0, -4) {};
		\node [style=X dot] (51) at (11.75, -4) {};
		\node [style=Z dot] (55) at (12.75, -2) {};
		\node [style=none] (56) at (12.75, -4) {};
		\node [style=none] (57) at (4.75, -2) {$\ldots$};
		\node [style=none] (58) at (15.75, -0.5) {};
		\node [style=none] (59) at (15.75, 1.25) {};
		\node [style=none] (60) at (26.75, 1.25) {};
		\node [style=none] (61) at (26.75, -1.25) {};
		\node [style=none] (62) at (18.25, -1.25) {};
		\node [style=none] (63) at (16.75, -0.5) {};
		\node [style=none] (64) at (16.75, 0.25) {};
		\node [style=none] (65) at (18.25, 0.25) {};
		\node [style=Z dot] (66) at (17.5, -0.25) {};
		\node [style=none] (67) at (18.25, -0.25) {};
		\node [style=none] (68) at (16.75, -0.25) {};
		\node [style=X dot] (69) at (17.5, -2) {};
		\node [style=X dot] (70) at (16.25, -2) {};
		\node [style=none] (71) at (19.5, -2) {};
		\node [style=none] (72) at (27.75, -1) {};
		\node [style=none] (73) at (26.75, -1) {};
		\node [style=none] (74) at (27.75, 1) {};
		\node [style=none] (75) at (26.75, 1) {};
		\node [style=none, rotate=90] (76) at (27.25, 0) {$\ldots$};
		\node [style=Z dot] (79) at (18.75, -2) {};
		\node [style=none] (80) at (18.75, -1.25) {};
		\node [style=none] (82) at (24, -2) {};
		\node [style=none] (83) at (23, -2) {$\ldots$};
		\node [style=none] (84) at (21, -1.25) {};
		\node [style=none] (85) at (21, -0.25) {};
		\node [style=none] (86) at (22.5, -0.25) {};
		\node [style=none] (87) at (22.5, -1.25) {};
		\node [style=Z dot] (88) at (21.75, -0.75) {};
		\node [style=none] (89) at (22.5, -0.75) {};
		\node [style=none] (90) at (21, -0.75) {};
		\node [style=Z dot] (92) at (25.75, -2) {};
		\node [style=Z dot] (93) at (24.75, -2) {};
		\node [style=X dot] (95) at (26.75, -2) {};
		\node [style=none] (104) at (15.75, -4) {};
		\node [style=X dot] (107) at (25.75, -4) {};
		\node [style=none] (113) at (20.5, -2) {$\ldots$};
		\node [style=X dot] (114) at (27.75, -2) {};
		\node [style=none] (115) at (26.75, -4) {};
		\node [style=none] (116) at (14.25, -1) {$\FaultEq$};
		\node [style=none] (117) at (24, -1.25) {};
		\node [style=none] (118) at (24, -0.25) {};
		\node [style=none] (119) at (25.5, -0.25) {};
		\node [style=none] (120) at (25.5, -1.25) {};
		\node [style=X dot] (121) at (24.75, -0.75) {};
		\node [style=none] (122) at (25.5, -0.75) {};
		\node [style=none] (123) at (24, -0.75) {};
		\node [style=none] (207) at (8.25, -2) {};
		\node [style=Z dot] (208) at (9, -2) {};
		\node [style=none] (209) at (8.25, -1.25) {};
		\node [style=none] (210) at (8.25, -0.25) {};
		\node [style=none] (211) at (9.75, -0.25) {};
		\node [style=none] (212) at (9.75, -1.25) {};
		\node [style=X dot] (213) at (9, -0.75) {};
		\node [style=none] (214) at (9.75, -0.75) {};
		\node [style=none] (215) at (8.25, -0.75) {};
		\node [style=sLabel] (216) at (21.75, 0.5) {CAT-PART};
		\node [style=none] (217) at (22.5, -1.5) {};
	\end{pgfonlayer}
	\begin{pgfonlayer}{edgelayer}
		\draw (2.center) to (3.center);
		\draw (7.center) to (6.center);
		\draw (4.center) to (7.center);
		\draw (6.center) to (5.center);
		\draw (5.center) to (0.center);
		\draw (0.center) to (1.center);
		\draw (1.center) to (2.center);
		\draw (10.center) to (8);
		\draw (8) to (9.center);
		\draw (8) to (11);
		\draw (12) to (11);
		\draw (13.center) to (11);
		\draw (14.center) to (15.center);
		\draw (17.center) to (16.center);
		\draw (22.center) to (21);
		\draw (29.center) to (28.center);
		\draw (28.center) to (27.center);
		\draw (27.center) to (26.center);
		\draw (26.center) to (4.center);
		\draw (36) to (35);
		\draw (30) to (39);
		\draw (40) to (35);
		\draw (40) to (39);
		\draw (41) to (39);
		\draw (42) to (40);
		\draw (51) to (48.center);
		\draw (36) to (51);
		\draw (55) to (36);
		\draw (56.center) to (51);
		\draw (32.center) to (33.center);
		\draw (60.center) to (61.center);
		\draw (65.center) to (64.center);
		\draw (62.center) to (65.center);
		\draw (64.center) to (63.center);
		\draw (63.center) to (58.center);
		\draw (58.center) to (59.center);
		\draw (59.center) to (60.center);
		\draw (68.center) to (66);
		\draw (66) to (67.center);
		\draw (66) to (69);
		\draw (70) to (69);
		\draw (71.center) to (69);
		\draw (72.center) to (73.center);
		\draw (75.center) to (74.center);
		\draw (80.center) to (79);
		\draw (87.center) to (86.center);
		\draw (86.center) to (85.center);
		\draw (85.center) to (84.center);
		\draw (84.center) to (62.center);
		\draw (107) to (104.center);
		\draw (92) to (107);
		\draw (89.center) to (90.center);
		\draw (95) to (82.center);
		\draw (114) to (95);
		\draw (115.center) to (107);
		\draw (120.center) to (119.center);
		\draw (119.center) to (118.center);
		\draw (118.center) to (117.center);
		\draw (122.center) to (123.center);
		\draw (121) to (93);
		\draw (61.center) to (120.center);
		\draw (117.center) to (87.center);
		\draw (212.center) to (211.center);
		\draw (211.center) to (210.center);
		\draw (210.center) to (209.center);
		\draw (214.center) to (215.center);
		\draw (213) to (208);
		\draw (207.center) to (35);
		\draw (29.center) to (209.center);
		\draw (212.center) to (3.center);
		\draw [in=-90, out=-180, looseness=1.25] (217.center) to (88);
		\draw [in=105, out=0, looseness=0.50] (217.center) to (95);
	\end{pgfonlayer}
\end{tikzpicture}

%% file: figures/well_ordered_extraction.tikz
\begin{tikzpicture}
	\begin{pgfonlayer}{nodelayer}
		\node [style=Z bold phase dot] (0) at (-7.5, 13) {9};
		\node [style=Z bold phase dot] (1) at (-5.25, 11.5) {8};
		\node [style=Z bold phase dot] (2) at (-9.75, 11.5) {3};
		\node [style=Z bold phase dot] (3) at (-3.5, 9) {1};
		\node [style=Z bold phase dot] (4) at (-11.5, 9) {11};
		\node [style=Z bold phase dot] (5) at (-5.25, 6.5) {7};
		\node [style=Z bold phase dot] (6) at (-9.75, 6.5) {10};
		\node [style=Z bold phase dot] (7) at (-7.5, 5) {2};
		\node [style=Z bold phase dot] (8) at (-7.5, 9) {0};
		\node [style=Z bold phase dot] (9) at (-8.75, 10) {4};
		\node [style=Z bold phase dot] (10) at (-5.75, 9.25) {5};
		\node [style=Z bold phase dot] (11) at (-7.75, 7.25) {6};
		\node [style=Z bold phase dot] (12) at (3, 13) {9};
		\node [style=Z bold phase dot] (13) at (5.25, 11.5) {8};
		\node [style=Z bold phase dot] (14) at (0.75, 11.5) {3};
		\node [style=Z bold phase dot] (15) at (7, 9) {1};
		\node [style=Z bold phase dot] (16) at (-1, 9) {11};
		\node [style=Z bold phase dot] (17) at (5.25, 6.5) {7};
		\node [style=Z bold phase dot] (18) at (0.75, 6.5) {10};
		\node [style=Z bold phase dot] (19) at (3, 5) {2};
		\node [style=Z bold phase dot] (20) at (3, 9) {0};
		\node [style=Z bold phase dot] (21) at (1.75, 10) {4};
		\node [style=Z bold phase dot] (22) at (4.75, 9.25) {5};
		\node [style=Z bold phase dot] (23) at (2.75, 7.25) {6};
		\node [style=Z bold phase dot] (24) at (13.5, 13) {9};
		\node [style=Z bold phase dot] (25) at (15.75, 11.5) {8};
		\node [style=Z bold phase dot] (26) at (11.25, 11.5) {3};
		\node [style=Z bold phase dot] (27) at (17.5, 9) {1};
		\node [style=Z bold phase dot] (28) at (9.5, 9) {11};
		\node [style=Z bold phase dot] (29) at (15.75, 6.5) {7};
		\node [style=Z bold phase dot] (30) at (11.25, 6.5) {10};
		\node [style=Z bold phase dot] (31) at (13.5, 5) {2};
		\node [style=Z bold phase dot] (32) at (13.5, 9) {0};
		\node [style=Z bold phase dot] (33) at (12.25, 10) {4};
		\node [style=Z bold phase dot] (34) at (15.25, 9.25) {5};
		\node [style=Z bold phase dot] (35) at (13.25, 7.25) {6};
		\node [style=cnot ctrl] (37) at (-11.75, -7.75) {};
		\node [style=not] (38) at (-11.75, -8.75) {};
		\node [style=none] (39) at (8.75, -16.75) {};
		\node [style=meter] (40) at (9.5, -16.75) {};
		\node [style=none, label={left:\scriptsize $\ket{+}$}] (42) at (-12.25, -7.75) {};
		\node [style=none, label={left:\scriptsize $\ket{0}$}] (43) at (-12.25, -8.75) {};
		\node [style=none, label={left:\scriptsize $\ket{0}$}] (44) at (-11.25, -9.75) {};
		\node [style=none, label={left:\scriptsize $\ket{0}$}] (45) at (-8.75, -10.75) {};
		\node [style=none, label={left:\scriptsize $\ket{0}$}] (46) at (-7.75, -11.75) {};
		\node [style=none, label={left:\scriptsize $\ket{0}$}] (47) at (-6.75, -12.75) {};
		\node [style=none, label={left:\scriptsize $\ket{0}$}] (48) at (-2.75, -13.75) {};
		\node [style=none, label={left:\scriptsize $\ket{0}$}] (49) at (4.25, -14.75) {};
		\node [style=none, label={left:\scriptsize $\ket{0}$}] (50) at (-4.75, -15.75) {};
		\node [style=none, label={left:\scriptsize $\ket{0}$}] (51) at (-0.75, -16.75) {};
		\node [style=none, label={left:\scriptsize $\ket{0}$}] (52) at (6.25, -16.75) {};
		\node [style=none] (53) at (12.5, -7.75) {};
		\node [style=none] (54) at (12.5, -8.75) {};
		\node [style=none] (55) at (12.5, -9.75) {};
		\node [style=none] (56) at (12.5, -10.75) {};
		\node [style=none] (57) at (12.5, -11.75) {};
		\node [style=none] (58) at (12.5, -12.75) {};
		\node [style=none] (59) at (12.5, -13.75) {};
		\node [style=none] (60) at (12.5, -14.75) {};
		\node [style=cnot ctrl] (61) at (-10.75, -7.75) {};
		\node [style=not] (62) at (-10.75, -9.75) {};
		\node [style=cnot ctrl] (63) at (-8.25, -9.75) {};
		\node [style=not] (64) at (-8.25, -10.75) {};
		\node [style=cnot ctrl] (65) at (-7.25, -8.75) {};
		\node [style=not] (66) at (-7.25, -11.75) {};
		\node [style=cnot ctrl] (67) at (-6.25, -7.75) {};
		\node [style=not] (68) at (-6.25, -12.75) {};
		\node [style=cnot ctrl] (69) at (-4.25, -7.75) {};
		\node [style=not] (70) at (-4.25, -15.75) {};
		\node [style=cnot ctrl] (71) at (-2.25, -7.75) {};
		\node [style=not] (72) at (-2.25, -13.75) {};
		\node [style=cnot ctrl] (73) at (-0.25, -7.75) {};
		\node [style=not] (74) at (-0.25, -16.75) {};
		\node [style=cnot ctrl] (75) at (1.25, -9.75) {};
		\node [style=not] (76) at (1.25, -16.75) {};
		\node [style=cnot ctrl] (77) at (4.75, -9.75) {};
		\node [style=not] (78) at (4.75, -14.75) {};
		\node [style=cnot ctrl] (79) at (6.75, -9.75) {};
		\node [style=not] (80) at (6.75, -16.75) {};
		\node [style=cnot ctrl] (81) at (8.25, -8.75) {};
		\node [style=not] (82) at (8.25, -16.75) {};
		\node [style=cnot ctrl] (83) at (11.25, -8.75) {};
		\node [style=not] (84) at (11.25, -15.75) {};
		\node [style=none] (85) at (1.75, -16.75) {};
		\node [style=meter] (86) at (2.5, -16.75) {};
		\node [style=none] (88) at (11.75, -15.75) {};
		\node [style=meter] (89) at (12.5, -15.75) {};
		\node [style=Z bold phase dot] (91) at (35.5, 7) {9};
		\node [style=Z bold phase dot] (92) at (37.5, 7) {8};
		\node [style=Z bold phase dot] (93) at (33, 7) {3};
		\node [style=Z bold phase dot] (94) at (40, 9) {1};
		\node [style=Z bold phase dot] (95) at (30.5, 11) {11};
		\node [style=Z bold phase dot] (96) at (42.75, 9) {7};
		\node [style=Z bold phase dot] (97) at (28.5, 11) {10};
		\node [style=Z bold phase dot] (98) at (26.5, 11) {2};
		\node [style=Z bold phase dot] (99) at (20, 9) {0};
		\node [style=Z bold phase dot] (100) at (23.5, 7) {4};
		\node [style=Z bold phase dot] (101) at (23.5, 9) {5};
		\node [style=Z bold phase dot] (102) at (23.5, 11) {6};
		\node [style=none] (103) at (21.5, 12) {};
		\node [style=none] (104) at (21.5, 5.75) {};
		\node [style=none] (105) at (25.5, 12) {};
		\node [style=none] (106) at (25.5, 5.75) {};
		\node [style=none] (107) at (27.5, 12) {};
		\node [style=none] (108) at (27.5, 5.75) {};
		\node [style=none] (109) at (29.5, 12) {};
		\node [style=none] (110) at (29.5, 5.75) {};
		\node [style=none] (111) at (31.5, 12) {};
		\node [style=none] (112) at (31.5, 5.75) {};
		\node [style=none] (113) at (34.5, 12) {};
		\node [style=none] (114) at (34.5, 5.75) {};
		\node [style=none] (115) at (36.5, 12) {};
		\node [style=none] (116) at (36.5, 5.75) {};
		\node [style=none] (117) at (38.5, 12) {};
		\node [style=none] (118) at (38.5, 5.75) {};
		\node [style=none] (119) at (41.5, 12) {};
		\node [style=none] (120) at (41.5, 5.75) {};
		\node [style=Z bold phase dot] (121) at (3.25, -3.25) {9};
		\node [style=Z bold phase dot] (122) at (5.25, -3.25) {8};
		\node [style=Z bold phase dot] (123) at (0.75, -3.25) {3};
		\node [style=Z bold phase dot] (124) at (7.75, -1.25) {1};
		\node [style=Z bold phase dot] (125) at (-1.75, 0.75) {11};
		\node [style=Z bold phase dot] (126) at (10.5, -1.25) {7};
		\node [style=Z bold phase dot] (127) at (-3.75, 0.75) {10};
		\node [style=Z bold phase dot] (128) at (-5.75, 0.75) {2};
		\node [style=Z bold phase dot] (129) at (-12.25, -1.25) {0};
		\node [style=Z bold phase dot] (130) at (-8.75, -3.25) {4};
		\node [style=Z bold phase dot] (131) at (-8.75, -1.25) {5};
		\node [style=Z bold phase dot] (132) at (-8.75, 0.75) {6};
		\node [style=none] (133) at (-7.25, -5) {};
		\node [style=none] (134) at (-7.25, -2.75) {};
		\node [style=none] (135) at (-7.25, -0.75) {};
		\node [style=none] (136) at (-2.25, -0.75) {};
		\node [style=none] (137) at (4.75, -4.75) {};
		\node [style=none] (138) at (12, -3.25) {};
		\node [style=none] (139) at (12, -1.25) {};
		\node [style=none] (140) at (-5.75, -4.25) {};
		\node [style=none] (141) at (10.5, -4.25) {};
		\node [style=none] (142) at (12, 0.75) {};
		\node [style=none] (143) at (-1.75, -5.25) {};
		\node [style=none] (144) at (0.75, -5.25) {};
		\node [style=none] (145) at (5.25, -5.75) {};
		\node [style=none] (146) at (7.75, -5.75) {};
		\node [style=none] (147) at (-5.5, 13) {};
		\node [style=none] (148) at (1, 13) {};
		\node [style=none] (149) at (-2.25, 14) {(2)};
		\node [style=none] (150) at (5, 13) {};
		\node [style=none] (151) at (11.5, 13) {};
		\node [style=none] (152) at (8.25, 14) {(3)};
		\node [style=none] (153) at (16.25, 13) {};
		\node [style=none] (154) at (22.75, 13) {};
		\node [style=none, align=center] (155) at (19.5, 14) {(4)};
		\node [style=none] (157) at (-6, 4.5) {};
		\node [style=none] (158) at (-6, 1.75) {};
		\node [style=none] (159) at (31.5, 4.5) {};
		\node [style=none] (160) at (31.5, 3.5) {};
		\node [style=none] (161) at (4, 3.5) {};
		\node [style=none] (162) at (4, 1.75) {};
		\node [style=none, align=center] (164) at (-1, 2.75) {(5)};
		\node [style=Z bold phase dot] (165) at (34.75, -3.25) {9};
		\node [style=Z bold phase dot] (166) at (36.75, -3.25) {8};
		\node [style=Z bold phase dot] (167) at (32.25, -3.25) {3};
		\node [style=Z bold phase dot] (168) at (39.25, -1.25) {1};
		\node [style=Z bold phase dot] (169) at (29.75, 0.75) {11};
		\node [style=Z bold phase dot] (170) at (42, -1.25) {7};
		\node [style=Z bold phase dot] (171) at (28, 0.75) {10};
		\node [style=Z bold phase dot] (172) at (25.75, 0.75) {2};
		\node [style=Z bold phase dot] (173) at (18, 0.75) {0};
		\node [style=Z bold phase dot] (174) at (22.5, -3.25) {4};
		\node [style=Z bold phase dot] (175) at (23.5, -1.25) {5};
		\node [style=Z bold phase dot] (176) at (24.5, 0.75) {6};
		\node [style=none] (177) at (24.875, -5) {};
		\node [style=none] (178) at (25.5, -6.75) {};
		\node [style=none] (179) at (26.25, -8.25) {};
		\node [style=none] (180) at (29.5, -9.75) {};
		\node [style=none] (181) at (36.25, -11.25) {};
		\node [style=none] (182) at (43.5, -3.25) {};
		\node [style=none] (183) at (43.5, -1.25) {};
		\node [style=X bold phase dot] (184) at (25.75, -12.75) {};
		\node [style=X bold phase dot] (185) at (42, -12.75) {};
		\node [style=none] (186) at (43.5, 0.75) {};
		\node [style=X bold phase dot] (187) at (29.75, -14.25) {};
		\node [style=X bold phase dot] (188) at (32.25, -14.25) {};
		\node [style=X bold phase dot] (189) at (36.75, -14.25) {};
		\node [style=X bold phase dot] (190) at (39.25, -14.25) {};
		\node [style=X bold phase dot] (191) at (19.5, -1.25) {};
		\node [style=Z bold phase dot] (192) at (19.5, 0.75) {};
		\node [style=Z bold phase dot] (193) at (21, 0.75) {};
		\node [style=X bold phase dot] (194) at (21, -3.25) {};
		\node [style=X bold phase dot] (195) at (18, -1.25) {};
		\node [style=X bold phase dot] (196) at (18, -3.25) {};
		\node [style=X bold phase dot] (197) at (22.5, -5) {};
		\node [style=X bold phase dot] (198) at (21, -5) {};
		\node [style=X bold phase dot] (199) at (21.75, -6.75) {};
		\node [style=X bold phase dot] (200) at (23.5, -6.75) {};
		\node [style=X bold phase dot] (201) at (24.5, -8.25) {};
		\node [style=X bold phase dot] (202) at (23, -8.25) {};
		\node [style=X bold phase dot] (203) at (24.25, -12.75) {};
		\node [style=X bold phase dot] (204) at (43.5, -12.75) {};
		\node [style=X bold phase dot] (205) at (28, -9.75) {};
		\node [style=X bold phase dot] (206) at (26.5, -9.75) {};
		\node [style=X bold phase dot] (207) at (34.75, -11.25) {};
		\node [style=X bold phase dot] (208) at (33.25, -11.25) {};
		\node [style=X bold phase dot] (209) at (28.25, -14.25) {};
		\node [style=X bold phase dot] (210) at (33.75, -14.25) {};
		\node [style=X bold phase dot] (211) at (35.25, -14.25) {};
		\node [style=X bold phase dot] (212) at (40.75, -14.25) {};
		\node [style=none] (213) at (11, -5.5) {};
		\node [style=none] (214) at (19.25, -5.5) {};
		\node [style=none] (215) at (15, -4.75) {(6)};
		\node [style=none] (216) at (23, -11.25) {};
		\node [style=none] (217) at (14.5, -11.25) {};
		\node [style=none] (218) at (18.75, -10.5) {(7)};
		\node [style=none] (219) at (-11.75, 13) {(1)};
	\end{pgfonlayer}
	\begin{pgfonlayer}{edgelayer}
		\draw [style=blue edge] (0) to (1);
		\draw [style=blue edge] (0) to (2);
		\draw [style=dashed-line] (2) to (4);
		\draw [style=dashed-line] (1) to (3);
		\draw [style=blue edge] (4) to (6);
		\draw [style=blue edge] (6) to (7);
		\draw [style=blue edge] (3) to (5);
		\draw [style=dashed-line] (7) to (5);
		\draw [style=blue edge] (9) to (2);
		\draw [style=blue edge] (9) to (8);
		\draw [style=blue edge] (8) to (11);
		\draw [style=blue edge] (8) to (10);
		\draw [style=blue edge] (10) to (3);
		\draw [style=blue edge] (11) to (7);
		\draw [style=graph arrow] (12) to (13);
		\draw [style=red dashed arrow] (13) to (15);
		\draw [style=graph arrow] (15) to (17);
		\draw [style=graph arrow] (21) to (14);
		\draw [style=graph arrow] (20) to (23);
		\draw [style=graph arrow] (20) to (22);
		\draw [style=graph arrow] (22) to (15);
		\draw [style=graph arrow] (23) to (19);
		\draw [style=red dashed arrow] (16) to (14);
		\draw [style=red dashed arrow] (17) to (19);
		\draw [style=graph arrow] (14) to (12);
		\draw [style=graph arrow] (20) to (21);
		\draw [style=graph arrow] (19) to (18);
		\draw [style=graph arrow] (18) to (16);
		\draw [style=graph arrow] (24) to (25);
		\draw [style=red dashed arrow] (25) to (27);
		\draw [style=graph arrow] (27) to (29);
		\draw [style=graph arrow] (33) to (26);
		\draw [style=graph arrow] (32) to (35);
		\draw [style=graph arrow] (32) to (34);
		\draw [style=graph arrow] (34) to (27);
		\draw [style=graph arrow] (35) to (31);
		\draw [style=red dashed arrow] (28) to (26);
		\draw [style=graph arrow] (26) to (24);
		\draw [style=graph arrow] (32) to (33);
		\draw [style=graph arrow] (31) to (30);
		\draw [style=graph arrow] (30) to (28);
		\draw (37) to (38);
		\draw (40) to (39.center);
		\draw (37) to (42.center);
		\draw (38) to (43.center);
		\draw (61) to (62);
		\draw (63) to (64);
		\draw (65) to (66);
		\draw (67) to (68);
		\draw (69) to (70);
		\draw (37) to (61);
		\draw (61) to (67);
		\draw (67) to (69);
		\draw (38) to (65);
		\draw (44.center) to (62);
		\draw (62) to (63);
		\draw (45.center) to (64);
		\draw (46.center) to (66);
		\draw (47.center) to (68);
		\draw (50.center) to (70);
		\draw (71) to (72);
		\draw (73) to (74);
		\draw (75) to (76);
		\draw (77) to (78);
		\draw (79) to (80);
		\draw (81) to (82);
		\draw (83) to (84);
		\draw (86) to (85.center);
		\draw (89) to (88.center);
		\draw (51.center) to (74);
		\draw (74) to (76);
		\draw (52.center) to (80);
		\draw (80) to (82);
		\draw (82) to (39.center);
		\draw (48.center) to (72);
		\draw (49.center) to (78);
		\draw (78) to (60.center);
		\draw (59.center) to (72);
		\draw (58.center) to (68);
		\draw (57.center) to (66);
		\draw (63) to (75);
		\draw (75) to (77);
		\draw (77) to (79);
		\draw (79) to (55.center);
		\draw (64) to (56.center);
		\draw (65) to (81);
		\draw (81) to (83);
		\draw (83) to (54.center);
		\draw (69) to (71);
		\draw (71) to (73);
		\draw (73) to (53.center);
		\draw (70) to (84);
		\draw (84) to (88.center);
		\draw (76) to (85.center);
		\draw [style=graph arrow] (91) to (92);
		\draw [style=red dashed arrow] (92) to (94);
		\draw [style=graph arrow] (94) to (96);
		\draw [style=graph arrow] (100) to (93);
		\draw [style=graph arrow] (99) to (102);
		\draw [style=graph arrow] (99) to (101);
		\draw [style=graph arrow] (101) to (94);
		\draw [style=graph arrow] (102) to (98);
		\draw [style=red dashed arrow] (95) to (93);
		\draw [style=graph arrow] (93) to (91);
		\draw [style=graph arrow] (99) to (100);
		\draw [style=graph arrow] (98) to (97);
		\draw [style=graph arrow] (97) to (95);
		\draw [style=lightdots] (104.center) to (103.center);
		\draw [style=lightdots] (106.center) to (105.center);
		\draw [style=lightdots] (108.center) to (107.center);
		\draw [style=lightdots] (110.center) to (109.center);
		\draw [style=lightdots] (112.center) to (111.center);
		\draw [style=lightdots] (114.center) to (113.center);
		\draw [style=lightdots] (116.center) to (115.center);
		\draw [style=lightdots] (118.center) to (117.center);
		\draw [style=lightdots] (120.center) to (119.center);
		\draw [style=graph line] (121) to (122);
		\draw [style=graph line] (124) to (126);
		\draw [style=graph line] (130) to (123);
		\draw [style=graph line] (129) to (132);
		\draw [style=graph line] (129) to (131);
		\draw [style=graph line] (131) to (124);
		\draw [style=graph line] (132) to (128);
		\draw [style=graph line] (123) to (121);
		\draw [style=graph line] (129) to (130);
		\draw [style=graph line] (128) to (127);
		\draw [style=graph line] (127) to (125);
		\draw [style=graph line] (130) to (133.center);
		\draw [style=graph line] (131) to (134.center);
		\draw [style=graph line] (132) to (135.center);
		\draw [style=graph line] (127) to (136.center);
		\draw [style=graph line] (121) to (137.center);
		\draw [style=graph line] (122) to (138.center);
		\draw [style=graph line] (139.center) to (126);
		\draw [style=graph line] (141.center) to (126);
		\draw [style=graph line] (141.center) to (140.center);
		\draw [style=graph line] (140.center) to (128);
		\draw [style=graph line] (142.center) to (125);
		\draw [style=graph line] (144.center) to (123);
		\draw [style=graph line] (144.center) to (143.center);
		\draw [style=graph line] (143.center) to (125);
		\draw [style=graph line] (146.center) to (124);
		\draw [style=graph line] (146.center) to (145.center);
		\draw [style=graph line] (145.center) to (122);
		\draw [style=arrow] (147.center) to (148.center);
		\draw [style=arrow] (150.center) to (151.center);
		\draw [style=arrow] (153.center) to (154.center);
		\draw [style=arrow] (157.center) to (158.center);
		\draw [style=arrow] (161.center) to (162.center);
		\draw (159.center) to (160.center);
		\draw (160.center) to (161.center);
		\draw [style=graph line] (165) to (166);
		\draw [style=graph line] (168) to (170);
		\draw [style=graph line] (174) to (167);
		\draw [style=graph line] (175) to (168);
		\draw [style=graph line] (176) to (172);
		\draw [style=graph line] (167) to (165);
		\draw [style=graph line] (172) to (171);
		\draw [style=graph line] (171) to (169);
		\draw [style=graph line] (166) to (182.center);
		\draw [style=graph line] (183.center) to (170);
		\draw [style=graph line] (185) to (170);
		\draw [style=graph line] (185) to (184);
		\draw [style=graph line] (184) to (172);
		\draw [style=graph line] (186.center) to (169);
		\draw [style=graph line] (188) to (167);
		\draw [style=graph line] (188) to (187);
		\draw [style=graph line] (187) to (169);
		\draw [style=graph line] (190) to (168);
		\draw [style=graph line] (190) to (189);
		\draw [style=graph line] (189) to (166);
		\draw [style=graph line] (192) to (191);
		\draw [style=graph line] (173) to (192);
		\draw [style=graph line] (192) to (193);
		\draw [style=graph line] (193) to (176);
		\draw [style=graph line] (196) to (194);
		\draw [style=graph line] (194) to (174);
		\draw [style=graph line] (195) to (191);
		\draw [style=graph line] (191) to (175);
		\draw [style=graph line] (193) to (194);
		\draw [style=graph line] (198) to (197);
		\draw [style=graph line] (197) to (177.center);
		\draw [style=graph line] (174) to (197);
		\draw [style=graph line] (178.center) to (200);
		\draw [style=graph line] (200) to (175);
		\draw [style=graph line] (200) to (199);
		\draw [style=graph line] (179.center) to (201);
		\draw [style=graph line] (201) to (176);
		\draw [style=graph line] (202) to (201);
		\draw [style=graph line] (203) to (184);
		\draw [style=graph line] (204) to (185);
		\draw [style=graph line] (206) to (205);
		\draw [style=graph line] (171) to (205);
		\draw [style=graph line] (205) to (180.center);
		\draw [style=graph line] (207) to (165);
		\draw [style=graph line] (181.center) to (207);
		\draw [style=graph line] (208) to (207);
		\draw [style=graph line] (209) to (187);
		\draw [style=graph line] (210) to (188);
		\draw [style=graph line] (211) to (189);
		\draw [style=graph line] (212) to (190);
		\draw [style=arrow] (213.center) to (214.center);
		\draw [style=arrow] (216.center) to (217.center);
	\end{pgfonlayer}
\end{tikzpicture}

%% file: figures/building-G-and-D.tikz
\begin{tikzpicture}
	\begin{pgfonlayer}{nodelayer}
		\node [style=none] (189) at (0, 0) {};
		\node [style=none] (190) at (2.5, 0) {};
		\node [style=none, align=center] (191) at (1.25, 1.25) {Add \\ roots};
		\node [style=Z bold phase dot] (265) at (-5.5, 6) {0};
		\node [style=Z bold phase dot] (266) at (-3.5, 5) {2};
		\node [style=Z bold phase dot] (267) at (-4.5, 4) {3};
		\node [style=Z bold phase dot] (268) at (-6.5, 5) {1};
		\node [style=Z bold phase dot] (269) at (-5, -2) {11};
		\node [style=Z bold phase dot] (270) at (-5, -0.5) {9};
		\node [style=Z bold phase dot] (271) at (-6.5, -0.5) {8};
		\node [style=Z bold phase dot] (272) at (-6.5, -2) {10};
		\node [style=Z bold phase dot] (273) at (-5.5, -6) {16};
		\node [style=Z bold phase dot] (274) at (-3.75, -5) {15};
		\node [style=Z bold phase dot] (275) at (-4.75, -4) {13};
		\node [style=Z bold phase dot] (276) at (-6.5, -5) {14};
		\node [style=X bold phase dot] (277) at (-2.5, 4) {4};
		\node [style=X bold phase dot] (278) at (-2.5, 0) {7};
		\node [style=X bold phase dot] (279) at (-2.5, -4) {12};
		\node [style=X bold phase dot] (280) at (-2.5, 2) {6};
		\node [style=X bold phase dot] (281) at (-4.5, 2) {5};
		\node [style=none] (282) at (-1.5, 6) {};
		\node [style=none] (283) at (-1.5, 4) {};
		\node [style=none] (284) at (-1.5, 2) {};
		\node [style=none] (285) at (-1.5, 0) {};
		\node [style=none] (286) at (-1.5, -2) {};
		\node [style=none] (287) at (-1.5, -6) {};
		\node [style=none] (288) at (-1.5, -4) {};
		\node [style=Z bold phase dot] (289) at (5, 6) {0};
		\node [style=Z bold phase dot] (290) at (6.5, 5) {2};
		\node [style=Z bold phase dot] (291) at (5.5, 4) {3};
		\node [style=Z bold phase dot] (292) at (4, 5) {1};
		\node [style=Z bold phase dot] (293) at (6.5, -2) {11};
		\node [style=Z bold phase dot] (294) at (6.5, -0.5) {9};
		\node [style=Z bold phase dot] (295) at (5, -0.5) {8};
		\node [style=Z bold phase dot] (296) at (5, -2) {10};
		\node [style=Z bold phase dot] (297) at (5, -6) {16};
		\node [style=Z bold phase dot] (298) at (7, -5) {15};
		\node [style=Z bold phase dot] (299) at (6, -4) {13};
		\node [style=Z bold phase dot] (300) at (4, -5) {14};
		\node [style=X bold phase dot] (301) at (9, 4) {4};
		\node [style=X bold phase dot] (302) at (9, 0) {7};
		\node [style=X bold phase dot] (303) at (9, -4) {12};
		\node [style=X bold phase dot] (304) at (9, 2) {6};
		\node [style=X bold phase dot] (305) at (7.5, 2) {5};
		\node [style=none] (306) at (10, 6) {};
		\node [style=none] (307) at (10, 4) {};
		\node [style=none] (308) at (10, 2) {};
		\node [style=none] (309) at (10, 0) {};
		\node [style=none] (310) at (10, -2) {};
		\node [style=none] (311) at (10, -6) {};
		\node [style=none] (312) at (10, -4) {};
		\node [style=Z bold phase dot] (313) at (4, 3.75) {$r_0$};
		\node [style=Z bold phase dot] (314) at (4, -1.25) {$r_4$};
		\node [style=Z bold phase dot] (315) at (4, -3.25) {$r_6$};
		\node [style=X bold phase dot] (316) at (7.5, 4) {$r_1$};
		\node [style=X bold phase dot] (317) at (7.5, 0) {$r_3$};
		\node [style=X bold phase dot] (318) at (7.5, -4) {$r_5$};
		\node [style=X bold phase dot] (319) at (6, 2) {$r_2$};
		\node [style=Z bold phase dot] (337) at (16.75, 6) {0};
		\node [style=Z bold phase dot] (338) at (18.25, 5) {2};
		\node [style=Z bold phase dot] (339) at (17.25, 4) {3};
		\node [style=Z bold phase dot] (340) at (15.75, 5) {1};
		\node [style=Z bold phase dot] (341) at (17, -2) {11};
		\node [style=Z bold phase dot] (342) at (17, -0.5) {9};
		\node [style=Z bold phase dot] (343) at (15.5, -0.5) {8};
		\node [style=Z bold phase dot] (344) at (15.5, -2) {10};
		\node [style=Z bold phase dot] (345) at (16.75, -6) {16};
		\node [style=Z bold phase dot] (346) at (18.75, -5) {15};
		\node [style=Z bold phase dot] (347) at (17.75, -4) {13};
		\node [style=Z bold phase dot] (348) at (15.75, -5) {14};
		\node [style=X bold phase dot] (349) at (20.75, 4) {4};
		\node [style=X bold phase dot] (350) at (20.75, 0) {7};
		\node [style=X bold phase dot] (351) at (20.75, -4) {12};
		\node [style=X bold phase dot] (352) at (20.75, 2) {6};
		\node [style=X bold phase dot] (353) at (19.25, 2) {5};
		\node [style=none] (354) at (21.75, 6) {};
		\node [style=none] (355) at (21.75, 4) {};
		\node [style=none] (356) at (21.75, 2) {};
		\node [style=none] (357) at (21.75, 0) {};
		\node [style=none] (358) at (21.75, -2) {};
		\node [style=none] (359) at (21.75, -6) {};
		\node [style=none] (360) at (21.75, -4) {};
		\node [style=Z bold phase dot] (361) at (15.75, 3.75) {$r_0$};
		\node [style=Z bold phase dot] (362) at (14.5, -1.25) {$r_4$};
		\node [style=Z bold phase dot] (363) at (15.75, -3.25) {$r_6$};
		\node [style=X bold phase dot] (364) at (19.25, 4) {$r_1$};
		\node [style=X bold phase dot] (365) at (19.25, 0) {$r_3$};
		\node [style=X bold phase dot] (366) at (19.25, -4) {$r_5$};
		\node [style=X bold phase dot] (367) at (17.75, 2) {$r_2$};
		\node [style=none] (368) at (11.25, 0) {};
		\node [style=none] (369) at (13.25, 0) {};
		\node [style=none] (370) at (12.25, 0.75) {$(Z_3, X_3)$};
		\node [style=Z bold phase dot] (371) at (10, -10.75) {0};
		\node [style=Z bold phase dot] (372) at (8, -9) {2};
		\node [style=Z bold phase dot] (373) at (6, -9) {3};
		\node [style=Z bold phase dot] (374) at (6, -10.75) {1};
		\node [style=Z bold phase dot] (375) at (10, -17.75) {11};
		\node [style=Z bold phase dot] (376) at (8, -19.5) {9};
		\node [style=Z bold phase dot] (377) at (6, -19.5) {8};
		\node [style=Z bold phase dot] (378) at (6, -17.75) {10};
		\node [style=Z bold phase dot] (379) at (10, -23) {16};
		\node [style=Z bold phase dot] (380) at (8, -24.75) {15};
		\node [style=Z bold phase dot] (381) at (6, -24.75) {13};
		\node [style=Z bold phase dot] (382) at (6, -23) {14};
		\node [style=Z bold phase dot] (383) at (4, -10.75) {$r_0$};
		\node [style=Z bold phase dot] (384) at (4, -17.75) {$r_4$};
		\node [style=Z bold phase dot] (385) at (4, -23) {$r_6$};
		\node [style=X bold phase dot] (386) at (6, -12.5) {4};
		\node [style=X bold phase dot] (387) at (6, -16) {7};
		\node [style=X bold phase dot] (388) at (6, -21.25) {12};
		\node [style=X bold phase dot] (389) at (8, -14.25) {6};
		\node [style=X bold phase dot] (390) at (6, -14.25) {5};
		\node [style=X bold phase dot] (391) at (4, -12.5) {$r_1$};
		\node [style=X bold phase dot] (392) at (4, -16) {$r_3$};
		\node [style=X bold phase dot] (393) at (4, -21.25) {$r_5$};
		\node [style=X bold phase dot] (394) at (4, -14.25) {$r_2$};
		\node [style=Z bold phase dot] (395) at (21, -10.75) {0};
		\node [style=Z bold phase dot] (396) at (19, -9) {2};
		\node [style=Z bold phase dot] (397) at (17, -9) {3};
		\node [style=Z bold phase dot] (398) at (17, -10.75) {1};
		\node [style=Z bold phase dot] (399) at (21, -17.75) {11};
		\node [style=Z bold phase dot] (400) at (19, -19.5) {9};
		\node [style=Z bold phase dot] (401) at (17, -19.5) {8};
		\node [style=Z bold phase dot] (402) at (17, -17.75) {10};
		\node [style=Z bold phase dot] (403) at (21, -23) {16};
		\node [style=Z bold phase dot] (404) at (19, -24.75) {15};
		\node [style=Z bold phase dot] (405) at (17, -24.75) {13};
		\node [style=Z bold phase dot] (406) at (17, -23) {14};
		\node [style=Z bold phase dot] (407) at (15, -10.75) {$r_0$};
		\node [style=Z bold phase dot] (408) at (15, -17.75) {$r_4$};
		\node [style=Z bold phase dot] (409) at (15, -23) {$r_6$};
		\node [style=X bold phase dot] (410) at (17, -12.5) {4};
		\node [style=X bold phase dot] (411) at (19, -16) {7};
		\node [style=X bold phase dot] (412) at (17, -21.25) {12};
		\node [style=X bold phase dot] (413) at (19, -14.25) {6};
		\node [style=X bold phase dot] (414) at (17, -14.25) {5};
		\node [style=X bold phase dot] (415) at (15, -12.5) {$r_1$};
		\node [style=X bold phase dot] (416) at (17, -16) {$r_3$};
		\node [style=X bold phase dot] (417) at (15, -21.25) {$r_5$};
		\node [style=X bold phase dot] (418) at (15, -14.25) {$r_2$};
		\node [style=Z bold phase dot] (419) at (32.75, -10.75) {0};
		\node [style=Z bold phase dot] (420) at (30.75, -9) {2};
		\node [style=Z bold phase dot] (421) at (28.75, -9) {3};
		\node [style=Z bold phase dot] (422) at (30.75, -10.75) {1};
		\node [style=Z bold phase dot] (423) at (32.75, -17.75) {11};
		\node [style=Z bold phase dot] (424) at (30.75, -19.5) {9};
		\node [style=Z bold phase dot] (425) at (28.75, -19.5) {8};
		\node [style=Z bold phase dot] (426) at (28.75, -17.75) {10};
		\node [style=Z bold phase dot] (427) at (32.75, -23) {16};
		\node [style=Z bold phase dot] (428) at (30.75, -24.75) {15};
		\node [style=Z bold phase dot] (429) at (28.75, -24.75) {13};
		\node [style=Z bold phase dot] (430) at (28.75, -23) {14};
		\node [style=Z bold phase dot] (431) at (26.75, -10.75) {$r_0$};
		\node [style=Z bold phase dot] (432) at (26.75, -17.75) {$r_4$};
		\node [style=Z bold phase dot] (433) at (26.75, -23) {$r_6$};
		\node [style=X bold phase dot] (434) at (28.75, -12.5) {4};
		\node [style=X bold phase dot] (435) at (30.75, -16) {7};
		\node [style=X bold phase dot] (436) at (28.75, -21.25) {12};
		\node [style=X bold phase dot] (437) at (30.75, -14.25) {6};
		\node [style=X bold phase dot] (438) at (28.75, -14.25) {5};
		\node [style=X bold phase dot] (439) at (26.75, -12.5) {$r_1$};
		\node [style=X bold phase dot] (440) at (28.75, -16) {$r_3$};
		\node [style=X bold phase dot] (441) at (26.75, -21.25) {$r_5$};
		\node [style=X bold phase dot] (442) at (26.75, -14.25) {$r_2$};
		\node [style=Z bold phase dot] (443) at (49.75, -10.75) {0};
		\node [style=Z bold phase dot] (444) at (47.75, -9) {2};
		\node [style=Z bold phase dot] (445) at (45.75, -9) {3};
		\node [style=Z bold phase dot] (446) at (43.75, -10.75) {1};
		\node [style=Z bold phase dot] (447) at (45.75, -17.75) {11};
		\node [style=Z bold phase dot] (448) at (43.75, -19.5) {9};
		\node [style=Z bold phase dot] (449) at (41.75, -19.5) {8};
		\node [style=Z bold phase dot] (450) at (41.75, -17.75) {10};
		\node [style=Z bold phase dot] (451) at (49.75, -23) {16};
		\node [style=Z bold phase dot] (452) at (43.75, -24.75) {15};
		\node [style=Z bold phase dot] (453) at (41.75, -24.75) {13};
		\node [style=Z bold phase dot] (454) at (47.75, -23) {14};
		\node [style=Z bold phase dot] (455) at (39.75, -10.75) {$r_0$};
		\node [style=Z bold phase dot] (456) at (39.75, -17.75) {$r_4$};
		\node [style=Z bold phase dot] (457) at (39.75, -23) {$r_6$};
		\node [style=X bold phase dot] (458) at (47.75, -12.5) {4};
		\node [style=X bold phase dot] (459) at (43.75, -16) {7};
		\node [style=X bold phase dot] (460) at (43.75, -21.25) {12};
		\node [style=X bold phase dot] (461) at (47.75, -14.25) {6};
		\node [style=X bold phase dot] (462) at (45.75, -14.25) {5};
		\node [style=X bold phase dot] (463) at (41.75, -12.5) {$r_1$};
		\node [style=X bold phase dot] (464) at (41.75, -16) {$r_3$};
		\node [style=X bold phase dot] (465) at (41.75, -21.25) {$r_5$};
		\node [style=X bold phase dot] (466) at (43.75, -14.25) {$r_2$};
		\node [style=none] (467) at (36.25, -16.25) {\LARGE$\ldots$};
		\node [style=none] (468) at (11.25, -17) {};
		\node [style=none] (469) at (13.25, -17) {};
		\node [style=none] (471) at (23, -17) {};
		\node [style=none] (472) at (25, -17) {};
		\node [style=none] (474) at (34.25, -17) {};
		\node [style=none] (475) at (38.75, -17) {};
		\node [style=Z bold phase dot] (476) at (28, 6) {0};
		\node [style=Z bold phase dot] (477) at (29.75, 5) {2};
		\node [style=Z bold phase dot] (478) at (28.75, 4) {3};
		\node [style=Z bold phase dot] (479) at (26.75, 5) {1};
		\node [style=Z bold phase dot] (480) at (28.5, -2) {11};
		\node [style=Z bold phase dot] (481) at (28.5, -0.5) {9};
		\node [style=Z bold phase dot] (482) at (27.25, -0.5) {8};
		\node [style=Z bold phase dot] (483) at (27.25, -2) {10};
		\node [style=Z bold phase dot] (484) at (28, -6) {16};
		\node [style=Z bold phase dot] (485) at (30, -5) {15};
		\node [style=Z bold phase dot] (486) at (29, -4) {13};
		\node [style=Z bold phase dot] (487) at (27, -5) {14};
		\node [style=X bold phase dot] (488) at (32, 4) {4};
		\node [style=X bold phase dot] (489) at (32, 0) {7};
		\node [style=X bold phase dot] (490) at (32, -4) {12};
		\node [style=X bold phase dot] (491) at (32, 2) {6};
		\node [style=X bold phase dot] (492) at (30.5, 2) {5};
		\node [style=none] (493) at (33.25, 6) {};
		\node [style=none] (494) at (33.25, 4) {};
		\node [style=none] (495) at (33.25, 2) {};
		\node [style=none] (496) at (33.25, 0) {};
		\node [style=none] (497) at (33.25, -2) {};
		\node [style=none] (498) at (33.25, -6) {};
		\node [style=none] (499) at (33.25, -4) {};
		\node [style=Z bold phase dot] (500) at (27, 3.75) {$r_0$};
		\node [style=Z bold phase dot] (501) at (26.25, -1.25) {$r_4$};
		\node [style=Z bold phase dot] (502) at (27, -3.25) {$r_6$};
		\node [style=X bold phase dot] (503) at (30.5, 4) {$r_1$};
		\node [style=X bold phase dot] (504) at (30.5, 0) {$r_3$};
		\node [style=X bold phase dot] (505) at (30.5, -4) {$r_5$};
		\node [style=X bold phase dot] (506) at (28.75, 2) {$r_2$};
		\node [style=none] (507) at (23, 0) {};
		\node [style=none] (508) at (25, 0) {};
		\node [style=none] (509) at (24, 0.75) {$(Z_1, X_3)$};
		\node [style=Z bold phase dot] (510) at (-5.5, -11.25) {$Z_1$};
		\node [style=Z bold phase dot] (511) at (-5.5, -19.25) {$Z_2$};
		\node [style=Z bold phase dot] (512) at (-5.5, -23.25) {$Z_3$};
		\node [style=X bold phase dot] (513) at (-2.25, -13.25) {$X_{\!1}$};
		\node [style=X bold phase dot] (514) at (-2.25, -15.25) {$X_{\!2}$};
		\node [style=X bold phase dot] (515) at (-2.25, -17.25) {$X_{\!3}$};
		\node [style=X bold phase dot] (516) at (-2.25, -21.25) {$X_{\!4}$};
		\node [style=none] (517) at (-0.75, -11.25) {};
		\node [style=none] (518) at (-0.75, -13.25) {};
		\node [style=none] (519) at (-0.75, -15.25) {};
		\node [style=none] (520) at (-0.75, -17.25) {};
		\node [style=none] (521) at (-0.75, -19.25) {};
		\node [style=none] (522) at (-0.75, -21.25) {};
		\node [style=none] (523) at (-0.75, -23.25) {};
		\node [style=none] (524) at (-3.25, -9.25) {Target state: $M_{\text{best}}$};
		\node [style=none] (525) at (-6.75, -10) {};
		\node [style=none] (526) at (-6.75, -24.5) {};
		\node [style=none] (527) at (0.25, -24.5) {};
		\node [style=none] (528) at (0.25, -10) {};
		\node [style=Z bold phase dot] (529) at (42, 6) {0};
		\node [style=Z bold phase dot] (530) at (44.25, 5) {2};
		\node [style=Z bold phase dot] (531) at (43.25, 4) {3};
		\node [style=Z bold phase dot] (532) at (41, 5) {1};
		\node [style=Z bold phase dot] (533) at (42.5, -2) {11};
		\node [style=Z bold phase dot] (534) at (42.5, -0.5) {9};
		\node [style=Z bold phase dot] (535) at (41, -0.5) {8};
		\node [style=Z bold phase dot] (536) at (41, -2) {10};
		\node [style=Z bold phase dot] (537) at (42.25, -6) {16};
		\node [style=Z bold phase dot] (538) at (44.25, -5) {15};
		\node [style=Z bold phase dot] (539) at (43.25, -4) {13};
		\node [style=Z bold phase dot] (540) at (41.25, -5) {14};
		\node [style=X bold phase dot] (541) at (46.25, 4) {4};
		\node [style=X bold phase dot] (542) at (46.25, 0) {7};
		\node [style=X bold phase dot] (543) at (46.25, -4) {12};
		\node [style=X bold phase dot] (544) at (46.25, 2) {6};
		\node [style=X bold phase dot] (545) at (42.5, 2) {5};
		\node [style=none] (546) at (47.75, 6) {};
		\node [style=none] (547) at (47.75, 4) {};
		\node [style=none] (548) at (47.75, 2) {};
		\node [style=none] (549) at (47.75, 0) {};
		\node [style=none] (550) at (47.75, -2) {};
		\node [style=none] (551) at (47.75, -6) {};
		\node [style=none] (552) at (47.75, -4) {};
		\node [style=Z bold phase dot] (553) at (41.25, 3.75) {$r_0$};
		\node [style=Z bold phase dot] (554) at (40, -1.25) {$r_4$};
		\node [style=Z bold phase dot] (555) at (41.25, -3.25) {$r_6$};
		\node [style=X bold phase dot] (556) at (45, 4) {$r_1$};
		\node [style=X bold phase dot] (557) at (44.25, 0) {$r_3$};
		\node [style=X bold phase dot] (558) at (45, -4) {$r_5$};
		\node [style=X bold phase dot] (559) at (41.25, 2) {$r_2$};
		\node [style=none] (560) at (36.5, 0.75) {\LARGE$\ldots$};
		\node [style=none] (561) at (34.25, 0) {};
		\node [style=none] (562) at (38.75, 0) {};
		\node [style=none] (563) at (12.25, -0.75) {$(13, r_3)$};
		\node [style=none] (564) at (24, -0.75) {$(1, 7)$};
	\end{pgfonlayer}
	\begin{pgfonlayer}{edgelayer}
		\draw [style=arrow, in=180, out=0] (189.center) to (190.center);
		\draw [style=graph line] (271) to (270);
		\draw [style=graph line] (269) to (270);
		\draw [style=graph line] (286.center) to (269);
		\draw [style=graph line] (269) to (272);
		\draw [style=graph line] (279) to (288.center);
		\draw [style=graph line] (287.center) to (273);
		\draw [style=graph line] (273) to (274);
		\draw [style=graph line] (274) to (275);
		\draw [style=graph line] (276) to (273);
		\draw [style=graph line] (278) to (285.center);
		\draw [style=graph line] (284.center) to (280);
		\draw [style=graph line] (280) to (281);
		\draw [style=graph line] (277) to (283.center);
		\draw [style=graph line] (282.center) to (265);
		\draw [style=graph line] (265) to (268);
		\draw [style=graph line] (267) to (266);
		\draw [style=graph line] (266) to (265);
		\draw [style=graph line] (268) to (267);
		\draw [style=graph line] (276) to (275);
		\draw [style=graph line] (272) to (271);
		\draw [style=graph line] (295) to (294);
		\draw [style=graph line] (293) to (294);
		\draw [style=graph line] (310.center) to (293);
		\draw [style=graph line] (293) to (296);
		\draw [style=graph line] (303) to (312.center);
		\draw [style=graph line] (311.center) to (297);
		\draw [style=graph line] (297) to (298);
		\draw [style=graph line] (298) to (299);
		\draw [style=graph line] (300) to (297);
		\draw [style=graph line] (302) to (309.center);
		\draw [style=graph line] (308.center) to (304);
		\draw [style=graph line] (304) to (305);
		\draw [style=graph line] (301) to (307.center);
		\draw [style=graph line] (306.center) to (289);
		\draw [style=graph line] (289) to (292);
		\draw [style=graph line] (291) to (290);
		\draw [style=graph line] (290) to (289);
		\draw [style=graph line] (313) to (291);
		\draw [style=graph line] (313) to (292);
		\draw [style=graph line] (319) to (305);
		\draw [style=graph line] (317) to (302);
		\draw [style=graph line] (314) to (296);
		\draw [style=graph line] (314) to (295);
		\draw [style=graph line] (318) to (303);
		\draw [style=graph line] (301) to (316);
		\draw [style=graph line] (315) to (300);
		\draw [style=graph line] (315) to (299);
		\draw [style=graph line] (343) to (342);
		\draw [style=graph line] (341) to (342);
		\draw [style=graph line] (358.center) to (341);
		\draw [style=graph line] (341) to (344);
		\draw [style=graph line] (351) to (360.center);
		\draw [style=graph line] (359.center) to (345);
		\draw [style=graph line] (345) to (346);
		\draw [style=graph line] (346) to (347);
		\draw [style=graph line] (348) to (345);
		\draw [style=graph line] (350) to (357.center);
		\draw [style=graph line] (356.center) to (352);
		\draw [style=graph line] (352) to (353);
		\draw [style=graph line] (349) to (355.center);
		\draw [style=graph line] (354.center) to (337);
		\draw [style=graph line] (337) to (340);
		\draw [style=graph line] (339) to (338);
		\draw [style=graph line] (338) to (337);
		\draw [style=graph line] (361) to (339);
		\draw [style=graph line] (361) to (340);
		\draw [style=graph line] (367) to (353);
		\draw [style=graph line] (365) to (350);
		\draw [style=graph line] (362) to (344);
		\draw [style=graph line] (362) to (343);
		\draw [style=graph line] (366) to (351);
		\draw [style=graph line] (349) to (364);
		\draw [style=graph line] (363) to (348);
		\draw [style=graph line] (363) to (347);
		\draw [style=arrow, in=180, out=0] (368.center) to (369.center);
		\draw [style=graph arrow] (374) to (371);
		\draw [style=red dashed arrow] (372) to (371);
		\draw [style=graph arrow] (373) to (372);
		\draw [style=graph arrow] (378) to (375);
		\draw [style=red dashed arrow] (376) to (375);
		\draw [style=graph arrow] (377) to (376);
		\draw [style=graph arrow] (382) to (379);
		\draw [style=red dashed arrow] (380) to (379);
		\draw [style=graph arrow] (381) to (380);
		\draw [style=graph arrow] (383) to (373);
		\draw [style=graph arrow] (383) to (374);
		\draw [style=graph arrow] (385) to (382);
		\draw [style=graph arrow] (385) to (381);
		\draw [style=graph arrow] (384) to (377);
		\draw [style=graph arrow] (384) to (378);
		\draw [style=graph arrow] (393) to (388);
		\draw [style=graph arrow] (392) to (387);
		\draw [style=graph arrow] (394) to (390);
		\draw [style=graph arrow] (390) to (389);
		\draw [style=graph arrow] (391) to (386);
		\draw [style=graph arrow] (398) to (395);
		\draw [style=red dashed arrow] (396) to (395);
		\draw [style=graph arrow] (397) to (396);
		\draw [style=graph arrow] (402) to (399);
		\draw [style=red dashed arrow] (400) to (399);
		\draw [style=graph arrow] (401) to (400);
		\draw [style=graph arrow] (406) to (403);
		\draw [style=red dashed arrow] (404) to (403);
		\draw [style=graph arrow] (405) to (404);
		\draw [style=graph arrow] (407) to (397);
		\draw [style=graph arrow] (407) to (398);
		\draw [style=graph arrow] (409) to (406);
		\draw [style=graph arrow] (409) to (405);
		\draw [style=graph arrow] (408) to (401);
		\draw [style=graph arrow] (408) to (402);
		\draw [style=graph arrow] (417) to (412);
		\draw [style=graph arrow] (416) to (411);
		\draw [style=graph arrow] (418) to (414);
		\draw [style=graph arrow] (414) to (413);
		\draw [style=graph arrow] (415) to (410);
		\draw [style=graph arrow] (422) to (419);
		\draw [style=red dashed arrow] (420) to (419);
		\draw [style=graph arrow] (421) to (420);
		\draw [style=graph arrow] (426) to (423);
		\draw [style=red dashed arrow] (424) to (423);
		\draw [style=graph arrow] (425) to (424);
		\draw [style=graph arrow] (430) to (427);
		\draw [style=red dashed arrow] (428) to (427);
		\draw [style=graph arrow] (429) to (428);
		\draw [style=graph arrow] (431) to (421);
		\draw [style=graph arrow] (431) to (422);
		\draw [style=graph arrow] (433) to (430);
		\draw [style=graph arrow] (433) to (429);
		\draw [style=graph arrow] (432) to (425);
		\draw [style=graph arrow] (432) to (426);
		\draw [style=graph arrow] (441) to (436);
		\draw [style=graph arrow] (440) to (435);
		\draw [style=graph arrow] (442) to (438);
		\draw [style=graph arrow] (438) to (437);
		\draw [style=graph arrow] (439) to (434);
		\draw [style=graph arrow] (446) to (443);
		\draw [style=red dashed arrow] (444) to (443);
		\draw [style=graph arrow] (445) to (444);
		\draw [style=graph arrow] (450) to (447);
		\draw [style=red dashed arrow] (448) to (447);
		\draw [style=graph arrow] (449) to (448);
		\draw [style=graph arrow] (454) to (451);
		\draw [style=red dashed arrow] (452) to (451);
		\draw [style=graph arrow] (453) to (452);
		\draw [style=graph arrow] (455) to (445);
		\draw [style=graph arrow] (455) to (446);
		\draw [style=graph arrow] (457) to (454);
		\draw [style=graph arrow] (457) to (453);
		\draw [style=graph arrow] (456) to (449);
		\draw [style=graph arrow] (456) to (450);
		\draw [style=graph arrow] (465) to (460);
		\draw [style=graph arrow] (464) to (459);
		\draw [style=graph arrow] (466) to (462);
		\draw [style=graph arrow] (462) to (461);
		\draw [style=graph arrow] (463) to (458);
		\draw [style={cnot_edge_dir}] (445) to (458);
		\draw [style={cnot_edge_dir}] (463) to (444);
		\draw [style={cnot_edge_dir}] (457) to (464);
		\draw [style={cnot_edge_dir}] (465) to (452);
		\draw [style={cnot_edge_dir}] (453) to (460);
		\draw [style={cnot_edge_dir}] (462) to (454);
		\draw [style={cnot_edge_dir}] (457) to (461);
		\draw [style={cnot_edge_dir}] (455) to (459);
		\draw [style={cnot_edge_dir}] (464) to (446);
		\draw [style={cnot_edge_dir}] (455) to (462);
		\draw [style={cnot_edge_dir}] (466) to (445);
		\draw [style={cnot_edge_dir}] (449) to (466);
		\draw [style={cnot_edge_dir}] (456) to (463);
		\draw [style=arrow, in=180, out=0] (468.center) to (469.center);
		\draw [style=arrow, in=180, out=0] (471.center) to (472.center);
		\draw [style=arrow, in=180, out=0] (474.center) to (475.center);
		\draw [style=graph line] (482) to (481);
		\draw [style=graph line] (480) to (481);
		\draw [style=graph line] (497.center) to (480);
		\draw [style=graph line] (480) to (483);
		\draw [style=graph line] (490) to (499.center);
		\draw [style=graph line] (498.center) to (484);
		\draw [style=graph line] (484) to (485);
		\draw [style=graph line] (485) to (486);
		\draw [style=graph line] (487) to (484);
		\draw [style=graph line] (489) to (496.center);
		\draw [style=graph line] (495.center) to (491);
		\draw [style=graph line] (491) to (492);
		\draw [style=graph line] (488) to (494.center);
		\draw [style=graph line] (493.center) to (476);
		\draw [style=graph line] (476) to (479);
		\draw [style=graph line] (478) to (477);
		\draw [style=graph line] (477) to (476);
		\draw [style=graph line] (500) to (478);
		\draw [style=graph line] (500) to (479);
		\draw [style=graph line] (506) to (492);
		\draw [style=graph line] (504) to (489);
		\draw [style=graph line] (501) to (483);
		\draw [style=graph line] (501) to (482);
		\draw [style=graph line] (505) to (490);
		\draw [style=graph line] (488) to (503);
		\draw [style=graph line] (502) to (487);
		\draw [style=graph line] (502) to (486);
		\draw [style=arrow, in=180, out=0] (507.center) to (508.center);
		\draw [style=graph line, in=150, out=-30] (510) to (513);
		\draw [style=graph line] (510) to (514);
		\draw [style=graph line] (510) to (515);
		\draw [style=graph line] (511) to (513);
		\draw [style=graph line] (511) to (514);
		\draw [style=graph line] (511) to (516);
		\draw [style=graph line] (512) to (514);
		\draw [style=graph line] (512) to (515);
		\draw [style=graph line] (512) to (516);
		\draw [style=graph line] (522.center) to (516);
		\draw [style=graph line] (520.center) to (515);
		\draw [style=graph line] (519.center) to (514);
		\draw [style=graph line] (518.center) to (513);
		\draw [style=graph line] (510) to (517.center);
		\draw [style=graph line, in=180, out=0] (512) to (523.center);
		\draw [style=graph line] (511) to (521.center);
		\draw (525.center) to (526.center);
		\draw (526.center) to (527.center);
		\draw (527.center) to (528.center);
		\draw (528.center) to (525.center);
		\draw [style=graph line] (535) to (534);
		\draw [style=graph line] (533) to (534);
		\draw [style=graph line] (550.center) to (533);
		\draw [style=graph line] (533) to (536);
		\draw [style=graph line] (543) to (552.center);
		\draw [style=graph line] (551.center) to (537);
		\draw [style=graph line] (537) to (538);
		\draw [style=graph line] (538) to (539);
		\draw [style=graph line] (540) to (537);
		\draw [style=graph line] (542) to (549.center);
		\draw [style=graph line] (548.center) to (544);
		\draw [style=graph line] (544) to (545);
		\draw [style=graph line] (541) to (547.center);
		\draw [style=graph line] (546.center) to (529);
		\draw [style=graph line] (529) to (532);
		\draw [style=graph line] (531) to (530);
		\draw [style=graph line] (530) to (529);
		\draw [style=graph line] (532) to (542);
		\draw [style=graph line] (544) to (540);
		\draw [style=graph line] (553) to (531);
		\draw [style=graph line] (553) to (532);
		\draw [style=graph line] (559) to (545);
		\draw [style=graph line] (557) to (542);
		\draw [style=graph line] (557) to (539);
		\draw [style=graph line] (554) to (536);
		\draw [style=graph line] (554) to (535);
		\draw [style=graph line] (558) to (543);
		\draw [style=graph line] (556) to (535);
		\draw [style=graph line] (530) to (541);
		\draw [style=graph line] (541) to (556);
		\draw [style=graph line] (555) to (540);
		\draw [style=graph line] (555) to (539);
		\draw [style=graph line] (545) to (531);
		\draw [style=graph line] (559) to (534);
		\draw [style=arrow, in=180, out=0] (561.center) to (562.center);
		\draw [style=graph line] (486) to (504);
		\draw [style=graph line] (347) to (365);
		\draw [style=graph line] (479) to (489);
		\draw [style={cnot_edge_dir}] (409) to (416);
		\draw [style={cnot_edge_dir}] (431) to (435);
		\draw [style={cnot_edge_dir}] (433) to (440);
		\draw [style={cnot_edge_dir}] (440) to (422);
		\draw [style=graph line] (543) to (538);
		\draw [style=graph line] (558) to (536);
		\draw [style={cnot_edge_dir}] (456) to (465);
	\end{pgfonlayer}
\end{tikzpicture}

%% file: figures/full-extraction.tikz
\begin{tikzpicture}
	\begin{pgfonlayer}{nodelayer}
		\node [style=X bold phase dot] (0) at (-8, 5) {4};
		\node [style=X bold phase dot] (1) at (-12.875, 1) {7};
		\node [style=X bold phase dot] (2) at (-12.5, -5) {12};
		\node [style=none] (3) at (-4.5, 7) {};
		\node [style=none] (4) at (-4.5, 5) {};
		\node [style=none] (5) at (-4.5, 3) {};
		\node [style=none] (6) at (-4.5, 1) {};
		\node [style=none] (7) at (-4.5, -1) {};
		\node [style=none] (8) at (-4.5, -5) {};
		\node [style=none] (9) at (-4.5, -7) {};
		\node [style=Z bold phase dot] (10) at (-12.875, 7) {1};
		\node [style=Z bold phase dot] (11) at (-6, 7) {0};
		\node [style=Z bold phase dot] (12) at (-10, 9) {3};
		\node [style=Z bold phase dot] (13) at (-8, 9) {2};
		\node [style=Z bold phase dot] (14) at (-8, -7) {14};
		\node [style=Z bold phase dot] (15) at (-6, -7) {16};
		\node [style=Z bold phase dot] (16) at (-14.375, -9) {13};
		\node [style=Z bold phase dot] (17) at (-12.5, -9) {15};
		\node [style=Z bold phase dot] (18) at (-15, -1) {10};
		\node [style=Z bold phase dot] (19) at (-10, -1) {11};
		\node [style=Z bold phase dot] (20) at (-12.125, -3) {9};
		\node [style=Z bold phase dot] (21) at (-15.625, -3) {8};
		\node [style=X bold phase dot] (22) at (-8, 3) {6};
		\node [style=X bold phase dot] (23) at (-10, 3) {5};
		\node [style=Z bold phase dot] (24) at (-17.5, 7) {$r_0$};
		\node [style=Z bold phase dot] (25) at (-17.5, -1) {$r_4$};
		\node [style=Z bold phase dot] (26) at (-17.5, -7) {$r_6$};
		\node [style=X bold phase dot] (27) at (-15.625, 5) {$r_1$};
		\node [style=X bold phase dot] (28) at (-14.375, 1) {$r_3$};
		\node [style=X bold phase dot] (29) at (-15, -5) {$r_5$};
		\node [style=X bold phase dot] (30) at (-12.125, 3) {$r_2$};
		\node [style=X bold phase dot] (55) at (12.5, 5) {4};
		\node [style=X bold phase dot] (56) at (7.625, 1) {7};
		\node [style=X bold phase dot] (57) at (8, -5) {12};
		\node [style=none] (58) at (16, 7) {};
		\node [style=none] (59) at (16, 5) {};
		\node [style=none] (60) at (16, 3) {};
		\node [style=none] (61) at (16, 1) {};
		\node [style=none] (62) at (16, -1) {};
		\node [style=none] (63) at (16, -5) {};
		\node [style=none] (64) at (16, -7) {};
		\node [style=Z bold phase dot] (65) at (7.625, 7) {1};
		\node [style=Z bold phase dot] (66) at (14.5, 7) {0};
		\node [style=Z bold phase dot] (67) at (10.5, 9) {3};
		\node [style=Z bold phase dot] (68) at (12.5, 9) {2};
		\node [style=Z bold phase dot] (69) at (12.5, -7) {14};
		\node [style=Z bold phase dot] (70) at (14.5, -7) {16};
		\node [style=Z bold phase dot] (71) at (6.125, -9) {13};
		\node [style=Z bold phase dot] (72) at (8, -9) {15};
		\node [style=Z bold phase dot] (73) at (5.5, -1) {10};
		\node [style=Z bold phase dot] (74) at (10.5, -1) {11};
		\node [style=Z bold phase dot] (75) at (8.375, -3) {9};
		\node [style=Z bold phase dot] (76) at (4.875, -3) {8};
		\node [style=X bold phase dot] (77) at (12.5, 3) {6};
		\node [style=X bold phase dot] (78) at (10.5, 3) {5};
		\node [style=Z bold phase dot] (79) at (3, 7) {$r_0$};
		\node [style=Z bold phase dot] (80) at (3, -1) {$r_4$};
		\node [style=Z bold phase dot] (81) at (3, -7) {$r_6$};
		\node [style=X bold phase dot] (82) at (4.875, 5) {$r_1$};
		\node [style=X bold phase dot] (83) at (6.125, 1) {$r_3$};
		\node [style=X bold phase dot] (84) at (5.5, -5) {$r_5$};
		\node [style=X bold phase dot] (85) at (8.375, 3) {$r_2$};
		\node [style=X bold phase dot] (86) at (3, -9) {};
		\node [style=X bold phase dot] (87) at (1.5, -9) {};
		\node [style=X bold phase dot] (88) at (4, -5) {};
		\node [style=X bold phase dot] (89) at (4, 1) {};
		\node [style=X bold phase dot] (90) at (3.25, 5) {};
		\node [style=X bold phase dot] (91) at (3, 9) {};
		\node [style=X bold phase dot] (92) at (1.5, 9) {};
		\node [style=X bold phase dot] (93) at (3, -3) {};
		\node [style=X bold phase dot] (94) at (1.5, -3) {};
		\node [style=Z bold phase dot] (95) at (1.5, -1) {};
		\node [style=X bold phase dot] (96) at (6.5, 3) {};
		\node [style=Z bold phase dot] (97) at (1.5, 7) {};
		\node [style=Z bold phase dot] (98) at (1.5, -7) {};
		\node [style=X bold phase dot] (99) at (14.5, 9) {};
		\node [style=X bold phase dot] (100) at (16, 9) {};
		\node [style=X bold phase dot] (101) at (10.5, -3) {};
		\node [style=X bold phase dot] (102) at (11.75, -3) {};
		\node [style=X bold phase dot] (103) at (14.5, -9) {};
		\node [style=X bold phase dot] (104) at (16, -9) {};
		\node [style=none] (311) at (-3, 0) {};
		\node [style=none] (312) at (0, 0) {};
		\node [style=not] (314) at (33.5, 5) {};
		\node [style=not] (315) at (28.625, 1) {};
		\node [style=not] (316) at (29, -5) {};
		\node [style=none] (317) at (37, 7) {};
		\node [style=none] (318) at (37, 5) {};
		\node [style=none] (319) at (37, 3) {};
		\node [style=none] (320) at (37, 1) {};
		\node [style=none] (321) at (37, -1) {};
		\node [style=none] (322) at (37, -5) {};
		\node [style=none] (323) at (37, -7) {};
		\node [style=cnot ctrl] (324) at (28.625, 7) {};
		\node [style=cnot ctrl] (325) at (35.5, 7) {};
		\node [style=cnot ctrl] (326) at (31.5, 9) {};
		\node [style=cnot ctrl] (327) at (33.5, 9) {};
		\node [style=cnot ctrl] (328) at (33.5, -7) {};
		\node [style=cnot ctrl] (329) at (35.5, -7) {};
		\node [style=cnot ctrl] (330) at (27.125, -9) {};
		\node [style=cnot ctrl] (331) at (29, -9) {};
		\node [style=cnot ctrl] (332) at (26.5, -1) {};
		\node [style=cnot ctrl] (333) at (31.5, -1) {};
		\node [style=cnot ctrl] (334) at (29.375, -3) {};
		\node [style=cnot ctrl] (335) at (25.875, -3) {};
		\node [style=not] (336) at (33.5, 3) {};
		\node [style=not] (337) at (31.5, 3) {};
		\node [style=cnot ctrl] (338) at (24, 7) {};
		\node [style=cnot ctrl] (339) at (24, -1) {};
		\node [style=cnot ctrl] (340) at (24, -7) {};
		\node [style=not] (341) at (25.875, 5) {};
		\node [style=not] (342) at (27.125, 1) {};
		\node [style=not] (343) at (26.5, -5) {};
		\node [style=not] (344) at (29.375, 3) {};
		\node [style=not] (345) at (24, -9) {};
		\node [style=none] (346) at (22.5, -9) {};
		\node [style=none] (347) at (25, -5) {};
		\node [style=none] (348) at (25, 1) {};
		\node [style=none] (349) at (24.25, 5) {};
		\node [style=not] (350) at (24, 9) {};
		\node [style=none] (351) at (22.5, 9) {};
		\node [style=not] (352) at (24, -3) {};
		\node [style=none] (353) at (22.5, -3) {};
		\node [style=none] (354) at (22.5, -1) {};
		\node [style=none] (355) at (27.5, 3) {};
		\node [style=none] (356) at (22.5, 7) {};
		\node [style=none] (357) at (22.5, -7) {};
		\node [style=not] (358) at (35.5, 9) {};
		\node [style=meter] (359) at (37, 9) {};
		\node [style=not] (360) at (31.5, -3) {};
		\node [style=meter] (361) at (32.75, -3) {};
		\node [style=not] (362) at (35.5, -9) {};
		\node [style=meter] (363) at (37, -9) {};
		\node [style=none] (364) at (17.5, 0) {};
		\node [style=none] (365) at (20, 0) {};
		\node [style=none] (367) at (22, -9) {$\ket{0}$};
		\node [style=none] (368) at (24.5, -5) {$\ket{0}$};
		\node [style=none] (369) at (24.5, 1) {$\ket{0}$};
		\node [style=none] (370) at (23.75, 5) {$\ket{0}$};
		\node [style=none] (371) at (22, 9) {$\ket{0}$};
		\node [style=none] (372) at (22, -3) {$\ket{0}$};
		\node [style=none] (373) at (27, 3) {$\ket{0}$};
		\node [style=none] (374) at (22, -1) {$\ket{+}$};
		\node [style=none] (375) at (22, 7) {$\ket{+}$};
		\node [style=none] (376) at (22, -7) {$\ket{+}$};
		\node [style=Z bold phase dot] (408) at (-29.25, 7) {0};
		\node [style=Z bold phase dot] (409) at (-27, 6) {2};
		\node [style=Z bold phase dot] (410) at (-28, 5) {3};
		\node [style=Z bold phase dot] (411) at (-30.25, 6) {1};
		\node [style=Z bold phase dot] (412) at (-28.75, -1) {11};
		\node [style=Z bold phase dot] (413) at (-28.75, 0.5) {9};
		\node [style=Z bold phase dot] (414) at (-30.25, 0.5) {8};
		\node [style=Z bold phase dot] (415) at (-30.25, -1) {10};
		\node [style=Z bold phase dot] (416) at (-29, -7) {16};
		\node [style=Z bold phase dot] (417) at (-27, -6) {15};
		\node [style=Z bold phase dot] (418) at (-28, -5) {13};
		\node [style=Z bold phase dot] (419) at (-30, -6) {14};
		\node [style=X bold phase dot] (420) at (-25, 5) {4};
		\node [style=X bold phase dot] (421) at (-25, 1) {7};
		\node [style=X bold phase dot] (422) at (-25, -5) {12};
		\node [style=X bold phase dot] (423) at (-25, 3) {6};
		\node [style=X bold phase dot] (424) at (-28.75, 3) {5};
		\node [style=none] (425) at (-23.5, 7) {};
		\node [style=none] (426) at (-23.5, 5) {};
		\node [style=none] (427) at (-23.5, 3) {};
		\node [style=none] (428) at (-23.5, 1) {};
		\node [style=none] (429) at (-23.5, -1) {};
		\node [style=none] (430) at (-23.5, -7) {};
		\node [style=none] (431) at (-23.5, -5) {};
		\node [style=Z bold phase dot] (432) at (-30, 4.75) {$r_0$};
		\node [style=Z bold phase dot] (433) at (-31.25, -0.25) {$r_4$};
		\node [style=Z bold phase dot] (434) at (-30, -4.25) {$r_6$};
		\node [style=X bold phase dot] (435) at (-26.25, 5) {$r_1$};
		\node [style=X bold phase dot] (436) at (-27, 1) {$r_3$};
		\node [style=X bold phase dot] (437) at (-26.25, -5) {$r_5$};
		\node [style=X bold phase dot] (438) at (-30, 3) {$r_2$};
		\node [style=none] (439) at (-22, 0) {};
		\node [style=none] (440) at (-19, 0) {};
		\node [style=none] (441) at (-27, -10.5) {(a)};
		\node [style=none] (442) at (-12, -10.5) {(b)};
		\node [style=none] (443) at (8, -10.5) {(c)};
		\node [style=none] (444) at (29, -10.5) {(d)};
	\end{pgfonlayer}
	\begin{pgfonlayer}{edgelayer}
		\draw [graph line] (8.center) to (2);
		\draw [graph line] (6.center) to (1);
		\draw [graph line] (4.center) to (0);
		\draw [graph line] (12) to (13);
		\draw [graph line] (13) to (11);
		\draw [graph line] (13) to (0);
		\draw [graph line] (10) to (1);
		\draw [graph line] (11) to (3.center);
		\draw [graph line] (15) to (14);
		\draw [graph line] (16) to (17);
		\draw [graph line] (17) to (15);
		\draw [graph line] (15) to (9.center);
		\draw [graph line] (17) to (2);
		\draw [graph line] (18) to (19);
		\draw [graph line] (19) to (20);
		\draw [graph line] (20) to (21);
		\draw [graph line] (19) to (7.center);
		\draw [graph line] (22) to (23);
		\draw [graph line] (22) to (5.center);
		\draw [graph line] (23) to (12);
		\draw [graph line] (22) to (14);
		\draw [style=graph line] (24) to (12);
		\draw [style=graph line] (11) to (10);
		\draw [style=graph line] (10) to (24);
		\draw [style=graph line] (25) to (21);
		\draw [style=graph line] (25) to (18);
		\draw [style=graph line] (26) to (14);
		\draw [style=graph line] (26) to (16);
		\draw [style=graph line] (0) to (27);
		\draw [style=graph line] (27) to (21);
		\draw [style=graph line] (28) to (1);
		\draw [style=graph line] (28) to (16);
		\draw [style=graph line] (29) to (2);
		\draw [style=graph line] (29) to (18);
		\draw [style=graph line] (30) to (23);
		\draw [style=graph line] (20) to (30);
		\draw [graph line] (63.center) to (57);
		\draw [graph line] (61.center) to (56);
		\draw [graph line] (59.center) to (55);
		\draw [graph line] (67) to (68);
		\draw [graph line] (68) to (55);
		\draw [graph line] (65) to (56);
		\draw [graph line] (66) to (58.center);
		\draw [graph line] (70) to (69);
		\draw [graph line] (71) to (72);
		\draw [graph line] (70) to (64.center);
		\draw [graph line] (72) to (57);
		\draw [graph line] (73) to (74);
		\draw [graph line] (75) to (76);
		\draw [graph line] (74) to (62.center);
		\draw [graph line] (77) to (78);
		\draw [graph line] (77) to (60.center);
		\draw [graph line] (78) to (67);
		\draw [graph line] (77) to (69);
		\draw [style=graph line] (66) to (65);
		\draw [style=graph line] (65) to (79);
		\draw [style=graph line] (80) to (73);
		\draw [style=graph line] (81) to (69);
		\draw [style=graph line] (55) to (82);
		\draw [style=graph line] (82) to (76);
		\draw [style=graph line] (83) to (56);
		\draw [style=graph line] (83) to (71);
		\draw [style=graph line] (84) to (57);
		\draw [style=graph line] (84) to (73);
		\draw [style=graph line] (85) to (78);
		\draw [style=graph line] (75) to (85);
		\draw [style=graph line] (71) to (86);
		\draw [style=graph line] (86) to (81);
		\draw [style=graph line] (87) to (86);
		\draw [style=graph line] (88) to (84);
		\draw [style=graph line] (89) to (83);
		\draw [style=graph line] (90) to (82);
		\draw [style=graph line] (91) to (67);
		\draw [style=graph line] (91) to (79);
		\draw [style=graph line] (92) to (91);
		\draw [style=graph line] (93) to (76);
		\draw [style=graph line] (93) to (80);
		\draw [style=graph line] (94) to (93);
		\draw [style=graph line] (95) to (80);
		\draw [style=graph line] (96) to (85);
		\draw [style=graph line] (97) to (79);
		\draw [style=graph line] (98) to (81);
		\draw [style=graph line] (99) to (66);
		\draw [style=graph line] (99) to (68);
		\draw [style=graph line] (100) to (99);
		\draw [style=graph line] (101) to (75);
		\draw [style=graph line] (101) to (74);
		\draw [style=graph line] (102) to (101);
		\draw [style=graph line] (72) to (103);
		\draw [style=graph line] (103) to (70);
		\draw [style=graph line] (104) to (103);
		\draw [style=arrow, in=180, out=0] (311.center) to (312.center);
		\draw [graph line] (322.center) to (316);
		\draw [graph line] (320.center) to (315);
		\draw [graph line] (318.center) to (314);
		\draw [graph line] (326) to (327);
		\draw [graph line] (327) to (314);
		\draw [graph line] (324) to (315);
		\draw [graph line] (325) to (317.center);
		\draw [graph line] (329) to (328);
		\draw [graph line] (330) to (331);
		\draw [graph line] (329) to (323.center);
		\draw [graph line] (331) to (316);
		\draw [graph line] (332) to (333);
		\draw [graph line] (334) to (335);
		\draw [graph line] (333) to (321.center);
		\draw [graph line] (336) to (337);
		\draw [graph line] (336) to (319.center);
		\draw [graph line] (337) to (326);
		\draw [graph line] (336) to (328);
		\draw [style=graph line] (325) to (324);
		\draw [style=graph line] (324) to (338);
		\draw [style=graph line] (339) to (332);
		\draw [style=graph line] (340) to (328);
		\draw [style=graph line] (314) to (341);
		\draw [style=graph line] (341) to (335);
		\draw [style=graph line] (342) to (315);
		\draw [style=graph line] (342) to (330);
		\draw [style=graph line] (343) to (316);
		\draw [style=graph line] (343) to (332);
		\draw [style=graph line] (344) to (337);
		\draw [style=graph line] (334) to (344);
		\draw [style=graph line] (330) to (345);
		\draw [style=graph line] (345) to (340);
		\draw [style=graph line] (346.center) to (345);
		\draw [style=graph line] (347.center) to (343);
		\draw [style=graph line] (348.center) to (342);
		\draw [style=graph line] (349.center) to (341);
		\draw [style=graph line] (350) to (326);
		\draw [style=graph line] (350) to (338);
		\draw [style=graph line] (351.center) to (350);
		\draw [style=graph line] (352) to (335);
		\draw [style=graph line] (352) to (339);
		\draw [style=graph line] (353.center) to (352);
		\draw [style=graph line] (354.center) to (339);
		\draw [style=graph line] (355.center) to (344);
		\draw [style=graph line] (356.center) to (338);
		\draw [style=graph line] (357.center) to (340);
		\draw [style=graph line] (358) to (325);
		\draw [style=graph line] (358) to (327);
		\draw [style=graph line] (359) to (358);
		\draw [style=graph line] (360) to (334);
		\draw [style=graph line] (360) to (333);
		\draw [style=graph line] (361) to (360);
		\draw [style=graph line] (331) to (362);
		\draw [style=graph line] (362) to (329);
		\draw [style=graph line] (363) to (362);
		\draw [style=arrow, in=180, out=0] (364.center) to (365.center);
		\draw [style=graph line] (414) to (413);
		\draw [style=graph line] (412) to (413);
		\draw [style=graph line] (429.center) to (412);
		\draw [style=graph line] (412) to (415);
		\draw [style=graph line] (422) to (431.center);
		\draw [style=graph line] (430.center) to (416);
		\draw [style=graph line] (416) to (417);
		\draw [style=graph line] (417) to (418);
		\draw [style=graph line] (419) to (416);
		\draw [style=graph line] (421) to (428.center);
		\draw [style=graph line] (427.center) to (423);
		\draw [style=graph line] (423) to (424);
		\draw [style=graph line] (420) to (426.center);
		\draw [style=graph line] (425.center) to (408);
		\draw [style=graph line] (408) to (411);
		\draw [style=graph line] (410) to (409);
		\draw [style=graph line] (409) to (408);
		\draw [style=graph line] (411) to (421);
		\draw [style=graph line] (423) to (419);
		\draw [style=graph line] (432) to (410);
		\draw [style=graph line] (432) to (411);
		\draw [style=graph line] (438) to (424);
		\draw [style=graph line] (436) to (421);
		\draw [style=graph line] (436) to (418);
		\draw [style=graph line] (433) to (415);
		\draw [style=graph line] (433) to (414);
		\draw [style=graph line] (435) to (414);
		\draw [style=graph line] (409) to (420);
		\draw [style=graph line] (420) to (435);
		\draw [style=graph line] (434) to (419);
		\draw [style=graph line] (434) to (418);
		\draw [style=graph line] (424) to (410);
		\draw [style=graph line] (438) to (413);
		\draw [style=arrow, in=180, out=0] (439.center) to (440.center);
		\draw [style=graph line] (422) to (417);
		\draw [style=graph line] (437) to (415);
		\draw [style=graph line] (422) to (437);
	\end{pgfonlayer}
\end{tikzpicture}

%% file: sim_table.tex
\begin{table*}[!h]
\centering
\fontsize{6}{6}\selectfont
\begin{tabular*}{\textwidth}{@{\extracolsep{\fill}}l l c c c c c c c}
\toprule
\makecell[l]{QEC Code \\ \& State} & Method & \makecell{CNOT \\ Count} & \makecell{Flag \\ Count} & \makecell{Qubit Reuse \\ Opt.\@ Target} & \makecell{Sim.\@ \\ Qubits} & Depth & LER & AR \\
\midrule
\multirow{4}{*}{\makecell[l]{$\code{7, 1, 3}$ Steane $\ket{\overline{0}}$}} & \multirow{2}{*}{Flag at Origin} & \multirow{2}{*}{15} & \multirow{2}{*}{3} & sim.\@ qubits & 8 & 15 & \multirow{2}{*}{\textbf{\boldmath $[2.1, \,\, 2.2]\! \times\! 10^{-5}$}} & \multirow{2}{*}{$[0.9870, \,\, 0.9871]$} \\
 & & & & depth & 10 & 8 & & \\
\cmidrule{2-9}
 & \multirow{2}{*}{SpiderCSS (Earliest Start)} & \multirow{2}{*}{15} & \multirow{2}{*}{3} & sim.\@ qubits & 8 & \textbf{12} & \multirow{2}{*}{$[2.3, \,\, 2.5]\! \times\! 10^{-5}$} & \multirow{2}{*}{$[0.9870, \,\, 0.9871]$} \\
 & & & & depth & 10 & \textbf{5} & & \\
\midrule
\multirow{4}{*}{\makecell[l]{$\code{9, 1, 3}$ rot.\@ surface $\ket{\overline{0}}$}} & \multirow{2}{*}{Flag at Origin} & \multirow{2}{*}{26} & \multirow{2}{*}{9} & sim.\@ qubits & 11 & 20 & \multirow{2}{*}{$[17.7, \,\, 19.5]\! \times\! 10^{-6}$} & \multirow{2}{*}{$[0.9622, \,\, 0.9623]$} \\
 & & & & depth & 18 & 6 & & \\
\cmidrule{2-9}
 & \multirow{2}{*}{SpiderCSS (Crit.\@ Path)} & \multirow{2}{*}{\textbf{12}} & \multirow{2}{*}{\textbf{2}} & sim.\@ qubits & \textbf{9} & \textbf{5} & \multirow{2}{*}{\textbf{\boldmath $[9.3, \,\, 10.5]\! \times\! 10^{-6}$}} & \multirow{2}{*}{\textbf{\boldmath $[0.9913, \,\, 0.9913]$}} \\
 & & & & depth & \textbf{11} & \textbf{4} & & \\
\midrule
\multirow{4}{*}{\makecell[l]{$\code{17, 1, 5}$ colour code $\ket{\overline{0}}$}} & \multirow{2}{*}{Flag at Origin} & \multirow{2}{*}{74} & \multirow{2}{*}{21} & sim.\@ qubits & 22 & 47 & \multirow{2}{*}{$[4.9, \,\, 8.2]\! \times\! 10^{-7}$} & \multirow{2}{*}{$[0.9009, \,\, 0.9010]$} \\
 & & & & depth & 38 & 17 & & \\
\cmidrule{2-9}
 & \multirow{2}{*}{SpiderCSS (Seq.\@ Dist.)} & \multirow{2}{*}{74} & \multirow{2}{*}{21} & sim.\@ qubits & 22 & 47 & \multirow{2}{*}{\textbf{\boldmath $[4.7, \,\, 8.0]\! \times\! 10^{-7}$}} & \multirow{2}{*}{\textbf{\boldmath $[0.9107, \,\, 0.9108]$}} \\
 & & & & depth & \textbf{36} & \textbf{13} & & \\
\midrule
\multirow{4}{*}{\makecell[l]{$\code{25, 1, 5}$ rot.\@ surface $\ket{\overline{0}}$}} & \multirow{2}{*}{Flag at Origin} & \multirow{2}{*}{92} & \multirow{2}{*}{28} & sim.\@ qubits & 33 & 64 & \multirow{2}{*}{$[5.0, \,\, 8.5]\! \times\! 10^{-7}$} & \multirow{2}{*}{$[0.8735, \,\, 0.8736]$} \\
 & & & & depth & 53 & 20 & & \\
\cmidrule{2-9}
 & \multirow{2}{*}{SpiderCSS (Crit.\@ Path)} & \multirow{2}{*}{\textbf{78}} & \multirow{2}{*}{\textbf{21}} & sim.\@ qubits & \textbf{29} & \textbf{38} & \multirow{2}{*}{\textbf{\boldmath $[3.3, \,\, 6.1]\! \times\! 10^{-7}$}} & \multirow{2}{*}{\textbf{\boldmath $[0.9101, \,\, 0.9102]$}} \\
 & & & & depth & \textbf{46} & \textbf{8} & & \\
\midrule
\multirow{4}{*}{\makecell[l]{$\code{49, 1, 5}$ triorthogonal $\ket{\overline{+}}$}} & \multirow{2}{*}{Flag at Origin} & \multirow{2}{*}{361} & \multirow{2}{*}{105} & sim.\@ qubits & 89 & \textbf{195} & \multirow{2}{*}{$[10.0, \,\, 10.5]\! \times\! 10^{-5}$} & \multirow{2}{*}{$[0.4968, \,\, 0.4970]$} \\
 & & & & depth & 154 & 56 & & \\
\cmidrule{2-9}
 & \multirow{2}{*}{SpiderCSS (Earliest Start)} & \multirow{2}{*}{361} & \multirow{2}{*}{105} & sim.\@ qubits & \textbf{80} & 232 & \multirow{2}{*}{\textbf{\boldmath $[7.8, \,\, 8.3]\! \times\! 10^{-5}$}} & \multirow{2}{*}{\textbf{\boldmath $[0.5281, \,\, 0.5283]$}} \\
 & & & & depth & \textbf{128} & \textbf{29} & & \\
\midrule
\multirow{4}{*}{\makecell[l]{$\code{20, 2, 6}$ self-dual $\ket{\overline{00}}$}} & \multirow{2}{*}{Flag at Origin} & \multirow{2}{*}{145} & \multirow{2}{*}{47} & sim.\@ qubits & 40 & 101 & \multirow{2}{*}{$[3.8, \,\, 11.3]\! \times\! 10^{-8}$} & \multirow{2}{*}{$[0.7895, \,\, 0.7896]$} \\
 & & & & depth & 67 & 51 & & \\
\cmidrule{2-9}
 & \multirow{2}{*}{SpiderCSS (Seq.\@ Dist.)} & \multirow{2}{*}{145} & \multirow{2}{*}{47} & sim.\@ qubits & \textbf{39} & \textbf{91} & \multirow{2}{*}{\textbf{\boldmath $[0.5, \,\, 4.4]\! \times\! 10^{-8}$}} & \multirow{2}{*}{\textbf{\boldmath $[0.8089, \,\, 0.8089]$}} \\
 & & & & depth & \textbf{60} & \textbf{16} & & \\
\midrule
\multirow{4}{*}{\makecell[l]{$\code{23, 1, 7}$ Golay $\ket{\overline{0}}$}} & \multirow{2}{*}{Flag at Origin} & \multirow{2}{*}{237} & \multirow{2}{*}{80} & sim.\@ qubits & \textbf{53} & \textbf{142} & \multirow{2}{*}{$[3.5, \,\, 5.6]\! \times\! 10^{-7}$} & \multirow{2}{*}{$[0.6441, \,\, 0.6442]$} \\
 & & & & depth & 103 & 29 & & \\
\cmidrule{2-9}
 & \multirow{2}{*}{SpiderCSS (Active Spider)} & \multirow{2}{*}{237} & \multirow{2}{*}{80} & sim.\@ qubits & 54 & 143 & \multirow{2}{*}{\textbf{\boldmath $[1.9, \,\, 3.4]\! \times\! 10^{-7}$}} & \multirow{2}{*}{\textbf{\boldmath $[0.6771, \,\, 0.6772]$}} \\
 & & & & depth & \textbf{83} & \textbf{22} & & \\
\midrule
\multirow{4}{*}{\makecell[l]{$\code{31, 1, 7}$ colour code $\ket{\overline{0}}$}} & \multirow{2}{*}{Flag at Origin} & \multirow{2}{*}{211} & \multirow{2}{*}{69} & sim.\@ qubits & 53 & \textbf{129} & \multirow{2}{*}{$[14.0, \,\, 27.3]\! \times\! 10^{-8}$} & \multirow{2}{*}{$[0.6954, \,\, 0.6956]$} \\
 & & & & depth & 100 & 55 & & \\
\cmidrule{2-9}
 & \multirow{2}{*}{SpiderCSS (Earliest Start)} & \multirow{2}{*}{211} & \multirow{2}{*}{69} & sim.\@ qubits & \textbf{45} & 136 & \multirow{2}{*}{\textbf{\boldmath $[5.5, \,\, 14.5]\! \times\! 10^{-8}$}} & \multirow{2}{*}{\textbf{\boldmath $[0.7189, \,\, 0.7190]$}} \\
 & & & & depth & \textbf{84} & \textbf{19} & & \\
\midrule
\multirow{4}{*}{\makecell[l]{$\code{49, 1, 7}$ rot.\@ surface $\ket{\overline{0}}$}} & \multirow{2}{*}{Flag at Origin} & \multirow{2}{*}{262} & \multirow{2}{*}{85} & sim.\@ qubits & 79 & 174 & \multirow{2}{*}{$[11.4, \,\, 24.5]\! \times\! 10^{-8}$} & \multirow{2}{*}{$[0.6215, \,\, 0.6216]$} \\
 & & & & depth & 134 & 43 & & \\
\cmidrule{2-9}
 & \multirow{2}{*}{SpiderCSS (Earliest Start)} & \multirow{2}{*}{\textbf{226}} & \multirow{2}{*}{\textbf{69}} & sim.\@ qubits & \textbf{65} & \textbf{129} & \multirow{2}{*}{\textbf{\boldmath $[4.7, \,\, 13.2]\! \times\! 10^{-8}$}} & \multirow{2}{*}{\textbf{\boldmath $[0.7121, \,\, 0.7122]$}} \\
 & & & & depth & \textbf{101} & \textbf{18} & & \\
\midrule
\multirow{4}{*}{\makecell[l]{$\code{95, 1, 7}$ triorthogonal $\ket{\overline{+}}$}} & \multirow{2}{*}{Flag at Origin} & \multirow{2}{*}{1175} & \multirow{2}{*}{380} & sim.\@ qubits & 243 & 804 & \multirow{2}{*}{\textbf{\boldmath $[1.6, \,\, 2.0]\! \times\! 10^{-4}$}} & \multirow{2}{*}{$[0.0259, \,\, 0.0260]$} \\
 & & & & depth & 475 & 386 & & \\
\cmidrule{2-9}
 & \multirow{2}{*}{SpiderCSS (Seq.\@ Dist.)} & \multirow{2}{*}{\textbf{1173}} & \multirow{2}{*}{\textbf{379}} & sim.\@ qubits & \textbf{184} & \textbf{752} & \multirow{2}{*}{$[3.0, \,\, 3.4]\! \times\! 10^{-4}$} & \multirow{2}{*}{\textbf{\boldmath $[0.0372, \,\, 0.0373]$}} \\
 & & & & depth & \textbf{332} & \textbf{44} & & \\
\midrule
\multirow{4}{*}{\makecell[l]{$\code{49, 1, 9}$ colour code $\ket{\overline{0}}$}} & \multirow{2}{*}{Flag at Origin} & \multirow{2}{*}{408} & \multirow{2}{*}{136} & sim.\@ qubits & 89 & \textbf{256} & \multirow{2}{*}{$[3.1, \,\, 13.1]\! \times\! 10^{-8}$} & \multirow{2}{*}{$[0.4413, \,\, 0.4414]$} \\
 & & & & depth & 185 & 120 & & \\
\cmidrule{2-9}
 & \multirow{2}{*}{SpiderCSS (Crit.\@ Path)} & \multirow{2}{*}{\textbf{402}} & \multirow{2}{*}{\textbf{133}} & sim.\@ qubits & \textbf{81} & 265 & \multirow{2}{*}{\textbf{\boldmath $[0.5, \,\, 6.3]\! \times\! 10^{-8}$}} & \multirow{2}{*}{\textbf{\boldmath $[0.4633, \,\, 0.4634]$}} \\
 & & & & depth & \textbf{144} & \textbf{24} & & \\
\midrule
\multirow{4}{*}{\makecell[l]{$\code{81, 1, 9}$ rot.\@ surface $\ket{\overline{0}}$}} & \multirow{2}{*}{Flag at Origin} & \multirow{2}{*}{614} & \multirow{2}{*}{206} & sim.\@ qubits & 135 & 375 & \multirow{2}{*}{$-$} & \multirow{2}{*}{$[0.2454, \,\, 0.2459]$} \\
 & & & & depth & 287 & 126 & & \\
\cmidrule{2-9}
 & \multirow{2}{*}{SpiderCSS (Active Spider)} & \multirow{2}{*}{\textbf{486}} & \multirow{2}{*}{\textbf{155}} & sim.\@ qubits & \textbf{116} & \textbf{308} & \multirow{2}{*}{$-$} & \multirow{2}{*}{\textbf{\boldmath $[0.3816, \,\, 0.3822]$}} \\
 & & & & depth & \textbf{198} & \textbf{25} & & \\
\midrule
\multirow{4}{*}{\makecell[l]{$\code{47, 1, 11}$ self-dual $\ket{\overline{0}}$}} & \multirow{2}{*}{Flag at Origin} & \multirow{2}{*}{1033} & \multirow{2}{*}{388} & sim.\@ qubits & \textbf{163} & 659 & \multirow{2}{*}{$[27.7, \,\, 46.0]\! \times\! 10^{-7}$} & \multirow{2}{*}{$[0.0672, \,\, 0.0672]$} \\
 & & & & depth & 435 & 289 & & \\
\cmidrule{2-9}
 & \multirow{2}{*}{SpiderCSS (Crit.\@ Path)} & \multirow{2}{*}{\textbf{971}} & \multirow{2}{*}{\textbf{355}} & sim.\@ qubits & 185 & \textbf{635} & \multirow{2}{*}{\textbf{\boldmath $[1.6, \,\, 6.9]\! \times\! 10^{-7}$}} & \multirow{2}{*}{\textbf{\boldmath $[0.0838, \,\, 0.0839]$}} \\
 & & & & depth & \textbf{303} & \textbf{38} & & \\
\midrule
\multirow{4}{*}{\makecell[l]{$\code{71, 1, 11}$ colour code $\ket{\overline{0}}$}} & \multirow{2}{*}{Flag at Origin} & \multirow{2}{*}{829} & \multirow{2}{*}{268} & sim.\@ qubits & 163 & 546 & \multirow{2}{*}{$-$} & \multirow{2}{*}{$[0.1344, \,\, 0.1348]$} \\
 & & & & depth & 339 & 279 & & \\
\cmidrule{2-9}
 & \multirow{2}{*}{SpiderCSS (Active Spider)} & \multirow{2}{*}{\textbf{727}} & \multirow{2}{*}{\textbf{251}} & sim.\@ qubits & \textbf{135} & \textbf{475} & \multirow{2}{*}{$-$} & \multirow{2}{*}{\textbf{\boldmath $[0.2060, \,\, 0.2065]$}} \\
 & & & & depth & \textbf{258} & \textbf{28} & & \\
\midrule
\multirow{2}{*}{\makecell[l]{$\code{144, 12, 12}$ gross $\ket{\overline0}^{\otimes 12}$}} & \multirow{2}{*}{SpiderCSS (Active Spider)} & \multirow{2}{*}{5018} & \multirow{2}{*}{1824} & sim.\@ qubits & 800 & 2993 & \multirow{2}{*}{$-$} & \multirow{2}{*}{$[0.0504, \,\, 0.0512]^{*}$} \\
 & & & & depth & 1191 & 97 & & \\
\midrule
\multirow{2}{*}{\makecell[l]{$\code{82, 2, 13}$ 2bga $\ket{\overline{00}}$}} & \multirow{2}{*}{SpiderCSS (Active Spider)} & \multirow{2}{*}{2586} & \multirow{2}{*}{970} & sim.\@ qubits & 488 & 1587 & \multirow{2}{*}{$-$} & \multirow{2}{*}{$[0.3453, \,\, 0.3472]^{*}$} \\
 & & & & depth & 640 & 76 & & \\
\midrule
\multirow{2}{*}{\makecell[l]{$\code{92, 2, 14}$ 2bga $\ket{\overline{00}}$}} & \multirow{2}{*}{SpiderCSS (Active Spider)} & \multirow{2}{*}{3647} & \multirow{2}{*}{1409} & sim.\@ qubits & 647 & 2164 & \multirow{2}{*}{$-$} & \multirow{2}{*}{$[0.1346, \,\, 0.1359]^{*}$} \\
 & & & & depth & 857 & 96 & & \\
\midrule
\multirow{2}{*}{\makecell[l]{$\code{79, 1, 15}$ self-dual $\ket{\overline{0}}$}} & \multirow{2}{*}{SpiderCSS (Active Spider)} & \multirow{2}{*}{3257} & \multirow{2}{*}{1270} & sim.\@ qubits & 576 & 2032 & \multirow{2}{*}{$-$} & \multirow{2}{*}{$[0.2005, \,\, 0.2020]^{*}$} \\
 & & & & depth & 709 & 98 & & \\
\midrule
\bottomrule
\end{tabular*}
\caption{
	Resource overhead, logical error rate, and acceptance rate for different CSS state preparations. Results in bold denote the best result for that particular state.
	Columns from left to right: QEC code and state, Method (SpiderCSS or FaO, i.e.\@ Flag at Origin~\cite{forlivesi2025flag}), number of CNOT gates in the circuit, number of flag measurements, optimization target of qubit reuse strategy (simultaneous qubits, or circuit depth), maximum simultaneous number of qubits necessary, circuit depth, and finally logical error rate and acceptance rates using Wilson confidence intervals of 95\%.
	The logical error rates of some codes were not estimated (marked $-$) as the lookup table was too large to store in memory.
	For the 4 largest codes, values marked with $^*$ indicate simulations performed with a physical error rate of $p=0.0001$ instead of the usual $p=0.001$.
}
\label{tab:sim_results}
\end{table*}